\documentclass[a4paper,12pt]{article}

\usepackage[utf8]{inputenc}
\usepackage{amsmath}
\usepackage{amsfonts}
\usepackage{amsthm}
\usepackage{amssymb}
\usepackage{lmodern}
\usepackage{enumitem}
\usepackage{tikz}
\usepackage{url}
\usepackage{stmaryrd}
\usepackage{xparse}
\usepackage{bbm}
\usepackage{mathtools}
\mathtoolsset{centercolon}
\usepackage{chngcntr}
\usepackage{subfig}
\usepackage[font=small]{caption}

\usepackage{fullpage}
\usepackage{xcolor}
\usepackage[backgroundcolor=white, textwidth=2cm]{todonotes}
\usepackage{url}

\usetikzlibrary{decorations.pathmorphing}
\usetikzlibrary{decorations.pathreplacing}
\usetikzlibrary{calc}
\usetikzlibrary{arrows.meta}
\tikzstyle{vertex} = [circle, fill=black, inner sep=0.6mm,  minimum size = 1mm]

\usepackage[colorlinks]{hyperref}
\definecolor{lightblue}{rgb}{0.5,0.5,1.0}
\definecolor{darkred}{rgb}{0.5,0,0}
\definecolor{darkgreen}{rgb}{0,0.5,0}
\definecolor{darkblue}{rgb}{0,0,0.5}
\definecolor{darkyellow}{rgb}{0.5,0.5,0}
\definecolor{purple}{rgb}{0.5,0,0.5}

\hypersetup{colorlinks,linkcolor=darkred,filecolor=darkgreen,urlcolor=darkred,citecolor=darkblue}

\renewcommand{\phi}{\varphi}
\newcommand{\nat}{\mathbb{N}}
\newcommand{\ZZ}{\mathbb{Z}}
\newcommand{\FF}{\mathbb{F}}

\newcommand{\defining}[1]{\emph{\textbf{#1}}}

\newcommand{\iso}{\cong}

\newcommand{\disunion}{\mathbin{\uplus}}

\DeclarePairedDelimiter\set{\lbrace}{\rbrace}
\DeclarePairedDelimiterX\setcond[2]{\{}{\}}{\mathchoice{\,}{}{}{}#1 \;\delimsize\vert\; #2\mathchoice{\,}{}{}{}}

\newcommand{\qspace}{\mathchoice{\;}{\,}{}{}}

\newcommand{\msetopen}[0]{\{\hspace{-3pt}\{}
\newcommand{\msetclose}[0]{\}\hspace{-3pt}\}}
\newcommand{\bigmsetopen}[0]{\big\{\hspace{-4pt}\big\{}
\newcommand{\bigmsetclose}[0]{\big\}\hspace{-4pt}\big\}}
\newcommand{\mset}[1]{\msetopen #1 \msetclose}

\newcommand{\bigmsetcondition}[2]{\bigmsetopen\mathchoice{\,}{}{}{} #1 \;\big\vert\; #2 \mathchoice{\,}{}{}{}\bigmsetclose}

\newcommand{\termA}{s}
\newcommand{\termB}{t}
\newcommand{\termC}{r}

\newcommand{\formA}{\Phi}
\newcommand{\formB}{\Psi}

\newcommand{\interpret}{\Theta}

\newcommand{\rel}{R}

\newcommand{\auto}{\phi}
\newcommand{\autoA}{\phi}
\newcommand{\autoB}{\psi}
\newcommand{\autGroup}[1]{\mathsf{Aut}(#1)}

\newcommand{\orbs}[2]{\mathsf{orbs}_{#1}(#2)}

\newcommand{\tup}[1]{\bar{#1}}

\newcommand{\vertA}{u}
\newcommand{\vertB}{v}
\newcommand{\vertC}{w}

\newcommand{\PTime}{\textsc{Ptime}}
\newcommand{\NPTime}{\textsc{NPtime}}

\newcommand{\CPT}{\ensuremath{\text{CPT}}}
\newcommand{\CPTWSC}{\ensuremath{\text{CPT+WSC}}}

\newcommand{\restrictVect}[2]{#1|_{#2}}

\DeclareMathSymbol{\shortminus}{\mathbin}{AMSa}{"39}
\newcommand{\inv}[1]{#1^{\shortminus 1}}

\newcommand{\hiddenEq}{\phantom{{}={}}}

\theoremstyle{plain}
\newtheorem{lemma}{Lemma}
\newtheorem{theorem}[lemma]{Theorem}
\newtheorem{corollary}[lemma]{Corollary}
\newtheorem{definition}[lemma]{Definition}
\newtheorem{claim}{Claim}
\counterwithin*{claim}{lemma} 

\newenvironment{claimproof}[1][\proofname]{\begin{proof}[#1]}{\end{proof}}

\newcommand\blfootnote[1]{%
	\begingroup
	\renewcommand\thefootnote{}\footnote{#1}%
	\addtocounter{footnote}{-1}%
	\endgroup
}

\newcommand{\spleq}{\preceq}

\newcommand{\StructA}{\mathfrak{A}}
\newcommand{\StructB}{\mathfrak{B}}
\newcommand{\StructC}{\mathfrak{H}}
\newcommand{\StructVA}{A}
\newcommand{\StructVB}{B}
\newcommand{\StructVC}{H}

\newcommand{\StructHFVA}{A_{HF}}

\newcommand{\Struct}{\StructA}
\newcommand{\StructV}{\StructVA}
\newcommand{\StructHFV}{\StructHFVA}

\newcommand{\reduct}[2]{#1 \upharpoonright #2}

\newcommand{\GraphClass}{\mathcal{K}}

\newcommand{\Atoms}{\textsf{Atoms}}
\newcommand{\Pair}{\textsf{Pair}}
\newcommand{\Union}{\textsf{Union}}
\newcommand{\Unique}{\textsf{Unique}}
\newcommand{\Card}{\textsf{Card}}
\newcommand{\HFsym}{\mathsf{HF}}
\newcommand{\HF}[1]{\HFsym(#1)}

\newcommand{\TC}[1]{\textsf{TC}(#1)}

\newcommand{\denotation}[1]{\llbracket #1 \rrbracket}

\newcommand{\choiceError}{\dag}

\newcommand{\sig}{\tau}
\newcommand{\sigA}{\tau}
\newcommand{\sigB}{\sigma}

\newcommand{\neighbors}[2]{N_{#1}(#2)}

\newcommand{\wscForm}[6]{\mathsf{WSC}^*#1#2.\qspace(#3,#4,#5,#6)}
\newcommand{\wscFormbig}[6]{\mathsf{WSC}^*#1#2.\qspace\big(#3,#4,#5,#6\big)}
\newcommand{\termStep}{\termA_{\text{step}}}
\newcommand{\termChoice}{\termA_{\text{choice}}}
\newcommand{\termWit}{\termA_{\text{wit}}}
\newcommand{\formOut}{\formA_{\text{out}}}

\newcommand{\stepF}[0]{f_{\text{step}}}
\newcommand{\choiceF}[0]{f_{\text{choice}}}
\newcommand{\autF}[0]{f_{\text{wit}}}
\newcommand{\tree}[0]{\mathcal{T}(\stepF^{\Struct,\tup{a}},\choiceF^{\Struct,\tup{a}})}
\newcommand{\paths}[0]{\mathcal{P}(\stepF^{\Struct,\tup{a}},\choiceF^{\Struct,\tup{a}})}
\newcommand{\wscStarSym}{\mathsf{WSC}^*}
\newcommand{\wsc}[3]{\wscStarSym(#1,#2,#3)}

\newcommand{\termOut}{\termA_{\text{out}}}

\newcommand{\indVar}{\iota}
\newcommand{\concat}[2]{#1 #2}

\newcommand{\canlabels}[2]{\mathsf{labels}(#1,#2)}
\newcommand{\canon}[2]{\mathsf{canon}(#1,#2)}

\newcommand{\ordVar}{o}

\newcommand{\dwlm}{M}

\newcommand{\dwla}{\mathcal{M}}

\newcommand{\dwlmwit}{M^{\text{wit}}}
\newcommand{\dwlmout}{M^{\text{out}}}
\newcommand{\hatdwlmwit}{\hat{M}^{\text{wit}}}
\newcommand{\hatdwlmout}{\hat{M}^{\text{out}}}

\newcommand{\sketch}[1]{D(#1)}
\newcommand{\symSubset}[2]{{\subseteq_{#1,#2}}}
\newcommand{\coConf}[1]{C(#1)}

\newcommand{\addPairSym}{\ensuremath{\mathtt{addPair}}}
\newcommand{\createSym}{\ensuremath{\mathtt{create}}}
\newcommand{\forgetSym}{\ensuremath{\mathtt{forget}}}
\newcommand{\contractSym}{\ensuremath{\mathtt{contract}}}

\newcommand{\addUPairSym}{\ensuremath{\mathtt{addUPair}}}
\newcommand{\sccSym}{\ensuremath{\mathtt{scc}}}
\newcommand{\renameSym}{\ensuremath{\mathtt{rename}}}

\newcommand{\addPair}[1]{\ensuremath{\addPairSym(#1)}}
\newcommand{\create}[1]{\ensuremath{\createSym(#1)}}

\newcommand{\contract}[1]{\ensuremath{\contractSym(#1)}}
\newcommand{\addUPair}[1]{\ensuremath{\addUPairSym(#1)}}
\newcommand{\scc}[1]{\ensuremath{\sccSym(#1)}}
\newcommand{\rename}[1]{\ensuremath{\renameSym(#1)}}

\newcommand{\refineSym}{\ensuremath{\mathtt{refine}}}
\newcommand{\refine}[2]{\ensuremath{\mathtt{refine}(#1, #2)}}
\newcommand{\choiceSym}{\ensuremath{\mathtt{choice}}}
\newcommand{\choice}[1]{\ensuremath{\mathtt{choice}(#1)}}

\newcommand{\run}[2]{\operatorname{run}(#1, #2)}

\newcommand{\relA}{E}
\newcommand{\relB}{F}
\renewcommand{\rel}{\relA}

\newcommand{\colA}{R}
\newcommand{\colB}{S}
\newcommand{\colC}{T}
\newcommand{\col}{\colA}

\newcommand{\relColA}{X}
\newcommand{\relColB}{Y}

\newcommand{\vcA}{C}
\newcommand{\vcB}{D}
\newcommand{\vc}{\vcA}

\newcommand{\ccA}{U}
\newcommand{\ccB}{V}
\newcommand{\cc}{\ccA}

\newcommand{\vertices}[1]{\mathsf{V}(#1)}
\newcommand{\verticesOf}[2]{\mathsf{V}_{#1}(#2)}
\newcommand{\plainEdges}[1]{\mathsf{E}_{\text{plain}}(#1)}
\newcommand{\crossingEdges}[1]{\mathsf{E}_{\text{cross}}(#1)}
\newcommand{\crossingVertices}[1]{\mathsf{V}_{\text{cross}}(#1)}
\newcommand{\diagrel}[1]{\operatorname{diag}(#1)}
\newcommand{\nonHF}[1]{#1^{\text{flat}}}
\newcommand{\leftT}{\text{left}}
\newcommand{\rightT}{\text{right}}
\newcommand{\memT}{\text{in}}

\newcommand{\pairVtx}[2]{\langle #1, #2 \rangle}
\def\labeledUnion{\Cup}

\newcommand{\inrel}{\rel_\in}
\newcommand{\iorel}{\rel_{\mathsf{io}}}
\newcommand{\iorelsub}[1]{\rel_{\mathsf{io},#1}}

\newcommand{\plvA}[1]{#1^{(1)}}
\newcommand{\plvB}[1]{#1^{(2)}}
\newcommand{\plvI}[2]{#1^{(#2)}}
\newcommand{\pRel}{\rel_p}
\newcommand{\pRelInv}{\inv{\rel}_p}

\newcommand{\CFI}[1]{\mathsf{CFI}(#1)}
\newcommand{\orig}[1]{\operatorname{orig}(#1)}
\newcommand{\bVertA}{\mathfrak{u}}
\newcommand{\bVertB}{\mathfrak{v}}

\newcommand{\bEdge}{\mathfrak{e}}

\title{Choiceless Polynomial Time with Witnessed Symmetric Choice}

\author{Moritz Lichter\\
	TU Darmstadt\\
	\texttt{lichter@mathematik.tu-darmstadt.de}
	\and
	Pascal Schweitzer\\
	TU Darmstadt\\
	\texttt{schweitzer@mathematik.tu-darmstadt.de}}

\begin{document}
	\maketitle

\begin{abstract}

	We extend Choiceless Polynomial Time (CPT),
	the currently only remaining promising candidate in the quest for a logic capturing \PTime{},
	so that this extended logic has the following property:
	for every class of structures for which isomorphism is definable, the logic automatically captures \PTime{}. 
	
	For the construction of this logic we extend CPT by a witnessed symmetric choice operator. This operator allows for choices from definable orbits.
	But, to ensure polynomial-time evaluation,
	automorphisms have to be provided to certify that the choice set is indeed an orbit.
	
	We argue that, in this logic, definable isomorphism implies definable canonization.
	Thereby, our construction removes the non-trivial step of extending isomorphism definability results to canonization.
	This step was a part of proofs that show that CPT or other logics capture \PTime{} on a particular class of structures. The step typically required substantial extra effort.
	\blfootnote{\noindent The research leading to these results has received funding from the European Research Council (ERC) under the European Union’s Horizon 2020 research and innovation programme (EngageS: grant agreement No.\ 820148).}
\end{abstract}

\section{Introduction}
One of the central open problems in descriptive complexity theory is the quest for a logic capturing~\PTime{}~\cite{Grohe2008}. This long-standing open problem~\cite{ChandraHarel82} asks for a logic in which precisely all polynomial-time decidable properties can be expressed as a sentence and for which all formulas can be evaluated in polynomial time. The alternative would be to prove the nonexistence of such a logic, which would however imply a separation of \PTime{} and \NPTime{}~\cite{MR0371622}.

While the general quest for a \PTime{}-logic remains wide open, progress generally comes in one of two flavors: research results
either show that some logic captures \PTime{} for an ever more extensive class of structures,
or a logic is separated from \PTime{}, ruling it out as a candidate for a \PTime{} logic.
In this article we are concerned with a third flavor, namely with reducing the question to a presumably simpler one. 

Historically, almost all results showing that \PTime{} is captured by some logic for some graph class
exploit the Immerman-Vardi Theorem~\cite{Immerman87}
which states that inflationary fixed-point logic IFP captures \PTime{} on ordered graphs.
To apply this theorem to a class of unordered graphs,
one defines canonization of that class inside a logic (which at least captures IFP).
Defining canonization is the task of defining an isomorphic, ordered copy of the input graph.
In algorithmic contexts, canonization is closely related to the problem of isomorphism testing.
While polynomial-time canonization provides polynomial-time isomorphism testing, a reduction
the other way is unknown. 
Granted, for most graph classes
for which a polynomial-time isomorphism algorithm is known, a polynomial-time canonization algorithm is known as well (see~\cite{DBLP:conf/stoc/SchweitzerW19}). However, we know of no reduction that is universally applicable.
Studying the analogous relationship for logics, we are interested in the question of whether definability of the isomorphism problem within a \PTime{} logic for some graph class provides us with a logic for all of \PTime{} on that graph class.

Since we do not know that graph isomorphism is polynomial-time solvable in general, it is of course not clear that in a logic for \PTime{} the isomorphism problem needs to be definable.  However, for all graphs 
on which a logic has been shown to capture \PTime{},
a polynomial-time isomorphism testing algorithm is known~\cite{Grohe2017,GroheN19,LichterS21,AbuZaidGraedelGrohePakusa2014}.
Crucially, in every \PTime{}-capturing logic for such graph classes, in particular isomorphism testing has to be definable.
With a reduction of canonization to isomorphism-testing,
we then obtain a necessary and sufficient condition to capture \PTime{}, namely the definability of the isomorphism problem.

Regarding the issue of isomorphism versus canonization, we should highlight that 
defining canonization often appears to require considerably more effort than defining the isomorphism problem~\cite{Grohe2017, GroheN19}.

In this article\footnote{An extended abstract of this article previously appeared in the Proceedings of the 36th Annual ACM/IEEE Symposium on Logic in Computer Science~\cite{LichterSchweitzer22}.} we present a logic
in which a definable isomorphism test automatically implies a definable canonization.
After rank logic was recently eliminated as a candidate of a logic capturing \PTime{}~\cite{Lichter21},
we consider an extension of the one major remaining candidate, namely Choiceless Polynomial Time (CPT)~\cite{bgs1999}.
The logic CPT operates on hereditarily finite sets formed from the vertices of the input graph. The construction of these sets is isomorphism-invariant, which guarantees that every CPT-term or formula evaluates to an isomorphism-invariant result. This is generally regarded as a requirement for a reasonable logic~\cite{Gurevich1988}. The requirement has an important consequence:
while in algorithms it is common to make choices that are not necessarily isomorphism-invariant (e.g., pick the first neighbor of a vertex within a DFS-transversal and then process it), 
this cannot be done in CPT -- one has to process all possible choices in equal fashion.
For algorithms making choices, we have to prove that they compute the correct (in particular~isomorphism-invariant) result. But in a logic, this property should be built-in.
One possibility to overcome this problem was
studied by Gire and Hoang~\cite{GireHoang98}
as well as by Dawar and Richerby~\cite{DawarRicherby03}.
They extended IFP with a symmetric choice construct,
which allows that during a fixed-point computation in every step one element can be chosen from a set that has been defined.
But, in order to ensure isomorphism-invariance,
these choice sets have to be orbits of the graph.
The output of such a fixed-point computation with choices is not necessarily isomorphism-invariant.
However, at least we are guaranteed that all possible outputs
are related via automorphisms of the graph, independent of the choices that were made.
Crucially, the logics are designed so that fixed-point computations are used only as ``intermediate results''
to define overall a property in the end. Since this property is either true or false,
the output of a formula is isomorphism-invariant after all.

While this approach of introducing symmetric choices yields a reasonable logic,
it is not clear whether its formulas can be evaluated in polynomial time.
Indeed, when a choice is to be made, it has to be verified that the choice set
is actually an orbit
and it is not known that orbits can be computed in polynomial time.
This is resolved in~\cite{GireHoang98}
by handing over the obligation to check that the choice sets are orbits to the formulas themselves.
For this, the formulas,  beside defining the choice set, also have to define automorphisms which can be used to check whether
the choice set is indeed an orbit.
That way, the logic can be evaluated in polynomial time.
We apply a similar approach to CPT.
A fixed-point operator is added, in which in every iteration a choice is made from a choice set. For each choice set, automorphisms certifying that the choice set is indeed an orbit have to be provided by the formula.
We call this witnessed symmetric choice (WSC).

So why should witnessed symmetric choice in CPT
suffice to show that isomorphism testing and canonization are equivalent?
Here we build on two existing results.
The first one~\cite{GroheSchweitzerWiebking2021}
shows that, in CPT, a definable isomorphism test
implies a definable complete invariant,
that is, an ordered object can be defined
which is equal for two input structures if and only if they are isomorphic.
The second, more classical result is due to Gurevich~\cite{Gurevich97}.
It shows how an algorithm computing complete invariants can be turned into an algorithm computing a canonization.
This algorithm requires that the class of graphs
is closed under individualization
(that is, under coloring individual vertices).
While being closed under individualization is a restriction in some contexts~\cite[Theorem 33]{DBLP:conf/mfcs/KieferSS15},
this is usually not the case~\cite{DBLP:conf/mfcs/KieferSS15, DBLP:journals/ipl/Mathon79}.
The canonization algorithm repeatedly uses the complete invariant to compute a canonical orbit,
chooses and individualizes one vertex in that orbit,
and proceeds until all vertices are individualized. Thereby, a total order on the vertices is defined.
And indeed, this algorithm can be expressed in CPT extended by witnessed symmetric choice
and a definable complete invariant can be turned into a definable canonization.

\paragraph{Results.}
We extend CPT with a fixed-point operator with witnessed symmetric choice 
and obtain the logic CPT+WSC.
Here some small, but important formal changes to~\cite{GireHoang98, DawarRicherby03} have to be made
so that we can successfully implement a variant of Gurevich's canonization algorithm in CPT+WSC. 
Using Gurevich's canonization algorithm,
we show that a CPT-definable complete invariant (and thus a CPT-definable isomorphism test~\cite{GroheSchweitzerWiebking2021})
implies a CPT+WSC-definable canonization and thus that CPT+WSC captures \PTime{}. However, we prefer to have a logic so that definability of the isomorphism problem implies that the same logic defines canonization and captures \PTime{} rather than just an extension.
To show precisely this property for CPT+WSC turns out to be rather difficult and formally intricate in several aspects.
Indeed, we lift a result of~\cite{GroheSchweitzerWiebking2021}
from CPT to CPT+WSC thereby showing the following:
if CPT+WSC defines isomorphism of a class of structures closed under individualization,
then it defines a complete invariant
and using the mentioned canonization algorithm
CPT+WSC defines a canonization, too. Overall, we obtain the following.

\begin{theorem}
	\label{thm:iso-implies-canon-and-ptime}
	If CPT+WSC defines isomorphism of a class of $\sig$-structures $\GraphClass$
	(closed under individualization),
	then CPT+WSC defines a canonization of $\GraphClass$-structures
	and captures \PTime{} on $\GraphClass$-structures.
\end{theorem}

Finally, we apply these results to the Cai-Fürer-Immerman query~\cite{CaiFI1992}
and construct a class of base graphs,
for which the CFI-query was not known to be definable in CPT.

\paragraph{Our Technique.}
To ensure that fixed-point operators with witnessed symmetric choice always yield isomorphism-invariant results, we follow the approach of~\cite{GireHoang98, DawarRicherby03}.
Every such fixed-point operator comes with a formula, called the output formula,
which is evaluated on the (not necessarily isomorphism-invariantly) defined fixed-point.
We use the formal techniques of~\cite{DawarRicherby03} to define the semantics of this fixed-point operator.
The most important difference is that the term producing the witnessing automorphisms
also has access to the defined fixed-point.
This turned out to be crucial to implement Gurevich's canonization algorithm
and in turn necessitates other minor formal differences to~\cite{DawarRicherby03} and~\cite{GireHoang98}.
Equipped with these changes,
extending Gurevich's canonization algorithm to also provide witnessing automorphism becomes rather straightforward.

This shows that a definable isomorphism problem in CPT implies that CPT+WSC is a logic capturing \PTime{}. For our theorem however, we require that the statement is true whenever isomorphism is definable in CPT+WSC rather than in CPT.
Therefore, we require that a CPT+WSC-definable isomorphism test implies the existence of a CPT+WSC-definable complete invariant.
We therefore consider the DeepWL computation model,
which is used in~\cite{GroheSchweitzerWiebking2021}
to show that a CPT-definable isomorphism tests implies a CPT-definable complete invariant,
and extend DeepWL with witnessed symmetric choice.
The proof of~\cite{GroheSchweitzerWiebking2021} is based on a translation
of CPT to DeepWL, a normalization procedure in DeepWL yielding the complete invariant, and a translation back into CPT.
Unfortunately, it turned out that this normalization procedure cannot be easily adapted to DeepWL with witnessed symmetric choice.
At multiple points we have to change small but essential parts of definitions
and so cannot reuse as many results of~\cite{GroheSchweitzerWiebking2021}
as one would have liked.
The philosophical reason for this is that DeepWL is based on constructing everything in parallel (for all possible inputs),
which is incompatible with choices.
We cannot compute with different possible choices at the same time in the same graph, as these choices influence each other.
This forced us to nest DeepWL-algorithms to, in some way,
resemble nested fixed-point operators with witnessed symmetric choice.

\paragraph{Related Work.}
In the quest for a \PTime{} logic,
IFP extended by counting quantifiers was shown not to capture \PTime{}~\cite{CaiFI1992},
but used to capture \PTime{} on various graph classes.
These include graphs with excluded minors~\cite{Grohe2017}
and graphs with bounded rank width~\cite{GroheN19}.
Both results take the route via canonization.
On the negative side,
rank logic~\cite{Lichter21} and the more general linear algebraic logic~\cite{DawarGraedelLichter22} were separated from~\PTime{}.

CPT was shown to capture \PTime{} on various classes of structures,
for example on padded structures~\cite{BlassGS02} 
(i.e.,~on disjoint unions of arbitrary structures and sufficiently large cliques),
structures with bounded color class size whose automorphism groups are abelian~\cite{AbuZaidGraedelGrohePakusa2014},
and on (some) structures with bounded color class size
whose automorphism groups are dihedral groups~\cite{LichterS21}. Also here, each of the results is obtained via canonization.
Philosophically, all these approaches are somewhat orthogonal to witnessed symmetric choice.
They use the fact that some set of objects, for which it is not known whether they form orbits, is small enough to try out all possible choices.
The CFI-query on ordered base graphs was shown to be CPT-definable~\cite{DawarRicherbyRossman2008} using deeply nested sets invariant under all isomorphisms.
These results were generalized to base graphs with logarithmic color class size
and to graphs with linear maximal degree~\cite{PakusaSchalthoeferSelman2016}.
Defining the CFI-query on ordered base graphs in CPT+WSC is comparatively easier
similar to IFP with witnessed symmetric choice in~\cite{GireHoang98}.
While it is in general still open whether CPT captures \PTime{},
there are isomorphism-invariant functions not definable in CPT~\cite{rossman2010}, see also~\cite{Pago21} for more recent work on limits of definability in CPT.

The extension of first order logic with non-witnessed symmetric choice
was studied in~\cite{GireHoang98, DawarRicherby03}.
Whenever the choice set is not an orbit, nothing is chosen.
The more general variant in~\cite{DawarRicherby03}
supports parameters for fixed-point operators with symmetric choice and
allows for nested fixed-point operators.
We followed most of these approaches
and generalized the usage of quantifiers
to output formulas
in order to wrap the calculated fixed-point into an isomorphism-invariant output.
While in the first-order setting the fixed-point operators are limited to only define relations, in the CPT setting we can of course define arbitrary hereditarily finite sets.
For these sets, using output formulas seems more suitable than solely using quantifiers.
Dropping the requirement of choosing from orbits is also studied in~\cite{DawarRicherby03Nondet}, called nondeterministic choice.
When only formulas are considered which always produce a deterministic result,
one captures \PTime{}.
But since it is not decidable whether a formula has this property,
this approach does not yield a logic in which formulas can be evaluated in polynomial time
and thus does not provide a \PTime{} logic.

\paragraph{Structure of the Article.}
After reviewing some preliminaries and in particular CPT
in Section~\ref{sec:prelimiaries},
we extend CPT with a fixed-point operator with witnessed symmetric choice in Section~\ref{sec:extension-wsc}.
In Section~\ref{sec:canon-in-cpt-wsc} we implement the canonization algorithm and show the equivalence between isomorphism and canonization,
apart from the one crucial point
that definable isomorphism implies a definable complete invariant in \CPTWSC{}.
This point is shown in Section~\ref{sec:iso-in-cpt+wsc},
which introduces and extends the DeepWL model.
Finally, we apply these techniques in Section~\ref{sec:cfi}
to the CFI-query.
We end with a discussion and open questions in Section~\ref{sec:conclusion}.

\section{Preliminaries}
\label{sec:prelimiaries}

We denote by $[k]$ the set $\set{1,\dots, k}$.
Let~$M$ and~$N$ be sets.
We denote $M$-indexed tuples with entries in~$N$ by $M^N$
and the $m$-indexed entry of $\tup{t} \in M^N$
by $\tup{t}(m)$ for every $m\in M$.
In the case of $M = [k]$,
the $i$-th entry of a tuple $\tup{t} \in N^k$ is $t_i$.
The concatenation of two tuples $\tup{t}_1 \in N^k$ and $\tup{t}_2 \in N^\ell$
is the $(k+\ell)$-tuple $\tup{t}_1\tup{t}_2 \in N^{k+\ell}$.
The disjoint union of~$N$ and~$M$ is $N \disunion M$.
We denote by $\mset{a_1,\dots, a_k}$
the multiset containing the elements $a_1, \dots, a_k$.

A \defining{(relational) signature}
$\sig = \set{\rel_1, \dots, \rel_\ell}$ consists of a set of relation symbols
with associated \defining{arities} $r_i \in \nat$ for all $i \in [\ell]$.
A \defining{$\sig$-structure}~$\Struct$ is a tuple
${\Struct = (\StructV, \rel_1^\Struct, \dots, \rel_\ell^\Struct)}$
where $\rel_i^\Struct \subseteq\StructV^{r_i}$ for all $i \in [\ell]$.
We always denote the\defining{ universe} of~$\Struct$  by~$\StructVA$
and call its elements \defining{atoms}.
The disjoint union of two structures~$\Struct$ and~$\StructB$ is $\Struct \disunion \StructB$.
The \defining{reduct} of a $\sig$-structure~$\Struct$ to a signature $\sig'\subseteq \sig$
is $\reduct{\Struct}{\sig'}$.
This article only considers finite structures.

The \defining{hereditarily finite sets over~$A$},
denoted by $\HF{A}$,
for some set of atoms~$A$
is the inclusion-wise smallest set such that
$A \subseteq \HF{A}$
and $a \in \HF{A}$ for every finite $a \subseteq \HF{A}$.
A~set $a \in \HF{A}$ is \defining{transitive},
if $ c\in b \in a$ implies $c \in a$ for every $c, b \in \HF{A}$.
The \defining{transitive closure} $\TC{a}$ of~$a$
is the least (with respect to set inclusion)
transitive set~$b$ with $a \subseteq b$.

Let $\Struct$ be a $\sig$-structure
and $\tup{a}$ be a tuple of $\HF{\StructV}$-sets.
We write $\autGroup{\Struct}$ for the \defining{group of automorphisms} of~$\Struct$
and $\autGroup{(\Struct, \tup{a})}$
for the group of automorphisms~$\autoA$ of~$\Struct$
which stabilize each of the $\HF{\StructV}$-sets~$a_i$,
i.e., $\autoA(a_i) = a_i$, for all $i \in [|\tup{a}|]$.
A set $b \in \HF{\StructV}$ is an \defining{orbit of $(\Struct, \tup{a})$}
if $b = \setcond{\auto(c)}{\auto \in \autGroup{(\Struct, \tup{a})}}$
for one (and thus every) $c \in b$.
We also use the orbit notion for other objects, e.g.,~tuples
(which could of course be encoded as hereditarily finite sets):
if~$b$ is a set of $k$-tuples of atoms that forms an orbit,
we call~$b$ a \defining{$k$-orbit}.
We write $\orbs{k}{\Struct}$ for the set of $k$-orbits of~$\Struct$.

\subsection*{Choiceless Polynomial Time}
The logic CPT was introduced by Blass, Gurevich, and Shelah~\cite{bgs1999}
using a pseudocode-like syntax and abstract state machines.
Later there were ``logical'' definitions using iteration terms or fixed points.
To give a concise definition of CPT,
we follow~\cite{GradelGrohe2015} and
use ideas of~\cite{pakusa2015} to enforce polynomial bounds.

Let~$\sig$ be a signature
and extend~$\sig$ by adding set-theoretic function symbols
\[\sig^\HFsym := \sig \disunion \set{\emptyset,\Atoms, \Pair,\Union, \Unique, \Card},\]
where~$\emptyset$ and~$\Atoms$ are constants,~%
$\Union$,~$\Unique$, and~$\Card$ are unary, and~$\Pair$ is binary.
The \defining{hereditarily finite expansion} $\HF{\Struct}$ 
of a $\sig$-structure~$\Struct$
is the $\sig^\HFsym$-structure over the universe $\HF{\StructV}$ defined as follows:
all relations in~$\sig$ are interpreted as they are in~$\Struct$.
The special function symbols have the expected set-theoretic interpretation:
\begin{itemize}
	\item $\emptyset^{\HF{\Struct}} = \emptyset$ and $\Atoms^{\HF{\Struct}} = \StructV$,
	\item $\Pair^{\HF{\Struct}}(a,b) = \set{a,b}$,
	\item $\Union^{\HF{\Struct}}(a) = \setcond{b}{\exists c \in a.\qspace b\in c}$,
	\item $\Unique^{\HF{\Struct}}(a) = \begin{cases}
		b & \text{if } a = \set{b}\\
		\emptyset & \text{otherwise}
	\end{cases}$, and
	\item $\Card^{\HF{\Struct}}(a) = \begin{cases}
		|a| &\text {if } a \notin \StructV\\
		\emptyset &\text{otherwise}
	\end{cases}$,\\
	where the number $|a|$ is encoded as a von Neumann ordinal.
\end{itemize} 
Note that the~$\Unique$ function is invariant under automorphisms 
because it only evaluates non-trivially when applied to singleton sets.

The logic CPT is obtained as the polynomial-time fragment
of the logic BGS (after Blass, Gurevich, and Shelah~\cite{bgs1999}):
A \defining{BGS-term} is composed of variables,
function symbols from~$\sig^{\HFsym}$,
and the two following constructs:
if $\termA(\tup{x},y)$ and $\termB(\tup{x})$ are terms with
a tuple of free variables~$\tup{x}$
(and an additional free variable~$y$ in the case of~$\termA$)
and $\formA(\tup{x}, y)$ is a formula with free variables~$\tup{x}$ and~$y$,
then $\termC(\tup{x}) = \setcond{\termA(\tup{x}, y)}{y \in \termB(\tup{x}), \formA(\tup{x}, y)}$
is a \defining{comprehension term} with free variables~$\tup{x}$.
For a term $\termA(\tup{x},y)$ with free variables~$\tup{x}$ and~$y$,
the \defining{iteration term} $\termA[y]^*(\tup{x})$ has free variables%
\footnote{Here we differ from the definition in~\cite{GradelGrohe2015},
in which~$\termA$ is only allowed to have one free variable $y$.
For CPT, allowing more free variables does not increase expressiveness,
but for our extensions later it is useful to allow additional free variables in an iteration term.}~$\tup{x}$.
\defining{BGS-formulas} are obtained as
$\rel(\termB_1, \dots, \termB_k)$ (for $\rel \in \sig$ of arity $k$ and BGS-terms $\termB_1, \dots, \termB_k$), as $\termB_1 = \termB_2$, and as the usual Boolean connectives.

Let~$\Struct$ be a $\sig$-structure.
BGS-terms and formulas are evaluated over $\HF{\StructV}$
by the denotation $\denotation{t}^\Struct \colon \HF{\StructV}^k \to \HF{\StructV}$
that maps values $\tup{a} = (a_1, \dots ,a_k) \in \HF{\StructV}^k$ for the free variables $\tup{x} = (x_1, \dots, x_k)$
of a term~$\termA$
to the value of~$\termA$ obtained if we interpret~$x_i$ with~$a_i$ (for every $i \in [k]$).
For a formula~$\formA$ with free variables~$\tup{x}$,
the denotation $\denotation{\formA}^\Struct$ is the set of all 
$\tup{a} = (a_1, \dots ,a_k) \in \HF{\StructV}^k$ satisfying~$\formA$.
The denotation of a comprehension term~$\termC$ as above is the following:
\[\denotation{\termC}^\Struct(\tup{a}) = 
\setcond*{\denotation{\termA}^\Struct(\tup{a}b)}{b \in \denotation{\termB}^\Struct(\tup{a}),
	(\tup{a}b) \in \denotation{\formA}^\Struct},\]
where~$\tup{a}b$ denotes the tuple $(a_1, \dots , a_k, b)$.
An iteration term $\termA[y]^*(\tup{x})$ for a tuple $\tup{b} \in \HF{\StructV}^{|\tup{x}|}$
with sets for the free variables
defines a sequence of sets
via
\begin{align*}
	a_0 &:= \emptyset,\\
	a_{i+1} &:= \denotation{\termA}^\Struct(\tup{b}a_i).
\end{align*}
Let $\ell := \ell(\termA[y]^*,\Struct,\tup{b})$
be the least number~$i$ such that $a_{i+1} = a_{i}$.
If such an~$\ell$ exists,
we set $\denotation{\termA[y]^*}^\Struct(\tup{b}) := a_\ell$ and we set
$\denotation{\termA[y]^*}^\Struct(\tup{b}) := \emptyset$ otherwise.

A CPT-term (or formula, respectively) is a tuple $(\termB,p)$ (or $(\formA, p)$, respectively) of a BGS-term (or formula) and a polynomial~$p(n)$.
The semantics of CPT is derived from BGS by replacing~$\termB$ with $(\termB,p)$
everywhere (or~$\formA$ with $(\formA, p)$)
with the following exception for iteration terms:
We define $\denotation{(\termA[y]^*,p)}^\Struct(\tup{b}) := \denotation{\termA[y]^*}^\Struct(\tup{b})$
if  $\ell(\termA[y]^*, \Struct,\tup{b}) \leq p(|\StructV|)$
and $|\TC{a_i}| \leq p(|\StructV|)$ for all~$i$,
where the sets~$a_i$ are defined as above.
Otherwise, we set $\denotation{(\termA[y]^*,p)}^\Struct(\tup{b}) := \emptyset$.
The size of~$a_i$ is measured by $|\TC{a_i}|$
because by transitivity $\TC{a_i}$
contains all sets $b_k \in  \dots \in b_1 \in a_i$
occurring somewhere in the structure of~$a_i$.
It suffices to put polynomial bounds on iteration terms
because all other terms increase the size of the defined sets only polynomially.

\subsection*{Logical Interpretations}
We use an easy notion of a CPT-interpretation in this article.
Let~$\sigA$ and $\sigB= \set{\rel_1, \dots, \rel_k}$ be relational signatures.
We write CPT$[\sigA]$ respectively CPT$[\sigB$]
for CPT-formulas or terms over the signatures~$\sigA$ or~$\sigB$.
A CPT$[\sigA,\sigB]$-interpretation~$\interpret$ with parameters~$\tup{x}$
is a tuple $\interpret(\tup{x}) = (\termA(\tup{x}), \formA_1(\tup{x}\tup{y}_1), \dots, \formA_k(\tup{x}\tup{y}_k))$
of a CPT$[\sigA]$-term~$\termA(\tup{x})$ and CPT$[\sigA]$-formulas $\formA_1(\tup{x}\tup{y}_1), \dots, \formA_k(\tup{x}\tup{y}_k)$
such that $|\tup{y}_i|$ equals the arity of~$\rel_i$ for every $i \in [k]$.

Let~$\StructA$ be a $\sigA$-structure and $\tup{a} \in \StructVA^{|\tup{x}|}$.
Then the $\sigB$-structure $\interpret(\StructA, \tup{a})$ 
has universe 
\begin{align*}
	\StructVB &:= \denotation{\termA}^\Struct(\tup{a})
	\intertext{and for every $i\in[k]$ the relation~$\rel_i$ of arity~$j$ defined by}
	\rel_i^{\interpret(\StructA, \tup{a})} &:= \setcond*{\tup{b} \in \StructVB^k}{\tup{a}\tup{b} \in \denotation{\formA_i}^\Struct}.
\end{align*}
For readers familiar with first-order interpretations,
we remark we neither need the notion of the dimension of a \CPT{}-interpretation
nor do we need a congruence relation in order to consider a quotient of the interpreted structure.
Both, tuples of fixed length and equivalence classes of such tuples,
can be defined directly in CPT using hereditarily finite sets.

\section{CPT with a Symmetric Choice Operator}
\label{sec:extension-wsc}	

We start by extending BGS with a fixed-point operator with witnessed symmetric choice (WSC-fixed-point operator).
The logic BGS+WSC is the extension of BGS logic
by the following operator to construct formulas:
\[\formB(\tup{z}) = \wscFormbig{x}{y}{\termStep(\tup{z}xy)}{\termChoice(\tup{z}x)}{\termWit(\tup{z}xy)}{\formOut(\tup{z}x)}.\]
Here,~$\termStep$,~$\termChoice$, and~$\termWit$ are BGS+WSC-terms
and~$\formOut$ is a BGS+WSC-formula.
The free variables of~$\termStep$ and~$\termWit$ apart from~$x$ and~$y$
and the free variables of~$\termChoice$ and~$\formOut$ apart from~$x$
are free in~$\formB$.
In particular,~$y$ is only bound in~$\termStep$ and~$\termWit$.
We call~$\termStep$ the \defining{step term},~%
$\termChoice$ the \defining{choice term},~%
$\termWit$ the \defining{witnessing term},~%
and~$\formOut$ the \defining{output formula}.
Intuitively,
we want to iterate the step term $\termStep(x,y)$ until we reach a fixed-point
for the set~$x$.
However, we choose before each step an element~$y$ of the choice set 
defined by the choice term $\termChoice(x)$.
Once a fixed-point is reached,
the witnessing term~$\termWit$ must provide automorphisms for every intermediate step
witnessing that we indeed chose from orbits (details later).
Finally,~$\formB$ is satisfied
if the output formula $\formOut(a)$ is satisfied where~$a$ is the set computed through iteration with choice.
Because we always choose from orbits,
$\formOut(a)$ is satisfied  by some fixed-point~$a$
if and only it is satisfied for every possible fixed-point~$a$
(details also later).
In that way, the evaluation of BGS+WSC-terms and formulas is still
deterministic, so does not depend on any choices made in the fixed-point computation.

We remark that here it seems reasonable to allow free variables in iteration terms.
Otherwise, they cannot be used in~$\termStep$,~$\termChoice$,~$\termWit$, and~$\formOut$.
While for CPT or BGS it is clear that nested iteration terms can be eliminated,
this is not clear for BGS+WSC.
	
\paragraph{An example.}
To illustrate the definition we discuss an example. 
A \emph{universal vertex} is a vertex adjacent to every other vertex. A graph~$G=(V,E)$ is a \emph{threshold} graph if we can reduce it to the empty graph by repeatedly removing
a universal or an isolated vertex. Our example describes a CPT+WSC-sentence that defines the class of threshold graphs
(which we only do for illustration as the class of threshold graphs is already
IFP-definable). 
The idea is simple: the set of vertices that are universal or isolated form an orbit (note that a graph on more than one vertex cannot have a universal and an isolated vertex at the same time).
Thus, we use a WSC-fixed-point operator to choose one such vertex, 
remove it, and repeat this, until no vertex can be removed anymore.

We start with the choice term $\termChoice(x)$. For a set~$x\subseteq V$ the following term defines the set of vertices that are universal or isolated in~$G[V\setminus x]$:
\begin{align*}
	\termC &:= \Atoms \setminus x,\\
	\termChoice(x) &:= \setcond*{y}{y \in \termC, \big(\forall z \in  \termC.\qspace y = z \lor E(y,z)\big) \lor \big(\forall z \in  \termC.\qspace y = z \lor \neg E(y,z)\big) }.
	\intertext{Then the step term $\termStep(x,y)$ adds a chosen vertex~$y$ to those already removed:}
	\termStep(x,y) &:=x\cup (\set{y} \cap \Atoms).
\end{align*}
The intersection with~$\Atoms$ is needed
to obtain a fixed-point when~$\termChoice$ defines the empty set
(in which case~$y$ is the empty set, too).
As certification, the term $\termB(z,z')$ defines the transposition of~$z$ and~$z'$ and the witnessing term~$\termWit(x)$ collects all transpositions of pairs of vertices in the choice set:
\begin{align*} 
	\termB(z,z') &:= \setcond[\big]{(x,x)}{x \in \Atoms \setminus \set{z,z'}} \cup \set[\big]{(z,z'),(z',z)}\\
	\termWit(x) &:= \setcond*{\termB(z,z')}{(z,z') \in x^2}.
\end{align*}
Finally, the output formula~$\formOut(x) := x = \Atoms$ checks whether all vertices have been removed.
Overall, the following formula defines the class of threshold graphs:
\[ \wscForm{x}{y}{\termStep}{\termChoice}{\termWit}{\formOut}.\]
The WSC-fixed-point operator will compute the fixed-point of the variable~$x$ starting with $x_0=\emptyset$ as initial value for~$x$.
First, the term $\termChoice(a_0)$ is evaluated to define the first choice set~$y_0$ of all universal or isolated vertices of~$G$.
Then one such vertex~$\vertA$ is chosen
and the step term $\termStep(x_0,\set{\vertA})$ is evaluated
yielding $x_1 = \set{\vertA}$.

Now inductively assume that~$x_i$ contains all vertices removed so far.
Then~$y_i$ is the set of all universal or isolated vertices of $G[V\setminus x_i]$, this is now an orbit of $(G, x_i)$ (in fact, an orbit of $(G,x_0, \dots , x_i)$).
So again, a vertex $\vertA \in y_i$ is chosen
and added by~$\termStep$ to~$x_i$ yielding the set~$x_{i+1}$.

So finally assume the case that~$y_i$ is empty.
Then nothing is chosen and the step term evaluated 
and yields $\termStep(x_i, \emptyset) = x_i$, so a fixed-point is reached.
Then the output formula defines whether it was possible to remove all vertices.

Now, the WSC-fixed-point operator evaluates the witnessing term
to certify that indeed all choices where made from orbits.
As argued before, all choice sets are indeed orbits
and the term $\termWit(y_i)$ outputs for every~$y_i$ a set of automorphisms,
that for every $\vertA,\vertB \in y_i$ contains an automorphism
mapping~$\vertA$ to~$\vertB$
and so it is certified that the choice sets indeed are orbits.

\subsection{Semantics of Symmetric Choice Operators}
\label{sec:semantics-choice-operators}
We now define the precise semantics of the WSC-fixed-point operator.
We define the evaluation of WSC-fixed-point operators
for arbitrary isomorphism-invariant functions in place of the choice, step, and witnessing terms.
This makes the definition independent of the semantics of CPT
and we can reuse it later in Section~\ref{sec:iso-in-cpt+wsc}.

\begin{definition}[Isomorphism-Invariant Function]
	For a $\sig$-structure~$\Struct$ and a tuple $\tup{a} \in \HF{\StructV}^*$,
	a function $f\colon \HF{\StructV}^k \to \HF{\StructV}$ is called
	\defining{$(\Struct, \tup{a})$-isomorphism-invariant} if
	every automorphism $\auto \in \autGroup{(\Struct, \tup{a})}$
	satisfies
	$f(\auto(\tup{b})) = \auto(f(\tup{b}))$  for every $\tup{b} \in \HF{\StructV}^k$. 
\end{definition}
\begin{definition}[Witnessing an Orbit]
	For a $\sig$-structure~$\Struct$ and a tuple $\tup{a} \in \HF{\StructV}^*$,
	a set~$M$ \defining{witnesses} a set~$N$ as orbit of $(\Struct, \tup{a})$
	if $M \subseteq \autGroup{(\Struct, \tup{a})}$
	and for every $b,c \in N$
	there is a $\phi \in M$
	satisfying $\phi(b) = c$.
\end{definition}
Note that this definition in principle also allows witnessing proper subsets of orbits.
However, the sets~$N$ of interest in the following
will always be given by an isomorphism-invariant function
and so~$N$ can never be a proper subset of an orbit.

Now fix an arbitrary $\sig$-structure~$\Struct$ and a tuple $\tup{a} \in \HF{\StructV}^*$.
Let $\stepF^{\Struct,\tup{a}}, \autF^{\Struct,\tup{a}}  \colon \HF{\StructV} \times \HF{\StructV} \to \HF{\StructV}$ 
and $\choiceF^{\Struct,\tup{a}} \colon \HF{\StructV} \to \HF{\StructV}$
be $(\Struct, \tup{a})$-isomorphism-invariant functions.

We define the (possibly infinite) unique least rooted tree~$\tree$ 
whose vertices are labeled with $\HF{\StructV}$-sets (so two nodes in the tree can have the same label)
and which satisfies the following:
\begin{itemize}
	\item The root is labeled with~$\emptyset$.
	\item A vertex labeled with $c \in \HF{\StructV}$
	has for every  $d\in \choiceF^\Struct(b)$
	a child labeled with $\stepF^\Struct(b, c)$.
\end{itemize}
Let~$\paths$ be the set of tuples $p = (b_1, \dots, b_n)$ of $\HF{\StructV}$-sets
such that ${n\geq 2}$, ${b_1 = \emptyset}$, $b_{n-1} = b_n$,
$b_{i-1} \neq b_{i}$ for all $1< i < n$,
and there is a path of length~$n$ in~$\tree$ starting at the root
and the $i$-th vertex in the path is labeled with $b_{i}$ for all $i \in [n]$.
That is,~$\tree$ models the computation for all possible choices
and~$\paths$ is the set of all possible labels yielding a fixed-point.
For sake of readability, we call the elements of~$\paths$ also paths.

We say that the function $\autF^{\Struct,\tup{a}}$ \defining{witnesses}
a path $p= (b_1, \dots, b_n) \in \paths$
if for every $i \in [n-1]$
it holds that
$\autF^{\Struct,\tup{a}}(b_n,b_i)$
witnesses $\choiceF^{\Struct,\tup{a}}(b_i)$ as an $(\Struct, \tup{a}, b_1, \dots, b_i)$-orbit.
Finally, we define 
\[\wsc{\stepF^{\Struct,\tup{a}}}{\choiceF^{\Struct,\tup{a}}}{\autF^{\Struct,\tup{a}}} :=
\setcond*{b_n}{(b_1,\dots,b_n) \in \paths}\]
if $\autF^{\Struct,\tup{a}}$ witnesses all paths in~$\paths$
and $\wsc{\stepF^{\Struct,\tup{a}}}{\choiceF^{\Struct,\tup{a}}}{\autF^{\Struct,\tup{a}}} := \emptyset$
otherwise.

\begin{lemma}
	\label{lem:if-witnessed-path-then-orbit}
	If $\autF^{\Struct,\tup{a}}$ witnesses some path in~$\paths$,
	then~$\paths$ is an orbit of $(\StructA, \tup{a})$.
\end{lemma}
\begin{proof}
	Let $p^* = (b_1^*, \dots , b_k^*) \in \paths$ be a witnessed path,
	let~$P_i$ be the set of prefixes of length~$i$ of the paths in~$\paths$,
	and let~$p^*_i$ be the prefix of length~$i$ of~$p^*$.
	We prove by induction on~$i$ that~$P_i$ is an orbit.	
	For the root~$\emptyset$ the claim trivially holds.
	
	We show that $p =(b_1, \dots, b_{i+1})$ is in $P_{i+1}$
	if and only if there is an automorphism $\auto \in \autGroup{(\Struct, \tup{a})}$ such that $\auto(p^*_{i+1})=p$.
	Let~$p_i$ be the prefix of length~$i$ of~$p$.
	
	First, assume that $p \in P_{i+1}$.
	By the induction hypothesis,
	there is an automorphism
	$\autoA \in \autGroup{(\Struct, \tup{a})}$
	such that $\autoA(p^*_i) = p_i$.
	By definition of~$\tree$,
	for some
	$c^*_{i} \in \choiceF^{\Struct,\tup{a}}(b^*_i)$ and some
	$c_{i} \in \choiceF^{\Struct,\tup{a}}(b_i)$
	it holds that
	$b^*_i = \stepF^{\Struct,\tup{a}}(b^*_i,c^*_{i})$ and
	$b_i = \stepF^{\Struct,\tup{a}}(b_i,c_{i})$.
	Because $\choiceF^{\Struct,\tup{a}}$ is isomorphism-invariant,
	we have that $\choiceF^{\Struct,\tup{a}}(b_i)=\autoA(\choiceF^{\Struct,\tup{a}}(b^*_i))$.
	Because~$p^*$ is witnessed,
	$\choiceF^{\Struct,\tup{a}}(b^*_i)$ is an orbit 
	of $(\StructA, \tup{a} p^*_i)$ 
	and so
	$\choiceF^{\Struct,\tup{a}}(b_i)$ is an orbit 
	of $(\StructA, \tup{a} p_i) = \autoA((\StructA, \tup{a} p^*_i))$.
	It follows that $\autoA(c^*_{i})$ and $c_{i}$ are in the same orbit of $(\StructA, \tup{a}  p_i)$.
	So let $\autoB \in \autGroup{(\StructA, \tup{a} p_i)}$
	such that $\autoB(\autoA(c^*_{i})) = c_{i}$.
	We now have
	\begin{align*}
		\autoB(\autoA(b^*_{i+1}))&=\autoB(\autoA(\stepF^{\Struct,\tup{a}}(b^*_i,c^*_{i})))\\
		&= \stepF^{\Struct,\tup{a}}(\autoB(\autoA(b^*_i)),\autoB(\autoA(c^*_{i})))\\
		&= \stepF^{\Struct,\tup{a}}(b_i,c_{i}) = b_{i+1}
	\end{align*}
	because $\stepF^{\Struct,\tup{a}}$ is isomorphism-invariant.
	Because~$\autoB$ fixes $b_1, \dots, b_i$,
	we finally have that $(\autoB \circ \autoA) (p^*_{i+1}) = p$.
	
	Second, assume that $\auto \in \autGroup{(\StructA, \tup{a})}$
	such that $p = \auto(p^*_i)$.
	Then by induction hypothesis,
	$p_i = \auto(p^*_i) \in P_i$.
	As before, let $c^*_{i} \in \choiceF^{\Struct,\tup{a}}(b^*_i)$
	such that $b^*_{i+1} = \stepF^{\Struct,\tup{a}}(b^*_i, c^*_i)$.
	Because $\choiceF^{\Struct,\tup{a}}$ is isomorphism-invariant,
	it holds that
	$\auto(c^*_{i+1}) \in \choiceF^{\Struct,\tup{a}}(\auto(b^*_i),\auto(c^*_i))$.
	Because $\stepF^{\Struct,\tup{a}}$ is isomorphism-invariant,
	it holds that 
	\[b_{i+1} = \auto(b^*_{i+1}) =
	\auto(\stepF^{\Struct,\tup{a}}(b^*_i,c^*_{i+1})) =
	\stepF^{\Struct,\tup{a}}(\auto(b^*_i), \auto(c^*_{i+1})) = \stepF^{\Struct,\tup{a}}(b_i, \auto(c^*_{i}).\]
	That is, the vertex corresponding to~$b_i$ (when following the tree~$\tree$ starting at the root)
	has a child labeled with~$b_{i+1}$
	and hence $p \in P_{i+1}$.
\end{proof}
\begin{corollary}\label{cor:witnessed-all-or-none}
	The function $\autF^{\Struct,\tup{a}}$ either witnesses all paths in~$\paths$ or none of them.
\end{corollary}
\begin{proof}
	If there is a witnessed path $p^* = (b_1^*, \dots , b_k^*)$ in~$\paths$,
	then the set $\paths$ is an orbit.
	The claim follows because $\choiceF^{\Struct,\tup{a}}$ is isomorphism-invariant:
	So for every other path~$p$, there is a
	$\auto \in \autGroup{(\Struct, \tup{a})}$
	such that $p = \auto(p^*)$
	and in particular $p = (b_1, \dots, b_k)$ (i.e.,~$p$ and~$p^*$ have the same length).
	Let $i \in [n-1]$.
	Then
	\begin{align*}
		\autF^{\Struct,\tup{a}}(b_n,b_i) = \autF^{\Struct,\tup{a}}(\auto(b^*_n), \auto(b^*_i))
	= \auto(\autF^{\Struct,\tup{a}}(b^*_n,b^*_i)) &\subseteq 
	 \auto(\autGroup{(\Struct, \tup{a}, b^*_1,\dots, b^*_i)})\\ &= \autGroup{\Struct,\tup{a},b_1,\dots, b_i}.
	\end{align*}
	Let $d,e \in \choiceF^{\Struct,\tup{a}}(b_i, c_i)$.
	Then again because $\choiceF^{\Struct,\tup{a}}$ is isomorphism-invariant,
	we have $\inv{\auto}(d),\inv{\auto}(e) \in \choiceF^{\Struct,\tup{a}}(b^*_i,c^*_i)$.
	Because~$p^*$ is witnessed,
	there is a $\autoB \in \autF^{\Struct,\tup{a}}(b^*_n, c^*_i)$
	such that $\autoB(\inv{\auto}(d)) = \inv{\auto}(e)$.
	Hence, $\autoB \circ \auto$ is contained in $ \choiceF^{\Struct,\tup{a}}(b_i,c_i)$
	and maps~$d$ to~$e$.
\end{proof}
	
\begin{corollary}\label{cor:wsc-orbit}
	$\wsc{\stepF^{\Struct,\tup{a}}}{\choiceF^{\Struct,\tup{a}}}{\autF^{\Struct,\tup{a}}}$
	is an $(\Struct,\tup{a})$-orbit.
\end{corollary}
\begin{proof}
	Assume that there is a non-witnessed path (or no path) in~$\paths$.
	In this case $\wsc{\stepF^{\Struct,\tup{a}}}{\choiceF^{\Struct,\tup{a}}}{\autF^{\Struct,\tup{a}}} = \emptyset$ trivially satisfies the claim.
	Otherwise, there is a witnessed path and~$\paths$ is an orbit,
	all paths have the same length,
	and $\wsc{\stepF^{\Struct,\tup{a}}}{\choiceF^{\Struct,\tup{a}}}{\autF^{\Struct,\tup{a}}}$
	is the set of all sets~$b$, which are the last entry of some path in~$\paths$.
	Because~$\paths$ is an orbit,
	in particular the last vertices in every root-to-leaf-path
	(which are necessarily at the same depth) form an orbit
	and the claim follows.
\end{proof}

Now we can define the denotation of WSC-fixed-point operators:
For a BGS+WCS-term~$\termA$ with free variables $x_1, \dots, x_k$
and a tuple $\tup{a} \in \HF{\StructV}^\ell$ for some $\ell \leq k$,
we write $\denotation{\termA}^\Struct(\tup{a})$
for the ``partial'' application of the function $\denotation{\termA}^\Struct$,
so for the function $\HF{\StructV}^{k-\ell} \to \HF{\StructV}$
defined by $\tup{b} \mapsto \denotation{\termA}^\Struct(\tup{a}\tup{b})$.

Let~$\termStep$ and~$\termWit$ be BGS+WSC-terms with free variables~$\tup{z}xy$,~%
$\termChoice$ be a BGS+WSC-term with free variables~$\tup{z}x$,
and~$\formOut$ be a BGS+WSC-formula with free variables~$\tup{z}x$.
We define
\begin{align*}
&\denotation{\wscForm{x}{y}{\termStep}{\termChoice}{\termWit}{\formOut}}^\Struct:=\\
	&\qquad\setcond*{\tup{a}}{\tup{a}b \in \denotation{\formOut}^\Struct \text{ for every } b \in \wsc{\denotation{\termStep}^\Struct(\tup{a})}{\denotation{\termChoice}^\Struct(\tup{a})}{\denotation{\termWit}^\Struct(\tup{a})}}.
\end{align*}
In~\cite{DawarRicherby03}, the fixed-point operator with symmetric choice
is not evaluated on~$\Struct$, but on the reduct $\reduct{\Struct}{\sig'}$,
where $\sig'\subseteq \sig$ is the subset of relations of the $\sig$-structure~$\Struct$ used in the fixed-point operator.
This ensures that adding unused relations to structures does not change the result of a formula (the additional relations potentially change the orbits of the structure and choices cannot be witnessed anymore),
which is a desirable property~\cite{ebbinghaus1985}.
We do not use the ``reduct semantics'' in this article.
We could in principle use it
but then Section~\ref{sec:iso-in-cpt+wsc} 
would get even more technical without providing further insights.

\paragraph{Failure on Non-Witnessed Choices.}
While the denotation defined as above results in a reasonable logic,
we want a special treatment of the case when choices cannot be witnessed.
Whenever during the evaluation of a formula
there is a path in~$\paths$ that is not witnessed,
we abort evaluation and output an error, indicating there was a non-witnessed choice.
Formally, we extend the denotation by an error-marker~$\choiceError$.
Then the denotation of a term becomes a function $\HF{\StructV}^k \to \HF{\StructV} \cup \set{\choiceError}$
and the denotation of a formula a function $\HF{\StructV}^k \to \set{\top, \bot, \choiceError}$.
Whenever a~$\choiceError$ occurs, it is just propagated.
We omit the formal definitions here.
Later, we will see that the error marker is necessary to guarantee polynomial-time evaluation.

\paragraph{Fixing Intermediate Steps.}
We defined the evaluation of choice terms similar to~\cite{DawarRicherby03}
using the tree~$\tree$. However, our definition is different in one crucial aspect:
in the setting of Lemma~\ref{lem:if-witnessed-path-then-orbit},
we require that $\choiceF(b_i)$ defines an orbit of $(\StructA, \tup{a}, b_1, \dots , b_i)$,
where in~\cite{DawarRicherby03} an orbit of $(\StructA, \tup{a}, b_i, c_i)$
is required.
That is, in BGS+WSC one has to respect in some sense all choices made in previous intermediate steps during the fixed-point computation.
This is crucial to prove Lemma~\ref{lem:if-witnessed-path-then-orbit}.
This is not required in~\cite{DawarRicherby03}
because the authors only need that the vertices in~$\tree$ on the same level
are in the same orbit.
We actually need that the paths in the tree in their entirety form an orbit
to establish Corollary~\ref{cor:witnessed-all-or-none}.
Due to this corollary, we can give the witnessing term
access to~$b_n$ when
witnessing $(\StructA, \tup{a}, b_1,\dots, b_i)$-orbits.
Accessing the defined fixed-point to witness intermediate choice sets
will become crucial in the following, namely to define Gurevich's algorithm in Section~\ref{sec:canon-in-cpt-wsc}.
We do not require that the actual chosen elements~$c_i$ are fixed by the automorphisms
because, in contrast to~\cite{DawarRicherby03}, the choice term only gets the~$b_i$ as input and not the~$c_i$.
So if~$c_i$ and~$c_i'$ result in the same next intermediate step~$b_{i+1}$, the subsequent computation will be the same for both choices.
For the very same reason and again in contrast to~\cite{DawarRicherby03},
it is sufficient to label the vertices in the tree~$\tree$ only with the intermediate steps~$b_i$ and not additionally with the chosen elements.

\subsection{CPT+WSC}
Similarly to how CPT is obtained from BGS, we obtain CPT+WSC by enforcing polynomial bounds on BGS+WSC terms and formulas:
A CPT+WSC-term (respectively, formula) is a pair $(\termA, p(n))$ (respectively, $(\formA, p(n))$)
of a BGS+WSC-term (respectively, formula) and a polynomial.

For BGS-operators, we add the same restrictions as in CPT.
For a WSC-fixed-point operator $\wscForm{x}{y}{\termStep}{\termChoice}{\termWit}{\formOut}$,
a structure~$\StructA$, and a tuple $\tup{a} \in \HF{\StructV}^*$,
we restrict $\paths$ to paths $(b_1, \dots, b_k)$ of length $k \leq p(|\StructV|)$
for which ${|\TC{b_i}| \leq p(|\StructV|)}$
for all $i \in [k]$.
If there is a path in~$\paths$ of length greater than $p(|\StructV|)$
or in some path there is a set not bounded by~$p$,
then the WSC-fixed-point operator has denotation~$\choiceError$.

It is important
that we do not require that $|\wsc{\denotation{\termStep}^\Struct(\tup{a})}{\denotation{\termChoice}^\Struct(\tup{a})}{\denotation{\termWit}^\Struct(\tup{a})}|$
is bounded by $p(|\StructV|)$.
In fact, using WSC-fixed-point operators
only makes sense if the set is allowed to be of superpolynomial
size, as otherwise we could define it
with a regular iteration term.
It is also important to output~$\choiceError$ and not~$\emptyset$
when the polynomial bound is exceeded
because in that case we cannot validate whether all choice sets
are orbits (and so it might depend on the choices whether the
bound is exceeded or not).
To evaluate the witnessing term we need access to the fixed-point,
which cannot be computed if the polynomial bound is exceeded.

Because the WSC-fixed-point operator can only choose from orbits,
CPT+WSC is isomorphism-invariant:
\begin{lemma}\label{lem:cpt-wsc-iso-invariant}
	For every structure~$\StructA$,
	every CPT+WSC-term~$\termA$,
	and every CPT+WSC-formula~$\formA$,
	the denotations
	$\denotation{\termA}^\StructA$ and 
	$\denotation{\formA}^\StructA$ are unions of $\StructA$-orbits.
\end{lemma}
\begin{proof}
	The proof is straightforward by structural induction on terms and formulas
	using Corollary~\ref{cor:wsc-orbit}.
\end{proof}
While Lemma~\ref{lem:cpt-wsc-iso-invariant}
only concerns automorphisms, it is easy to see
that CPT+WSC also respects isomorphism between different structures.
Using multiple structures would make Section~\ref{sec:semantics-choice-operators}
formally more complicated without providing new insights.

Using Lemma~\ref{lem:cpt-wsc-iso-invariant} we can show that
model checking for CPT+WSC can be done in polynomial time.
Naively computing the denotation is not possible
because, as we have seen earlier,
the sets $\wsc{\denotation{\termStep}^\Struct(\tup{a})}{\denotation{\termChoice}^\Struct(\tup{a})}{\denotation{\termWit}^\Struct(\tup{a})}$
are possibly not of polynomial size.

\begin{lemma}
	For every CPT+WSC term $(\termA, p(n))$ or formula $(\formA,p(n))$,
	we can compute in polynomial time on input~$\StructA$ and $\tup{a} \in \HF{\StructVA}^k$ the denotation $\denotation{\termA}^\Struct(\tup{a})$ or  $\denotation{\formA}^\Struct(\tup{a})$ respectively.
\end{lemma}
\begin{proof}
	The proof is by structural induction on~$\termA$ or~$\formA$.
	We show that $\denotation{\termA}^\Struct(\tup{a})$ or  $\denotation{\formA}^\Struct(\tup{a})$ can be computed in polynomial time for every tuple $\tup{a} \in \HF{\StructVA}$ of suitable length.
	
	Assume by induction hypothesis that for CPT+WSC-terms~$\termA$ and~$\termB$ and formulas~$\formA$ and~$\formB$,
	the denotation can be computed in polynomial time.
	Then we can surely do so for comprehension terms,
	iterations terms, and all formulas composed of~$\termA$,~$\termB$,~$\formA$, and~$\formB$ apart from WSC-fixed-point operators.
	
	So we have to consider a WSC-fixed-point operator
	\[\wscForm{x}{y}{\termStep}{\termChoice}{\termWit}{\formOut}.\]
	Because $M = \wsc{\denotation{\termStep}^\Struct(\tup{a})}{\denotation{\termChoice}^\Struct(\tup{a})}{\denotation{\termWit}^\Struct(\tup{a})}$
	is an orbit of $(\StructA,\tup{a})$ by Corollary~\ref{cor:wsc-orbit},
	it suffices to compute one $b \in M$
	(or determine that none exists)
	and check whether $\tup{a}b \in \denotation{\formOut}^\Struct$
	by Lemma~\ref{lem:cpt-wsc-iso-invariant}.
	Given~$b$, the check can be done in polynomial time by induction hypothesis.
	Some $b\in M$ can be computed
	by iteratively evaluating~$\termChoice$ to define a choice set,
	selecting one arbitrary element out of it,
	and then evaluating~$\termStep$ with this choice
	until either a fixed-point~$b$ is reached or more than 
	$p(|\StructVA|)$ iterations are performed.
	In the later case output~$\choiceError$.
	If this is not the case,
	we check whether~$\termWit$ witnesses the computed path.
	If the path is not witnessed, then we abort with output~$\choiceError$
	and otherwise by Corollary~\ref{cor:witnessed-all-or-none}
	we computed one $b \in M$.
	If, at any point during the evaluation,
	we have to construct a set~$c$ with $|\TC{c}|> p(|\StructVA|)$,
	then output~$\choiceError$ as well.
	If we always chose from orbits we would have constructed such an excessively large set~$c$ for all possible choices.
	Otherwise, if some choice set would not be an orbit
	we would fail to witness the orbits and output~$\choiceError$ as well.
	
	We need to evaluate~$\termStep$,~$\termChoice$, and~$\termWit$
	at most $p(|\StructVA|)$ many times,
	so computing~$b$ is also done in polynomial time.
\end{proof}

\subsection{Defining Sets}
\label{sec:defining-sets}

The WSC-fixed-point operator can only output truth values.
These are, by design, isomorphism-invariant.
We now discuss alternatives:
Let $f \colon \HF{\StructV}^k \to \HF{\StructV}$
be some function which we want to define with an iteration term with choice
(the domain is the set of possible parameter values).
To obtain a deterministic logic,
we need that $\autoA(f(\tup{a})) = f(\tup{a})$
for every $\autoA \in \autGroup{(\Struct, \tup{a})}$.
This clearly holds if~$f$ only returns truth values
(e.g., encoded by~$\emptyset$ and $\set{\emptyset}$).
We do not know how to decide in polynomial time
whether the condition $\autoA(f(\tup{a})) = f(\tup{a})$
is satisfied for all $\autoA \in \autGroup{(\Struct, \tup{a})}$
during the evaluation.
So we consider functions $f \colon \HF{\StructV}^k \to \HF{\emptyset}$,
which generalizes the case of truth values
but still is syntactically isomorphism-invariant.
We define an iteration term for this case:
\[\wscForm{x}{y}{\termStep}{\termChoice}{\termWit}{\termOut},\]
where~$\termStep$,~$\termChoice$,~$\termWit$, and~$\termOut$
are BGS+WSC-terms.
The only difference to the iteration term with choice seen so far is
that the output formula is replaced with an \defining{output term}.
Let $W = \wsc{\denotation{\termStep}^\Struct(\tup{a})}{\denotation{\termChoice}^\Struct(\tup{a})}{\denotation{\termWit}^\Struct(\tup{a})}$
to define the denotation as follows:
\begin{align*}
	&\phantom{:= {}} \denotation{\wscForm{x}{y}{\termStep}{\termChoice}{\termWit}{\termOut}}^\Struct(\tup{a})\\
	&:=
	\begin{cases}
		\bigcup_{b \in W} \denotation{\termOut}^\Struct(\tup{a}b) &\text{if } \denotation{\termOut}^\Struct(\tup{a}b) \in \HF{\emptyset} \text { for all }b \in W,\\
		\emptyset &\text{otherwise.}
	\end{cases}
\end{align*}
That is, if $\denotation{\termOut}^\Struct(\tup{a}b) \in \HF{\emptyset}$  for all $b \in W$,
then $\denotation{\termOut}^\Struct(\tup{a}b) = \denotation{\termOut}^\Struct(\tup{a}b')$ for all $b,b' \in W$
and the iteration term evaluates to
$\denotation{\termOut}^\Struct(\tup{a}b)$ for every $b \in W$
if $W \neq \emptyset$ and to~$\emptyset$ otherwise.

We now show that the extended WSC-fixed-point operator
does not increase the expressive power of CPT+WSC:
\begin{lemma}
	\label{lem:extended-iteration-term}
	For all CPT+WSC-terms~$\termStep$,~$\termChoice$,~$\termWit$, and~$\termOut$,
	there is a CPT+WSC-term~$\termB$
	such that for all structures~$\StructA$
	it holds that
	\[\denotation{\wscForm{x}{y}{\termStep}{\termChoice}{\termWit}{\termOut}}^\Struct = \denotation{\termB}^\StructA.\]
\end{lemma}
\begin{proof}
	Let $p(n)$ be the polynomial bound of the CPT+WSC-term
	and $n = |\StructVA|$.
	Let~$a$ be a set constructed during the evaluation.
	Because of the polynomial bound, we have $|\TC{a}| \leq p(n)$.
	The set~$a$ corresponds to a directed acyclic graph (DAG),
	where the leaves are either atoms or~$\emptyset$.
	By the condition $|\TC{a}| \leq p(n)$,
	the DAG has at most~$p(n)$ many vertices.
	Note that $\HF{\emptyset}$ can be totally ordered in CPT.
	In particular, 
	if $a \in \HF{\emptyset}$,
	the DAG corresponding to~$a$ can be totally ordered.
	Given the DAG, we can reconstruct~$a$ in CPT.

	Let $\termB_{\text{toDAG}}(z)$ be a CPT-term,
	which given a set $a \in \HF{\emptyset}$
	outputs the totally ordered DAG corresponding to~$a$
	(that is, we can assume that its vertex set is $[|\TC{a}|]$)
	as a set containing the edges of the DAG.
	If $a \not\in \HF{\emptyset}$, then $\termB_{\text{toDAG}}$
	outputs~$\emptyset$.
	Furthermore, let $\termB_{\text{fromDAG}}(z)$
	be the term recovering~$a$ from this set.
	First,
	\begin{align*}
		\formA &:= \wscFormbig{x}{y}{\termStep}{\termChoice}{\termWit}{(\termOut \neq \emptyset) \Rightarrow (\termB_{\text{toDAG}}(\termOut) \neq \emptyset)}
		\intertext{defines whether $\termOut$ only outputs $\HF{\emptyset}$-sets.
	Second,}
	\termC &:= \setcond[\big]{(i,j)}{i,j \in [p(|\Atoms|)],
	\wscForm{x}{y}{\termStep}{\termChoice}{\termWit}{
		(i,j) \in \termB_{\text{toDAG}}(\termOut)}}
	\end{align*}
	defines the DAG given by $\termB_{\text{toDAG}}$ (with possible some isolated vertices, which can easily be ignored).
	Last,
	$\Unique(\setcond{\termB_{\text{fromDAG}}(\termC)}{\formA})$
	is equivalent to $\wscForm{x}{y}{\termStep}{\termChoice}{\termWit}{\termOut}$,
	where we use $p(n)^2$ as new polynomial bound
	(because we just try all pairs $i,j$).
\end{proof}
With an easy inductive argument one sees
that also nesting the extended iteration terms
does not increase the expressive power.

\section{Canonization in CPT+WSC}
\label{sec:canon-in-cpt-wsc}

In the following section,
we work with classes of relational $\sig$-structures~$\GraphClass$.
We always assume that these classes are closed under isomorphisms and only
contain connected structures
because we are interested in isomorphism testing and canonization.
The case of unconnected structures reduces to connected ones.

The process of individualizing certain atoms in structures
plays a crucial role
and we only want to consider classes of structures
closed under individualizing atoms.
Conceptually, individualizing an atom means to give it a unique color.
The formulas we are going to define in fact iteratively individualize atoms.
So it will be convenient to capture the individualized atoms
by CPT+WSC ``internal'' tuples,
that is, many formulas will have a free variable~$\indVar$
to which we can pass a tuple (encoded using sets)
containing the tuple of individualized atoms.
Thus, instead of assuming that the classes of structures
are closed under individualization,
we work with the free variable~$\indVar$ to which all possible tuples can be passed and require in our definitions that certain properties hold for all possible tuples passed into~$\indVar$.
At crucial  points we remind the reader that 
in that sense we work with individualization-closed classes of structures.

In what follows, we will always assume
that tuples do not contain duplicates.
Moreover, we will freely switch between the ``internal'' representation of tuples
in CPT and the ``external'' tuples of individualized atoms whenever needed.
For sake of shorter formulas,
we introduce a slightly special concatenation operation for tuples of atoms in CPT+WSC.
This operation is shorthand notation for a more complex but uninteresting CPT-term.
Let~$x$ and~$y$ be variables.
We write $\concat{x}{y}$ for a CPT term satisfying the following equations:
\begin{align*}
	\denotation{\concat{x}{y}}^\Struct(\tup{a},b) &= \tup{a}b & \tup{a} \in\StructV^*,b\in\StructV, b \notin \tup{a},\\
	\denotation{\concat{x}{y}}^\Struct(\tup{a},b) &= \tup{a} & \tup{a} \in\StructV^*,b\in\StructV, b \in \tup{a},\\	\denotation{\concat{x}{y}}^\Struct(\tup{a},\set{b}) &= \denotation{\concat{x}{y}}^\Struct(\tup{a},b)  & \tup{a} \in\StructV^*,b\in\StructV,\\
	\denotation{\concat{x}{y}}^\Struct(\tup{a},\emptyset) &= \tup{a} & \tup{a} \in\StructV^*.
\end{align*}
We also assume that~$\concat{x}{y}$ concatenates two tuples
by also removing duplicates.
We write~$x_i$ for the term extracting the $i$-th position of a tuple
or the empty set if $i$ is larger than the length of the tuple
($i$ is encoded as a von Neumann ordinal).
We extend the notation from variables to arbitrary CPT+WSC-terms~$\termA$ and~$\termB$
and write~$\concat{\termA}{\termB}$
for the CPT+WSC-term appending the result of~$\termB$ to the result of~$\termA$ in the way defined above.
Clearly, all these terms can be defined in CPT+WSC
(or in CPT if~$\termA$ and~$\termB$ are CPT-terms, too).

We now introduce various notions related to defining isomorphism
and canonization.
In the end, we show that all of them are equivalent.
As already mentioned, it will be important to not only consider all structures~$\Struct$
of a given class of structures~$\GraphClass$,
but also to consider all pairs $(\Struct, \tup{a})$ for all $\tup{a} \in \StructV^*$.
In what follows, let~$L$ be one of the logics CPT or CPT+WSC.

\begin{definition}[Definable Isomorphism]
	\label{def:definable-isomorphism}
	A logic~$L$ \defining{defines isomorphism} for a class of $\sig$-structures~$\GraphClass$,
	if there is an $L$-formula $\formA(\indVar_1,\indVar_2)$,
	such that, for every $\StructA,\StructB \in \GraphClass$
		and $\tup{a} \in \StructVA^*, \tup{b} \in \StructVB^*$, on the disjoint union $\StructA\disunion\StructB$
	it holds that $(\tup{a},\tup{b}) \in \denotation{\formA}^{\StructA\disunion\StructB}$
	if and only if $(\StructA,\tup{a}) \iso (\StructB,\tup{b})$.
\end{definition}
In the case that~$L$ is CPT+WSC,
we require in the previous definition that~$\formA$
never outputs~$\choiceError$.
So we can write $(\tup{a},\tup{b}) \in \denotation{\formA}^{\StructA\disunion\StructB}$,
because as~$\choiceError$ never occurs,
we can regard $\denotation{\formA}^{\StructA\disunion\StructB}$ again as the set of tuples satisfying~$\formA$.
In all definitions that follow, we also require without further mentioning that~$\choiceError$ never occurs for any input.

\begin{definition}[Distinguishable Orbits]
	We say that a class of $\tau$-structures~$\GraphClass$ has
	\defining{$L$\nobreakdash-distinguishable $k$-orbits}
	if there is an $L$-formula $\formA(\indVar,x,y)$
	such that  for every $\Struct \in \GraphClass$
	and every $\tup{a} \in \StructVA^*$
	the relation~$\spleq$ on $k$-tuples 
	satisfying  $\tup{b} \spleq \tup{c}$  if and only if $(\tup{a},\tup{b},\tup{c}) \in \denotation{\formA}^\Struct$
	is a total preorder
	and its equivalence classes are the $k$-orbits of $(\StructA,\tup{a})$.
\end{definition}
Note that because~$\spleq$ is a total preorder and not just some equivalence relation, it defines a total order on the $k$-orbits.

\begin{definition}[Complete Invariant]
	\label{def:definable-complete-invariant}
	An \defining{$L$-definable complete invariant} of a class of $\sig$-structures~$\GraphClass$
	is an $L$-term $\termA(\indVar)$ 
	which satisfies the following:
	$\denotation{\termA}^\StructA(\tup{a}) = \denotation{\termA}^\StructB(\tup{b})$
	if and only if $(\StructA,\tup{a}) \iso (\StructB, \tup{b})$
	for every $\StructA,\StructB \in \GraphClass$
	and every $\tup{a} \in \StructVA^*, \tup{b}\in \StructVB^*$.
\end{definition}

\begin{lemma}
	\label{lem:invariants-order}
	If an $L$-term $\termA(\indVar)$ is a complete invariant for a class of $\sig$-structures~$\GraphClass$,
	then there is a CPT-term
	defining a total order 
	on $\setcond{\denotation{\termA}^\StructA(\tup{a})}{\StructA \in \GraphClass, \tup{a} \in \StructVA^*}$.
\end{lemma}
\begin{proof}
	Let $\StructA \in \GraphClass$.
	Then there is another structure $\StructB \in \GraphClass$
	such that $\StructA \iso \StructB$ and $\StructVA \cap \StructVB = \emptyset$
	because~$\GraphClass$ is closed under isomorphism.
	Let $\auto \colon \StructVA\to \StructVB$ be an isomorphism.
	Then $\denotation{\termA}^\StructA = \denotation{\termA}^\StructB = \auto(\denotation{\termA}^\StructA)$,
	but this implies $\denotation{\termA}^\StructA \in \HF{\emptyset}$
	for every $\StructA \in \GraphClass$.
	Finally, $\HF{\emptyset}$ can be ordered in CPT.
\end{proof}

\begin{lemma}
	\label{lem:complete-invariant-implies-definable-orbits}
	If there is an $L$-definable complete invariant for a class
	of $\sig$-structures~$\GraphClass$,
	then $\GraphClass$ has $L$-distinguishable $k$-orbits
	for every $k \in \nat$.
\end{lemma}
\begin{proof}
	Let $k \in \nat$, $\termA_{\text{inv}}(\indVar)$ be an $L$-definable complete invariant,
	and let $\formA_{\text{inv}}(x,y)$ be a CPT-formula defining a total order
	on the invariant by Lemma~\ref{lem:invariants-order}.
	We define
\begin{align*}
	\formA_{\text{orb}}(\indVar, x,y)&:=\formA_{\text{inv}}(\termA_{\text{inv}}(\concat{\indVar}{x}), \termA_{\text{inv}}(\concat{\indVar}{y})).
\end{align*}
	The $L$-formula $\formA_{\text{orb}}(\indVar, x,y)$
	orders two $k$-tuples $\tup{\vertA}$ and $\tup{\vertB}$
	according to the order of $\formA_{\text{inv}}(x,y)$
	on the complete invariants when individualizing $\tup{\vertA}$ respectively $\tup{\vertB}$.
	So $\formA_{\text{orb}}$ is a total preorder.
	For a structure~$\StructA \in \GraphClass$ and an $\tup{a} \in \StructVA^*$,
	two $k$-tuples $\tup{\vertA},\tup{\vertB}\in\StructVA^k$
	are in the same $k$-orbit of $(\StructA,\tup{a})$
	if and only if $(\StructA,\tup{a}\tup{\vertA}) \iso (\StructA,\tup{a}\tup{\vertB})$,
	which is the case if and only if
	$\denotation{\termA_{\text{inv}}}^\Struct(\tup{a}\tup{\vertA}) = \denotation{\termA_{\text{inv}}}^\Struct(\tup{a}\tup{\vertB})$,
	that is, $\formA_{\text{orb}}$ defines and orders the $k$-orbits of $(\Struct, \tup{a})$.
\end{proof}

\begin{definition}[Ready for Individualization]
	A class of $\sig$-structures~$\GraphClass$ is \defining{ready for individualization} in~$L$
	if there is an $L$-term $\termA(\indVar)$ that for every structure~$\Struct \in \GraphClass$ and every $\tup{a} \in \StructV^*$
	defines a $1$-orbit~$O$ of $(\Struct, \tup{a})$, that is, $O = \denotation{\termA}^\Struct(\tup{a}) \in \orbs{1}{(\Struct, \tup{a})}$,
	such that if there is a $1$-orbit disjoint with~$\tup{a}$,
	then~$O$ is disjoint with $\tup{a}$.
\end{definition}

\begin{definition}[Canonization]
	A logic~$L$ defines a \defining{canonization}
	for a class of \linebreak $\sig$\nobreakdash-structures~$\GraphClass$,
	if there is an $L[\sig,\sig \disunion{\leq}]$-interpretation $\interpret(\indVar)$
	mapping $\sig$-structures to
	$(\sig \disunion{\leq})$-structures such that
	\begin{enumerate}
		\item $\leq$ is a total order on the universe of $\interpret(\StructA,\tup{a})$, for every $\StructA \in \GraphClass$ and $\tup{a} \in \StructV^*$,
		\item $(\StructA,\tup{a}) \iso \reduct{\interpret(\StructA,\tup{a})}{\sig}$, for every $\StructA \in \GraphClass$ and $\tup{a} \in \StructV^*$, and
		\item $\interpret(\StructA,\tup{a}) = \interpret(\StructB,\tup{b})$ if and only if 
		$(\StructA,\tup{a}) \iso (\StructB,\tup{b})$, for every $\StructA,\StructB \in \GraphClass$, $\tup{a} \in \StructVA^*$, and $\tup{b} \in \StructVB^*$.
	\end{enumerate}
\end{definition}

We now show that the algorithm of~\cite{Gurevich97},
which turns a complete invariant into a canonization,
is CPT+WSC-definable.
Intuitively,
we iteratively individualize an atom of a non-trivial orbit
which is minimal according to an isomorphism-invariant total order on the orbits.
We continue this procedure, until we individualized all atoms
and thereby defined a total order on the atoms (see Figure~\ref{fig:gurevich-example} for an example).
While the order itself is not unique,
the isomorphism type of the ordered structure is unique,
because we always chose from orbits.
This way, we obtain the canon by renaming the atoms to be just numbers.
We now define this approach formally:
Let $\GraphClass$ be a class of $\sig$-structures
ready for individualization in CPT+WSC
and let $\termA_{\text{orb}}(\indVar)$ be the corresponding CPT+WSC-term
defining a $1$-orbit with the required properties.

\newcommand{\figcolA}{blue}
\newcommand{\figcolB}{darkgreen}
\newcommand{\figcolC}{darkyellow}
\newcommand{\figcolD}{purple}
\newcommand{\figcolE}{red}

\def\vertexColor{gray}

\newcommand{\gggraph}[1]{}
\newcommand{\ggraph}[9]{
	\renewcommand{\gggraph}[1]{
		\begin{scope}[scale = 0.9, every node/.style={circle, fill=\vertexColor!40!white, inner sep=0.1mm,   minimum size=4mm}]
			\node[#1!40!white, draw=black, text=black, font = \scriptsize] (V1) at (0,0) { #2};
			\node[#3!40!white, draw=black,text=black, font = \scriptsize] (V2) at (-1,-1) { #4};
			\node[#5!40!white, draw=black,text=black, font = \scriptsize] (V3) at (1, -1) { #6};
			\node[#7!40!white, draw=black,text=black, font = \scriptsize] (V4) at (1,1)   { #8};
			\node[#9!40!white, draw=black,text=black, font = \scriptsize] (V5) at (-1,1) { ##1};
			\path[-, black, draw]
			(V1) -- (V2)
			(V1) -- (V3)
			(V1) -- (V4)
			(V1) -- (V5)
			(V2) -- (V3) -- (V4) -- (V5) -- (V2);
		\end{scope}
	}%
	\gggraph%
}

\newcommand{\ggraphil}[5]{
	\tikz[baseline=-0.6ex]{
		\begin{scope}[scale = 1, every node/.style={transform shape}, transform shape]
			\ggraph{#1}{}{#2}{}{#3}{}{#4}{}{#5}{}
		\end{scope}
	}
}
\newcommand{\gggraphillabels}{}
\newcommand{\ggraphillabels}[9]{
	\renewcommand{\gggraphillabels}[2]{
		\tikz[baseline=-0.5ex]{
			\begin{scope}[scale = 0.28, every node/.style={anchor=center}]
				\node[vertex, #1, font=\tiny, inner sep = 0, text =white,] (V1) at (0,0) {#2};
				\node[vertex, #3, font=\tiny, inner sep = 0, text =white] (V2) at (-1,-1){#4};
				\node[vertex, #5, font=\tiny, inner sep = 0, text =white] (V3) at (1, -1){#6};
				\node[vertex, #7, font=\tiny, inner sep = 0, text =white] (V4) at (1,1){#8};
				\node[vertex, #9, font=\tiny, inner sep = 0, text =white] (V5) at (-1,1){##1};
				\path[-, black, draw]
				(V1) -- (V2)
				(V1) -- (V3)
				(V1) -- (V4)
				(V1) -- (V5)
				(V2) -- (V3) -- (V4) -- (V5) -- (V2);
				##2
			\end{scope}
		}
	}
	\gggraphillabels
}

\begin{figure}
	\centering
	\begin{tikzpicture}

		\begin{scope}[scale=0.5]
			
			\node[centered, font = \footnotesize, align=center ] (stepslabel) at (4,2.5) {individualized vertices};
			\node[centered, font = \footnotesize, align=center] (orbitslabel) at (4,-6.5) {ordered orbits};
			
			\begin{scope}[scale=1, shift={(-1.5,0)}] 		
				\ggraph{\vertexColor}{}{\vertexColor}{}{\vertexColor}{}{\vertexColor}{}{\vertexColor}{};
			\end{scope}
			
			\begin{scope}[scale=1,shift={(-1.5,-4)}]
				\ggraph{\vertexColor}{2}{\vertexColor}{1}{\vertexColor}{1}{\vertexColor}{1}{\vertexColor}{1};
			\end{scope}
			
			\path[-{To[width=3mm,length=1.5mm]}, black, thick]
			(0, -2) edge node[pos=0.4,font=\footnotesize, align=center] {choice:\\orbit 1} (2.5, -2);

			\begin{scope}[scale=1, shift={(4,0)}] 		
				\ggraph{\vertexColor}{}{\figcolA}{I}{\vertexColor}{}{\vertexColor}{}{\vertexColor}{};
			\end{scope}
			
			\begin{scope}[scale=1,shift={(4,-4)}]
				\ggraph{\vertexColor}{5}{\figcolA}{I}{\vertexColor}{3}{\vertexColor}{4}{\vertexColor}{3};
			\end{scope}
			
			\path[-{To[width=3mm,length=1.5mm]}, black, thick]
			(5.5, -2) edge node[pos=0.4, font=\footnotesize, align=center] {choice:\\orbit 3} (8, -2);
			
			\begin{scope}[scale=1, shift={(9.5,0)}] 		
				\ggraph{\vertexColor}{}{\figcolA}{I}{\figcolB}{II}{\vertexColor}{}{\vertexColor}{};
			\end{scope}
			
			\begin{scope}[scale=1,shift={(9.5,-4)}]
				\ggraph{\vertexColor}{6}{\figcolA}{I}{\figcolB}{II}{\vertexColor}{7}{\vertexColor}{8};
			\end{scope}
			
			\path[decorate, draw, -{To[width=3mm,length=1.5mm]}, black, thick, decoration={snake, amplitude=0.6mm, post length = 0.75mm}]
			(11, -2) to  (15.75, 0);
			
			\begin{scope}[scale=1, shift={(17.25,0)}] 		
				\ggraph{\figcolC}{III}{\figcolA}{I}{\figcolB}{II}{\figcolD}{IV}{\figcolE}{V};
			\end{scope}

			\node[centered, font = \footnotesize, align = center] (stepslabel) at (17.25,-7) {order obtained\\ from different\\ choices};
			
			\begin{scope}[scale=1, shift={(17.25,-4)}] 		
				\ggraph{\figcolC}{III}{\figcolB}{II}{\figcolA}{I}{\figcolE}{V}{\figcolD}{IV};
			\end{scope}

			\node[centered, font = \footnotesize, align = center] (canonlabel) at (23,-4.25) {ordered\\canon};
			
			\path[-{To[width=3mm,length=1.5mm]}, black, thick]
			(18.75, 0) edge (20.5,-1.25)
			(18.75, -4) edge (20.5,-2.75);
			
			\begin{scope}[scale = 1, shift={(23,-2)}, every node/.style={circle, fill=black, inner sep=0.1mm,   minimum size=4mm}]
				\node[\figcolC!40!white, draw=black, font=\scriptsize, inner sep = 0, text =black] (V1) at (0,0) {III};
				\node[\figcolB!40!white, draw=black, font=\scriptsize, inner sep = 0, text =black] (V2) at (-1,0) {II};
				\node[\figcolD!40!white, draw=black, font=\scriptsize, inner sep = 0, text =black] (V3) at (1,0) {IV};
				\node[\figcolE!40!white, draw=black, font=\scriptsize, inner sep = 0, text =black] (V4) at (2,0)   {V};
				\node[\figcolA!40!white, draw=black, font=\scriptsize, inner sep = 0, text =black] (V5) at (-2,0) {I};
				\path[-, black, draw]  
				(V5) edge (V2)
				(V2) edge (V1)
				(V1) edge (V3)
				(V3) edge (V4)
				(V5) edge[bend left = 50] (V1)
				(V1) edge[bend left = 50] (V4)
				(V2) edge[bend right = 50] (V3)
				(V5) edge[bend right = 50] (V4);
			\end{scope}
		\end{scope}
	\end{tikzpicture}
	\caption{Gurevich's canonization algorithm on an example graph:
		The top row shows the sequence of graphs with individualized vertices (shown by Roman numerals) in the algorithm.
		The bottom row shows the orbit partition of the graphs
		(vertices with the same number are in the same orbit).
		In each step, one vertex of the orbit with minimal number is chosen and individualized.
		This process yields the total order shown on the top right.
		Another order obtained from different choices in shown below.
		Both orders induce to the same canon.
	}
	\label{fig:gurevich-example}
\end{figure}
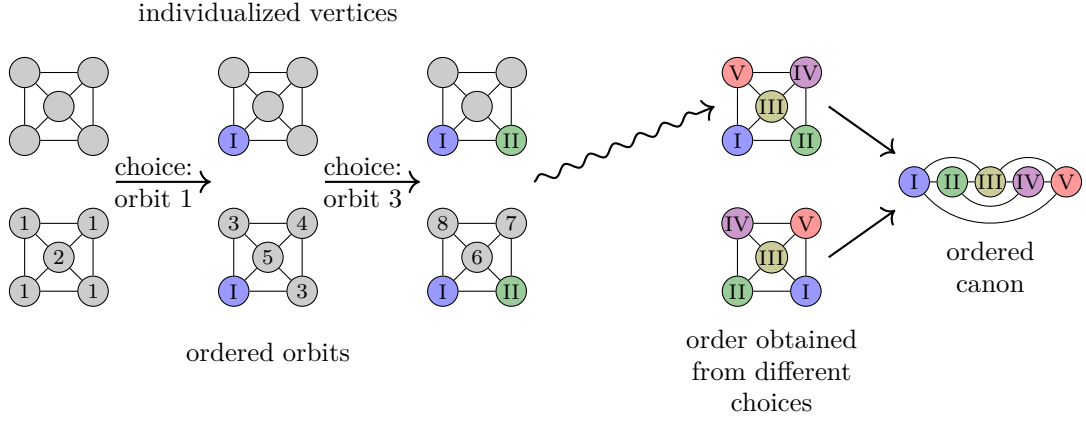

We define for every $\StructA \in \GraphClass$ and $\tup{a} \in \StructVA^*$ a set $\canlabels{\StructA}{\tup{a}}$ as follows:
If all atoms are individualized, i.e., every atom is contained in~$\tup{a}$,
we set
\[\canlabels{\StructA}{\tup{a}} := \set{\tup{a}}.\]
Otherwise, let $O= \denotation{\termA_{\text{orb}}}^{\StructA}(\tup{a})$
be the $1$-orbit given by $\termA_{\text{orb}}$.
In particular, this orbit is disjoint with~$\tup{a}$.
We define
\[\canlabels{\StructA}{\tup{a}} := \bigcup_{\vertA \in O} \canlabels{\StructA}{\tup{a}\vertA}.\]
For $\tup{b} \in \canlabels{\StructA}{\tup{a}}$,
let $\auto_{\tup{b}} \colon \StructVA \to [|\StructVA|]$
be defined via $\vertA \mapsto i$ if and only if $\vertA = b_i$. 
It is easy to see that $\canlabels{\StructA}{\tup{a}}$ is an $(\StructA,\tup{a})$-orbit.
Hence, the definition
\[\canon{\StructA}{\tup{a}} := \autoA_{\tup{b}}((\StructA,\tup{a}))\]
is well-defined and independent of the choice of $\tup{b} \in \canlabels{\StructA}{\tup{a}}$ because  $\autoA_{\tup{b}}((\StructA,\tup{a})) = \autoA_{\tup{c}}((\StructA,\tup{a}))$ for every $\tup{b}, \tup{c} \in \canlabels{\StructA}{\tup{a}}$.

\begin{lemma}
	\label{lem:canon-is-canonization}
	For every $\StructA,\StructB \in \GraphClass$,
	every $\tup{a} \in \HF{\StructA}^*$, and
	every $\tup{b} \in \HF{\StructB}^*$, it holds that
	$\canon{\StructA}{\tup{a}} \iso (\StructA,\tup{a})$
	and
	$\canon{\StructA}{\tup{a}} = \canon{\StructB}{\tup{b}}$
	if and only if $(\StructA,\tup{a}) \iso (\StructB, \tup{b})$.
\end{lemma}
\begin{proof}
	Let $\StructA,\StructB \in \GraphClass$,
	$\tup{a} \in \HF{\StructA}^*$, and
	$\tup{b} \in \HF{\StructB}^*$.
	First, because the $\autoA_{\tup{c}}$
	are bijections for every $\tup{c} \in \canlabels{\StructA}{\tup{a}}$,
	it follows that $\canon{\StructA}{\tup{a}} \iso (\StructA,\tup{a})$.
	Second, because the order on the $1$-orbits defined by
	$\termA_{\text{orb}}$
	is isomorphism-invariant,
	$\canon{\StructA}{\tup{a}} = \canon{\StructB}{\tup{b}}$
	if and only if $(\StructA,\tup{a}) \iso (\StructB, \tup{b})$.
	This is essentially the argument why Gurevich's canonization algorithm
	in~\cite{Gurevich97} is correct.
\end{proof}

\begin{lemma}
	\label{lem:definable-orbits-implies-canonization}
	If a class of $\sig$-structures~$\GraphClass$ is ready for individualization in CPT+WSC,
	then CPT+WSC defines a canonization for~$\GraphClass$-structures.
\end{lemma}
\begin{proof}
	To implement the former approach in CPT+WSC,
	we first introduce some notation:
	we define a fixed-point operator with deterministic choice
	similar to the WSC-fixed-point operator,
	but which can be simulated just with plain iteration terms:
	\[\mathsf{DC}^*xy.\qspace(\termStep,\termChoice,\termA_{\text{order}}),\]
	where~$\termStep$,~$\termChoice$, and~$\termA_{\text{order}}$ are CPT+WSC-terms.
	The first terms $\termStep(x,y)$ and $\termChoice(x)$ behave exactly as in the 
	symmetric choice operator:~%
	$\termStep$ defines a step function and~$\termChoice$ a choice set (of atoms).
	But the third term~$\termA_{\text{order}}$ defines a total order on the atoms
	and is used to resolve the choices deterministically
	by picking the minimal one.
	The operator evaluates to the fixed-point obtained in that manner
	or to~$\emptyset$ if the polynomial bound is exceeded.
	
	We define a CPT+WSC-interpretation $\interpret(\indVar)$ whose universe is $[|\StructV|]$.
	The total order~$\leq$ is just the natural order on $[|\StructV|]$.
	For every $k$-ary relation $\rel \in \sig$,
	we define a
	formula $\formB_R(\indVar,i_1,\dots, i_k)$ as follows:
	\begin{align*}
		\termB_{\text{label}}(\ordVar,\indVar) &:= \mathsf{DC}^*xy.\qspace\big(\concat{\indVar}{\concat{x}{y}},\termA_{\text{orb}}(\concat{\indVar}{x}),\ordVar\big),\\
		\termB_{\text{wit}}(\ordVar,\indVar) &:= \setcond[\Big]{\setcond[\big]{( \termB_{\text{label}}(\ordVar, \concat{\indVar}{x})_j, \termB_{\text{label}}(\ordVar, \concat{\indVar}{y})_j)}{j \in [\Card(\Atoms)]}}{x,y \in \termA_{\text{orb}}(\indVar)},\\
		\formB_R(\indVar, i_1,\dots,i_k) &:= \wscFormbig{x}{y}{\concat{\indVar}{\concat{x}{y}}}{\termA_{\text{orb}}(\concat{\indVar}{x})}{\termB_{\text{wit}}(y,x)}{R(x_{i_1},\dots,x_{i_k})}.
	\end{align*}
	Fix an arbitrary structure $\StructA \in \GraphClass$.
	Set~$\StructVA^<$ to be the set of all $\StructVA$-tuples of length
	$|\StructVA|$ containing all atoms exactly once
	(i.e., the set of all total orders on~$\StructVA$).
	Additionally, fix an arbitrary $\tup{a} \in \StructVA^*$.
	\begin{claim}\label{clm:labels-correct}
		If $\tup{b} \in \StructVA^<$,
		then 
		$\denotation{\termB_{\text{label}}(\ordVar,\indVar)}^\Struct(\tup{b},\tup{a})
		\in \canlabels{\StructA}{\tup{a}}$.
	\end{claim}
	\begin{claimproof}
		Recall here that the tuple operation $\concat{\indVar}{\concat{x}{y}}$
		discards duplicates from~$x$ and~$y$.
		That is, in the first iteration
		we add~$\tup{a}$ and another atom to~$x$ 
		(or just~$\tup{a}$, if already all atoms are individualized).
		In every iteration, we choose 
		the minimal atom according to~$\tup{b}$
		from the orbit given by~$\termA_{\text{orb}}$
		until all atoms are individualized.
		Then nothing is added to~$x$ anymore and a fixed-point is reached.
		This follows exactly the definition of $\canlabels{\StructA}{\tup{a}}$.
	\end{claimproof}
	
	\begin{claim}\label{clm:aut-correct}
		If $\tup{b} \in \StructVA^<$,
		then 
		$\denotation{\termB_{\text{wit}}(\ordVar,\indVar)}^\Struct(\tup{b},\tup{a})$
		witnesses that 
		$\denotation{\termA_{\text{orb}}(\indVar)}^\Struct(\tup{a})$
		is an orbit of $(\StructA,\tup{a})$.
	\end{claim}
	\begin{claimproof}
		Let $\tup{b} \in \StructVA^<$,
		$O = \denotation{\termA_{\text{orb}}(\indVar)}^\Struct(\tup{a})$,
		and $\vertA,\vertB \in O$.
		Furthermore, let $\tup{c} \in \canlabels{\StructA}{\tup{a}\vertA}$
		and 
		$\tup{d} \in \canlabels{\StructA}{\tup{a}\vertB}$.
		Then $\autoA_{\tup{c}}$ is an isomorphism
		$(\StructA,\tup{a}\vertA) \to \canon{\StructA}{\tup{a}\vertA}$
		and also an isomorphism
		$(\StructA,\tup{a}) \to \canon{\StructA}{\tup{a}}$
		because~$O$ is the orbit given by~$\termA_{\text{orb}}$.
		Likewise,~$\autoA_{\tup{d}}$ is an isomorphism
		$(\StructA,\tup{a}) \to \canon{\StructA}{\tup{a}}$.
		Then $\inv{\autoA_{\tup{d}}}\circ\autoA_{\tup{c}} = \setcond{(c_j,d_j)}{j \in |\StructVA|} \in \autGroup{(\StructA,\tup{a})}$
		and satisfies $(\inv{\autoA_{\tup{d}}}\circ\autoA_{\tup{c}})(\vertA)=\vertB$.
		It easy to see using Claim~\ref{clm:labels-correct} that~%
		$\termB_{\text{wit}}$ exactly defines such automorphisms
		for every pair of atoms in the orbit given by~$\termA_{\text{orb}}$
		(see Figure~\ref{fig:witnessing-automorphisms-gurevich-example} for an illustration of witnessing automorphisms for the example
		in Figure~\ref{fig:gurevich-example}).
	\end{claimproof}

	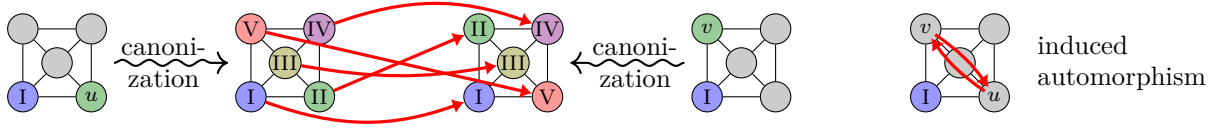
\begin{figure}
		\begin{tikzpicture}
			\begin{scope} [scale=0.5]
			
			\begin{scope}[scale=1, shift={(1,0)}] 		
				\ggraph{\vertexColor}{}{\figcolA}{I}{\figcolB}{$u$}{\vertexColor}{}{\vertexColor}{};
			\end{scope}

			\path[decorate, draw, -{To[width=3mm,length=1.5mm]}, black, thick, decoration={snake,amplitude=0.4mm,post length = 0.75mm}]
			(2.5, 0) to  (5.5, 0);
			\node[font=\footnotesize, align=center] (discretelabel) at(3.8,0) {canoni-\\zation};
			
			\begin{scope}[scale=1, shift={(7,0)}, name prefix =B-] 		
				\ggraph{\figcolC}{III}{\figcolA}{I}{\figcolB}{II}{\figcolD}{IV}{\figcolE}{V};
			\end{scope}

			\begin{scope}[scale=1, shift={(13,0)}, name prefix =A-] 		
				\ggraph{\figcolC}{III}{\figcolA}{I}{\figcolE}{V}{\figcolD}{IV}{\figcolB}{II};
			\end{scope}
			
			\path[decorate, draw, -{To[width=3mm,length=1.5mm]}, black, thick, decoration={snake,amplitude=0.4mm, post length = 0.75mm}]
			(17.5, 0) to  (14.5, 0);
			\node[font=\footnotesize, align=center] (discretelabel) at(16.2,0) {canoni-\\zation};
			
			\begin{scope}[scale=1, shift={(19,0)}] 		
				\ggraph{\vertexColor}{}{\figcolA}{I}{\vertexColor}{}{\vertexColor}{}{\figcolB}{$v$};
			\end{scope}
			
			\path[-{Latex[width=2mm,length=1.5mm]}, draw, red, very thick] (B-V1) edge [bend right = 10] (A-V1)
			(B-V3) edge  (A-V5)
			(B-V2) edge [bend right= 20] (A-V2)
			(B-V4) edge [bend left =20]  (A-V4)
			(B-V5) edge  (A-V3);

			\node [right, align=left, font=\footnotesize](autolabel) at (26.5,0) {induced\\automorphism};
			
			\begin{scope}[scale=1, shift={(24.75,0)}] 		
				\ggraph{\vertexColor}{}{\figcolA}{I}{\vertexColor}{$u$}{\vertexColor}{}{\vertexColor}{$v$};
				\path[red, very thick, -{Latex[width=2mm,length=1.5mm]}]
				(-0.65,0.75) 
				edge[bend left = 10] (0.75,-0.65)
				(0.65,-0.75) 
				edge[bend left = 10] (-0.75,0.65);
			\end{scope}
			\end{scope}
			
		\end{tikzpicture}
	\caption{Witnessing automorphisms for the example in Figure~\ref{fig:gurevich-example}:
		The goal is to witness that $\set{\vertA,\vertB}$ is an orbit
		in the graph with one individualized vertex (I).
		Gurevich's canonization algorithm is used to 
		canonize the two graphs, in which~$\vertA$ respectively~$\vertB$ is individualized
		(formally, we use the already defined order to resolve choices in this instance of the algorithm).
		The obtained total orders induce the witnessing automorphism.
	}
	\label{fig:witnessing-automorphisms-gurevich-example}
	\end{figure}
	
	\begin{claim}
		\label{clm:wsc-correct}
		$\wsc{\denotation{\concat{\indVar}{\concat{y}{z}}}^\Struct(\tup{a})}{\denotation{\termA_{\text{orb}}(\concat{\indVar}{y})}^\Struct(\tup{a})}{\denotation{\termB_{\text{wit}}}^\Struct(\tup{a})} = \canlabels{\Struct}{\tup{a}}$.
	\end{claim}
	\begin{claimproof}
		As in Claim~\ref{clm:labels-correct},
		the WSC-fixed-point operator
		expresses precisely the definition of  $\canlabels{\StructA}{\tup{a}}$.
		By Claim~\ref{clm:aut-correct}, all choices are witnessed
		because every automorphism stabilizing some tuple $\tup{c} \in \StructA^*$
		stabilizes all prefixes of~$\tup{c}$. It thus also stabilizes all tuples defined earlier
		in the iteration.
		Because we consider the result under all possible choices,
		the claim follows.
	\end{claimproof}
	Finally, we show that $\reduct{\interpret(\StructA,\tup{a})}{\sig} = \canon{\StructA}{\tup{a}}$.
	By Claim~\ref{clm:wsc-correct},
	$(\tup{a}, i_1,\dots, i_\ell) \in \denotation{\formB_R}^\Struct$ if
	and only if $(i_1,\dots, i_\ell) \in R^{\auto_{\tup{b}}(\StructA,\tup{a})}$
	for some (and thus every) $\tup{b} \in \canlabels{\StructA}{\tup{a}}$.
	This is exactly the definition of $\canon{\StructA}{\tup{a}}$
	and hence we obtained a CPT+WSC-definable canonization
	(Lemma~\ref{lem:canon-is-canonization}).
\end{proof}

\begin{lemma}[\cite{GroheSchweitzerWiebking2021}]
	\label{cpt-iso-implies-inv}
	If CPT defines isomorphism of a class of binary $\sig$-structures~$\GraphClass$
	(closed under individualization),
	then there is a CPT-term defining a complete invariant for~$\GraphClass$.
\end{lemma}
While this lemma is only for binary structures,
it can also be applied to arbitrary structures.
Every $\sig$-structure can be encoded by a binary structure using a CPT-interpretation~$\interpret$
(in fact, an FO-interpretation suffices)
such that $\interpret(\StructA) \iso \interpret(\StructB)$
if and only if $\StructA \iso \StructB$
and given a definable isomorphism test for a class of $\sig$-structures~$\GraphClass$,
we can define an isomorphism test on $\interpret(\GraphClass)$ and vice versa.

\begin{corollary}
	If CPT defines isomorphism of a class of $\sig$-structures~$\GraphClass$,
	then CPT+WSC defines canonization of $\GraphClass$-structures
	and captures \PTime{} on $\GraphClass$-structures.
\end{corollary}
This corollary is asymmetric in the sense
that we turn an isomorphism-defining CPT-formula
into a canonization-defining CPT+WSC-formula.
The next goal is to prove the symmetric version,
which starts with an isomorphism-defining CPT+WSC-formula (rather than a CPT-formula).
We begin with the following theorem
(which, similarly to Lemma~\ref{cpt-iso-implies-inv} can also be used for non-binary structures).
\begin{theorem}
	\label{thm:cpt+wsc-iso-implies-inv}
	If CPT+WSC defines isomorphism of a class of binary $\sig$-structures~$\GraphClass$,
	then CPT+WSC defines a complete invariant for $\GraphClass$-structures.
\end{theorem}
The long proof of this theorem is deferred to  Section~\ref{sec:iso-in-cpt+wsc}.
Assuming Theorem~\ref{thm:cpt+wsc-iso-implies-inv} for now, we conclude:

\begin{theorem}
	\label{thm:cpt+wsc-all-equiv}
	Let~$\GraphClass$ be a class of $\sig$-structures (closed under individualization).
	The following are equivalent:
	\begin{enumerate}
		\item \label{alleq-ready} $\GraphClass$ is ready for individualization in CPT+WSC.
		\item \label{alleq-one-orbit} $\GraphClass$ has CPT+WSC-distinguishable $1$-orbits.
		\item \label{alleq-k-orbit} $\GraphClass$ has CPT+WSC-distinguishable $k$-orbits for every $k \in \nat$.
		\item \label{alleq-iso} CPT+WSC defines isomorphism of~$\GraphClass$.
		\item \label{alleq-inv} CPT+WSC defines a complete invariant for $\GraphClass$.
		\item \label{alleq-canon} CPT+WSC defines a canonization for $\GraphClass$.
	\end{enumerate}
\end{theorem}
\begin{proof}	
	We show \ref{alleq-iso}$\Rightarrow$\ref{alleq-inv}$\Rightarrow$\ref{alleq-k-orbit}$\Rightarrow$\ref{alleq-one-orbit}$\Rightarrow$\ref{alleq-ready}$\Rightarrow$\ref{alleq-canon}$\Rightarrow$\ref{alleq-iso}.
	
	Theorem~\ref{thm:cpt+wsc-iso-implies-inv} proves \ref{alleq-iso}$\Rightarrow$\ref{alleq-inv},
	Lemma~\ref{lem:complete-invariant-implies-definable-orbits} proves \ref{alleq-inv}$\Rightarrow$\ref{alleq-k-orbit},
	\ref{alleq-k-orbit}$\Rightarrow$\ref{alleq-one-orbit} is trivial, and
	Lemma~\ref{lem:definable-orbits-implies-canonization} proves
	\ref{alleq-ready}$\Rightarrow$\ref{alleq-canon}.
	To show \ref{alleq-one-orbit}$\Rightarrow$\ref{alleq-ready},
	one can pick the minimal orbit according to the given preorder
	satisfying the requirement of ready for individualization,
	and finally 
	\ref{alleq-canon}$\Rightarrow$\ref{alleq-iso}
	is done by comparing the two canons
	of the two structures given as the disjoint union.
	This is done as follows:
	Let~$\termA_{\text{canon}}$ be a closed CPT+WSC-term defining a canonization,
	that is, it evaluates the interpretation defining the ordered copy
	as an $\HF{\emptyset}$-set using numbers as atoms.
	Further, let~$\termA_{\text{cc}}$ be a CPT+WSC-term defining the set of the two connected components of the disjoint union,
	i.e., the set of the two universes of the structures to test for isomorphism.
	To evaluate~$\termA_{\text{canon}}$ on a single component of the disjoint union,
	we need to forbid automorphisms exchanging the components (in the case that they are isomorphic),
	so~$\termA_{\text{canon}}$ indeed can ignore one component and just canonize the other.
	The idea is to add a free variable~$x$ to~$\termA_{\text{canon}}$
	which will hold a set of atoms forming a component.
	Let $\termA'_{\text{canon}}(x)$ be the CPT+WSC-term with a free variable~$x$ (unused in~$\termA_{\text{canon}}$) obtained from~$\termA_{\text{canon}}$ in the following way:
	we replace every occurrence of~$\Atoms$ in~$\termA_{\text{canon}}$ with~$x$
	and in every WSC-fixed-point operator
	we add the tautology $x = x$ in the step, choice, and witnessing term.
	This way, every WSC-fixed-point operator
	has~$x$ as free variable and in particular all witnessing automorphisms need to stabilize the component held by~$x$.
	Then the formula
	$\Card(\setcond{\termA'_{\text{canon}}(x)}{x \in \termA_{\text{cc}}}) = 1$
	is satisfied if and only if both components of the disjoint union are isomorphic.
\end{proof}

Finally, we can prove Theorem~\ref{thm:iso-implies-canon-and-ptime}.
\begin{proof}[Proof of Theorem~\ref{thm:iso-implies-canon-and-ptime}]
	Let~$\GraphClass$ be a class of $\sig$-structures,
	for which CPT+WSC defines isomorphism.
Canonization of~$\GraphClass$ is CPT+WSC-definable by Theorem~\ref{thm:cpt+wsc-all-equiv} for~$\GraphClass$
and so by the Immerman-Vardi Theorem~\cite{Immerman87}
CPT+WSC captures \PTime{}.
\end{proof}

\begin{corollary}
	If graph isomorphism is in \PTime{},
	then CPT+WSC defines isomorphism on all structures
	if and only if CPT+WSC captures \PTime{}.
\end{corollary}

\section{Isomorphism Testing in CPT+WSC}
\label{sec:iso-in-cpt+wsc}
The goal of this section is to prove Theorem~\ref{thm:cpt+wsc-iso-implies-inv},
which states that a CPT+WSC-definable isomorphism tests
implies a CPT+WSC-definable complete invariant.
The proof of Lemma~\ref{cpt-iso-implies-inv} in~\cite{GroheSchweitzerWiebking2021},
which proves the same statement for CPT,
uses the equivalence between CPT and the DeepWL computation model.
This model ensures that a Turing machine can only access and modify a relational structure in an isomorphism-invariant way.
For DeepWL, the authors of~\cite{GroheSchweitzerWiebking2021} show that on input $\StructA\disunion\StructB$
a DeepWL-algorithm never needs to ``mix'' atoms of the two structures.
This implies that if there is a DeepWL-algorithm to decide isomorphism,
it essentially suffices not to compute on input $\StructA \disunion \StructB$
but to consider the output of another algorithm on input~$\StructA$
and on input~$\StructB$. 
The run of the Turing machine in the latter DeepWL-algorithm turns out 
to be a complete invariant if the DeepWL-algorithm decides isomorphism.
To prove Theorem~\ref{thm:cpt+wsc-iso-implies-inv},
we extend DeepWL with witnessed symmetric choice
and essentially follow the same proof idea albeit with necessary adaptions.
For a more elaborate introduction into DeepWL we refer to~\cite{GroheSchweitzerWiebking2021}.

In the rest of this section, we assume that all structures are binary relational structures.
Moreover, we see all relation symbols as binary strings,
so Turing machines with a fixed alphabet can compute with relation symbols.

\paragraph{Coherent Configurations.}

Before introducing DeepWL,
we need some background on coherent configuration.
We start with introducing terminology for general binary relational structures,
which can be seen as edge-colored graphs.
Let~$\StructA$ be a binary $\sigA$-structure.
The inverse of a relation~$\rel^\Struct$ (for some $\rel \in \sig$) is
\[\inv{(\rel^\Struct)} := \setcond*{(\vertB,\vertA)}{(\vertA,\vertB) \in \rel^\Struct}.\]
We call~$\rel$ \defining{undirected} if $\rel^\Struct = \inv{(\rel^\Struct)}$
and \defining{directed} otherwise.
We use $\set{\vertA,\vertB} \in \rel^\Struct$
as notation for $\set{(\vertA,\vertB),(\vertB,\vertA)} \subseteq \rel^\Struct$.
For $\pi \subseteq \sigA$,
two atoms $\vertA, \vertB \in \StructV$
are \defining{$\pi$-connected}
if there is a path from~$\vertA$ to~$\vertB$ only using edges
contained in a relation in~$\pi$.
Similarly, we define $\pi$-connected components
and strongly $\pi$-connected components ($\pi$-SCCs).

We now turn to coherent configurations.
Let~$\StructC$ be a binary $\sigB$-structure.
A relation $\col \in \sigB$ is called \defining{diagonal}, if 
$\col^\StructC\subseteq \diagrel{\StructVC}  := \setcond{(\vertA,\vertA)}{\vertA \in \StructVC}$.
The structure~$\StructC$ is a \defining{coherent configuration}
if it satisfies the following properties:
\begin{enumerate}
	\item The $\sigB$-relations partition~$\StructVC^2$, that is, $\setcond{\col^\StructC}{\col \in \sigB}$ is a partition of~$\StructVC^2$.
	In particular all relations~$\col^\StructC$ are nonempty.
	\item Every relation $\col \in \sigB$ is either disjoint from or a subset of $\diagrel{\StructVC}$.
	\item Every relation $\col \in \sigB$ has an inverse $\inv{\col} \in \sigB$,
	i.e., $\inv{(\col^\StructC)} = (\inv{\col})^\StructC$.
	\item For every triple $(\colA, \colB, \colC) \in \sigB$,
	 there is a number $q(\colA,\colB,\colC) \in \nat$ such that whenever
	$(\vertA,\vertB) \in \colA^\StructC$,
	there are exactly $q(\colA,\colB,\colC)$
	many $\vertC \in \StructVC$ such that $(\vertA,\vertC) \in \colB^\StructC$
	and $(\vertC,\vertB) \in \colC^\StructC$.
\end{enumerate}
The number $q(\colA,\colB,\colC)$ is called the \defining{intersection number} of $(\colA,\colB,\colC)$.
The function $q \colon \sigB^3 \to \nat$ given by $(\colA,\colB,\colC) \mapsto q(\colA,\colB,\colC)$ is called the \defining{intersection function} of~$\StructC$.
The $\sigB$-relations are called \defining{colors}.
Diagonal $\sigB$-relations are called \defining{fibers}.
We say that a color $\colA\in \sigB$ \defining{has an $(\colB, \colC)$-colored triangle},
if $q(\colA,\colB,\colC) \geq 1$.
This can be extended to paths.
We say that a color~$\colA$ \defining{has an $(\colB_1, \dots, \colB_k)$-colored path},
if for every $(\vertA,\vertB) \in \colA^\Struct$
there is a path $(\vertC_1,\dots,\vertC_{k+1})$ such that
$\vertC_1 = \vertA$, $\vertC_{k+1} = \vertB$,
and for every $i \in [k]$ it holds that $(\vertC_i, \vertC_{i+1}) \in \colB_i^\StructC$.
By the properties of coherent configuration, this is either the case for every $(\vertA,\vertB) \in \colA^\Struct$ or for no such edge.
We also say that a color has~$k$ many colored triangles or colored paths,
if we want to specify the exact number of these triangles or paths.

The coherent configuration~$\StructC$ \defining{refines} a $\sig$-structure~$\StructA$
if $\StructVC = \StructVA$ and 
for every $\sigB$-relation~$\colA$
and every $\sigA$-relation~$\relA$
it holds that either $\colA^\StructC \subseteq \relA^\StructA$
or $\colA^\StructC \cap \relA^\StructA = \emptyset$.
A coherent configuration~$\StructC$ refining a structure~$\StructA$
is a \defining{coarsest} coherent configuration refining~$\StructA$
if every coherent configuration~$\StructC'$ refining~$\StructA$
also refines~$\StructC$.
Given a $\sigA$-structure~$\StructA$,
a coarsest coherent configuration refining~$\StructA$ can be computed canonically
with the two-dimensional Weisfeiler-Leman algorithm.
We denote this configuration by~$\coConf{\StructA}$.

\paragraph{Structures with Sets as Vertices.}
In the following, relational structures
in which some ``atoms'' are obtained as HF-sets of other atoms
play an important role.
We formalize this as follows:

For a signature $\sig = \set{\rel_1,\dots, \rel_k}$,
a finite binary \defining{$\sig$-HF-structure}~$\Struct$ is a tuple $(\StructV, \StructHFV, \rel_1^\Struct, \dots , \rel_k^\Struct)$,
where
$\StructV$ is a finite set of atoms,
$\StructHFV \in \HF{\StructV} \setminus \StructV$
is a finite set of $\HF{\StructV}$-sets,
and $\rel_i^\Struct \subseteq (\StructV\cup \StructHFV)^2$ for all $i \in [k]$,
that is, the universe of~$\Struct$ is a set of atoms~$\StructV$ and some $\HF{\StructV}$-sets~$\StructHFV$.
We call~$\StructV$ \defining{atoms}
and $\vertices{\Struct} := \StructV \cup \StructHFV$ \defining{vertices}.
In that sense, every $\sig$-HF-structure~$\Struct$
can be turned into a $\sig$-structure~$\nonHF{\Struct}$,
where the sets in~$\StructHFV$ become fresh atoms. Conversely, every $\sig$-structure is also a $\sig$-HF-structure,
where the set~$\StructHFV$ is empty.

An automorphism of the $\sig$-HF-structure~$\Struct$
is a permutation~$\autoA$ of the atoms~$\StructV$
such that $\autoA(\StructHFV) = \StructHFV$ and $(\vertA,\vertB) \in \rel_i^\Struct$ if and only if
$\autoA(\vertA, \vertB) \in \rel_i^\Struct$
for every $i \in [k]$ and every $\vertA,\vertB \in \vertices{\Struct}$.
That is, an automorphism of~$\Struct$
has to respect the HF-structure of the vertices.
So a $\sig$-HF-structure~$\Struct$ has potentially fewer automorphisms
than the $\sig$-structure~$\nonHF{\Struct}$.
Using this notion of automorphisms,
$\Struct$-orbits and $(\Struct, \tup{a})$-orbits (for a tuple $\tup{a}$ of $\HF{\StructV}$-sets) are defined as before.

The disjoint union of two HF-structures~$\StructA$ and~$\StructB$
is the structure $\StructA \disunion \StructB$ with atom set
$\StructVA \disunion \StructVB$ that is defined as expected.
For an HF-structure~$\StructA$ and a set $M \subseteq \vertices{\StructA}$
such that~$M \subseteq \HF{\StructVA\cap M}$,
that is,~$M$ contains only HF-sets formed over atoms contained in~$M$,
the substructure of~$\StructA$ induced by~$M$ is denoted~$\StructA[M]$.

We define $\coConf{\Struct} := \coConf{\nonHF{\Struct}}$, that is, coherent configurations are always computed with respect to~$\nonHF{\Struct}$.

\subsection{DeepWL}
\label{sec:deepwl}

We are going to introduce the notion of a DeepWL-algorithm from~\cite{GroheSchweitzerWiebking2021}:
A \defining{DeepWL-algorithm} is a two-tape Turing machine
using the alphabet $\set{0,1}$
with three special states~%
$q_{\addPairSym}$,~$q_{\sccSym}$, and~$q_{\createSym}$.
The first tape is called the \defining{work-tape}
and the second one the \defining{interaction-tape}.
The Turing machine computes on a binary relational $\sig$-HF-structure~$\StructA$,
but it has no direct access to it.
Instead, the structure is put in the so-called ``cloud''
which maintains the pair $(\Struct, \coConf{\Struct})$.
The Turing machine only has access to the \defining{algebraic sketch} $\sketch{\StructA} = (\sigA, \sigB, \symSubset{\sigA}{\sigB}, q)$,
which gets written on the interaction-tape and consists of the following objects:
\begin{enumerate}
	\item $\sigA$ is the signature of the HF-structure~$\StructA$.
	\item $\sigB$ is the signature of the canonical coarsest coherent configuration $\coConf{\StructA}$ refining~$\nonHF{\StructA}$.
	\item  $\symSubset{\sigA}{\sigB} := \setcond{(\col, \rel) \in \sigA \times \sigB}{\col^{\coConf{\StructA}} \subseteq \rel^\StructA}$
	is the \defining{symbolic subset relation}.
	It relates a $\sigB$-color~$\col$ to the $\sigA$-relation~$\rel$ which is refined by~$\sigB$, i.e.,~$\col^{\coConf{\StructA}} \subseteq \rel^\StructA$.
	\item $q$ is the intersection function of~$\coConf{\Struct}$.
\end{enumerate}
In the following and unless stated otherwise, 
we use~$\sigA$ for the signature of the HF-structure~$\Struct$ in the cloud 
and~$\sigB$ for the signature of~$\coConf{\Struct}$,
which we assume to be disjoint from~$\sigA$.
We call relations $\col \in \sigB$ \defining{colors}
and relations $\rel \in \sigA$ just \defining{relations}.
If~$\rel$ (respectively~$\col$) is a diagonal relation,
we identify~$\rel$ (or~$\col$) with the set $\setcond{\vertA}{(\vertA,\vertA)\in \rel^\Struct}$
and call~$\rel$ a \defining{vertex class} (or~$\col$ a \defining{fiber}).
We use the letters~$\colA$,~$\colB$, and~$\colC$ for colors and
the letters~$\relA$ and~$\relB$ for relations.
We use the letters~$\vcA$ and~$\vcB$ for vertex classes and
the letters~$\ccA$ and~$\ccB$ for fibers.
Although the cloud contains the pair  $(\Struct, \coConf{\Struct})$,
we will just say that~$\Struct$ is in the cloud
and interpret $\sigB$-colors~$\col$ in~$\Struct$,
i.e.,~just write $\col^\Struct$ for $\col^{\coConf{\Struct}}$.

With the special states~$q_{\addPairSym}$,~$q_{\sccSym}$, and~$q_{\createSym}$,
the Turing machine can add vertices to the structure in the cloud
in an isomorphism-invariant manner.
If~$\Struct$ is the input structure to the DeepWL-algorithm,
the vertices of the HF-structure in the cloud
will be pairs $\pairVtx{a}{i}$ of an $\HF{\StructV}$-set (or atom)~$a \in \HF{\StructV}$
and a number~$i \in \nat$.
The number~$i$ is encoded as an $\HF{\emptyset}$-set and $\pairVtx{a}{i}$ denotes the Kuratowski encoding of pairs%
\footnote{While in the previous sections we only used the pair encoding implicitly, here we use the explicit $\pairVtx{\cdot}{\cdot}$ notation for sake of readability.}).
Using the number~$i$,
we can create multiple vertices for the same~$a\in \HF{\StructV}$ as follows.
Whenever~$i$ many vertices for the set~$a$ exist (possible zero many),
then we add the vertex  $\pairVtx{a}{i+1}$.
So, when describing how vertices are added, we can identify them with its $\HF{\StructV}$-set and assume that the numbers are picked as described.

Now assume that~$\StructA$ is the $\sigA$-HF-structure in the cloud
at some point during the execution of the DeepWL-algorithm.
To enter the states~$q_{\addPairSym}$ and~$q_{\sccSym}$,
the Turing machine has to write a single relation symbol $\relColA \in \sigA \cup \sigB$ on the interaction-tape.
To enter~$q_{\createSym}$, a set $\pi \subseteq \sigB$ has to be written on the interaction-tape.
We say that the machine \defining{executes} $\addPair{\relColA}$, $\scc{\relColA}$, and $\create{\pi}$.
\begin{enumerate}[label=\alph*)]
	\item $\addPair{\relColA}$:
	For every $(\vertA,\vertB) \in \relColA^\Struct$
	a new vertex $\pairVtx{\vertA}{\vertB}$ is added to the structure
	(by the former convention, actually a vertex $\pairVtx{\pairVtx{\vertA}{\vertB}}{i}$ is added).
	Additionally, new relations~$\rel_{\leftT}$ and~$\rel_{\rightT}$
	are added to~$\sigA$ containing the pairs
	$(\pairVtx{\vertA}{\vertB},\vertA)$
	and $(\pairVtx{\vertA}{\vertB},\vertB)$ respectively.
	We call these relations the \defining{component relations}.
	\item $\scc{\relColA}$: 
	For every strongly $\relColA$-connected component~$c$, a new vertex~$c$ is added
	(note that~$c$ is itself an $\HF{\StructV}$-set).
	A new \defining{membership relation} symbol~$\rel_{\memT}$
	is added to~$\sigA$ containing the pairs
	$(c, \vertA)$ for every $\relColA$-SCC~$c$ and every $\vertA \in c$.
	\item $\create{\pi}$: A new relation symbol~$\rel$ is added to~$\sigA$,
	which is interpreted as the union of all $\col \in \pi$.
\end{enumerate} 
Whenever new relation symbols have to be picked, we choose the smallest unused one
according to the lexicographical order (recall that relation symbols are binary strings).
Each of these three operations modify the HF-structure~$\StructA$ in the cloud.
After that, the coherent configuration~$\coConf{\StructA}$
is recomputed and the new algebraic sketch~$\sketch{\StructA}$
is written onto the interaction-tape.
Then the Turing machine continues.
A DeepWL-algorithm accepts~$\StructA$,
if the head of the work-tape points to a~$1$
when the Turing machine halts and rejects otherwise.
For a more detailed definition and description of a DeepWL-algorithm we refer to~\cite{GroheSchweitzerWiebking2021}.

\paragraph{Differences and Equivalence to~\cite{GroheSchweitzerWiebking2021}.}
Our definition of a DeepWL-algorithm differs at various places from the one given in~\cite{GroheSchweitzerWiebking2021}, which we discuss now:

	 We omit the $\forgetSym$-operation, which allows the machine to remove a $\sigA$-relation from the structure in the cloud.
	 However,~\cite{GroheSchweitzerWiebking2021} proves that this operation is not needed
	 because the algebraic sketch $\sketch{\Struct'}$ of the structure~$\Struct'$ obtained from~$\Struct$ by removing some relations
	 can be computed from the sketch $\sketch{\Struct}$:
	\begin{lemma}[{\cite[Lemma~22]{GroheSchweitzerWiebking2021}}]
		\label{lem:restrict-sketch}
		There is a polynomial-time algorithm that for every algebraic sketch $\sketch{\StructA}$ of a binary structure~$\StructA$,
		every subset $\tilde{\sigA}\subseteq \sigA$,
		and every vertex class $\vcA \in \tau$,
		computes the algebraic sketch $\sketch{\reduct{\StructA[\vcA^\Struct]}{\tilde{\sigA}}}$.
	\end{lemma}	
	Note that the algorithm from the lemma has no access to~$\StructA$ but solely to $\sketch{\StructA}$ and so can be executed by a
	DeepWL-algorithm
	without modifying the cloud.
	So it suffices to remember the set of relations to remove.
	In particular, for every DeepWL-algorithm deciding isomorphism 
	the equivalent DeepWL-algorithm constructed in Theorem~11 of~\cite{GroheSchweitzerWiebking2021}
	does not use the $\forgetSym$-operation.
	
	Instead of the $\sccSym$-operation a $\contractSym$-operation is given in~\cite{GroheSchweitzerWiebking2021}. The $\contractSym$-operation contracts the vertices of the SCCs and does not create a new one for each SCC.
	A structure~$\Struct$ resulting from a $\contract{\relColA}$ operation
	can be obtained from the structure~$\Struct'$ resulting from the $\scc{\relColA}$-operation by removing the
	vertices incident to~$\relColA$.
	While our operations do not allow removing these vertices,
	they form a union of fibers
	and the algebraic sketch $\sketch{\Struct}$ can be computed
	from $\sketch{\Struct'}$ by Lemma~\ref{lem:restrict-sketch}.
	For polynomial-time DeepWL-algorithms,
	not removing these vertices only creates a polynomial overhead.
	The other way around, the $\sccSym$-operation can be simulated by first
	copying the vertices incident to $\relColA$ (using an $\addPairSym$-operation for the fibers incident to~$\relColA$)
	and then executing a $\contractSym$-operation on the copies.
	
	Our operations produce new component or membership relations,
	while in~\cite{GroheSchweitzerWiebking2021} a global component relation is maintained but the added vertices are put in a new vertex class.
	Surely, one is computable from the other.
	
	Our version of DeepWL computes on HF-structures,
	while~\cite{GroheSchweitzerWiebking2021} uses plain relational structures
	and new vertices are just added as fresh atoms.
	For now, this does not make a difference
	because the coherent configuration~$\coConf{\Struct} = \coConf{\nonHF{\Struct}}$ is computed on $\nonHF{\Struct}$.
	But seeing the vertices as hereditarily finite sets
	possibly removes automorphisms of the structure.
	However, we will now show that with our altered operations
	the automorphisms of the cloud as HF-structure
	coincide with the automorphisms of the cloud as plain relational structure.
	That is one major reason why we use different operations
	when dealing with symmetric choice:
	algebraic sketches can be computed with respect to plain relational structures
	while still maintaining automorphisms of the HF-structure.
	This is crucial later when we will simulate an extension of DeepWL with witnessed symmetric choice in CPT+WSC.

\paragraph{Automorphisms and HF-Structures.}
The following lemma justifies to compute $\coConf{\Struct}$
on the structure~$\nonHF{\Struct}$
and that DeepWL does not need to access the HF-structure of~$\Struct$
if the original input was a $\sig'$-(non-HF)-structure.
\begin{lemma}
	\label{lem:auts-flat-vs-hf}
	Let~$\Struct_0$ be a binary relational structure
	and let~$\Struct$ be a HF-structure obtained by a DeepWL-algorithm
	in the cloud on input~$\Struct_0$.
	Then for every automorphism ${\autoA \in \autGroup{\nonHF{\Struct}}}$
	it holds that $\restrictVect{\autoA}{\StructV} \in \autGroup{\Struct}$.
\end{lemma}
\begin{proof}
	The vertex set of the HF-structure~$\Struct$ is the union of
	the atoms~$\StructV = \StructV_0$
	and some $\HF{\StructV}$-sets~$\StructHFV$. 
	Every $\sccSym$-operation creates for every corresponding SCC~$c$ a vertex for the HF-set~$c$ and the membership relation coinciding with ``$\in$'' on the HF-sets.
	Similarly, every $\addPairSym$-operation introduces (Kuratowski-encoded) pairs
	and the component relations identifying the single entries in the pairs.
	That is, the structure of the vertices as HF-sets
	is encoded by the membership and component relations as a DAG using~$\StructV$ as sinks.
	Because a DeepWL-algorithm only adds but never removes relations,
	the membership and component relations are still present in~$\Struct$.
	
	Now let $\autoA \in \autGroup{\nonHF{\Struct}}$.
	To show that $\restrictVect{\autoA}{\StructV} \in \autGroup{\Struct}$,
	we have to show that $\restrictVect{\autoA}{\StructV}$ is a permutation of~$\StructV$ that satisfies
	$\restrictVect{\autoA}{\StructV}(\StructHFV) = \StructHFV$.
	We first note that~$\autoA$ cannot map an atom in~$\StructV$ to a vertex in~$\StructHFV$
	because atoms have no outgoing edge of a component or membership relation
	while all vertices in~$\StructHFV$ have one.
	That is,~$\autoA$ set-wise stabilizes~$\StructV$ and~$\StructHFV$, i.e.,
	$\autoA(\StructV) = \StructV$ and $\autoA(\StructHFV) = \StructHFV$.
	In particular, $\restrictVect{\autoA}{\StructV}$ is a permutation of~$\StructV$.

	Now let $a \in \StructHFV$. The $\HF{\StructV}$-set~$a$
	is the unique vertex of the DAG representing~$a$ via the membership and component relations.
	So~$\autoA(a)$ is the unique vertex of an isomorphic DAG,
	which we obtain by applying~$\autoA$ to the DAG representing~$a$.
	That is, $\autoA(a) = \restrictVect{\autoA}{\StructV} (a)$
	(note that~$\autoA$ permutes $\StructV \cup \StructHFV$ and~%
	$\restrictVect{\autoA}{\StructV}$ permutes~$\StructV$ and is applied to the atoms in~$a$).
	Because, as already seen, $\autoA(\StructHFV) = \StructHFV$,
	it follows that ${\autoA(\StructHFV) = \restrictVect{\autoA}{\StructV}(\StructHFV) = \StructHFV}$.
\end{proof}
This lemma is in particular important when we extend DeepWL with a witnessed choice operator: to compute (HF-set respecting) orbits of~$\Struct$,
it suffices to compute orbits of~$\nonHF{\Struct}$
and so to consider $\coConf{\nonHF{\Struct}}$.
Note that the lemma does not hold if the $\forgetSym$-operation of~\cite{GroheSchweitzerWiebking2021} is available (which can be used to forget a component or membership relation)
or the $\sccSym$-operation is replaced the by $\contractSym$-operation
(then also a membership-relation is lost).

\subsection{DeepWL with Witnessed Symmetric Choices}
\label{sec:deepwl-wsc}

We extend the DeepWL computation model with witnessed symmetric choice.
Here, we need two different notions:
we start with DeepWL+WSC-machines which are
then composed to DeepWL+WSC-algorithms.
A \defining{DeepWL+WSC-machine}~$\dwlm$ is a DeepWL-algorithm,
whose Turing machine has two additional special states~$q_{\choiceSym}$
and~$q_{\refineSym}$.
To enter~$q_{\choiceSym}$, the machine~$\dwlm$ has to write a relation symbol $\relColA \in \sigA \cup \sigB$
on the interaction-tape.
To enter~$q_{\refineSym}$,~$\dwlm$ has to write a relation symbol $\relColA \in \sigA \cup \sigB$ and a number $i\in \nat$ on the interaction-tape.
We say that the machine~$\dwlm$ executes $\choice{\relColA}$
and $\refine{\relColA}{i}$.
The DeepWL+WSC-machine~$\dwlm$ is \defining{choice-free},
if it never syntactically enters~$q_{\choiceSym}$.
That is,~$q_{\choiceSym}$
is not in the range of the transition function of the underlying Turing machine of~$\dwlm$.
DeepWL+WSC-algorithms are defined inductively:

\begin{definition}[DeepWL+WSC-algorithm]
	If~$\dwlmout$ is a DeepWL+WSC-machine,~%
	$\dwlmwit$ is a choice-free DeepWL+WSC-machine, and
	$\dwla_1, \dots, \dwla_\ell$ is a possibly empty sequence of
	DeepWL+WSC-algorithms,
	then the tuple
	$\dwla = (\dwlmout,\dwlmwit, \dwla_1, \dots, \dwla_\ell)$
	is a \linebreak[100] \defining{DeepWL+WSC-algorithm}.
	The machine~$\dwlmout$ is called the \defining{output machine} of~$\dwla$
	and the machine~$\dwlmwit$ is called the \defining{witnessing machine} of~$\dwla$.
\end{definition} 
Note that the base case of the former definition is the case $\ell = 0$,
i.e., the sequence of nested DeepWL+WSC-algorithms is empty.
The nested algorithms can be used by the machines~$\dwlmout$ and~$\dwlmwit$
as subroutines.
We first discuss the execution of a DeepWL+WSC-algorithm
and in particular the use of subroutines intuitively.
A formal definition will follow.
Let $\dwla = (\dwlmout,\dwlmwit, \dwla_1, \dots, \dwla_\ell)$ be a DeepWL+WSC-algorithm.
As first step, we consider the executions of the DeepWL+WSC-machines~$\dwlmout$ and~$\dwlmwit$.
\begin{enumerate}[label=\alph*)]
	\item Assume a DeepWL+WSC-machine $\dwlm \in \set{\dwlmout, \dwlmwit}$ executes $\refine{\relColA}{j}$.
	If $j > \ell$, then~$\dwlm$ just continues.
	If otherwise $j \in [\ell]$, then the DeepWL+WSC-algorithm~$\dwla_j$ is used to refine the relation~$\relColA$:
	Let~$\Struct$ be the content of the cloud of~$\dwlm$ when~$\dwlm$ executes the $\refineSym$-operation.
	If~$\relColA$ is directed,
	the algorithm~$\dwla_j$ is executed on 
	$(\Struct, \vertA\vertB)$
	($\vertA$ and~$\vertB$ are individualized
	by putting them into singleton vertex classes)
	for each $(\vertA, \vertB) \in \relColA^\StructA$.
	Otherwise~$\relColA$ is undirected,  so $\relColA^\Struct = \inv{(\relColA^\Struct)}$,
	then for each $\set{\vertA,\vertB} \in \relColA^\StructA$,
	the algorithm~$\dwla_j$ is executed on 
	$(\Struct, \set{\vertA\vertB})$
	(the undirected edge is individualized
	by creating a new vertex class only containing~$\vertA$ and~$\vertB$).
	The algorithm~$\dwla_j$ modifies its own cloud independently of the cloud of~$\dwlm$.
	If~$\dwla_j$ accepts $(\Struct, \vertA\vertB)$ for every $(\vertA,\vertB) \in \relColA^\StructA$
	(respectively $(\Struct, \set{\vertA\vertB})$
	for every $\set{\vertA,\vertB} \in \relColA^\StructA$),
	then nothing happens and~$\dwlm$ continues.
	Otherwise, a new relation~$\rel'$ is added to the cloud of~$\dwlm$,
	where~$\rel'$ consists of all $(\vertA,\vertB) \in \relColA^\StructA$
	(respectively $\set{\vertA,\vertB} \in \relColA^\StructA$),
	for which~$\dwla_j$ accepts the input.

	\item Assume a DeepWL+WSC-machine $\dwlm \in \set{\dwlmout, \dwlmwit}$ executes $\choice{\relColA}$.
	If~$\relColA$ is a directed relation,
	then an arbitrary $(\vertA, \vertB) \in \relColA^\Struct$ is individualized and~$\dwlm$ continues.
	If otherwise~$\relColA$ is undirected,
	then an undirected edge $\set{\vertA,\vertB} \in X^\Struct$
	is individualized.
	The edges are individualized as described in the $\refineSym$-operation.
\end{enumerate}
The algebraic sketch is recomputed and written onto the interaction-tape
whenever the structure in the cloud is modified by $\refine{\relColA}{j}$ or $\choice{\relColA}$.
The machine~$\dwlm$ accepts the input,
if the symbol under the head on the work-tape is a~$1$
when~$\dwlm$ halts and rejects otherwise.
Later, we will formally define the execution with choices using a tree,
similar to the definition of the iteration terms with choice in CPT+WSC in Section~\ref{sec:semantics-choice-operators}.

We now turn to the DeepWL+WSC-algorithm~$\dwla$.
To execute the algorithm~$\dwla$ on input~$\StructA_0$,
the output machine~$\dwlmout$ is executed  on~$\StructA_0$.
Let~$\Struct$ be the content of the cloud when~$\dwlmout$ halts.
For every $\choiceSym$-operation executed by~$\dwlmout$,
the witnessing machine~$\dwlmwit$ is executed.
Let~$k$ be a number not exceeding the number of $\choiceSym$-operations,
$\choice{\relColA_i}$ be the $i$-th executed $\choiceSym$-operation 
(for some~$\relColA_i$ in the current signature) for every $i \in [k]$,
and~$\Struct_i$ be the content of the cloud,
when the $i$-th $\choiceSym$-operation is executed
for every $i \in [k]$.
For the $k$-th $\choiceSym$-operation $\choice{\relColA_k}$,
the machine~$\dwlmwit$ has to provide automorphisms
witnessing that~$\relColA^{\Struct_k}$ is an
$(\Struct_0, \Struct_1, \dots, \Struct_{k})$-orbit
(details follow later).
That is, similarly to the WSC-fixed-point operator, all intermediate steps of the fixed-point computation have to be fixed by the witnessing automorphisms.
Recall again that we are working with HF-structures with the same set of atoms,
so all vertices are $\HF{\StructV_0}$-sets,
so the notion of an $(\Struct_0, \Struct_1, \dots, \Struct_{k})$-orbit
is well-defined.

The input of~$\dwlmwit$ is the \defining{labeled union}
$\Struct \labeledUnion \Struct_k$,
which is the union $\Struct \cup \Struct_k$
equipped with two fresh relation symbols~$\rel_1$ and~$\rel_2$
labeling the vertices of~$\Struct$ and~$\Struct_k$,
i.e.,
\begin{align*}
	\rel_1^{\Struct \labeledUnion \Struct_k} &:= \vertices{\Struct} \text{ and}\\
	\rel_2^{\Struct \labeledUnion \Struct_k} &:= \vertices{\Struct_k}.
\end{align*}
So~$\dwlmwit$ is able to reconstruct~$\Struct$ and~$\Struct_k$
and to determine how~$\Struct$ and~$\Struct_k$ relate to each other:
since the atoms of both~$\Struct$ and~$\Struct_k$ are $\HF{\StructA_0}$-sets,
common vertices are ``merged'' in the union.
When the witnessing machine~$\dwlmwit$ halts,
it has to write a relation symbol onto the interaction-tape.
The relation has to encode a set of witnessing automorphisms (details on the encoding follow later).

If all choices are successfully witnessed,~%
$\dwla$ accepts~$\StructA_0$ if~$\dwlmout$ accepts~$\StructA_0$ the input and rejects otherwise.
If some choice could not be witnessed, we abort the computation and output~$\choiceError$.
If an executed subalgorithm~$\dwla_i$
outputs~$\choiceError$, then~$\dwla$ also outputs~$\choiceError$.
We also say that~$\dwla$ \defining{fails}
if~$\dwla$ outputs~$\choiceError$.

\paragraph{Semantics of the new Operations.}
Both, the $\refineSym$- and the $\choiceSym$-operation
contain some special cases in its semantics.
First, both operations treat directed and undirected relations differently.
For undirected relations~$\relColA$,
an undirected edge ${\set{\vertA,\vertB} \in \relColA^\Struct}$
is individualized and not just a directed one.
Second, a $\refineSym$-operation
does not create a new relation 
in the case that every $(\vertA,\vertB) \in \relColA^\Struct$ or every $\set{\vertA,\vertB} \in \relColA^\Struct$
are accepted by~$\dwla_j$.

Creating a new relation containing the same edges as~$\relColA$ would seem more natural.
Indeed, for general DeepWL+WSC-algorithms, these special cases
do not change the expressive power.
But later in Section~\ref{sec:normalized-deepwl},
we will introduce the notion of a normalized DeepWL+WSC-algorithm,
which will put additional restrictions on e.g.~$\addPairSym$-operations.
Here the precise semantics of the $\refineSym$- and $\choiceSym$-operations will matter,
in particular it will be crucial to prove Lemma~\ref{lem:deepwl-wsc-simulate-normalized}.

\paragraph{Encoding Automorphisms.}
We now discuss how sets of automorphisms are encoded.
Let $\Struct_0$ be the input structure
and $\Struct$ be the current content of the cloud.
A tuple of relations $(\rel_{\text{aut}}, \rel_{\text{dom}}, \rel_{\text{img}})$
and a vertex $\vertC_\auto$
\defining{encode the
partial map} $\auto \colon \StructV\to \StructV$ as follows:
we have $\auto(\vertA) = \vertB$ in the case that~$\vertB$ is the only vertex for which 
there exists exactly one~$\vertC$
such that $(\vertC_\auto, \vertC) \in \rel_{\text{aut}}^\Struct$,
$(\vertC,\vertA) \in \rel_{\text{dom}}^\Struct$, and $(\vertC,\vertB) \in \rel_{\text{img}}^\Struct$.
A tuple of relations $(\rel_{\text{aut}}, \rel_{\text{dom}}, \rel_{\text{img}})$
\defining{encodes the set of partial maps}
\[N = \setcond*{\auto}{\vertC_\auto \text{ encodes } \auto \text{ for some }(\vertC_\auto, \vertC) \in \rel_{\text{aut}}^\Struct}.\]
The tuple $(\rel_{\text{aut}}, \rel_{\text{dom}}, \rel_{\text{img}})$
witnesses a relation $\rel_{\text{orb}}^\Struct$ as $(\Struct_0,\tup{a})$-orbit
for some ${\tup{a} \in \HF{\StructA_0}^*}$,
if the set~$N$ witnesses~$\rel_{\text{orb}}^\Struct$ as $(\Struct_0,\tup{a})$-orbit.

\paragraph{Execution with Choices.}

Intuitively, we have to nest DeepWL+WSC-algorithms,
to ensure that DeepWL+WSC-algorithms ``return'' an isomorphism-in\-variant result
(accept or reject).
This corresponds to the output formula in a WSC-fixed-point operator,
which ensures that it defines an isomorphism-invariant property.
We now define the execution of a DeepWL+WSC-algorithm formally.
We need to deal with choices to obtain a well-defined notion.
In the following, we will assume that all considered Turning machines always terminate.
We can do so because we will be only interested in polynomial-time Turing machines in this article.

A \defining{configuration}~$c$ of a DeepWL+WSC-machine~$\dwlm$ is
a tuple of a state~$q(c)$ of the Turing machine contained in of~$\dwlm$ and the content of the two tapes.
Suppose ${\dwla = (\dwlmout, \dwlmwit, \dwla_1, \dots, \dwla_\ell)}$ is an arbitrary  DeepWL+WSC-algorithm and~$\Struct_0$ is an input HF-structure to~$\dwla$.
Let $\delta^{\Struct_0}$ be the transition function of~$\dwlmout$:
for configurations~$c$ and~$c'$ of~$\dwlmout$
and HF-structures~$\StructA$ and~$\StructA'$, both with atom set~$\StructVA_0$,
we have
$\delta^{\Struct_0}(c, \StructA) = (c',\StructA')$
if~$\dwlmout$, started in configuration~$c$ with~$\StructA$ in the cloud,
executes the first $\choiceSym$-operation
(or halts if no $\choiceSym$-operation is executed)
in the configuration~$c'$
with~$\StructA'$ in the cloud. 
In particular, no choice operation is executed in the computation from 
$(c, \StructA)$ to $(c', \StructA')$.
If all DeepWL+WSC-algorithms $\dwla_1, \dots, \dwla_\ell$ are deterministic,
i.e., accept, reject, or fail independent of the choices made during the execution of the~$\dwla_i$,
then all $\refineSym$-operations executed by~$\dwlmout$ are deterministic
and so~$\delta^{\Struct_0}$ is indeed a well-defined function.
We view a tuple $(c, \StructA)$ as an $\HF{\StructV_0}$-set:~%
$c$ is encoded as an $\HF{\emptyset}$-set and 
the vertices of~$\StructA$ itself are $\HF{\StructV_0}$-sets.

To define the run of the DeepWL+WSC-algorithm~$\dwla$ in the presence of choices,
we reuse the $\wscStarSym$-operator from Section~\ref{sec:semantics-choice-operators}.
The set of possible runs of~$\dwla$ is the set
\[\wsc{\delta_{\text{step}}^{\Struct_0}}{\delta_{\text{choice}}^{\Struct_0}}{\delta_{\text{wit}}^{\Struct_0}},\]
where we define the functions~$\delta_{\text{step}}^{\Struct_0}$,~$\delta_{\text{choice}}^{\Struct_0}$, and~$\delta_{\text{wit}}^{\Struct_0}$ using~$\delta^{\Struct_0}$ as follows.
In the beginning,
the function~$\delta_{\text{choice}}^{\Struct_0}$ outputs the empty choice set (i.e., performs no choice)
and the function~$\delta_{\text{step}}^{\Struct_0}$ starts the 
output machine in the initial configuration~$c^{\text{out}}_0$ 
(initial state, empty work-tape, and $\sketch{\StructA_0}$ written on the interaction-tape) with the initial input structure~$\Struct_0$ in the cloud, that is,
\begin{align*}
	\delta_{\text{step}}^{\Struct_0}(\emptyset, \emptyset) &:= (c^{\text{out}}_0, \Struct_0),\\
	\delta_{\text{choice}}^{\Struct_0} (\emptyset) &:= \emptyset.
\end{align*}
Next, whenever a $\choice{\relColA}$-operation is executed,
$\delta_{\text{choice}}^{\Struct_0}$ outputs the relation $\tilde{\relColA}^\Struct$ in the current content of the cloud~$\StructA$,
where $\tilde{\relColA}^\Struct := \relColA^\Struct$ if the relation~$\relColA$ is directed
and $\tilde{\relColA}^\Struct := \setcond{\set{\vertA,\vertB}}{(\vertA,\vertB) \in \relColA^\Struct}$  if~$\relColA$ is undirected:
\begin{align*}
	\delta_{\text{choice}}^{\Struct_0} ((c^{\text{out}},\Struct)) &:= \begin{cases}
		\tilde{\relColA}^\Struct & \text{if } q(c^{\text{out}}) = q_{\choiceSym} \text{ and } \relColA \text{ is on interaction tape in } c^{\text{out}},\\
		\emptyset & \text{otherwise}.
	\end{cases}
\end{align*}
After choosing~$a$ (a set of size at most~$1$
encoding a directed or undirected edge),~%
$a$ is individualized as described earlier and the resulting structure is
denoted by $(\Struct, a)$.
Then~$\delta_{\text{step}}^{\Struct_0}$ continues with the next configuration~$\hat{c}^{\text{out}}$ 
according to the transition function of~$\dwlmout$ in the structure
(and the new algebraic sketch written onto the interaction-tape):
\begin{align*}
	\delta_{\text{step}}^{\Struct_0}((c^{\text{out}},\Struct), a) &:= \begin{cases}
		\delta^{\Struct_0}(\hat{c}^{\text{out}},(\Struct,a)) & \text{if } q(c^{\text{out}}) = q_{\choiceSym},\\
		(c^{\text{out}},\Struct) & \text{otherwise}.
	\end{cases}
\end{align*}
If in the end of the computation a halting state is reached
(so $q(c^{\text{out}}) \neq q_{\choiceSym}$),~%
$\delta_{\text{choice}}^{\Struct_0}$ returns the empty set,
nothing is chosen, and~$\delta_{\text{step}}^{\Struct_0}$ reaches a fixed-point.

The function~$\delta_{\text{wit}}^{\Struct_0}$
maps a pair $((c,\Struct),(c',\Struct_k))$
to the set of partial maps
encoded by the relation which is written on the interaction-tape 
when~$\dwlmwit$ halts
on input $\Struct \labeledUnion \Struct_k$.
Because~$\dwlmwit$ is choice-free,
we do not have to deal with choices here.

Recall from the $\wscStarSym$-operator,
that for a $\choice{\relColA}$-operation,
the relation~$\relColA$ has to be an orbit which fixed all intermediate steps:
Let $\choice{\relColA_1}, \dots, \choice{\relColA_k}$
be the sequence of all already executed $\choiceSym$-operations and
let $\Struct_1, \dots, \Struct_k$ be the contents of the cloud
when the corresponding $\choiceSym$-operation is executed.
Then the set~$\relColA_k^{\Struct_k}$ has to be an
$(\Struct_0, \Struct_1, \dots, \Struct_{k})$-orbit%
\footnote{Formally, by the definition of the $\wscStarSym$-operator,
	$\relColA_k^{\Struct_k}$ has to be an
	$(\Struct_0, (c_1,\Struct_1), \dots, (c_k,\Struct_{k}))$-orbit,
	where~$c_i$ is the configuration of the Turing machine
	at the moment when the $i$-th choice operation is executed. But since the~$c_i$ are encoded as $\HF{\emptyset}$-set,
	they are invariant under all permutations of the atoms.
}
(recall again, that all structures $\Struct_0, \dots, \Struct_k$
have the same atom set~$\StructV_0$),
which~$\delta_{\text{wit}}^{\StructV_0}$ has to witness.
Also, recall from the definition of the $\wscStarSym$-operator,
that the input to $\delta_{\text{wit}}^{\Struct_0}$
is indeed the pair $((c,\Struct),(c',\Struct_k))$
of the reached fixed-point $(c,\Struct)$
and the intermediate step $(c',\Struct_k)$
on which the $k$-th $\choiceSym$-operation is executed.

Let $W = \wsc{\delta_{\text{step}}^{\Struct_0}}{\delta_{\text{choice}}^{\Struct_0}}{\delta_{\text{wit}}^{\Struct_0}}$.
The algorithm~$\dwla$ \defining{accepts}~$\Struct_0$,
if for some (and thus for every) $(c, \Struct) \in  W$,
the head of the work-tape in~$c$
points to a~$1$,
\defining{fails} if $W = \emptyset$,
and \defining{rejects} otherwise.
Note that $W = \emptyset$ if and only if there are non-witnessed choices
because we assumed that our DeepWL+WSC-algorithms always terminate
and so a non-empty fixed-point is always reached.

\begin{lemma}
	\label{lem:dwsc-run-invariant}
	Every DeepWL+WSC-algorithm $\dwla = (\dwlmout, \dwlmwit, \dwla_1, \dots \dwla_\ell)$ satisfies the following:
	\begin{enumerate}
		\item The class of structures accepted by~$\dwla$ is closed under isomorphisms. 
		\item The algorithm~$\dwla$ always accepts (or respectively rejects or fails) independent of the choices made in the execution of~$\dwlmout$.
		\item If all choices were witnessed,
		the series of configurations of~$\dwlmout$ is the same for all possible choices.
	\end{enumerate}
\end{lemma}
\begin{proof}
	The proof is by induction on the nesting depth of the algorithm.
	Let $\dwla = (\dwlmout, \dwlmwit, \dwla_1, \dots \dwla_\ell)$
	be a DeepWL+WSC-algorithm,~$\StructA_0$ be the input structure, and
	assume by induction that the claim holds for 
	the DeepWL+WSC-algorithms $\dwla_1, \dots, \dwla_\ell$.
	So every $\refineSym$-operation executed by~$\dwlmout$ is isomorphism-invariant.
	Additionally, 
	all other operations apart from $\choiceSym$ modify the cloud in an isomorphism-invariant manner and
	the algebraic sketch itself is isomorphism-invariant~\cite{GroheSchweitzerWiebking2021}.
	Because~$\delta^{\Struct_0}$ ``stops'' the execution
	when the first  $\choiceSym$-operation is encountered,~%
	$\delta^{\Struct_0}$ is well-defined (i.e., deterministic) and isomorphism-invariant.
	That was exactly our assumption on~$\delta^{\Struct_0}$ in the former paragraph.
	Likewise, the function~$\delta_{\text{wit}}^{\Struct_0}$ is isomorphism-invariant
	because~$\dwlmwit$ is choice-free.
	
	Thus,~$\delta_{\text{choice}}^{\Struct_0}$ is isomorphism-invariant
	because~$\delta_{\text{step}}^{\Struct_0}$ is isomorphism-invariant,
	which is the case
	since the relation used as choice set
	only depends on the configuration returned by~$\delta_{\text{step}}^{\Struct_0}$.
	That is, we can apply the lemmas in Section~\ref{sec:semantics-choice-operators}.
	By Corollary~\ref{cor:wsc-orbit},
	the set $\wsc{\delta_{\text{step}}^{\Struct_0}}{\delta_{\text{choice}}^{\Struct_0}}{\delta_{\text{wit}}^{\Struct_0}}$ is an orbit of $\StructA_0$.
	If $\wsc{\delta_{\text{step}}^{\Struct_0}}{\delta_{\text{choice}}^{\Struct_0}}{\delta_{\text{wit}}^{\Struct_0}} = \emptyset$,
	then some choice could not be witnessed,
	which by Corollary~\ref{cor:witnessed-all-or-none}
	is either the case for all possible choices or never occurs.
	Otherwise, the configuration~$c$ is the same for all tuples $(c, \Struct) \in \wsc{\delta_{\text{step}}^{\Struct_0}}{\delta_{\text{choice}}^{\Struct_0}}{\delta_{\text{wit}}^{\Struct_0}}$
	because the configuration~$c$ is invariant under all permutations of the atoms (because~$c$ is an $\HF{\emptyset}$-set).
	So the algorithm either accepts or rejects for all possible choices
	(and Parts~1 and~2 are proven).
	
	To see Part~3,
	note that we can replace the configuration~$c$ of the Turing machine in
	the tuples $(c,\Struct)$ with the sequence of all visited configurations so far without breaking one of the arguments before.
	Then not only the last configuration is an orbit,
	but also the sequence of visited configurations,
	which is thus invariant under all automorphisms
	(given that all choices were witnessed).
\end{proof}

The former lemma has the consequence that,
apart from possibly $\choiceSym$-operations,
none of the operations change the automorphisms of the structure in the cloud.
\begin{corollary}
	\label{cor:operations-apart-choice-auto-invariant}
	Let $\dwla = (\dwlmout, \dwlmwit, \dwla_1, \dots \dwla_\ell)$
	be a DeepWL+WSC-algorithm,
	$\dwlm \in \set{\dwlmout, \dwlmwit}$,
	and~$\StructA$ be the current content of the cloud of~$\dwlm$ (at some point during its execution).
	Assume that~$\dwlm$ executes $\addPairSym$, $\createSym$, $\sccSym$, or $\refineSym$
	and let~$\StructA'$ be the content of the cloud of~$\dwlm$ after the execution.
	Then $\autGroup{\StructA} = \autGroup{\StructA'}$.
\end{corollary}
\begin{proof}
	Recall that an automorphism of an HF-structure is a permutation of the \emph{atoms}, which extends to all vertices.
	By construction, the $\addPairSym$-, $\createSym$-, and $\sccSym$-operations
	are isomorphism-invariant.
	By Lemma~\ref{lem:dwsc-run-invariant},
	also the $\refineSym$-operation is isomorphism-invariant
	because the class of structures accepted by DeepWL+WSC-algorithm used to refine a relation is isomorphism-closed.
\end{proof}

\paragraph{Internal Run.}
The internal run of a DeepWL-algorithm
is the sequence of configurations of the Turing machine
(the state and the content of the two tapes)
during the run.
Note that in particular the algebraic sketches
computed during the computations are part of the internal run
because they are written onto the interaction-tape.
As we have already seen,
the internal run of a DeepWL-algorithm is 
isomorphism-invariant~\cite{GroheSchweitzerWiebking2021}.

To define the internal run of a DeepWL+WSC-algorithm,
we have -- beside the internal run of the output machine --
to take all internal runs of the witnessing machine to witness the different choice sets and the internal runs of the subalgorithms into account.
The internal run is defined inductively over the nesting depth of DeepWL+WSC-algorithms.

Let $\dwla = (\dwlmout, \dwlmwit, \dwla_1, \dots, \dwla_\ell)$
be a DeepWL+WSC-algorithm
and assume that we have defined the internal run $\run{\dwla_i}{\StructA}$
for every $i \in [\ell]$ and HF-structure~$\StructA$.
The \defining{internal run} $\run{\dwlm}{\StructA}$ of a DeepWL+WSC-machine $\dwlm \in \set{\dwlmout, \dwlmwit}$ on input~$\StructA$
is the following:
Let $c_1, \dots, c_n$ be the sequence of configurations of the Turing machine of~$\dwlm$.
By Lemma~\ref{lem:dwsc-run-invariant}, the sequence is unique if all choices will be witnessed, which we assume for now.
Furthermore, let $k_1 < \dots < k_m$
be all indices such that 
$q(c_{k_j}) = q_{\refineSym}$ 
and~$\dwlm$ executes $\refine{\relColA_j}{i_j}$ such that $i_j \leq \ell$.
Let~$\StructA_j$ be the content of the cloud 
before executing $\refine{\relColA_j}{i_j}$ for all $j \in [m]$.
We define
\newcommand{\trippleAlignWithTwoCases}[8]{%
	#1 &:= \mathrlap{#2}\hphantom{\begin{cases} #5 &\\#7 &\end{cases}}\kern-\nulldelimiterspace #3\\
	#4 &:= \begin{cases} #5 & #6\\#7 & #8\end{cases}
}
\begin{align*}
	\trippleAlignWithTwoCases
	{r_j}
	{\bigmsetcondition{\run{\dwla_{i_j}}{(\StructA_j,x)}}{x\in \tilde{\relColA}_j^{\StructA_j}}}{\text{for every } j \in [m],}
	{\run{\dwlm}{\StructA}}
	{\choiceError}{\text{if } \choiceError \in r_i \text { for some } i \in [m],}
	{c_1, \dots, c_{k_1}, r_1, c_{k_1+1}, \dots, c_{k_m}, r_m,c_{k_m+1},\dots, c_\ell}{\text{otherwise.}}
\end{align*}
The internal run, denoted $\run{\dwla}{\StructA}$, of the DeepWL+WSC-algorithm~$\dwla$ on input~$\StructA$ is defined as follows:
Let $c_1, \dots, c_\ell$ be the sequence of configurations of the output Turing machine~$\dwlmout$.
Furthermore, let $k_1 < \dots < k_m$
be all indices such that 
$q(c_{k_j}) = c_{\choiceSym}$.
Let the corresponding $\choiceSym$-executions be $\choice{\relColA_j}$,
where the current content of the cloud  is~$\StructA_j$, for every $j \in [m]$.
Let~$\StructA_{m+1}$ be the final content of the cloud when~$\dwlmout$ halts.
We define
\begin{align*}
	\trippleAlignWithTwoCases
	{\hat{r}_j}
	{\run{\dwlmwit}{\StructA_{m+1} \labeledUnion \StructA_j}}{\text{for every } j \in [m],} 
	{\run{\dwla}{\StructA}}
	{\run{\dwlmout}{\StructA}, \hat{r}_1, \dots , \hat{r}_m}{\text{if all choices are witnessed,}}
	{\choiceError} {\text{otherwise.}}
\end{align*}
Using Lemma~\ref{lem:dwsc-run-invariant},
it is easy to see that $\run{\dwla}{\StructA}$ is isomorphism-invariant
and thus it can be canonically encoded as a 0/1-string.

\paragraph{Computability.}
Let  $\dwla = (\dwlmout, \dwlmwit, \dwla_1, \dots, \dwla_\ell)$ be 
a DeepWL+WSC-algorithm.
The algorithm~$\dwla$ \defining{decides a property}~$P$ of a class of $\sig$-structures~$\GraphClass$
if~$\dwla$ accepts $\StructA \in \GraphClass$ whenever~$\StructA$ satisfies~$P$
and rejects otherwise (and in particular never fails).
The algorithm~$\dwla$ \defining{computes a function} $f\colon \GraphClass \to \set{0,1}^*$,
if for every $\StructA \in \GraphClass$
the machine~$\dwlmout$ has written $f(\StructA)$ onto the work-tape when it halts on input~$\StructA$ and~$\dwla$ never fails.

Definitions~\ref{def:definable-isomorphism} and~\ref{def:definable-complete-invariant}
of definable isomorphism and complete invariant
are easily adapted to DeepWL+WSC: to do so,
for a $\sig$-structure~$\StructA$,
a tuple $\tup{a} \in \StructVA^*$,
and a fresh relation symbol~$\leq$,
the structure $(\StructA, \tup{a})$
is encoded by the $(\tau\cup{\leq})$-structure
$(\StructA, \leq^\Struct)$
such that $\vertA \leq \vertB$ if and only if
$\vertA = a_i$ and $\vertB = a_j$ for some $i \leq j$.

Next, we define the \defining{runtime} of a  DeepWL+WSC-machine $\dwlm \in \set{\dwlmout, \dwlmwit}$ on input~$\StructA_0$.
Every transition taken by the Turing machine counts as one time step.
Whenever a cloud-modifying operating is executed
and the algebraic sketch $\sketch{\StructA}$ of the new structure in the cloud~$\StructA$
is written onto the interaction-tape,
we count $|\sketch{\StructA}|$ many time steps,
where $|\sketch{\StructA}|$ is the encoding length of $\sketch{\StructA}$.
Following~\cite{GroheSchweitzerWiebking2021}, the encoding of $\sketch{\StructA}$
is unary and so the runtime of~$\dwlm$ is at least $|\StructVA_0|$.
When~$\dwlm$ executes $\refine{\relColA}{i}$ for $i \leq \ell$
and the current content of the cloud is~$\Struct$,
we count the sum of runtimes of $\dwla_i$ on input $(\StructA, x)$
for every $x \in \tilde{\relColA}^\Struct$.
The runtime of~$\dwla$ is the sum of the runtime of the output machine~$\dwlmout$
and the runtimes of the witnessing machine~$\dwlmwit$ to witness all choices.
A DeepWL+WSC-algorithm (or machine) runs in \defining{polynomial time},
if there exists a polynomial $p(n)$,
such that $p(|\sketch{\Struct_0}|)$ bounds the runtime on input~$\Struct_0$
for every HF-structure~$\Struct_0$ (or possibly for every $\Struct_0$ in a class of HF-structures of interest).

Note that if $\dwlmout, \dwlmwit, \dwla_1, \dots, \dwla_\ell$
run in polynomial time,
then~$\dwla$ runs in polynomial time:
the size of the structure in the cloud is polynomially bounded
and the machines~$\dwlmout$ and~$\dwlmwit$
can execute only polynomially bounded many $\refineSym$-operations.

\subsubsection{Derived Operations}
From the special operations $\addPairSym$, $\sccSym$, and $\choiceSym$ we can derive additional operations,
which then can be used for convenience.
These are:
\begin{enumerate}
	\item Ordered inputs:
	Currently, the structure~$\Struct$ put in the cloud
	is the only input to a DeepWL+WSC-machine or algorithm.
	We allow a pair $(\Struct, \tup{z})$
	of a structure~$\StructA$ and a binary string $\tup{z} \in \set{0,1}^*$ placed on the working tape as input.
	
	\item Ordered input for $\refineSym$-operations:
	We similarly extend $\refineSym$-operations to support
	$\refine{\vc}{i,\tup{z}}$ executions, where $\tup{z} \in \set{0,1}^*$ is an additional binary string initially put on the working tape.
	
	\item $\rename{\relA, \relB}$:
	If~$\relB$ is an unused relation symbol in the signature~$\sig$,
	the machine exchanges a relation symbol~$\rel$ with~$\relB$.
	Otherwise, no change is made.
	
	\item $\addUPair{\relColA}$: 
	This operation creates \emph{unordered pairs}, i.e.,~sets of size at most two for the $\relColA$-edges:
	Let~$\Struct$ be the current structure in the cloud. For every $\set{\vertA,\vertB}$ such that $(\vertA,\vertB) \in \relColA^\Struct$ a new vertex is added 
	and the membership relation $\bigcup_{(\vertA,\vertB) \in \relColA^\Struct} \set{(\set{\vertA,\vertB}, \vertA), (\set{\vertA,\vertB}, \vertB)}$
	is created.
\end{enumerate}
While adding the ability to rename relations seems useless,
it simplifies proofs
because we do not have to maintain bijections between
relation symbols of different structures. Instead, we just can rename them accordingly.

\begin{lemma}
	\label{lem:deepwl-wsc-plus}
	Ordered input to DeepWL+WSC-algorithms and the
	operations $\addUPairSym$,  $\renameSym$, and  $\refine{\vc}{i,\tup{z}}$
	can be simulated with the existing operations
	with only a polynomial overhead.
\end{lemma}
\begin{proof}

	We describe how the new operations can be simulated.
	\begin{enumerate}
		\item Ordered inputs can be simulated in the original DeepWL+WSC model as follows:
		Let~$\Struct$ be an input HF-structure,
		$\tup{z} \in \set{0,1}^*$,
		and $m = |\tup{z}|$.
		We encode~$\tup{z}$ into the string $\tup{z}' := 0^41^4z_10z_21z_30z_41\dots z_m$.
		There is only one position at which $0^41^4$
		occurs as substring of $\tup{z}'$, namely the first.
		We encode~$\tup{z}'$ into~$\StructA$ using vertex classes.
		Let $m' = |\tup{z}'|$.
		We add~$m'$ many vertex classes $\vc_1, \dots , \vc_{m'}$
		such that
		\[\vc_i^\Struct =  \begin{cases}
			\emptyset & \text{if } z'_i = 0,\\
			\StructVA & \text{if } z'_i = 1,
		\end{cases}\]
		for every $i \in [m']$.
		The string~$\tup{z}'$ can be read off the algebraic sketch:
		The algorithm starts at the lexicographically greatest relation symbol.
		Using the algebraic sketch the algorithm decides whether it encodes a~$0$ or a~$1$.
		This process is continued until $0^41^4$ is found.
		Then the algorithm computes~$\tup{z}$ from~$\tup{z}'$.
		Surely, the new vertex classes do not change
		the automorphisms of the structure~$\StructA$,
		so exactly the same orbits are witnessed after adding $\vc_1, \dots , \vc_{m'}$.
		Using Lemma~\ref{lem:restrict-sketch},
		we can always compute the algebraic sketch of the structure
		without the new vertex classes.
		
		\item $\refine{\vc}{i,\tup{z}}$:
		We proceed in the very similar way
		and encode~$\tup{z}$ by the string~$\tup{z}'$ as defined before
		using vertex classes.
		But now, the vertex classes 
		$\vc_1, \dots , \vc_{m'}$, where $m' = |\tup{z}'|$,
		have to be created by the DeepWL+WSC-algorithm.
		This can easily be done by executing $\create{\pi_i}$ for every $i \in [m']$,
		where
		$\pi_i = \emptyset$ if $z_i'= 0$
		and~$\pi_i$ is the set of all fibers if $z_i' =1$.

		The DeepWL+WSC-algorithm used to refine~$\vc$
		first reads off $\tup{z}'$ from the algebraic sketch
		and then, by Lemma~\ref{lem:restrict-sketch},
		computes the algebraic sketch
		without the vertex classes $\vc_1, \dots , \vc_{m'}$.
		
		\item $\rename{\relA, \relB}$:
		The DeepWL+WSC-machine maintains the bijection between current and renamed relation symbols
		on its tape.
		The bijection is passed through $\refineSym$-operations using ordered input.
		
		\item $\addUPair{\relColA}$:
		We first execute $\addPair{\relColA}$, so we obtain vertices 
		$\pairVtx{\vertA}{\vertB}$ for every $(\vertA,\vertB) \in \relColA^\Struct$
		and component relations~$\rel_{\leftT}$ and~$\rel_{\rightT}$.
		Let~$\pi$ be the set of colors
		which have an $(\rel_{\leftT}, \inv{\rel_{\rightT}})$-
		and an $(\rel_{\rightT}, \inv{\rel_{\leftT}})$-colored triangle.
		Then $\create{\pi}$ is executed and a relation~$\relB$ obtained.
		The relation~$\relB$ contains precisely all pairs $(\pairVtx{\vertA}{\vertB},\pairVtx{\vertB}{\vertA})$
		for every $\set{\vertA,\vertB} \in \relColA^\Struct$.
		We execute $\scc{\rel}$
		and obtain for each $\relB$-SCC a new vertex
		and the membership relation~$\rel_{\memT}$.
		Although the new vertices are created for the sets $\set{\pairVtx{\vertA}{\vertB},\pairVtx{\vertB}{\vertA}}$,
		they resemble the sets $\set{\vertA,\vertB}$.
		The membership relation~$\relB_{\memT}$ for these sets is obtained as the union of colors,
		which have an $(\rel_{\memT}, \rel_{\leftT})$- and an $(\rel_{\memT}, \rel_{\rightT})$-colored triangle.
		For the pairs $(\vertA,\vertB) \in \relColA^\Struct$
		for which $(\vertB,\vertA) \notin \relColA^\Struct$,
		we have not created an $\relB$-SCC vertex because these pair vertices are not incident to~$\relB$.
		Instead, we can just use the pair vertices itself
		(no automorphism can map~$\vertA$ to~$\vertB$ or vice versa in this case).
		We can define the relations refining~$\rel_{\leftT}$ and~$\rel_{\rightT}$
		incident to such pair vertices and also include them in~$\relB_{\memT}$.
		As before, the sketch without the additionally created pair vertices
		can be computed using Lemma~\ref{lem:restrict-sketch}.
		
		Regarding witnessing choices,
		the $\addUPairSym$-operation is isomorphism-invaraint.
		Because we did not use a $\choiceSym$-operation in the simulation,
		we have not changed the automorphisms of the HF-structure in the cloud
		by Corollary~\ref{cor:operations-apart-choice-auto-invariant}.
		So exactly the choices in the execution with the $\addUPairSym$-operation
		are witnessed as in the execution with the simulation.
	\end{enumerate}
	Finally, it is easy to see that the simulations run in polynomial time,
	where the size of a tuple $(\StructA, \tup{z})$ is $|\sketch{\StructA}|+|\tup{z}|$.
\end{proof}

\subsection{From CPT+WSC to DeepWL+WSC}

In this section, we translate a CPT+WSC-formula
into an equivalent polynomial-time DeepWL+WSC-algorithm.
The following translation is based on the translation
of CPT into interpretation logic in~\cite{GradelPSK15}
and the translation of interpretation logic in DeepWL in~\cite{GroheSchweitzerWiebking2021}.
However, we avoid the route through interpretation logic by directly
implementing the ideas of~\cite{GradelPSK15} in DeepWL+WSC.
To avoid case distinctions,
we assume in the following that the occurring CPT+WSC-terms or formulas never
output $\choiceError$.
Although a $\choiceError$-respecting translation can be given,
it is not needed for our purpose
(all isomorphism-defining CPT+WSC-formulas never output $\choiceError$).

	We simulate the evaluation of a BGS+WSC-term (or formula)
	with a DeepWL+WSC-algorithm.
	Recall that the vertices of the structure in the cloud during the execution
	of a DeepWL+WSC-machine on input structure~$\Struct$ are pairs $\pairVtx{a}{i}$ of a set $a \in \HF{\StructV}$ and a number~$i \in \nat$,
	where $\pairVtx{a}{i}$ is encoded as an $\HF{\StructV}$-set itself.
	For the simulation of a BGS+WSC-term (or formula) on input structure~$\Struct$,
	we encode two types of objects
	by vertices of the structure in the cloud:
	\begin{enumerate}
		\item 
		A hereditarily finite set $a \in \HF{\StructV}$
		is encoded by a vertex $\pairVtx{a}{i}$ for some number~$i$.
		We maintain a relation~$\inrel$,
		which serves as \defining{containment relation} between the vertices,
		that is, it corresponds to ``$\in$'' on the encoded sets.
		There is an exception for the empty set:
		because a DeepWL+WSC-algorithm cannot create a vertex for the empty set,
		we encode~$\emptyset$ by $\pairVtx{\StructV}{i}$ for some number~$i$. Here we require that $i \neq j$ for the number $j$
		for which $\pairVtx{\StructV}{j}$ encodes~$\StructV$.
		We apply this recursively: whenever~$\emptyset$ is part of an $\HF{\StructV}$-set, e.g., $\set{\emptyset}$, we consider
		the set obtained from replacing~$\emptyset$ by~$\StructV$.
		We will be able to distinguish the encoding vertices, e.g.,
		the vertices encoding~$\set{\emptyset}$ and~$\set{\StructV}$,
		using the containment relation.
		Note that the containment relation is different from the membership relation obtained from $\sccSym$-operations.
	
		\item
		A tuple $\tup{a} \in \HF{\StructVA}^k$ of hereditarily finite sets is encoded
		by a vertex $\pairVtx{\tup{a}}{i}$ for some number~$i$
		(for an appropriate encoding of tuples as hereditarily finite sets).
		We maintain a sequence of incident relations $\rel_1, \dots, \rel_k$,
		such that every~$\rel_j$ has out-degree~$1$ and
		associates to $\pairVtx{\tup{a}}{i}$
		a vertex $\pairVtx{a_j}{\ell}$ for some number~$\ell$.
		We call these relations the \defining{tuple relations}.
		We remark that the tuple relations in general will be different from the component relation of the $\addPairSym$-operations.
	\end{enumerate}
	We now introduce an intermediate version of DeepWL+WSC-algorithms.
	This intermediate version is needed for the recursive translation of CPT+WSC-terms and formulas.
	A \defining{choice-free DeepWL+WSC-algorithm} is a tuple
	$(\dwlmout, \dwla_1, \dots, \dwla_\ell)$ consisting
	of a choice-free DeepWL+WSC-machine~$\dwlmout$
	and DeepWL+WSC-algorithms $\dwla_1, \dots, \dwla_\ell$.
	The computation of a choice-free DeepWL+WSC-algorithm
	is exactly the same as the one of a regular DeepWL+WSC-algorithm,
	but since the output machine~$\dwlmout$ is choice-free,
	no witnessing machine is needed.
	Note here that for a choice-free DeepWL+WSC-algorithm
	the subalgorithms $\dwla_1, \dots, \dwla_\ell$ are not choice-free.
	The benefit of choice-free algorithms is that
	they can be composed more easily,
	e.g.,~they can be executed one after the other without worrying about witnessing choices.

	\begin{definition}[Simulating CPT+WSC with DeepWL+WSC]
		Let~$\StructA$ be a binary (non-HF) $\sig$\nobreakdash-structure.
		An HF-structure~$\StructA'$ is called \defining{compatible} with~$\StructA$
		if $\StructVA' =\StructV$ and $\autGroup{\StructA} = \autGroup{\StructA'}$.
		\begin{enumerate}[label=\alph*)]
			\item
			A choice-free DeepWL+WSC-algorithm~$\dwla$ \defining{simulates a CPT+WSC$[\sig]$-term} $\termA(\tup{x})$ on~$\StructA$, if
			for every HF-structure~$\StructA'$ compatible with~$\StructA$ and
			for every vertex class~$\vcA$ of~$\Struct'$ encoding $\HF{\StructVA}^{|\tup{x}|}$-tuples,
			the algorithm~$\dwla$ on input~$\Struct'$ and~$\vcA$ computes a vertex class~$\vc_\termA$ and a relation~$\iorel$ with the following property:
			\begin{itemize}
				\item $\iorel^{\StructA''} \subseteq \vcA^{\Struct'} \times \vcA_\termA^{\Struct'}$ is functional and surjective (where~$\StructA''$ is the content of the cloud after the execution of~$\dwla$) and
				\item $(\vertA,\vertB) \in \iorel^{\StructA''}$ if and only if
				the vertex~$\vertA$ encodes a tuple $\tup{a} \in \HF{\StructVA}^{|\tup{x}|}$,
				the vertex~$\vertB$ encodes a set $b \in \HF{\StructVA}$,
				and  $b = \denotation{\termA}^\Struct(\tup{a})$.
			\end{itemize}	
			The relation~$\iorel$ is called the \defining{input-output relation}.
			
			\item Likewise, a choice-free DeepWL+WSC-algorithm~$\dwla$ \defining{simulates a CPT+WSC$[\sig]$-formula}~$\formA(\tup{x})$
			which is not a WSC-fixed-point operator (but may contain such operators as subformulas) on~$\StructA$, if
			for every HF-structure~$\StructA'$ compatible with~$\StructA$
			and every vertex class~$\vcA$ of~$\StructA'$ encoding $\HF{\StructVA}^{|\tup{x}|}$-tuples,
			the algorithm~$\dwla$ on input~$\StructA'$ and~$\vcA$ defines a vertex class $\vc_\formA^{\StructA''} \subseteq \vcA^{\Struct'}$ of all $\vcA$-vertices
			encoding a tuple $\tup{a} \in \denotation{\formA}^\Struct$.
		
			\item Last, a (non choice-free) DeepWL+WSC-algorithm~$\dwla$ \defining{simulates a WSC-fixed-point operator}~$\formA(\tup{x})$ on~$\StructA$,
			if, 
			for every HF-structure~$\StructA'$ compatible with~$\StructA$
			and every singleton vertex class~$\vcA$ of~$\StructA'$ encoding an $\HF{\StructVA}^{|\tup{x}|}$-tuple~$\tup{a}$,
			the algorithm~$\dwla$ accepts on input~$\StructA'$ and~$\vcA$ if and only if $\tup{a} \in \denotation{\formA}^\Struct$.
		\end{enumerate}
		We say that a (choice-free) DeepWL+WSC-algorithm~$\dwla$ simulates a
		CPT+WSC$[\sig]$-term or formula~$\termA$,
		if~$\dwla$ simulates~$\termA$ on every $\sig$-structure.
	\end{definition}
	The former definition is elaborate because of the free variables.
	For a BGS+WSC-sentence~$\formA$,
	a DeepWL+WSC-algorithm~$\dwla$ simulating~$\formA$ accepts exactly the structures satisfying~$\formA$.
	
	We will translate  CPT+WSC-formulas and terms to polynomial time DeepWL+WSC-algorithms
	by induction on the nesting structure of the
	WSC-fixed-point operators.
	We call a WSC-fixed-point operator~$\formB$
	\defining{directly nested} in a CPT+WSC-formula (or a term)~$\formA$,
	if there is no WSC-fixed-point operator~$\formB' \neq \formB$
	such that~$\formB$ is a subformula of~$\formB'$
	and~$\formB'$ is a subformula of~$\formA$.
	
	\begin{lemma}
		\label{clm:simulate-cpt-wsc-term}
		For every CPT+WSC-formula~$\formA$ (respectively term)
		which is not a WSC-fixed-point operator,
		there is a choice-free polynomial-time DeepWL+WSC-algorithm simulating~$\formA$
		if for every WSC-fixed-point operator~$\formA'$ directly nested in~$\formA$,
		there exists a polynomial-time DeepWL+WSC-algorithm simulating~$\formA'$.
	\end{lemma}
	\begin{proof}
		In the following,
		we use relation symbols $\relColA \in \sigA\cup \sigB$
		for the sets~$\relColA^\StructA$,
		where~$\StructA$ is the current content of the cloud of the DeepWL+WSC-machine we are going to define.
		In that sense, for example $\relColA \setminus \relColB$ has a well-defined notion.
		This simplifies notation because we do not
		have to always introduce the current content of the cloud.
		
		The proof is by induction on the CPT+WSC formula or term.
		Let~$\formA$ be a CPT+WSC formula (respectively~$\termC$ be a CPT+WSC-term).  Let the directly nested WSC-fixed-point operators of~$\formA$ (respectively~$\termC$)
		be $\hat{\formA}_1, \dots , \hat{\formA}_\ell$
		and let, inductively, $\dwla_1, \dots, \dwla_\ell$
		be DeepWL+WSC-algorithms,
		such that~$\dwla_i$ simulates $\hat{\formA}_i$ for every $i\in[\ell]$.
		We now define a  DeepWL+WSC-machine~$\dwlmout$
		such that the choice-free DeepWL+WSC-algorithm
		$(\dwlmout, \dwla_1, \dots, \dwla_\ell)$
		simulates~$\formA$ (respectively~$\termC$).	
		
		We start with a vertex class~$\vcA$ whose vertices encode the values for the free variables.
		Then we perform the computations to simulate the formulas and terms.
		During the simulation we want to ensure that for every encoded set~$a$
		there is exactly one vertex of the form $\pairVtx{a}{i}$.
		We thus want to avoid duplicates
		(with the already mentioned exception for the empty set).
		So whenever we want to execute an $\addPairSym$-, $\addUPairSym$-, or $\sccSym$-operation,
		we actually have to check whether the resulting vertices already exist.
		This can be done using the containment relation~$\inrel$
		(and for the $\sccSym$-operation due to the fact that SCCs can be computed by DeepWL~\cite[Lemma~4]{GroheSchweitzerWiebking2021}).
		From now on, we implicitly assume that these checks are always done
		and just say that we execute the $\addPairSym$-, $\addUPairSym$-, or $\sccSym$-operations.
		We perform the following case distinction for~$\formA$ and~$\termC$:
		\begin{itemize}
			
			\item $\formB = \rel_{CPT}(x,y)$: 
			Here $\rel_{CPT}$ is a relation symbol of the signature of the CPT+WSC-formula,
			which is thus contained in the current structure in the cloud.
			Given a vertex class~$\vcA$ encoding the values
			for the pair~$xy$, 
			we return the vertex class~$\vcA'$ consisting of all $\vcA$-vertices
			for which the tuple relations are adjacent to atoms contained in the relation~$\rel_{CPT}$ as follows:
			Let~$\rel_1$ and~$\rel_2$ be the tuple relations for the first and second
			component of~$\vcA$.
			Then we want to return the vertex class
			containing the vertices
			which have an $(\rel_1,\rel_{CPT},\inv{\rel_2})$-colored cycle.
			Due to coherence, there is a union of fibers $\pi \subseteq \sigB$,
			which contains precisely these vertices.
			We execute $\create{\pi}$ and output the resulting relation.
			
			In the following, we will for readability not mention the required
			$\createSym$-operations or that we can find the union
			of fibers~$\pi$ if we want to obtain a relation of pairs with a $(\relB_1,\dots, \relB_k)$-colored path (for some relation symbols $\relB_1,\dots, \relB_k$).
			
			\item $\formB$ is $\termA(\tup{x}) = \termB(\tup{y})$:
			Given a vertex class $\vcA$ corresponding to the free variables $\tup{x} \cup \tup{y}$,
			we can use~$\vcA$ for the free variables~$\tup{x}$ and~$\tup{y}$
			by only using a subset of the tuple relations.
			Then we use the choice-free DeepWL+WSC-algorithms simulating~$\termA$ and~$\termB$
			to obtain vertex classes~$\vcA_\termA$ and~$\vcA_\termB$
			together with the input-output relations~$\iorelsub{\termA}$ and~$\iorelsub{\termB}$.
			Then precisely one $\vcA_\termA$-vertex~$\vertA$
			and one $\vcA_\termB$-vertex~$\vertB$
			are obtained from the same $\vcA$-vertex~$\vertC$
			if the edge $(\vertA,\vertB)$ has an $(\inv{\iorelsub{\termA}},\iorelsub{\termB})$-colored triangle.
			We now create a vertex class~$\vcA'$
			from the fibers refining~$\vcA$
			which have an $(\inv{\iorelsub{\termA}},\iorelsub{\termB})$-colored cycle, that is, the two vertices~$\vertA$ and~$\vertB$ are actually the same.
			This is exactly then the case when
			the outputs of~$\termA$ and~$\termB$ are equal.
			Note that here it is crucial to perform the cleanup steps
			unifying vertices encoding the same set
			after every step.
			
			\item $\formB = \neg \formA$:
			Given a vertex class~$\vcA$,
			we obtain by induction a vertex class~$\vcA_\formA$.
			Then we return $\vcA \setminus \vcA_\formA$.
			
			\item $\formB = \formA_1(\tup{x}) \land \formA_2(\tup{x})$:
			Given a vertex class~$\vcA$,
			we obtain by induction vertex classes~$\vcA_{\formA_1}$
			and~$\vcA_{\formA_2}$
			and we just return $\vcA_{\formA_1} \cap \vcA_{\formA_2}$.
			
			\item $\formB = \hat{\formA}_i(\tup{x})$:
			Let~$\vcA$ be a vertex class encoding the values
			for the variables~$\tup{x}$.
			The machine executes $\refine{\vcA}{i}$
			and returns the resulting vertex class
			(note that~$\vcA$ is a vertex class
			and so the fact that $\refineSym$-operations treat
			directed and undirected relations differently does not matter here).
			By induction hypothesis, this vertex class
			contains precisely the vertices encoding the tuples satisfying~$\hat{\formA}_i$.
			
			\item $\termC = \Atoms$: 
			We create a relation~$\rel$ connecting all atoms of the input structure
			(all vertices, which have no outgoing edge of some membership or component relation)
			and execute $\scc{\rel}$.
			We obtain a single additional vertex
			in a vertex class~$\vcA_\Atoms$
			and a membership relation connecting it to the atoms,
			which is merged into the containment relation~$\inrel$
			(using a $\createSym$- and a $\renameSym$-operation).
			(Of course, as described before,
			 we create the vertex class~$\vcA_\Atoms$ only once.)
			
			\item $\termC = \emptyset$:
			In the same way as in the $\Atoms$ case, we obtain a single vertex
			in its own vertex class~$\vcA_\emptyset$,
			but now we do not merge the created membership relation into the containment relation~$\inrel$.
			While~$\emptyset$ is encoded by a vertex for the atoms,
			with respect to automorphism this does not make a difference
			because every automorphism fixes the set of all atoms
			and~$\vcA_\emptyset$ and~$\vcA_\Atoms$
			contain different vertices $\pairVtx{\StructV}{i}$ and $\pairVtx{\StructV}{j}$
			for $i \neq j$, where~$\StructV$ is the set of atoms.
			So no automorphism can exchange these two vertices.

			\item $\termC = \Pair(\termA(\tup{x}), \termB(\tup{y}))$:
			Given a vertex class $\vcA$ encoding values for the free variables $\tup{x} \cup \tup{y}$,
			we can use~$\vcA$ for the free variables~$\tup{x}$ and~$\tup{y}$
			by only using a subset of the tuple relations.
			Then we use the choice-free DeepWL+WSC-algorithms simulating~$\termA$ and~$\termB$
			to obtain vertex classes~$\vcA_\termA$ and~$\vcA_\termB$
			together with the input-output relations~$\iorelsub{\termA}$ and~$\iorelsub{\termB}$.
			Then precisely the $\vcA_\termA$-vertices~$\vertA$
			and $\vcA_\termB$-vertices~$\vertB$
			are obtained for the same $\vcA$-vertex~$\vertC$
			if the pair $(\vertA,\vertB)$ has an $(\inv{\iorelsub{\termA}}, \iorelsub{\termB})$-colored triangle.
			That is, there is a union of colors~$\rel_=$ containing exactly all such pairs $(\vertA,\vertB)$.
			Then the machine executes $\addUPair{\rel_=}$
			and obtains the vertex class~$\vcA_\termC$
			and a new membership relation~$\relB$.
			The membership relation~$\relB$
			connects $\vcA_\termC$-vertices to $\vcA_\termA$- and $\vcA_\termB$-vertices.
			The input-output relation $\iorelsub{\termC}$ is the union of colors
			which have an $(\iorelsub{\termA}, \inv{\relB})$- and an $(\iorelsub{\termB}, \inv{\relB})$-colored triangle.
			Finally, the machine merges~$\relB$ into the containment relation~$\inrel$.
			
			\item $\termC = \Unique(\termA(\tup{x}))$:
			Given a vertex class~$\vcA$,
			we define the vertex class~$\vcA_\termA$
			with input-output relation~$\iorelsub{\termA}$
			and containment relation~$\inrel$ using the induction hypothesis.
			The vertex class~$\vcA_\termA$ can be partitioned
			into~$\vcA_\termA^{1}$ and~$\vcA_\termA^{\neq 1}$,
			where $\vcA_\termA^{1}$-vertices have precisely one outgoing $\inrel$-edge
			and $\vcA_\termA^{\neq 1}$-vertices do not.
			Let $\iorelsub{\termA}^1 := \iorelsub{\termA} \cap (\vcA \times \vcA_\termA^1)$ be the subset of $\iorelsub{\termA}$ leading to a $\vcA_\termA$-vertex encoding a singleton set (which is DeepWL-computable).
			We similarly partition~$\vcA$ into~$\vcA^1$ and~$\vcA^{\neq 1}$,
			where~$\vcA^1$ contains all $\vcA$-vertices incident to~$\iorelsub{\termA}^1$.
			
			The machine makes the following case distinction:
			If $|\vcA^{\neq 1}| = 0$, set $\vcA_\termC := \vcA'$
			and otherwise set $\vcA_\termC := \vcA' \cup \vcA_\emptyset$.
			For the $\vcA^1$-vertices,
			we obtain the input-output relation~$\iorelsub{\termC}^1$
			as the relation with an $(\iorelsub{\termA}^1 , \inrel)$-colored triangle.
			For the $\vcA^{\neq 1}$\nobreakdash-vertices, let
			$\iorelsub{\termC}^{\neq 1} = \vcA^{\neq 1} \times \vcA_\emptyset$ be the relation containing all edges between~$\vcA^{\neq 1}$ and the singleton vertex class~$\vcA_\emptyset$.
			Finally, output $\iorelsub{\termC} := \iorelsub{\termC}^1 \cup \iorelsub{\termC}^{\neq 1}$.
			
			\item $\termC = \Card(\termA(\tup{x}))$:
			Define for a given vertex class~$\vcA$ 
			the vertex class~$\vcA_\termA$,
			the input-output relation~$\iorelsub{\termA}$,
			and the containment relation~$\inrel$ using the induction hypothesis.
			We partition~$\vcA_\termA$ into $\vcA^1_\termA, \dots, \vcA^k_\termA$,
			such that $\vcA^i_\termA$-vertices have~$i$ many outgoing $\inrel$-edges.
			Then the machine defines for all $i\in [k]$ such that $\vcA^i_\termA \neq \emptyset$ the $i$-th von Neumann ordinal,
			defines~$\vcA_\termC$ as the union of these ordinals,
			and outputs~$\iorelsub{\termC}$ accordingly.
			
			\item $\termC = \setcond{\termA(\tup{x}y)}{y \in \termB(\tup{x}), \formA(\tup{x}y)}$:
			\begin{figure}
				\centering
				\begin{tikzpicture}
					
					\draw[fill=white!70!gray, draw=none] (4.5, -0.1) ellipse (0.5 cm and 1.35 cm);
					\draw[fill=white!50!gray, draw=none] (4.5, 0.5) ellipse (0.4 cm and 0.75 cm);
					
					\draw[fill=white!50!gray, draw=none] (6, 0.5) ellipse (0.4 cm and 0.75 cm);
					
					\node[vertex, label=left:{$\vcA$}] (C) at (0,0) {};
					\node[vertex, label=below:{$\vcA_\termB$}] (Ct) at (1.5,0) {};
					
					\node[vertex] (Ctin1) at (3,0) {};
					\node[vertex] (Ctin2) at (3,0.5) {};
					\node[vertex] (Ctin3) at (3,-0.5) {};
					
					\node[vertex] (Cxy1) at (4.5,0.0) {};
					\node[vertex, label=above:{$\vcA_\formA$}] (Cxy2) at (4.5,0.5) {};
					\node[vertex, label={[label distance=0.1cm]below:{$\vcA_{\tup{x}y}$}}] (Cxy3) at (4.5,-0.5) {};

					\node[vertex] (Cs1) at (6, 0) {};
					\node[vertex, label=above:{$\vcA_\termA$}] (Cs2) at (6, 0.5) {}; 
					
					\node[vertex, label=right:{$\vcA_\termC$}] (Cr) at (7.5, 0) {};
					
					\path[->, darkgreen, dashed]
						(C) edge node[below] {$\iorelsub{\termB}$} (Ct);
					\path[->, blue]
						(Ct) edge (Ctin1)
						(Ct) edge (Ctin2)
						(Ct) edge node[below] {$\inrel$} (Ctin3);
						
					\path[->,red]
						(C) edge [bend left=30]  (Ctin1)
						(C) edge [bend left=30] node[above, xshift=0.3cm] {$\rel_{\tup{x}y}$}(Ctin2)
						(C) edge [bend left=30]  (Ctin3);
						
					\path[<-]
						(C) edge [bend right= 55] (Cxy1)
						(C) edge [bend right= 55] (Cxy2)
						(C) edge [bend right= 55] node [below] {$\rel_{\tup{x}}$}  (Cxy3);
					\path[<-, darkgreen]
						(Ctin1) edge (Cxy1)
						(Ctin2) edge node[above] {$\rel_{y}$} (Cxy2)
						(Ctin3) edge (Cxy3);
						
					\path[->, red, dashed]
						(Cxy1) edge (Cs1)
						(Cxy2) edge node[above]{$\iorelsub{\termA}$} (Cs2);
						
					\path[darkyellow]
					 	(Cs1) edge node[below, yshift=-0.5cm] {$\relB$} (Cs2);
					 	
					 \path[->, blue, dashed]
					 	(Cr) edge  node[below, xshift=0.3cm] {$\inrel'$} (Cs1)
					 	(Cr) edge (Cs2);
					 	
					 \path[->]
					 	(C) edge [bend left = 60] node [below]{$\iorelsub{\termC}$} (Cr);
					
				\end{tikzpicture}\hspace{0.4cm}%
				\begin{tikzpicture}
					\node[vertex, label=above:{$\vcA$}] (C) at (0,0) {};
					
					\node[vertex] (P1) at (-0.2, -1.5){};
					\node[vertex] (P2) at (0.2, -1.5){};
					
					\node[vertex, label=above:{$\vcA^0$}] (C0) at (1.5,0) {};
					\node[vertex, label=below:{$\vcA_\emptyset$}] (Cempty) at (1.5,-1.5) {};
					
					\node[vertex, label=above:{$\vcA^i$}] (Ci) at (3.5,0) {};
					\node[vertex] (Pi) at (3.5,-1.5) {};
					\node[vertex, label=below:{$\vcA^i_\termA$}] (Csi) at (5,-1.5) {};
					\node[vertex, label=above:{$\ \vcA^{i+1}$}] (CiS) at (5,0) {};
					
					\node (L1) at (2.75,0) {$\dots$};
					\node (L1) at (2.75,-1.5) {$\dots$};
					
					\path[->]
						(C) edge (P1)
						(C) edge (P2);
					
					\path[->,blue]
						(C0) edge node[right] {$\relB^0_x$} (Cempty)
						(Ci) edge node[right] {$\relB^i_x$} (Pi)
						(CiS) edge node[right] {$\relB^{i+1}_x$} (Csi);
						
					\path[->, darkgreen]
						(Ci) edge[bend left = 10] node[right, yshift=0.3cm, xshift=-0.26cm] {$\iorelsub{\termA}^i$} (Csi);
						
					\path[->,red]
						(C) edge [bend left = 40] node[below] {$\rel^0$} (C0)
						(C) edge [bend left = 40] node[below, xshift=0.7cm, yshift=-0.1cm] {$\rel^i$} (Ci)
						(C) edge [bend left = 40] node[above, xshift=1cm] {$\rel^{i+1}$} (CiS);
				\end{tikzpicture}
				\caption{On the left, the translation of the comprehension term 
				$\termC = \setcond{\termA(\tup{x}y)}{y \in \termB(\tup{x}), \formA(\tup{x}y)}$ in DeepWL.
				The figure is drawn for a single vertex~$\vertA$ in 
				the input vertex class~$\vcA$
				and contains all vertices created for~$\vertA$.
				On the right, the translation of the iteration term $\termC = \termA[x]^*$.}
				\label{fig:translation-comprehension-term}
			\end{figure}
			We start with the vertex class~$\vcA$
			and define by induction the vertex class~$\vcA_\termB$,
			the input-output relation~$\iorelsub{\termB}$,
			and the containment relation~$\inrel$ (cf.\ Figure~\ref{fig:translation-comprehension-term}).
			Then the machine defines vertex classes encoding the tuples
			for the variables~$\tup{x}y$:
			Let~$\rel_{\tup{x}y}$ be the relations with an $(\iorelsub{\termB},\inrel)$-colored triangle.
			These relations associate an input tuple with an element contained in the output set of~$\termB$.
			The machine executes $\addPair{\rel_{\tup{x}y}}$
			and obtains a vertex class~$\vcA_{\tup{x}y}$
			with component relations $\rel_{\tup{x}} := \rel_{\leftT}$
			and $\rel_{y} := \rel_{\rightT}$. 
			
			By induction hypothesis, we obtain the vertex class $\vcA_\formA \subseteq \vcA_{\tup{x}y}$
			of $\vcA_{\tup{x}y}$-vertices satisfying~$\formA$.
			With~$\vcA_\formA$ as input vertex class,
			the machine obtains~$\vcA_\termA$ again by induction
			and obtains the input-output relation $\iorelsub{\termA}$.

			We define a relation $\relB$
			connecting $\vcA_\termA$-vertices originating from the same $\vcA_\termB$-vertex:~$\relB$ is given by the union of all colors~$\colA$
			which have an $(\inv{\iorelsub{\termA}}, \rel_y, \inv{\inrel}, \inrel, \inv{\rel_y}, \iorelsub{\termA})$-colored path.
			We then execute $\scc{\relB}$ to
			obtain a vertex class~$\vcA_\termC$
			and a new membership relation~$\inrel'$.
			Finally, the input-output relation $\iorelsub{\termC}$
			consists of the edges with an $(\iorelsub{\termB}, \inrel, \inv{\rel_y}, \iorelsub{\termA}, \inv{\inrel'})$-colored path and the machine merges~$\inrel'$ into the containment relation~$\inrel$.
				
			\item $\termC= \termA[x]^*$:
			Let~$\termA$ have free variables~$\tup{y}x$.
			We start with the vertex class $\vcA$ for the free variables~$\tup{y}$
			(cf.\ Figure~\ref{fig:translation-comprehension-term}).
			We then execute $\addPair{\rel_\emptyset}$
			for the relation $\rel_\emptyset = \vcA \times \vcA_\emptyset$.
			We obtain the vertex class~$\vcA^0$
			and the component relations $\rel^0 := \inv{\rel_{\leftT}}$ 
			between a $\vcA$-vertex and the corresponding $\vcA^0$-vertex
			and $\relB_x^0 := \rel_{\rightT}$
			defining the tuple relation to the entry for~$x$.
			
			Let~$\vcA^i$ be the vertex class containing the current input values 
			after the~$i$-th iteration.
			Let the relation~$\rel^i$ relate the values for~$\tup{y}$ with
			the extended tuples~$\tup{y}x$ in~$\vcA^i$ (maybe not all~$\tup{y}$ are related because a fixed-point for them is already computed).
			If~$i$ exceeds the maximal number of iterations (given by the polynomial bound of the CPT-term),
			set $\vcA^{i+1}_* := \emptyset$ and~$\rel^{i+1}_*$
			to be the relations with an $(\rel^i, \relB)$-colored triangle,
			where $\relB = \vcA^i \times \vcA_\emptyset$.
			
			Otherwise, the machine defines, starting with~$\vcA^i$,
			a vertex class~$\vcA^i_\termA$ and 
			the input-output relation~$\iorelsub{\termA}^i$ using the induction hypothesis.
			Let~$\relB^i_x$ be the tuple relation of~$\vcA^i$ for~$x$.
			We partition~$\rel^i_\termA$ into $\rel^i_{\text{fix}} := \rel^i_\termA \cap \relB^i_x$ and $\rel^i_{\text{continue}}:= \rel^i_\termA \setminus \relB^i_x$.
			In~$\rel^i_{\text{fix}}$, the set for the variable~$x$ contained in $\vcA^i$
			is equal to the output after applying~$\termA$ once more,
			so we reached a fixed-point.
			Set~$\vcA^i_*$ to be the subset of~$\vcA^i_\termA$ which is incident to~$\rel^i_{\text{fix}}$ and set~$\rel^i_*$ to be the relation with an $(\rel^i, \rel^i_{\text{fix}})$-colored triangle incident (so $\rel^i_*$ associates the input values of~$y$ with
			the computed fixed-point).
			
			If $\rel^i_{\text{continue}} = \emptyset$, the machine stops looping.
			Otherwise, the machine executes the operation $\addPair{\relB^i}$,
			where~$\relB^i$ are the edges with an
			$(\rel^i, \rel^i_{\text{continue}})$-colored triangle,
			and obtains the vertex class~$\vcA^{i+1}$,
			the relation~$\rel^{i+1}$ relating a $\vcA$-vertex with the corresponding $\vcA^{i+1}$-vertex
			and the relation~$\relB_x^i$ for the value of~$x$.
			That is, the machine combines the input variables for~$\tup{y}$ with the new set for~$x$.
			Now the machine continues looping.
			
			When the loop is finished,
			the machine outputs the union of the~$\vcA^i_*$
			and the union of the~$\rel^i_*$ is the input-output relation.
			
		\end{itemize}
	Note that we did not use a $\choiceSym$-operation
	and so, by Corollary~\ref{cor:operations-apart-choice-auto-invariant},
	we indeed can use the induction hypothesis
	because we do not change the automorphisms of the input structure.
	Hence, the structure in the cloud is always compatible with the input structure.
	
	We finally have to argue that the constructed choice-free DeepWL+WSC-algorithm runs in polynomial time.
	Let~$p$ be the polynomial given with the CPT+WSC-formula or term.
	In every translation step, the machine executes at most a polynomial number of steps
	because the number of iterations in the iteration term is bounded by~$p$.
	Also, the size of the structure in the cloud is bounded
	by a polynomial in the size of the transitive closure of all sets
	encoded so far, which itself is bounded by a polynomial.
	Whenever a directly nested WSC-fixed-point operator
	is simulated,
	one of the simulating DeepWL+WSC-algorithms~$\dwla_i$
	is called a polynomial number of times.
	Because all the~$\dwla_i$ run in polynomial time,
	the whole simulation runs in polynomial time.
	\end{proof}

	\begin{lemma}
		\label{clm:simulate-cpt-wsc-iteration-term-with-choice}
		Assume that $\formA(\tup{z}) =\wscForm{x}{y}{\termStep(\tup{z}xy)}{\termChoice(\tup{z}x)}{\termWit(\tup{z}xy)}{\formOut(\tup{z}x)}$
		is a CPT+WSC-formula
		and that for every WSC-fixed-point operator~$\formA'$ directly nested in~$\formA$
		there is a polynomial-time DeepWL+WSC-algorithm simulating~$\formA'$.
		Then there is a polynomial-time DeepWL+WSC-algorithm
		simulating~$\formA$.
	\end{lemma}
	\begin{proof}
		We will construct DeepWL+WSC-machines~$\dwlmout$ and~$\dwlmwit$
		in a way such that the DeepWL+WSC-algorithm
		$(\dwlmout, \dwlmwit, \dwla_1, \dots, \dwla_\ell)$
		simulates~$\formA$.
		The DeepWL+WSC-algorithms~$\dwla_i$ are the ones simulating
		the directly nested WSC-fixed-point operators
		and will be used by choice-free DeepWL+WSC-algorithms
		simulating~$\termStep$,~$\termChoice$,~$\termWit$, and~$\formOut$.
		
		The machine~$\dwlmout$ proceeds similarly to the plain iteration term
		in Lemma~\ref{clm:simulate-cpt-wsc-term}:
		For two singleton vertex classes~$\vcA$ and~$\vcB$
		we say that~$\dwlmout$ \defining{creates a pair of}~$\vcA$ and~$\vcB$
		when~$\dwlmout$ executes $\addPairSym$ for the single edge between the $\vcA$- and the $\vcB$-vertex.
		
		Given a singleton vertex class~$\vcA$
		and the step term~$\termStep$,
		the machine first creates a pair of~$\vcA$ and~$\vcA_\emptyset$
		resulting in the singleton vertex class~$\vcA^0$.
		
	 	Now, in the $i$-th iteration,
	 	let~$\vcA^i$ be the singleton vertex class
	 	encoding the current set in the fixed-point computation.
	 	The machine~$\dwlmout$ first checks whether~$i$ exceeds the maximal number of
	 	iterations.
	 	If so, it sets $\vcA_{\text{out}} := \vcA^0$
	 	(the input values paired with the empty set).
	 	Otherwise,
	 	it simulates, using Lemma~\ref{clm:simulate-cpt-wsc-term},
	 	the choice term~$\termChoice$
	 	on input~$\vcA^i$
	 	yielding a singleton vertex class~$\vcA^i_{\termChoice}$
	 	encoding the choice set.
	 	Let~$\vcA^i_{\text{orb}}$ be the vertex class
	 	incident via the containment relation to the ~$\vcA^i_{\termChoice}$-vertex.
	 	Then the machine executes $\choice{\vcA^i_{\text{orb}}}$
	 	yielding a singleton vertex class~$\vcA^i_{\text{ind}}$
	 	(because we choose from a vertex class, we do not have to deal
	 	with the difference of the $\choiceSym$-operation
	 	for directed or undirected relations).
	 	Next,~$\dwlmout$ creates a pair of~$\vcA^i$ and~$\vcA^i_{\text{ind}}$
	 	and simulates the step term~$\termStep$
	 	(Lemma~\ref{clm:simulate-cpt-wsc-term})
	 	with the newly created tuple vertex as input.
	 	This results in a singleton vertex class~$\vcA^i_{\termStep}$.
	 	Then the machine creates again a pair
	 	of~$\vcA$ and~$\vcA^i_{\termStep}$ 
	 	resulting in the singleton vertex class~$\vcA^{i+1}$.
	 	If $\vcA^i=\vcA^{i+1}$ (again it is important
	 	that vertices encoding the same set are not created twice),
	 	then a fixed-point is reached and the machine
	 	sets $\vcA_{\text{out}} := \vcA^{i+1}$.
	 	Otherwise, it starts the next iteration.
	 	Once~$\vcA_{\text{out}}$ is computed
	 	(note that~$\vcA_{\text{out}}$ is always a singleton vertex class because~$\vcA$ is),%
	 	~$\dwlmout$ simulates~$\formA_{\text{out}}$ on input~$\vcA_{\text{out}}$
	 	(again using Lemma~\ref{clm:simulate-cpt-wsc-term}
	 	for which we can w.l.o.g.~assume that~$\formA_{\text{out}}$ is not a WSC-fixed-point operator by considering $\formA_{\text{out}} \land \formA_{\text{out}}$).
	 	If the simulation outputs~$\vcA_{\text{out}}$ (so~$\formA_{\text{out}}$ is satisfied),%
	 	~$\dwlmout$ halts and accepts.
	 	Otherwise,~$\dwlmout$ halts and does not accept.
	 	Note that the polynomial bound is not exceeded because~$\formA$ does not output~$\choiceError$ by assumption.
		
		The machine~$\dwlmwit$ just simulates the term~$\termWit$
		(again by Lemma~\ref{clm:simulate-cpt-wsc-term})
		given the vertex classes~$\vcA^i$ and~$\vcA^n$,
		where~$n$ is the last iteration of the loop.
		Here is a subtle difference between CPT+WSC and DeepWL+WSC.
		The CPT+WSC-term~$\termWit$ gets the $i$\nobreakdash-th and $n$\nobreakdash-th step
		of the fixed-point computation as input
		to witness the choices.
		In DeepWL+WSC, we get the labeled union $\Struct_n \labeledUnion \Struct_i$
		of the final content of the cloud~$\Struct_n$ and the content of the cloud~$\Struct_i$ of the $i$-th $\choiceSym$-execution, which is to be witnessed.
		In our setting, this means
		that $\Struct_n \labeledUnion \Struct_i$
		contains the vertex classes~$\vcA^i$ and~$\vcA^n$
		but also the vertex class~$\vcA^i_{\termChoice}$ 
		encoding the choice set%
		\footnote{If a relation named~$\vcA^i$ is also present in~$\Struct_n$,
		then~$\vcA^i$ is not a singleton vertex class in $\Struct_n \labeledUnion \Struct_i$,
		but we can obtain the correct one by taking the~$\vcA^i$
		vertex which is also contained in the special~$\rel_2$ relation given by the labeled union. We similarly process the other vertex classes.}.
		The witnessing term~$\termWit$ defines automorphisms
		witnessing that~$\vcA^i_{\termChoice}$ is an orbit 
		of~$\Struct_i$ and fixing the~$\vcA^j$ for all $j \in [i]$.
		But in DeepWL, also the choice-sets have to be fixed.
		By Lemma~\ref{lem:cpt-wsc-iso-invariant},
		the CPT+WSC term~$\termChoice$
		is isomorphism-invariant,
		so all the $\vcA^j_{\termChoice}$
		are indeed fixed for all $j \in [i]$, too.
	 	That is, exactly those orbits are witnessed in the computation of
	 	the DeepWL+WSC-algorithm
	 	which are witnessed in the evaluation of the WSC-fixed-point operator
	 	(from which we assumed that all are witnessed).
	 	
	 	Arguing  that the constructed DeepWL+WSC-algorithm
	 	runs in polynomial time
	 	is similar to Lemma~\ref{clm:simulate-cpt-wsc-term}.
	\end{proof}

\begin{lemma}
	\label{lem:cpt-wsc-to-deepwl-wsc}
	Let~$\GraphClass$ be a class of binary $\sig$-structures.
	If a property~$P$ of $\GraphClass$-structures is CPT+WSC-definable,
	then there is a polynomial-time DeepWL+WSC-algorithm deciding~$P$.
\end{lemma}
\begin{proof}
	Let~$\formA$ be a CPT+WSC formula defining the property~$P$.
	Then in particular~$\formA$ never outputs~$\choiceError$ on $\GraphClass$-structures.
	By applying Lemmas~\ref{clm:simulate-cpt-wsc-term} and~\ref{clm:simulate-cpt-wsc-iteration-term-with-choice}
	recursively on the nesting structure of the WSC-fixed-point operators,
	we can translate~$\formA$ into a polynomial time and
	(possibly choice-free) DeepWL+WSC-algorithm~$\dwla$	simulating~$\formA$.
	Whenever needed, we extend choice-free DeepWL+WSC-algorithms
	to DeepWL+WSC-algorithms by adding a witnessing machine
	which immediately halts.
\end{proof}

\subsection{Normalized DeepWL+WSC}
\label{sec:normalized-deepwl}

For our aim to prove Theorem~\ref{thm:cpt+wsc-iso-implies-inv},
we translated an isomorphism-defining CPT+WSC-formula into a polynomial time isomorphism-deciding DeepWL+WSC-algorithm in the last section.
The next goal, which we address in this section,
is to show that the isomorphism-deciding DeepWL+WSC-algorithm
can be turned into a DeepWL+WSC-algorithm computing a complete invariant.

Recall that an isomorphism-deciding DeepWL+WSC-algorithm
gets as input the disjoint union of the two connected structures,
for which it has to decide whether they are isomorphic.
Using the $\addPairSym$- and $\sccSym$-operations,
the algorithms is able to create vertices composed of vertices from both structures,
i.e., it mixes the structures.
The ultimate goal in this section is to show that this mixing is not necessary:
We will show that every DeepWL+WSC-algorithm computing on a disjoint of structures
can be simulated by a DeepWL+WSC-algorithm not mixing the two structures.
We will call such non-mixing algorithms normalized.
A normalized algorithm can compute on each structure separately.
Exploiting this, we will obtain a complete invariant.

We will follow the idea of~\cite{GroheSchweitzerWiebking2021}
to show that every DeepWL+WSC-algorithm can be simulated by a normalized one.
However, we have to differ from the construction in many points.
These changes are crucial in the presence of choices.

\subsubsection{Pure DeepWL+WSC}

As first step to simulate one DeepWL+WSC-algorithm with another DeepWL+WSC-algorithm, we show that we can make some simplifying assumptions on the algorithm to be simulated.
This will simplify simulating it.
We adapt the notion of a pure DeepWL-algorithm from~\cite{GroheSchweitzerWiebking2021} to DeepWL+WSC-algorithms.
\begin{definition}[Pure DeepWL+WSC-algorithm]
	A DeepWL+WSC-algorithm is called \defining{pure}
	if $\addPair{\col}$ and $\scc{\col}$
	are only executed for colors (or fibers)
	and $\refine{\cc}{i}$ and $\choice{\cc}$ are only executed for fibers.
\end{definition}

\begin{lemma}
	\label{lem:pure}
	For every polynomial-time DeepWL+WSC-algorithm deciding a property~$P$ or
	computing a function~$f$,
	there is a pure polynomial-time DeepWL+WSC-algorithm deciding~$P$ or computing~$f$.
\end{lemma}
\begin{proof}
	We show that $\addPair{\relColA}$-, $\scc{\relColA}$-, $\refine{\relColA}{i}$-, and $\choice{\relColA}$-operations can be simulated by
	respective operations applied only to colors and fibers. 
	The case for $\addPair{\relColA}$ and $\scc{\relColA}$ is similar to the proof of Lemma~7 in~\cite{GroheSchweitzerWiebking2021}:
	
	\paragraph{$\addPair{\relColA}$:}
	First decompose~$\relColA$ into its colors $\colA_1, \dots, \colA_k$ using the symbolic subset relation.
	Then execute $\addPair{\colA_i}$ for every $i \in [k]$,
	obtain membership relations~$\rel_i$ for every $i \in [k]$.
	Create a new relation which is the union of all~$\rel_i$
	and serves as membership relation of the simulated  $\addPair{\relColA}$-operation.
	This way, we create the same vertices in the cloud
	as the $\addPair{\relColA}$-execution would do.
	By Lemma~\ref{lem:restrict-sketch}, we can compute
	the algebraic sketch without the relations~$\rel_i$.
	
	\paragraph{$\scc{\relColA}$:}
	By Lemma~4 of~\cite{GroheSchweitzerWiebking2021}, we can compute a relation
	which exactly contains the pairs of vertices in the same $\relColA$-SCC
	using a pure DeepWL-algorithm (and this algorithm only executes $\createSym$
	and no other operations, so can be directly transferred into our DeepWL-definition).
	
	Now, if there is an $\relColA$-SCC which is not discrete, i.e., it contains at least two vertices from the same fiber,
	we can identify a color~$\colA$ refining~$\relColA$,
	such that the $\colA$-SCCs are non-trivial.
	We execute $\scc{\colA}$ and obtain new vertices and a new membership relation~$\relB$.
	Next, we consider the relation~$\relColA'$,
	which consists of
	\begin{enumerate}[label=\alph*)]
		\item the edges of~$\relColA$ not incident to an $\colA$-SCC and 
		\item the edges $(\vertA, \vertC)$ for which
		there is an edge $(\vertA, \vertB)$ in~$\relColA$ such that~$\vertB$ is contained in an $\colA$-SCC and~$\vertC$ is the $\colA$-SCC vertex for~$\vertB$,
		that is,~$\vertB$ is adjacent to~$\vertC$ via the membership relation~$\relB$.
	\end{enumerate}
	The relation~$\relColA'$ is a union of colors, which can be identified using the symbolic subset relation.
	We create~$\relColA'$ using $\createSym$ for the appropriate set of colors.
	If the~$\relColA'$-SCCs are still not all discrete, we repeat the procedure.
	
	Now consider the case that all $\relColA$-SCCs are discrete.
	Conceptually, we want to pick from every $\relColA$-SCC the vertex in the minimal fiber~$\ccA$ as a representative of that SCC.
	We create a copy of the $\ccA$-vertices by executing $\addPair{\ccA}$.
	While iteratively creating the SCC-vertices, we maintain a relation
	which associates the SCC-vertices to the vertices in the $\relColA$-SCCs,
	for which they were created.
	Using this relation, we can define the membership relation.
	
	So, to simulate $\scc{\relColA}$ we do not create exactly the same vertex (as HF-set) because we pick the minimal fibers as representative.
	Regarding witnessing automorphisms, neither picking only a representative nor creating the intermediate SCC-vertices makes a difference:
	Because we did not use a $\choiceSym$-operation to simulate the $\refineSym$-operation, the structure in the cloud has the same automorphisms as the structure that is simulated
	by Corollary~\ref{cor:operations-apart-choice-auto-invariant}.
	Again by Lemma~\ref{lem:restrict-sketch}, we can compute
	the algebraic sketch without the additional vertices and relations.
	
	\paragraph{$\refine{\relColA}{k}$:}
	Let~$\relColA$ consist of the colors $\col_1, \dots , \col_m$,
	which we find using the symbolic subset relation.
	We execute $\refine{\col_1}{k}, \dots, \refine{\col_m}{k}$
	and obtain new relations $\relA_1, \dots, \relA_m$.
	The union of these relations precisely corresponds to the relation
	outputted by $\refine{\relColA}{k}$.
	So we can assume that~$\relColA$ is a color~$\col$.
	
	If~$\col$ is not a fiber, we do the following:
	If~$\col$ is an undirected color,
	then execute $\addUPair{\col}$ and obtain a new vertex class~$\vc$
	and a membership relation~$\relB$.
	We decompose~$\vc$ into its fibers $\ccA_1, \dots, \ccA_\ell$,
	execute $\refine{\ccA_i}{k}$ for every $i \in [\ell]$,
	and obtain the vertex classes $\vcB_1, \dots, \vcB_\ell$
	(in case $\refine{\ccA_i}{k}$ does not create a vertex class,
	i.e., the algorithm used to refine~$\ccA_i$ accepts all vertices in~$\ccA_i$,
	we set $\vcB_i := \ccA_i$).
	Let~$\vcB$ be the union of the~$\vcB_i$ for all $i\in[\ell]$.
	The relation~$\rel$ refining~$\col$ is obtained
	as the set of $\col$-edges with an $(\inv{\relB}, \vc, \relB$)-colored path.
	If~$\rel$ contains the same edges as~$\col$, we do nothing,
	otherwise~$\rel$ is the output of the simulated $\refineSym$-operation.
	
	If~$\col$ is a directed color,
	then execute $\addPair{\col}$ and
	obtain the vertex class~$\vc$ and the component relations~$\relB_{\leftT}$ and~$\relB_{\rightT}$.
	Next,
	we execute $\refine{\vc}{k}$
	(for which we again decompose~$\vc$ into fibers),
	and obtain the vertex class~$\vcB$ refining~$\vc$.
	Similarly to the unordered case, we
	obtain the relation~$\rel$ refining~$\col$
	as $\col$-edges with an $(\inv{\relB_{\leftT}}, \vcB, \relB_{\rightT}$)-colored path.
	If~$\rel$ and~$\col$ contain the same edges, we do nothing,
	otherwise~$\rel$ is the output of the simulated $\refineSym$-operation.
	
	The algorithm used to refine has to be altered
	to also perform the same operations
	to first create a new relation corresponding to the individualized 
	$\vcB$-vertex.
	Again by Corollary~\ref{cor:operations-apart-choice-auto-invariant},
	the simulation does not alter the automorphisms of the structure in the cloud and by Lemma~\ref{lem:restrict-sketch},
	the algebraic sketch without the additional vertices and relations can be computed.

	\paragraph{$\choice{\relColA}$:}
	If $\choice{\relColA}$ is executed for a relation~$\rel$
	consisting of more than one color,
	then~$\rel$ is not an orbit and the given algorithm does not decide a property
	or computes a function because it is going to fail
	(recall that failing is not allowed when deciding a property or computing a function).
	
	So we can assume that~$\relColA$ is a color~$\col$.
	Analogously to the $\refineSym$-case,
	we execute $\addUPair{\col}$ or $\addPair{\col}$,
	depending on whether $\col$ is undirected or not,
	and obtain a vertex class~$\vcA$.
	Next, we execute $\choice{\vcA}$
	yielding a singleton vertex class containing a $\vcA$-vertex~$\vertA$.
	Using the vertex~$\vertA$,
	we can define a relation containing a single (directed or undirected) $\col$-edge analogously as in the $\refineSym$-case.
	
	It is easy to see that a set of automorphisms
	witnesses~$\col$
	(seen as directed or undirected edges depending on whether~$\col$ is directed or not)
	as orbit if and only if it witnesses~$\vcA$ as orbit.
	Additionally, both individualizing an $\col$-edge and individualizing the corresponding $\vcA$-vertex results in an HF-structure with the same automorphisms.
	By Lemma~\ref{lem:restrict-sketch},
	we can compute the algebraic sketch without the additional vertices and relations.
\end{proof}

\subsubsection{Normalized DeepWL+WSC}
\label{sec:normalized-normalized}

We now formalize the notion of an DeepWL+WSC-algorithm
not mixing the two connected components of the structure in the cloud.
To do so, we adapt the notion of a normalized DeepWL-algorithm from~\cite{GroheSchweitzerWiebking2021} to DeepWL+WSC.
Recall that, for the aim of testing isomorphisms,
the input structure is the disjoint union of two connected structures.

\begin{definition} (Normalized HF-Structure)
	An HF-structure~$\Struct$ is \defining{normalized}
	if there are connected HF-structures~$\StructA_1$ and~$\StructA_2$
	such that $\Struct = \Struct_1 \disunion \Struct_2$.
	The~$\Struct_i$ are the \defining{components} of~$\Struct$.
	The vertices of a normalized structure are called \defining{plain}.
	The edges $\plainEdges{\StructA} := \verticesOf{1}{\Struct}^2 \cup \verticesOf{2}{\Struct}^2$ are called \defining{plain}
	and the edges
	$\crossingEdges{\StructA} := \verticesOf{1}{\Struct} \times \verticesOf{2}{\Struct} \cup \verticesOf{2}{\Struct} \times \verticesOf{1}{\Struct}$ are called \defining{crossing}.
	A relation or color is called plain (respectively crossing)
	if it only contains plain 
	(respectively crossing) edges.
\end{definition}
Observe that also the atom set of~$\Struct$
is the disjoint union $\StructVA = \StructVA_1 \disunion \StructVA_2$
and that
the components~$\StructA_1$ and~$\StructA_2$ of~$\StructA$
are unique because the~$\StructA_i$ are connected.
Also observe that if a DeepWL+WSC-machine is executed on $\StructA_1 \disunion \StructA_2$
and at some point during its execution the structure in the cloud is still normalized,
then every vertex~$\vertC$ is contained in $\HF{\StructVA_i}$ for some $i \in [2]$.
That is, one of the following holds:
\begin{itemize}
	\item $\vertC \in \vertices{\StructA_i}$,
	\item $\vertC$ was added by an $\addPairSym$-execution for a pair $(\vertA,\vertB)$ of vertices~$\vertA,\vertB \in \HF{\StructVA_i}$, or 
	\item$\vertC$ was obtained as a vertex for an SCC~$c\subseteq \HF{\StructVA_i}$
	during an $\sccSym$-execution.
\end{itemize}
In particular, $\addPair{\relColA}$- and $\sccSym{\relColA}$-operations were only executed for plain~$\relColA$.

\begin{lemma}[\cite{GroheSchweitzerWiebking2021}]
	The relation containing all plain (respectively crossing)
	edges is DeepWL-computable.
\end{lemma}

\begin{definition}[Normalized DeepWL+WSC]
	A DeepWL+WSC-machine~$\dwlm$ is \defining{normalized}
	if for every normalized HF-structure~$\StructA$,
	the structure in the cloud of~$\dwlm$ is normalized
	all every point in time during the execution of~$\dwlm$ on~$\Struct$.
	A DeepWL+WSC-algorithm $\dwla=(\dwlmout,\dwlmwit, \dwla_1, \dots, \dwla_\ell)$ is normalized
	if $\dwlmout$, $\dwlmwit$, and $\dwla_1, \dots, \dwla_\ell$ are normalized.
\end{definition}
The current definition has a severe issue, namely 
the way in which sets of witnessing automorphisms are encoded by the witnessing machine:
With the encoding of automorphisms we previously described,
a set of automorphisms can only be encoded by non-plain vertices.
So normalized DeepWL+WSC-algorithms need to use a different encoding of automorphisms to witness choices.
For the moment, we do not need to give a precise definition for this encoding because it is irrelevant for the following lemmas.
We will thus define the encoding later in Section~\ref{sec:simulation}
once we have established the required formalism.

Normalized structures have the important property that every crossing color
is actually a ``direct product'' of two fibers.
That is, crossing colors do not provide additional information
and this is the reason why general DeepWL+WSC-algorithms
can ultimately be simulated by normalized ones.

\begin{lemma}[{\cite[Lemma~8]{GroheSchweitzerWiebking2021}}]
	\label{lem:normalized-direct-product}
	Let $\Struct = \StructA_1 \disunion \StructA_2$ be a normalized HF-structure.
	\begin{enumerate}
		\item For every crossing color~$\col$ of~$\Struct$,
		there are two plain fibers~$\ccA$ and~$\ccB$
		such that $\col^\Struct = (\ccA^\Struct \times \ccB^\Struct) \cap \crossingEdges{\Struct}$.
		\item $\coConf{\StructA_i}$ is equal to $\reduct{\coConf{\StructA}[\vertices{\StructA_i}]}{\sigB_i}$ up to renaming colors for every $i \in [2]$, where~$\sigB_i$ is the set of all colors $\col\in\sigB$ for which $\col^\Struct \cap \vertices{\StructA_i}^2 \neq \emptyset$.
		\item There is a polynomial-time algorithm that on input $\sketch{\StructA_1}$ and $\sketch{\StructA_2}$
		computes $\sketch{\Struct}$.
		\item The two sketches $\sketch{\StructA_1}$ and $\sketch{\StructA_2}$
		determine $\sketch{\Struct}$:
		if $\sketch{\StructA} \neq \sketch{\StructB}$ for another HF-structure $\StructB = \StructB_1 \disunion \StructB_2$,
		then $\set{\sketch{\StructA_1},\sketch{\StructA_2}} \neq \set{\sketch{\StructB_1},\sketch{\StructB_2}}$.
	\end{enumerate}
\end{lemma}

\subsubsection{The ``Direct Product'' Property and \texorpdfstring{$\refineSym$}{refine}-Operations}

We now develop tools to deal with normalized DeepWL+WSC-algorithms.
First, we investigate the $\refineSym$-operation.
Intuitively, our goal is to preserve the ``direct product'' property for crossing colors as stated in Lemma~\ref{lem:normalized-direct-product}.
In principle, if we only create plain vertices, this property is preserved,
but a~$\refineSym$-operation can violate it.

\begin{figure}
	\centering
	\def\n{3}
	\def\colors{{"red", "blue", "green"}}
	\subfloat[][]{
		\begin{tikzpicture}
			
			\begin{scope}[yscale=0.75, xscale=0.5]
				\foreach \i in {1,...,\n}{
					\node [vertex, black] (v\i1) at (-2,\i) {};
					\node [vertex, black] (v\i2) at (2,\i)  {};
				}
				\draw [draw, black]
				(-1,0.5*\n+0.5) arc (0:20:5)
				(-1,0.5*\n+0.5) arc (0:-20:5)
				(+1,0.5*\n+0.5) arc (180:160:5)
				(+1,0.5*\n+0.5) arc (180:200:5);
				\foreach \i in {1,...,\n}{
					\foreach \j in {1,...,\n} {
						\path[-,draw, cyan, thick]
						(v\i1) to (v\j2);
					}
				}		
				
			\end{scope}
		\end{tikzpicture}
		\label{fig:refine-and-direct-product-property-sub-begin}	
	}%
	\hspace{3em}%
	\subfloat[][]{
		\begin{tikzpicture}
				\begin{scope}[yscale=0.75, xscale=0.5]
				
				\foreach \i in {1,...,\n}{
					\node [vertex, black] (v\i1) at (-2,\i) {};
					\node [vertex, black] (v\i2) at (2,\i)  {};
				}
				\draw [draw, black]
				(-1,0.5*\n+0.5) arc (0:20:5)
				(-1,0.5*\n+0.5) arc (0:-20:5)
				(+1,0.5*\n+0.5) arc (180:160:5)
				(+1,0.5*\n+0.5) arc (180:200:5);
				
				\foreach \i in {1,...,\n}{
					\foreach \j in {1,...,\n} {
						\ifthenelse{\i=\j}{
						}{
							\path[-,draw, cyan, thick]
							(v\i1) to (v\j2);
						}
					}
				}
				\foreach \i in {1,...,\n}{
					\path[-,draw, magenta, thick]
					(v\i1) to (v\i2);
				}
			\end{scope}
		\end{tikzpicture}	
		\label{fig:refine-and-direct-product-property-sub-refined-relation}
	}%
	\hspace{3em}%
	\subfloat[][]{
		\begin{tikzpicture}
			\begin{scope}[yscale=0.75, xscale=0.5]
				\foreach \i in {1,...,\n}{
					\pgfmathparse{\colors[\i-1]};
					\let\c\pgfmathresult
					\node [vertex, \c] (v\i1) at (-2,\i) {};
					\node [vertex, \c] (v\i2) at (2,\i)  {};
				}
				\draw [draw, black]
				(-1,0.5*\n+0.5) arc (0:20:5)
				(-1,0.5*\n+0.5) arc (0:-20:5)
				(+1,0.5*\n+0.5) arc (180:160:5)
				(+1,0.5*\n+0.5) arc (180:200:5);
				
				\foreach \i in {1,...,\n}{
					\foreach \j in {1,...,\n} {
						\ifthenelse{\i=\j}{
						}{
							\path[-,draw, cyan, thick]
							(v\i1) to (v\j2);
						}
					}
				}
				\foreach \i in {1,...,\n}{
					\path[-,draw, magenta, thick]
					(v\i1) to (v\i2);
				}
			\end{scope}
		\end{tikzpicture}
		\label{fig:refine-and-direct-product-property-sub-refined-vertices}
	}%
	\hspace{3em}%
	\subfloat[][]{
		\begin{tikzpicture}
			\begin{scope}[yscale=0.75, xscale=0.5]
				\foreach \i in {1,...,\n}{
					\pgfmathparse{\colors[\i-1]};
					\let\c\pgfmathresult
					\node [vertex, \c] (v\i1) at (-2,\i) {};
					\node [vertex, \c] (v\i2) at (2,\i)  {};
				}
				\draw [draw, black]
				(-1,0.5*\n+0.5) arc (0:20:5)
				(-1,0.5*\n+0.5) arc (0:-20:5)
				(+1,0.5*\n+0.5) arc (180:160:5)
				(+1,0.5*\n+0.5) arc (180:200:5);
				
				\foreach \i in {1,...,\n}{
					\foreach \j in {1,...,\n} {
						\ifthenelse{\i=\j}{
						}{
							\pgfmathparse{\colors[\i-1]};
							\let\ca\pgfmathresult
							\pgfmathparse{\colors[\j-1]};
							\let\cb\pgfmathresult
							\path[-,draw=\ca!50!\cb!50!white, thick]
							(v\i1) to (v\j2);
						}
					}
				}
				\foreach \i in {1,...,\n}{
					\pgfmathparse{\colors[\i-1]};
					\let\c\pgfmathresult
					\path[-,draw, \c!80!black, thick]
					(v\i1) to (v\i2);
				}
			\end{scope}
		\end{tikzpicture}
		\label{fig:refine-and-direct-product-property-sub-decompose}
	}
	\caption{
		An example for $\refineSym$-operations and the ``direct product'' property:
		In~\protect\subref{fig:refine-and-direct-product-property-sub-begin}, a coherent structure $\StructA_1 \disunion\StructA_2$ with a black fiber
		satisfying the ``direct product'' property,
		i.e., all edges between black vertices are light blue.
		Assume that in each component the black vertices can be ordered.
		In~\protect\subref{fig:refine-and-direct-product-property-sub-refined-relation}, the structure  obtained after
		refining the light blue color with an algorithm
		accepting exactly the edges between the $i$-th black vertices in both structures.
		This yields a red color.
		The structure is again coherent but does not have the ``direct product'' property.
		Ordering the vertices is not passed through the $\refineSym$-operation.
		In~\protect\subref{fig:refine-and-direct-product-property-sub-refined-vertices},
		the structure  resulting after
		first ordering the black vertices 
		and then executing the same $\refineSym$-operation.
		This structure has the ``direct product'' property:
		\protect\subref{fig:refine-and-direct-product-property-sub-decompose} shows the decomposition of the light blue relation and the red relation
		into colors.
	}
	\label{fig:refine-and-direct-product-property}	
\end{figure}
We give an illustrating example (cf.\ Figure~\ref{fig:refine-and-direct-product-property}):
Assume we are given a normalized HF-structure $\Struct = \StructA_1\disunion\StructA_2$,
both components of the same size,
where all vertices are in the same fiber and consequently all crossing edges are in the same color~$\col$ (Figure~\ref{fig:refine-and-direct-product-property-sub-begin})
Further, assume that some DeepWL+WSC-algorithm can linearly order the vertices in each component.
Then we can define an algorithm~$\dwla$ with the following property. For $\set{\vertA,\vertB} \subseteq \col^\Struct$, the algorithm accepts $(\StructA_1\disunion\StructA_2, \set{\vertA,\vertB})$
if and only if~$\vertA$ and~$\vertB$ are each the $i$-th vertex in their component for some~$i$.
If we refine~$\col$ with~$\dwla$,
a relation~$\rel$ containing a perfect matching between the components is added (Figure~\ref{fig:refine-and-direct-product-property-sub-refined-relation}).
But the information, that the atoms can be ordered, is not ``returned''.
The result is again a coherent configuration with only one fiber and so the ``direct product'' property is violated.
To avoid this problem, we need to distinguish the atoms first and
then execute the $\refineSym$-operation (Figure~\ref{fig:refine-and-direct-product-property-sub-refined-vertices}).

We will show now that this strategy generalizes to  arbitrary $\refineSym$-operations.
To do so, we need to establish the following technical lemmas.
The first lemma states that instead of computing on $\StructA_1 \disunion \StructA_2$,
we can also compute on the structure $(\StructA_1 \disunion \StructA_2, \StructVA_1)$
(that is, we add a vertex class containing all vertices of~$\StructA_1$),
where we artificially distinguish the two components,
i.e.,~we remove automorphisms exchanging the components.

\begin{lemma}
	\label{lem:normalized-distinguish-components}
	For every normalized DeepWL+WSC-algorithm~$\dwla$,
	there is a normalized algorithm~$\hat{\dwla}$ 
	such that for all connected HF-structures~$\StructA_1$ and~$\StructA_2$ the algorithm~$\dwla$ accepts (or respectively rejects) $\StructA_1 \disunion \StructA_2$
	if and only if~$\hat{\dwla}$ accepts $(\StructA_1 \disunion \StructA_2, \StructVA_1)$.
	Polynomial running time is preserved.
	For every HF-structures $\StructA_1$,~$\StructA_2$,~$\StructB_1$, and~$\StructB_2$ it holds that
	\begin{align*}
		\text{if } \run{\dwla}{\StructA_1 \disunion \StructA_2} &\neq \run{\dwla}{\StructB_1 \disunion \StructB_2},\\
		\text{then } \run{\hat{\dwla}}{(\StructA_1 \disunion \StructA_2, \StructVA_1)} &\neq \run{\hat{\dwla}}{(\StructB_1 \disunion \StructB_2, \StructVB_1)}.
	\end{align*}
\end{lemma}
\begin{proof}
	The proof is by induction on the nesting depth of DeepWL+WSC-algorithms.
	Let $\dwla = (\dwlmout, \dwlmwit, \dwla_1, \dots, \dwla_\ell)$
	be a normalized DeepWL+WSC-algorithm.
	By induction hypothesis, let
	$\hat{\dwla}_1, \dots, \hat{\dwla}_\ell$ be normalized DeepWL+WSC-algorithms
	satisfying the claim for the algorithms $\dwla_1, \dots, \dwla_\ell$.
	
	We construct DeepWL+WSC-machines~$\hatdwlmout$ and~$\hatdwlmwit$,
	i.e., for $\dwlm \in \set{\dwlmout,\dwlmwit}$
	we construct a DeepWL+WSC-machine~$\hat{\dwlm}$.
	Let $\StructA_1\disunion\StructA_2$ be a normalized HF-structure
	on which~$\dwlm$ is executed.
	By Lemma~\ref{lem:restrict-sketch},~$\hat{\dwlm}$ can compute
	the algebraic sketch $\sketch{\StructA_1\disunion\StructA_2}$
	from $\sketch{(\StructA_1\disunion\StructA_2, \StructVA_1)}$
	and thus track the run of~$\dwlm$.
	Note that by Lemma~\ref{lem:normalized-direct-product},
	the sketches of the individual components 
	in $\sketch{\StructA_1\disunion\StructA_2}$
	and $\sketch{(\StructA_1\disunion\StructA_2, \StructVA_1)}$ are equal possibly up to renaming of the colors.
	Additionally, the crossing colors of $(\StructA_1\disunion\StructA_2, \StructVA_1)$ are always directed.
	
	If~$\dwlm$ executes $\create{\pi}$ and~$\pi$ contains a color~$\col$ occurring in both components,
	then~$\hat{\dwlm}$ uses the two colors~$\col_1$ and~$\col_2$, each occurring in one component, whose union is equal to~$\col$.
	Every $\addPairSym$-, $\sccSym$-, and $\refineSym$-operation executed by~$\dwlm$
	is executed in the same way by~$\hat{\dwlm}$ respecting the renamed colors:
	whenever~$\dwlm$ uses a color or relation, which is now split into 
	two sets (one in each component),~$\hat{\dwlm}$
	creates the union of these colors/relations and uses this union.
	The $\refineSym$-operations yield the correct result by the induction hypothesis.
	To continue to track the run of~$\dwlm$,
	we again exploit Lemma~\ref{lem:restrict-sketch}
	to compute the sketch without the additional created  relations.

	Whenever~$\dwlm$ executes $\choice{\col}$
	(if $\choiceSym$ is executed for a relation consisting of multiple relations,
	then the algorithm will fail),
	we make the following case distinction.
	Let $\StructA' = (\StructA_1' \disunion \StructA_2', \StructVA_1)$
	be the current content of the cloud. 
	If~$\col$ is a crossing color,~$\hat{\dwlm}$ executes $\choice{\colB}$
	for the color~$\colB$ satisfying $\colB^{\Struct'} = \col^{\Struct'} \cap (\StructVA_1' \times \StructVA_2')$ (which is DeepWL-computable).
	If~$\col$ was a directed color, there is nothing more to do
	because the individualized edge $(\vertA,\vertB)$ is contained in $\col^{\Struct'}$.
	If~$\col$ was an undirected color,
	then a directed edge $(\vertA,\vertB)$ 
	instead of the undirected edge $\set{\vertA,\vertB}$
	is individualized.
	Recall that $(\vertA,\vertB)$ is individualized by creating two singleton vertex classes and that $\set{\vertA,\vertB}$ is individualized by creating
	a two-element vertex class.
	So we create the union of the two singleton color classes.
	This corresponds to individualizing the undirected edge
	and we can continue to track the run of~$\dwlm$.
	
	If~$\col$ is a plain color containing edges of both components,~$\hat{\dwlm}$ executes $\choice{\colB}$
	for the color~$\colB$ satisfying $\colB^{\StructA'} = \col^{\StructA'} \cap (\StructVA_1')^2$.
	Otherwise,~$\col$ is a plain color containing only edges of one component
	and~$\hat{\dwlm}$ simply executes $\choice{\col}$.
	It is easy to see that if~$\col$ is an orbit of $\StructA_1' \disunion \StructA_2'$, then~$\colB$ is an orbit of $(\StructA_1' \disunion \StructA_2', \StructVA_1)$.
	Again by Lemma~\ref{lem:restrict-sketch},
	we compute the sketch without the additional created relations
	to track the run of~$\dwlm$.
	
	We finally have to modify the witnessing machine~$\hatdwlmwit$.
	Independent of the precise way we will encode sets of witnessing automorphisms for normalized DeepWL+WSC-machines,
	we may proceed as follows (an explicit description of the encoding is given in Section~\ref{sec:simulation}):~%
	$\hatdwlmwit$ first computes the set of automorphisms given by~$\dwlmwit$.
	Then it removes the automorphisms
	exchanging the two components
	and outputs the set of remaining automorphisms.
	This way, all choices are witnessed
	if all choices in the execution of~$\dwla$ are witnessed. 
	
	Let $\StructA_1, \StructA_2, \StructB_1$, and $\StructB_2$ be connected HF-structures.
	Assume $\run{\dwla}{\StructA_1 \disunion \StructA_2} \neq \run{\dwla}{\StructB_1 \disunion \StructB_2}$.
	The machine~$\hat{\dwlm}$ tracks the run of~$\dwlm$.
	Hence, the run of~$\hat{\dwlm}$ contains the configurations of the Turing machines in~$\dwlm$ and also the algebraic sketches in the run of~$\dwlm$. These sketches
	were computed by~$\hat{\dwlm}$ using Lemma~\ref{lem:restrict-sketch}.
	So $\run{\hat{\dwla}}{\StructA_1 \disunion \StructA_2}$
	can be extracted from $\run{\hat{\dwla}}{(\StructA_1 \disunion \StructA_2, \StructVA_1)}$ (and likewise for $\StructB_1 \disunion \StructB_2$).
	We conclude that $\run{\hat{\dwla}}{(\StructA_1 \disunion \StructA_2, \StructVA_1)} \neq \run{\hat{\dwla}}{(\StructB_1 \disunion \StructB_2, \StructVB_1)}$.
\end{proof}

We now identify sets of runs
as a possibility to distinguish vertices,
so that 
executing a $\refineSym$-operation on a crossing relation
preserves the ``direct product'' property.

\begin{lemma}
	\label{lem:distinguish-cross-relation-distinguish-vertices}
	For every normalized DeepWL+WSC-algorithm~$\dwla$,
	every normalized HF-structure $\Struct =\Struct_1 \disunion \Struct_2$,
	every crossing color $\col^\Struct = (\ccA^\Struct \times \ccB^\Struct)\cap \crossingEdges{\Struct}$
	for some fibers~$\ccA$ and~$\ccB$,
	every $\vertA \in \ccA^\Struct \cap \verticesOf{i}{\Struct}$ for some $i \in [2]$,
	and
	every $\set{\vertA,\vertB},\set{\vertA,\vertB'} \in \col^\Struct\cup (\inv{\col})^\Struct$,
	we have the following:
	I~$\dwla$ accepts $(\Struct, \set{\vertA,\vertB})$
	and~$\dwla$ does not accept $(\Struct, \set{\vertA,\vertB'})$,
	then \[\setcond[\big]{\run{\dwla}{(\Struct, \set{\vertC,\vertB})}}{\vertC \in \ccA \cap \verticesOf{i}{\Struct}} \neq
	\setcond[\big]{\run{\dwla}{ (\Struct, \set{\vertC,\vertB'})}}{\vertC \in \ccA \cap \verticesOf{i}{\Struct}}.\] 
\end{lemma}
\begin{proof}
	Let~$\dwla$ be a normalized DeepWL+WSC-algorithm,
	$\Struct =\Struct_1 \disunion \Struct_2$ be a normalized HF-structure,
	and suppose $\col$, $\vertA\in \ccA^\Struct \cap \verticesOf{i}{\Struct}$, $(\vertA,\vertB)$, and $\set{\vertA,\vertB}$ satisfy the conditions of the lemma.
	Assume that~%
	$\dwla$ accepts $(\Struct, \set{\vertA,\vertB})$ and~%
	$\dwla$ does not accept $(\Struct, \set{\vertA,\vertB'})$.
	We show $\run{\dwla}{(\Struct, \set{\vertA,\vertB})} \notin \setcond{\run{\dwla}{(\Struct, \set{\vertC,\vertB'})}}{\vertC \in \ccA^\StructA \cap \verticesOf{i}{\Struct}}$.
	Assume for sake of contradiction that there is a vertex $\vertA' \in \ccA^\StructA \cap \verticesOf{i}{\Struct}$
	such that $\run{\dwla}{(\Struct, \set{\vertA,\vertB})} = \run{\dwla}{(\Struct, \set{\vertA',\vertB'})}$.
	At some point during the execution of~$\dwla$,
	$\run{\dwla}{(\Struct, \set{\vertA',\vertB'})}$
	has to differ from $\run{\dwla}{(\Struct, \set{\vertA,\vertB'})}$
	because $\run{\dwla}{(\Struct, \set{\vertA,\vertB})} = \run{\dwla}{(\Struct, \set{\vertA',\vertB'})}$ is accepting
	and $\run{\dwla}{(\Struct, \set{\vertA,\vertB'})}$ is not accepting.
	Because~$\dwla$ is normalized, at every point in time the current content of the cloud~$\StructB$ satisfies $\vertices{\StructB} = \verticesOf{1}{\StructB} \cup  \verticesOf{2}{\StructB}$.
	Before the first moment where the two runs differ,
	the component of~$\vertB'$ is equal in both clouds
	because the same operations were executed
	and the initial components of~$\vertB'$ were equal
	(in $(\Struct, \set{\vertA,\vertB'})$ and in $(\Struct, \set{\vertA',\vertB'})$
	the same vertex is individualized in this component).
	In particular, the components have the same sketches.
	So~$\dwla$ can only have a run on $\set{\vertA',\vertB'}$ that is different from its run on $\set{\vertA,\vertB'}$
	if the sketches of the components of~$\vertA$ and~$\vertA'$ differ
	(Lemma~\ref{lem:normalized-direct-product}).
	But this contradicts that $\run{\dwla}{(\Struct, \set{\vertA,\vertB})} = \run{\dwla}{ (\Struct, \set{\vertA',\vertB'})}$
	because the sketches are contained in the run.
\end{proof}

Now we want to compute the set $\setcond{\run{\dwla}{(\Struct, \set{\vertC,\vertB})}}{\vertC \in \ccA \cap \verticesOf{i}{\Struct}}$
with a DeepWL+WSC-algorithm.
We start with the case where a directed edge $(\vertC,\vertB)$
instead of the undirected one $\set{\vertC,\vertB}$ 
is individualized.
\begin{lemma}
	\label{lem:compute-distinguish-vertices}
	For every normalized DeepWL+WSC-algorithm  $\dwla$,
	every normalized HF-structure $\Struct =\Struct_1 \disunion \Struct_2$,
	every crossing color $\col^\Struct = (\ccA^\Struct \times \ccB^\Struct)\cap \crossingEdges{\Struct}$
	for some fibers~$\ccA$ and~$\ccB$, and
	every $\vertB \in \ccB^\Struct \cap \verticesOf{j}{\Struct}$ for some $\set{i,j}=[2]$,
	there is a normalized DeepWL+WSC-algorithm
	computing on input~$\vertB$ (given as a singleton vertex class) the set
	\[\setcond[\big]{\run{\dwla}{(\Struct, \vertC\vertB)}}{\vertC \in \ccA \cap \verticesOf{i}{\StructA}}\]
	if for every $(\vertA,\vertB) \in \col^\Struct$ 
	the algorithm~$\dwla$ does not fail
	on input $(\Struct, \vertA\vertB)$. 
	Polynomial runtime of the algorithm is preserved.
\end{lemma}
\begin{proof}
	Let~$\dwla$ be a normalized DeepWL+WSC-algorithm.
	For~$\dwla$, let~$\dwla'$ be a normalized DeepWL+WSC-algorithm,
	which on input~$i$,~$\vertB$, and~$\vertC$ (given as singleton vertex classes)
	simulates~$\dwla$ and
	decides whether the $i$-th position of the run $\run{\dwla}{(\Struct, \vertC\vertB)}$ is a~$1$ if~$\dwla$ does not fail.
	Similarly, let~$\dwla''$ be a normalized DeepWL+WSC-algorithm,
	which on input~$i$,~$\vertB$, and~$\vertC$
	decides whether $|\run{\dwlm}{(\Struct, \vertC\vertB)}| \geq i$
	unless~$\dwlm$ fails.
	
	Let $\Struct =\Struct_1 \disunion \Struct_2$ be a normalized HF-structure,
	$\col^\Struct = (\ccA^\Struct \times \ccB^\Struct)\cap \crossingEdges{\Struct}$ be a crossing color
	and $\vertB \in \ccB^\Struct \cap \verticesOf{j}{\Struct}$ for some $\set{i,j}=[2]$
	such that~$\dwla$ does not fail on input $(\StructA, \vertA\vertB)$
	for every $(\vertA,\vertB) \in \col^\Struct$.

	We construct a machine 
	which computes the runs $\run{\dwlm}{(\Struct, \vertC\vertB)}$
	in parallel for all $\vertC \in \ccA^\Struct$.
	The runs are encoded as binary 0/1-strings,
	which are encoded by vertex classes $\vcA_1, \dots, \vcA_m$,
	and $\vcB_1, \dots, \vcB_m$,
	where~$\vcA_i$ contains all $\ccA$-vertices
	for which the $i$-th bit of the run is a~$1$,
	and~$\vcB_i$ contains all $\ccA$-vertices for which
	the run has length at least~$i$.
	
	Initialize $\vcB_{0} := \cc$,
	that is,~$\vcB_{0}$ contains exactly the $\cc$-vertices
	on which the run of~$\dwla$ has length at least zero.
	Starting from $i=0$,
	we first determine the vertex class~$\vcB_{i+1}$ 
	by executing $\refine{\vcB_{i}}{\dwla'',i +1}$
	(we actually have to execute $\refine{\vcB_{i}}{k,i +1}$,
	where~$k$ is the index of~$\dwla''$
	in the list of nested subalgorithms).
	By definition of~$\dwla''$,~%
	$\vcB_{i+1}$  contains exactly the $\cc$-vertices,
	on which the run of~$\dwla$ has length at least~$i$.
	If $\vcB_{i+1} = \emptyset$, we stop because we computed all runs.
	Otherwise, we determine~$\vcA_{i+1}$ as a subset of~$\vcB_{i+1}$ by executing
	$\refine{\vcB_{i+1}}{\dwla', i+1}$ (again, we need to replace~$\dwla'$ by the corresponding number).
	By definition of~$\dwla'$, the class~$\vcA_{i+1}$ contains exactly the $\cc$-vertices
	for which the $(i+1)$-th bit in the run of~$\dwlm$ is a~$1$.
	Finally, we increment~$i$ and repeat.
	
	We now have encoded the runs $\run{\dwlm}{(\Struct, \vertC\vertB)}$
	with the vertex classes $\vcA_1, \dots, \vcA_m$
	and $\vcB_1, \dots, \vcB_m$.
	So we can determine which runs occur.
	Now it is easy to compute the set $\setcond{\run{\dwla}{(\Struct, \vertC\vertB)}}{\vertC \in \ccA}$ and to write it onto the tape
	because the runs can be ordered lexicographically.
	
	Regarding the runtime,
	its easy to see that~$\dwla'$ and~$\dwla''$
	run in polynomial time
	because they just simulate~$\dwla$ and check the run of~$\dwla$.
	Because the length of the run of~$\dwla$ is bounded by a polynomial,
	so is the number of iterations
	and the number of needed relations.
	That is, also the size of the algebraic sketch is bounded by a polynomial.
\end{proof}

\begin{lemma}
	\label{lem:compute-distinguish-vertices-function}
	For every normalized DeepWL+WSC-algorithm ~$\dwla$,
	every normalized HF-structure $\Struct =\Struct_1 \disunion \Struct_2$,
	every crossing color $\col^\Struct = (\ccA^\Struct \times \ccB^\Struct)\cap \crossingEdges{\Struct}$
	for some fibers~$\ccA$ and~$\ccB$,
	and every $i \in [2]$,
	we have the following:
	There is a normalized DeepWL+WSC-algorithm computing a function~$f$
	such that 
	if for every $(\vertA,\vertB) \in \col^\Struct$ 
	the algorithm~$\dwla$ does not fail
	on input $(\Struct, \vertA\vertB)$,
	then
	\begin{align*}
		f(\vertB) &\neq f(\vertB') 	\text{ whenever}\\
		\setcond[\big]{\run{\dwla}{(\Struct, \set{\vertC,\vertB}))}}{\vertC \in \ccA \cap \verticesOf{i}{\Struct}} &\neq
		\setcond[\big]{\run{\dwla}{ (\Struct, \set{\vertC,\vertB'})}}{\vertC \in \ccA \cap \verticesOf{i}{\Struct}}
	\end{align*}
	for every $\vertB,\vertB' \in \ccB^\Struct \cap\verticesOf{j}{\Struct}$
	such that $\set{i,j}=[2]$
	(the input vertex to~$f$ is given as singleton vertex class).
	Polynomial runtime is preserved.
\end{lemma}
\begin{proof}
	Let~$\dwla$ be a normalized DeepWL+WSC-algorithm and
	let~$\hat{\dwla}$ be the normalized DeepWL+WSC-algorithm given by Lemma~\ref{lem:normalized-distinguish-components} for~$\dwla$,
	that is,~$\hat{\dwla}$ accepts $(\Struct, \vertA\vertB)$
	if and only if~$\dwla$ accepts $(\Struct, \set{\vertA,\vertB})$
	for every $(\vertA,\vertB) \in \col^\Struct$
	for every crossing color $\col^\Struct = (\ccA^\Struct \times \ccB^\Struct)\cap \crossingEdges{\Struct}$ of every normalized HF-structure $\Struct$.
	Note here that individualizing~$\vertA$ implicitly distinguishes the two components and we do not need to create a vertex class for~$\StructVA_1$.
	We then,  given~$\vertB$ as singleton vertex class, compute using Lemma~\ref{lem:compute-distinguish-vertices}
	the set $f(\vertB) := \setcond{\run{\hat{\dwla}}{(\Struct, \vertC\vertB)}}{\vertC \in \ccA \cap \verticesOf{i}{\StructA}}$.
	By the properties of~$\hat{\dwla}$ granted by Lemma~\ref{lem:normalized-distinguish-components}, it
	follows that 
	\begin{align*}
		\text{if }\hspace{-0.5pt}\setcond[\big]{\run{\dwla}{(\Struct, \set{\vertC,\vertB})}}{\vertC \in \ccA \cap \verticesOf{i}{\Struct}} &\neq
	\setcond[\big]{\run{\dwla}{ (\Struct, \set{\vertC,\vertB'})}}{\vertC \in \ccA \cap \verticesOf{i}{\Struct}}, \\
	\text{then }\hspace{-0.5pt}\setcond[\big]{\run{\hat{\dwla}}{(\Struct, (\vertC,\vertB))}}{\vertC \in \ccA \cap \verticesOf{i}{\Struct}} &\neq
	\setcond[\big]{\run{\hat{\dwla}}{ (\Struct, (\vertC,\vertB'))}}{\vertC \in \ccA \cap \verticesOf{i}{\Struct}}.
	\end{align*}
\end{proof}

Finally, we want to use this function~$f$ to refine vertex classes.
That is, we want to execute a $\refine{\vc}{f}$-operation,
which splits a vertex class~$\vc$
such that two $\vc$-vertices~$\vertA$ and~$\vertB$
end up in different classes if and only if $f(\vertA) \neq f(\vertB)$.

\begin{lemma}
	\label{lem:deepwl-refine-with-function}
	Let~$\dwla$ be a normalized and  polynomial-time DeepWL+WSC-algorithm,
	which computes a function~$f$.
	Then we can simulate a $\refine{\vc}{f}$-execution 
	using a normalized polynomial-time DeepWL+WSC-algorithm.
\end{lemma}
\begin{proof}
	The proof is similar to the one of Lemma~\ref{lem:compute-distinguish-vertices}.
	We use vertex classes encoding the values of~$f$.
	We create a DeepWL+WSC-algorithm~$\dwla'$,
	which for every normalized structure~$\StructA$
	and every vertex class~$\vc$
	takes a $\vc$-vertex $\vertA$ (as singleton vertex class) and an
	additional number~$i$ as input.
	It decides whether the $i$-th bit of~$f(\vertA)$ is~$1$.
	We create another algorithm~$\dwla''$
	taking  a $\vc$-vertex~$\vertA$ and a number~$i$ and decides whether the $|f(\vertA)| \geq i$.
	
	Use~$\dwla'$ and~$\dwla''$ iteratively to obtain vertex classes refining~$\vc$
	encoding $f(\vertA)$ for every $\vc$-vertex~$\vertA$ analogously to the proof of Lemma~\ref{lem:compute-distinguish-vertices}.
	We then create new vertex classes for distinct output of~$f$
	containing the $\vc$-vertices of that value.
	
	Because~$f$ is computed by a polynomial-time machine,
	the length of~$f$ is bounded by a polynomial
	and so again is the number of iterations and created relations.
\end{proof}

In Section~\ref{sec:simulation},
we will use the lemmas of this section to simulate a
$\refineSym$-operation as already discussed:
first split the vertex classes properly suitably and then execute the $\refineSym$-operation.

\subsubsection{Building Plans}

To simulate arbitrary DeepWL+WSC-algorithms by normalized ones,
we need to compute the algebraic sketch of the non-normalized HF-structure in the cloud of the simulated algorithm
from the normalized HF-structure in the cloud of the simulating algorithm.
For this, we introduce the notion of a building plan.

\newcommand{\buildplan}{\Omega}
\newcommand{\applyplan}[2]{#1(#2)}

\newcommand{\crossVertexRel}[1]{\tilde{\rel}_{#1}}

\newcommand{\applyplancross}[2]{#1^{\text{cross}}(#2)}

\begin{definition}[Building Plan]
	Let $\Struct = \Struct_1 \disunion \Struct_2$
	be a normalized HF-structure.
	A \defining{vertex plan} for~$\Struct$
	is a set $\set{\vcA,\vcB}$ of vertex classes $\vcA,\vcB \in \sigA$. 
	A \defining{relation plan} for~$\Struct$
	defining a relation symbol $\rel \notin \sigA$
	is a pair $(\rel, \setcond{\set{\relB_{1,i}, \relB_{2,i}}}{i \in [k]})$,
	where $\relB_{j,i} \in \sigA$ for all $j \in [2]$ and $i \in [k]$.
	A \defining{building plan} for~$\Struct$
	is a pair $\buildplan = (\omega, \kappa)$,
	where~$\omega$ is a finite set of vertex plans
	and~$\kappa$ is a finite set of relation plans
	where each relation plan defines a different relation symbol.
\end{definition}

Intuitively, a building plan~$\buildplan$ describes a non-normalized HF-structure
$\applyplan{\buildplan}{\Struct}$ in terms of a normalized HF-structure~$\Struct$
and a recipe to construct vertices mixing the two components of~$\Struct$.
A vertex plan $\set{\vcA, \vcB}$
says that for every $\vertA \in \vcA^\Struct$ and $\vertB \in \vcB^\Struct$
in different components of~$\Struct$,
the vertex  $\set{\vertA,\vertB}$ is added to the structure
(recall that we are defining an HF-structure).
A relation plan $(\rel, \setcond{\set{\relB_{1,i}, \relB_{2,i}}}{i \in [k]})$
specifies a new relation~$\rel$ between the created vertices:
a pair $(\set{\vertA,\vertB},\set{\vertA',\vertB'})$
is contained in $\rel^{\applyplan{\buildplan}{\Struct}}$
if and only if there is some $i \in [k]$
such that $(\vertA,\vertA') \in \relB_{1,i}^\Struct$
and $(\vertB,\vertB') \in \relB_{2,i}^\Struct$.
Lastly, a special relation~$\pRel$ relating every vertex $\set{\vertA,\vertB}$ that is created to the original vertices~$\vertA$ and~$\vertB$ is added.
Formally, the structure $\applyplan{\buildplan}{\Struct} $ is defined as follows:

\begin{definition}[Structure Defined by a Building Plan]
Let $\buildplan = (\omega, \kappa)$ be a building plan for a normalized HF-structure~$\StructA$.
The structure $\applyplan{\buildplan}{\Struct}$ is defined as follows:
For every vertex plan $\set{\vcA,\vcB} \in \omega$,
define
\begin{align*}
	\rel_{\set{\vcA,\vcB}}^\Struct &:= 
	(\vcA^\Struct \times \vcB^\Struct \cup \vcB^\Struct \times \vcA^\Struct) \cap \crossingEdges{\Struct},\\
	\crossVertexRel{\set{\vcA,\vcB}}^\Struct &:=
	\setcond*{\set{\vertA,\vertB}}{(\vertA,\vertB) \in \rel_{\set{\vcA,\vcB}}^\Struct},
\end{align*}
and $\crossVertexRel{\omega}^\Struct :=  \bigcup_{\set{\vcA,\vcB} \in \omega} \crossVertexRel{\set{\vcA,\vcB}}^\Struct$.
The atoms of $\applyplan{\buildplan}{\Struct}$ are the ones of~$\Struct$,
the vertices are $\vertices{\applyplan{\buildplan}{\Struct}} := \vertices{\Struct} \cup \crossVertexRel{\omega}^\Struct$, and
the signature of $\applyplan{\buildplan}{\Struct}$ is
$\sigA \disunion \set{\pRel} \disunion \setcond{\rel}{(\rel,M) \in \kappa}$.
The relations are defined via
\begin{align*}
	\pRel^{\applyplan{\buildplan}{\Struct}} &:= 
	\bigcup_{\set{\vertA,\vertB} \in \crossVertexRel{\omega}^\Struct} \set[\big]{(\set{\vertA,\vertB}, \vertA), (\set{\vertA,\vertB}, \vertB)},\\
	\rel^{\applyplan{\buildplan}{\Struct}} &:=
	\setcond*{(\set{\vertA,\vertB},\set{\vertA',\vertB'})}{
		(\vertA,\vertA') \in \relB_{1,i}^\Struct,
		(\vertB,\vertB') \in \relB_{2,i}^\Struct,
		i \in [k]}\\ 
	&\hspace{5cm} \text{for every } (\rel, \setcond{\set{\relB_{1,i}, \relB_{2,i}}}{i \in [k]}) \in \kappa.
\end{align*}
The added vertices $\set{\vertA,\vertB}$ are called \defining{crossing}.
The set of all crossing vertices of $\applyplan{\buildplan}{\StructA}$ is $\crossingVertices{\StructA}$.
We call edges $(\vertC,\vertC') \in \crossingVertices{\StructA}^2$
\defining{inter-crossing}.
A relation or color is  inter-crossing if it only contains inter-crossing edges. A vertex class or fiber is called crossing, if it only contains crossing vertices.
\end{definition}
Later, we are interested only in the substructure of $\applyplan{\buildplan}{\Struct}$ induced by the crossing vertices.
But for HF-structures, this is actually ill-defined
because the crossing vertices do not contain the atoms of~$\Struct$.
So we turn to the non-HF-structure $\nonHF{\applyplan{\buildplan}{\Struct}}$
and define 
\[\applyplancross{\buildplan}{\Struct} := \nonHF{\applyplan{\buildplan}{\Struct}}[\crossingVertices{\nonHF{\applyplan{\buildplan}{\Struct}}}].\]
Here, we refer with $\crossingVertices{\nonHF{\applyplan{\buildplan}{\Struct}}}$
to the set of atoms in $\nonHF{\applyplan{\buildplan}{\Struct}}$
which are crossing vertices in $\applyplan{\buildplan}{\Struct}$.

We now show that the fibers of crossing vertices,
respectively the colors of inter-crossing edges,
are already determined by plain fibers and plain relations.

\begin{lemma}
	\label{lem:buildplan-almost-normalized}
	For every normalized HF-structure~$\Struct$ 
	and every building plan~$\buildplan$ for~$\Struct$,
	the structure $\applyplan{\buildplan}{\Struct}$
	satisfies the following:
		\begin{enumerate}[label=(A\arabic*)]
		\item \label{prop:an-plain-crossing} Every vertex of $\applyplan{\buildplan}{\Struct}$ is either plain or crossing.
		\item 	$\applyplan{\buildplan}{\Struct}[\vertices{\Struct}] = \StructA$.
		\item Every relation of $\applyplan{\buildplan}{\Struct}$ is either plain or inter-crossing or the special relation~$\pRel$.
		\item \label{prop:an-component-vertices} For every vertex $\vertA \in \crossingVertices{\StructA}$ and $i \in [2]$,
		there is exactly one vertex $\plvI{\vertA}{i} \in \verticesOf{i}{\StructA}$ such that $(\vertA,\plvI{\vertA}{i}) \in \pRel^{\StructA}$.
		Moreover, $\set{\plvA{\vertA},\plvB{\vertA}} \neq \set{\plvA{\vertB},\plvB{\vertB}}$
		for all distinct $\vertA,\vertB \in \crossingVertices{\StructA}$.
		\item \label{prop:an-direct-product-vertex}
		For every crossing vertex $\vertA \in \crossingVertices{\StructA}$
		its fiber~$\cc_\vertA$ 
		is uniquely determined by the set of fibers 
		$\set{\cc_{\plvA{\vertA}}, \cc_{\plvB{\vertA}}}$
		of the plain vertices~$\plvA{\vertA}$ and~$\plvB{\vertA}$,
		that is,
		whenever $\set{\cc_{\plvA{\vertA}}, \cc_{\plvB{\vertA}}} = \set{\cc_{\plvA{\vertB}}, \cc_{\plvB{\vertB}}}$,
		then $\cc_\vertA = \cc_\vertB$.
		Additionally, for every pair of vertices 
		$\vertB_1 \in \cc_{\plvA{\vertA}}^{\Struct} \cap \verticesOf{1}{\Struct}$
		and $\vertB_2 \in \cc_{\plvB{\vertA}}^{\Struct} \cap \verticesOf{2}{\Struct}$,
		there is a crossing vertex~$\vertC$ 
		in the fiber~$\cc_\vertA$ such that
		$\plvA{\vertC} = \vertB_1$ and $\plvB{\vertC} = \vertB_2$.
		\item \label{prop:an-direct-product-edge}
		In the same sense,
		the color~$\col_{(\vertA,\vertB)}$ of every inter-crossing edge $(\vertA,\vertB)$
		is uniquely determined by the set of plain colors
		$\set{\col_{(\plvA{\vertA},\plvA{\vertB})}, \col_{(\plvB{\vertA},\plvB{\vertB})}}$.
	\end{enumerate}
\end{lemma}
\begin{proof}
	Let $\StructA = \Struct_1 \disunion \Struct_2$ be a normalized HF-structure and $\buildplan= (\omega, \kappa)$ be a building plan for~$\StructA$.
	Properties~\ref{prop:an-plain-crossing} to~\ref{prop:an-component-vertices}
	immediately follow from the construction of $\applyplan{\buildplan}{\Struct}$.
	To show Property~\ref{prop:an-direct-product-vertex},
	let~$\vertA$ be a crossing vertex in fiber~$\cc_\vertA$
	and let~$\cc_{\plvI{\vertA}{i}}$ be the fiber of~$\plvI{\vertA}{i}$ for every $i \in [2]$.
	First consider the claim that for every ${\vertB_1 \in \cc_{\plvA{\vertA}}^{\Struct} \cap \verticesOf{1}{\Struct}}$
	and ${\vertB_2 \in \cc_{\plvB{\vertA}}^{\Struct} \cap \verticesOf{2}{\Struct}}$,
	there is a crossing vertex~$\vertC$ such that
	$\plvA{\vertC} = \vertB_1$ and ${\plvB{\vertC} = \vertB_2}$.
	Because the crossing vertex~$\vertA$ exists,
	there is a vertex plan $\set{\vcA,\vcB} \in \omega$
	such that ${\set{\plvA{\vertA},\plvB{\vertA}} \in \crossVertexRel{\set{\vcA,\vcB}}^\Struct}$.
	Assume w.l.o.g.~that $\plvA{\vertA} \in \vcA^\Struct$ and that $\plvB{\vertA} \in \vcB^\Struct$.
	Because ${\vertB_1 \in \cc_{\plvA{\vertA}}^{\Struct} \subseteq \vcA^\Struct}$
	and ${\vertB_2 \in \cc_{\plvB{\vertA}}^{\Struct} \subseteq \vcB^\Struct}$,
	it follows that $\set{\vertB_1, \vertB_2} \in \crossVertexRel{\set{\vcA,\vcB}}^\Struct$
	and so by definition of $\applyplan{\buildplan}{\Struct}$
	there is a crossing vertex~$\vertC$ such that
	$\plvA{\vertC} = \vertB_1$ and $\plvB{\vertC} = \vertB_2$.
	To show that the fiber~$\cc_\vertA$ is uniquely determined by $\set{\cc_{\plvA{\vertA}},\cc_{\plvB{\vertA}}}$,
	it suffices to consider the special case of~\ref{prop:an-direct-product-edge}
	when considering loop colors.
	
	It remains to prove Property~\ref{prop:an-direct-product-edge}.
	Let $(\vertA,\vertB)$ be an inter-crossing edge of color~$\col_{(\vertA,\vertB)}$.
	Let the color of $(\plvI{\vertA}{i},\plvI{\vertB}{i})$ be $\col_{(\plvI{\vertA}{i},\plvI{\vertB}{i})}$ for every $i \in [2]$.
	
	For the first direction,
	let $(\vertA', \vertB')$ be a crossing edge such that
	$\col_{(\vertA',\vertB')} = \col_{(\vertA,\vertB)}$.
	Because every crossing vertex is adjacent via~$\pRel$ to exactly one plain vertex in each component, for every $i \in[2]$
	every inter-crossing edge has exactly two $(\pRel, \colB_i, \inv{\pRel})$-colored paths, where~$\colB_i$ is a plain color.
	For $(\vertA, \vertB)$, we have that $\set{\colB_1, \colB_2} = \set{\col_{(\plvA{\vertA},\plvA{\vertB})}, \col_{(\plvB{\vertA},\plvB{\vertB})}}$.
	Because $(\vertA',\vertB')$ has the same color as $(\vertA,\vertB)$, the same holds for $(\vertA',\vertB')$,
	i.e.,
	\[\setcond*{\col_{(\plvI{\vertA'}{i},\plvI{\vertB'}{i})}}{i \in [2]} = \setcond*{\col_{(\plvI{\vertA}{i},\plvI{\vertB}{i})}}{i \in [2]}.\]
	
	For the remaining direction,
	we have to show that two inter-crossing edges
	$(\vertA_1, \vertB_1)$ and $(\vertA_2, \vertB_2)$
	for which $\setcond{\col_{\plvI{\vertA_1}{i},\plvI{\vertB_1}{i}}}{i \in [2]} = \setcond{\col_{\plvI{\vertA_2}{i},\plvI{\vertB_2}{i}}}{i \in [2]}$
	have the same color $\col_{(\vertA_1, \vertB_1)} = \col_{(\vertA_2, \vertB_2)}$, where
	$(\vertA_1, \vertB_1)$ and $(\vertA_2, \vertB_2)$ are arbitrary.
	To do so,
	we show that there is a coherent configuration~$\StructC$
	refining $\coConf{\applyplan{\buildplan}{\Struct}}$
	such that $\StructC[\verticesOf{i}{\Struct}] = \coConf{\applyplan{\buildplan}{\Struct}}[\verticesOf{i}{\Struct}]$
	for every $i \in [2]$
	and~$\StructC$ satisfies the required property.
	Then every coarser coherent configuration satisfies the same statement, too.
	We only sketch the construction, the idea is based on the proof of Lemma~10 in~\cite{GroheSchweitzerWiebking2021}.
	The main difference is that~\cite{GroheSchweitzerWiebking2021} gives
	the construction for the~$\rel_\omega$ relation
	and not the~$\crossVertexRel{\omega}$ relation
	(cf.\ the proof of Lemma~\ref{lem:compute-sketch-building-plan}).
	Essentially, we replace ordered pairs of two colors with a set of at most two colors.
	
	The coherent configuration is defined as follows:
	The plain edges are colored according to~$\coConf{\Struct}$.
	Every inter-crossing edge $(\vertA, \vertB)$
	is colored with the set of colors $\set{\col_{(\plvA{\vertA}, \plvA{\vertB})}, \col_{(\plvB{\vertA}, \plvB{\vertB})}}$ of its corresponding plain edges.
	Note that by Lemma~\ref{lem:normalized-direct-product},
	this also determines the color of the edges $(\plvI{\vertA'}{i}, \plvI{\vertB'}{j})$ whenever $\vertA',\vertB' \in \set{\vertA,\vertB}$ and $\set{i,j}=[2]$, i.e.,~$\plvI{\vertA'}{i}$ and $\plvI{\vertB'}{j}$ are in different components, because the color of an edge determines the fiber of its endpoints.
	Similarly, an edge $(\vertA,\vertB)$ of a crossing vertex~$\vertA$ and a plain vertex~$\vertB$ is colored 
	with the fiber of~$\vertA$ (which is defined in the inter-crossing case for loops),
	the fiber of~$\vertB$,
	and whether $(\vertA, \vertB)$ is contained in~$\pRel$.
	Clearly, the sketched coherent configuration has the required property by construction.
\end{proof}
We say that the plain vertices~$\plvA{\vertA}$ and~$\plvB{\vertA}$
are the vertices \defining{corresponding to} the crossing vertex~$\vertA$.
Likewise, the plain edges $(\plvA{\vertA},\plvA{\vertB})$
and $(\plvB{\vertA},\plvB{\vertB})$ correspond to the
inter-crossing edge $(\vertA,\vertB)$.
The set of colors $\set{\rel_{(\plvA{\vertA},\plvA{\vertB})}, \rel_{(\plvB{\vertB},\plvB{\vertB})}}$ corresponds to the color $\rel_{(\vertA,\vertB)}$ (and similarly for crossing vertices).

In some sense, Properties~\ref{prop:an-direct-product-vertex}
and~\ref{prop:an-direct-product-edge}
mean that we actually do not need to construct the crossing vertices
because all information is determined by the corresponding plain vertices and edges.
We now show this formally.
\begin{lemma}
	\label{lem:compute-sketch-building-plan}
	There is a normalized DeepWL-algorithm
	that for
	every normalized HF-structure~$\StructA$
	and every building plan $\buildplan=(\omega, \kappa)$ for~$\StructA$
	computes $\sketch{\applyplan{\buildplan}{\Struct}}$ in polynomial time.
\end{lemma}
\begin{proof}
	
	We show that there is a (non-normalized) DeepWL algorithm~$\dwla$
	which on input~$\StructA$ in the cloud and $\buildplan=(\omega, \kappa)$ on the work-tape
	constructs the structure $\applyplan{\buildplan}{\Struct}$
	in the cloud by first executing an $\addPairSym$-operation on a crossing relation and then a $\contractSym$-operation. This is done as follows.
	
	The crossing relations $\rel_{\set{\vcA,\vcB}}$ and $\crossVertexRel{\set{\vcA,\vcB}}$ 
	are DeepWL-computable for every vertex plan $\set{\vcA,\vcB} \in \omega$.
	In particular, the crossing relation $\crossVertexRel{\omega}^\Struct = \bigcup_{\set{\vcA,\vcB} \in \omega} \crossVertexRel{\set{\vcA,\vcB}}$ is DeepWL-computable.
	So~$\dwla$ executes $\addUPair{\crossVertexRel{\omega}}$ and
	obtains the relation~$\pRel$ as membership relation.
	For every relation plan $(\rel, \setcond{\set{\relB_1^i, \relB_2^i}}{i \in [k]}) \in \kappa$, the algorithm $\dwla$ defines $\rel^{\applyplan{\buildplan}{\Struct}}$ as follows.
	For every $i \in [k]$,~%
	$\dwla$ defines the relation~$\rel_i$ to be the relation with an $(\pRel, \relB_1^i, \pRelInv)$- and an $(\pRel, \relB_2^i, \pRelInv)$-colored path.
	Then~$\rel$ is obtained as the union of all~$\rel_i$.
	
	Similar to Lemma~\ref{lem:deepwl-wsc-plus},
	the $\addUPairSym$-operation can be simulated in the DeepWL-model of~\cite{GroheSchweitzerWiebking2021}
	by first an $\addPairSym$-execution and then 
	a $\contractSym$-execution.
	Recall that $\contract{\col}$ contracts every $\col$-SCC to a singleton vertex.
	So we first execute $\addPair{\crossVertexRel{\omega}}$
	and then contract the edges between the two vertices $\pairVtx{\vertA}{\vertB}$
	and $\pairVtx{\vertB}{\vertA}$ to obtain vertices $\set{\pairVtx{\vertA}{\vertB},\pairVtx{\vertB}{\vertA}}$.
		
	We now argue that there is a normalized DeepWL-algorithm~$\hat{\dwla}$
	computing the sketch of $\applyplan{\buildplan}{\Struct}$ without executing the operations (in fact, without modifying the cloud at all).
	Using Lemma~10 of~\cite{GroheSchweitzerWiebking2021},~%
	$\hat{\dwla}$ computes  the sketch after the $\addPairSym$-execution in polynomial time (which is possible because $\crossVertexRel{\omega}$ is crossing).
	Next,~$\hat{\dwla}$ computes the the sketch after the $\contractSym$-execution
	using Lemma~9 of~\cite{GroheSchweitzerWiebking2021}.
	Finally, the inter-crossing relations described by~$\kappa$
	are unions of inter-crossing colors
	and the subset-relation can be computed in polynomial time.
\end{proof}

\begin{definition}[Efficient Building Plan]
	A building plan $\buildplan=(\omega, \kappa)$ for a normalized $\sig$-HF-structure~$\Struct$
	is called \defining{efficient},
	if for every $\vertA \in \vertices{\Struct}$
	there is a vertex class $\vcA \in \sig$ such that
	$\vertA \in \vcA^\Struct$ and there is a vertex plan 
	$\set{\vcA,\vcB} \in \omega$ for some $\vcB \in \sig$.
\end{definition}
Intuitively, an efficient building plan makes use of all plain vertices
of~$\Struct$, which means
that~$\Struct$ does not contain unnecessary  vertices
to construct $\applyplan{\buildplan}{\Struct}$.

\begin{lemma}
	\label{lem:building-plan-efficient-size}
	If~$\buildplan$ is an efficient building plan for a normalized HF-structure~$\Struct$,
	then $|\Struct| \leq  2|\crossingVertices{\applyplan{\buildplan}{\Struct}}|$.
\end{lemma}
\begin{proof}
	Every crossing vertex in $\crossingVertices{\applyplan{\buildplan}{\Struct}}$
	is, by Property~\ref{prop:an-component-vertices},
	incident to exactly two plain vertices.
	So there are at most twice as many plain as crossing vertices.
\end{proof}

\begin{lemma}
	\label{lem:buildplan-automorphisms}
	Let~$\Struct$ be a normalized HF-structure
	and~$\buildplan$ be a building plan for~$\Struct$.
	Then $\autGroup{\Struct} = \autGroup{\applyplan{\buildplan}{\Struct}}$
	(note that both structures are HF-structures with atom set~$\StructV$).
	Every crossing fiber~$\ccA$ satisfies that
	if $\ccA^{\applyplan{\buildplan}{\Struct}}$ is an $\applyplan{\buildplan}{\Struct}$-orbit,
	then for the corresponding plain fibers $\set{\ccA_1, \ccA_2}$ of~$\ccA$
	the set $\crossVertexRel{\set{\ccA_1,\ccA_2}}^\Struct$ 
	is an $\Struct$-orbit.
\end{lemma}
\begin{proof}
	We first show that $\autGroup{\Struct} \subseteq \autGroup{\applyplan{\buildplan}{\Struct}}$.
	The structure $\applyplan{\buildplan}{\Struct}$
	is defined in an isomorphism-invariant manner.
	Whenever a vertex or relation is added,
	it is done for all vertices/edges of a given vertex class/relation
	(cf.\ the proof of Lemma~\ref{lem:compute-sketch-building-plan} that shows
	that $\applyplan{\buildplan}{\Struct}$ can be obtained
	from~$\Struct$ by a DeepWL-algorithm).
	So every automorphism of~$\Struct$ extends to an automorphism of $\applyplan{\buildplan}{\Struct}$.
	To show $\autGroup{\applyplan{\buildplan}{\Struct}} \subseteq \autGroup{\Struct}$, 
	note that~$\Struct$ is contained in $\applyplan{\buildplan}{\Struct}$.
	Also, note that all relations added in $\applyplan{\buildplan}{\Struct}$
	are new ones and no relation of~$\Struct$ is changed.
	It is never possible that an automorphism maps a crossing vertex to a plain vertex because the~$\pRel$ relation is directed from crossing to plain vertices.	
	So every automorphism of $\applyplan{\buildplan}{\Struct}$
	is an automorphism of~$\Struct$.
	
	For the second part,
	let~$\ccA$ be a crossing fiber
	such that $\ccA^{\applyplan{\buildplan}{\Struct}}$ 
	is an $\applyplan{\buildplan}{\Struct}$-orbit.
	By construction of the~$\pRel$ relation,
	an automorphism~$\auto$ satisfies 
	$\auto(\vertA) = \vertB$
	if and only if
	$\auto(\set{\plvA{\vertA}, \plvB{\vertA}}) = \set{\plvA{\vertB},\plvB{\vertB}}$
	for every $\vertA,\vertB \in \ccA^{\applyplan{\buildplan}{\Struct}}$ 
	(cf.\ Property~\ref{prop:an-component-vertices}).
	Because $\crossVertexRel{\set{\ccA_1,\ccA_2}}^\Struct$
	contains exactly these pairs $\set{\plvA{\vertA}, \plvB{\vertA}}$,
	it is an $\applyplan{\buildplan}{\Struct}$-orbit, too,
	and thus also an $\Struct$-orbit since $\autGroup{\Struct} = \autGroup{\applyplan{\buildplan}{\Struct}}$.
\end{proof}

\paragraph{Comparison to~\cite{GroheSchweitzerWiebking2021}.}
To show that every DeepWL-algorithm can be simulated by an equivalent
normalized DeepWL-algorithm,
almost normalized structures are used as intermediate step in~\cite{GroheSchweitzerWiebking2021}.
The structures obtained from building plans
differ at some points from the
almost normalized structures:
\begin{itemize}
	\item Our notion of crossing vertices is defined only for building plans and not for general DeepWL-algorithms (and is for \emph{undirected} crossing edges).
	\item Our relation~$\pRel$
	does not distinguish between $\plvA{\vertA}$ and $\plvB{\vertA}$.
	This is needed so that a $\choiceSym$-operation on
	crossing vertices does not necessarily distinguish the two components.
	\item The relation~$\pRel$ assigns to every crossing vertex
	exactly one plain vertex in each component.
	This ensures that choice sets can still be witnessed.
	\item Properties similar to~\ref{prop:an-direct-product-vertex}
	and~\ref{prop:an-direct-product-edge}
	are always satisfied
	for almost normalized structures in~\cite{GroheSchweitzerWiebking2021}.
	In combination with $\refineSym$- and $\choiceSym$-operations,
	this would not be the case anymore.
\end{itemize}

\subsubsection{Building Plans and \texorpdfstring{$\sccSym$}{scc}-Operations}
Before we can start to design normalized DeepWL+WSC-algorithms,
we have to investigate $\sccSym$-operations.
Assume that for an inter-crossing relation~$\rel$ 
the $\scc{\rel}$-operation is to be executed
and we want to find a building plan simulating this.
The challenge is to construct new plain vertices
so that we find plain vertex classes which correspond
to the vertex class obtained by $\scc{\rel}$.
The first step is to analyze the SCCs
of the corresponding plain colors of~$\rel$ in the components.
In a second step, we show how we can define a building plan
to simulate the $\scc{\rel}$-operation.
We start with a lemma regarding SCCs in coherent configurations.

\begin{lemma}
	\label{lem:coconf-scc-cc}
	Let~$\StructC$ be a coherent configuration with signature~$\sigB$
	and $\colA,\colB \in \sigB$ be colors connecting vertices of the same fiber,~i.e., $\colA^\StructC \subseteq (\ccA^\StructC)^2$ 
	and $\colB^\StructC \subseteq (\ccA^\StructC)^2$ 
	for some fiber $\ccA \in \sigB$.
	Then 
	\begin{enumerate}
		\item \label{itm:coconf-scc-cc-one} every $\colA$-connected component is strongly $\colA$-connected and
		\item \label{itm:coconf-scc-cc-two} every $\set{\colA,\colB}$-connected component is strongly $\set{\colA,\colB}$-connected.
	\end{enumerate}
\end{lemma}
\begin{proof}
	We start with Claim~\ref{itm:coconf-scc-cc-one}:
	If~$\col$ is itself a fiber, all connected components are trivial
	and the claim follows.
	Otherwise, let~$c$ be an $\col$-connected component.
	Then we can assume that~$\StructC$
	is a primitive coherent configuration,
	otherwise we can restrict~$\StructC$ to~$c$
	(all edges leaving~$c$ have different colors than edges contained in~$c$).
	Finally, it follows from Theorem 3.1.5.~in~\cite{ChenPonomarenko2019}
	that~$c$ is strongly $\col$-connected because~$\col$ is not a fiber.
	
	To show Claim~\ref{itm:coconf-scc-cc-two},
	let~$c$ be an $\set{\colA,\colB}$-connected component.
	To prove that~$c$ is strongly $\set{\colA,\colB}$-connected,
	let $(\vertA_1, \dots, \vertA_k)$ be a $\set{\colA,\colB}$-path in~$c$.
	We show that there is an $\set{\colA,\colB}$-path
	$(\vertB_1, \dots, \vertB_m)$ such that $\vertB_1 = \vertA_k$ and $\vertB_m = \vertA_1$.
	To do so, it suffices to show that for every $i \in [k-1]$
	there is a $\set{\colA,\colB}$-path from $\vertA_{i+1}$ to $\vertA_i$.
	So let $i \in [k-1]$ and w.l.o.g.~assume that $(\vertA_i, \vertA_{i+1}) \in \colA^\StructC$.
	Then by Claim~\ref{itm:coconf-scc-cc-one},
	the vertices~$\vertA_i$ and~$\vertA_{i+1}$ are in 
	the same $\colA$-SCC and in particular there is an $\colA$-path
	and so also an $\set{\colA,\colB}$-path from~$\vertA_{i+1}$ to~$\vertA_i$.
\end{proof}

\begin{lemma}
	\label{lem:contract-components-unique}
	Let $\Struct = \Struct_1 \disunion \Struct_2$ be a normalized HF-structure,~%
	$\buildplan$ be a building plan for~$\StructA$,
	$\colA$ be an inter-crossing color of $\StructB = \applyplan{\buildplan}{\StructA}$
	such that $\colA$-edges connect vertices of the same fiber $\ccA_\col$, i.e., $\col^\StructB \subseteq (\ccA_\col^\StructB)^2$,
	and $\colB, \colC \in \sigB$ be the corresponding plain colors of~$\col$.
	If~$c$ is an $\colA$-SCC,
	then $\setcond{\plvI{\vertA}{i}}{\vertA \in c}$ is a $\set{\colB,\colC}$-SCC
	for every $i \in [2]$.
\end{lemma}
\begin{proof}
	\setcounter{claim}{0}
	Define \[K_i := \setcond*{(\vertA,\vertB) \in \colA^\StructB} {(\plvI{\vertA}{i},\plvI{\vertB}{i}) \in \colB^\StructA}.\]
	That is, if $\colB\neq \colC$, we partition 
	$\colA^\StructB$ into~$K_1$ and~$K_2$ depending on whether the 
	corresponding $\colB$-edge is in~$\StructA_1$ or in~$\StructA_2$.
	If $\colB = \colC$, we just have $K_1 = K_2 = \colA^\StructB$.
	For a set~$c$ of crossing vertices we define $\plvI{c}{i} := \setcond{\plvI{\vertA}{i}}{\vertA \in c}$.
	We start to analyze the SCCs formed by the $K_i$-edges.
	
	\begin{claim}
		\label{clm:plain-vertices-sccs-one-side}
		Let $i \in [2]$ and $c$ be a $K_i$-SCC.
		Then the set~$\plvI{c}{i}$
		is an $\colB$-SCC.
		For $j \in [2]$  such that $\set{i,j} = [2]$,
		the set $\plvI{c}{j}$
		is an $\colC$-SCC.
	\end{claim}
	\begin{claimproof}
		We consider the part of the claim regarding $\colB$-SCCs.
		The part regarding $\colC$-SCCs is symmetric.		
		It is clear that if $(\vertA_1, \dots, \vertA_k)$ is a $K_i$-path,
		then $(\plvI{\vertA_1}{i}, \dots, \plvI{\vertA_k}{i})$ is an $\colB$-path.
		So~$\plvI{c}{i}$
		is contained in an $\colB$-SCC.
		If $\colB$-edges connect vertices in different fibers,
		then all $\colB$-SCCs are singletons.
		That is~$\plvI{c}{i}$ cannot be strictly contained in an $\colB$-SCC.
		
		So it remains to consider the case when~$\colB$ connects vertices in the same fiber.
		Then also~$\colC$ has to connect vertices in the same fiber
		(otherwise~$\colA$ would not be a color).
		Let $\set{\ccA, \ccB}$ be the fibers corresponding to~$\ccA_\col$.
		If~$\ccA$ has an incident $\colB$-edge, then~$\ccB$ has an incident $\colC$-edge (because otherwise~$\col$ would be empty, which is not allowed for colors).
		
		We show that~$\plvI{c}{i}$ is an $\colB$-connected component.
		This implies by Lemma~\ref{lem:coconf-scc-cc} that~$\plvI{c}{j}$ is an $\colB$-SCC.
		Let $(\vertA_1, \vertA_2) \in \colB^\Struct$
		be such that $\vertA_1 \in \plvI{c}{i}$.
		We show that there is an $\colA$-edge $(\vertB_1, \vertB_2)$
		such that $(\plvI{\vertB_1}{i}, \plvI{\vertB_2}{i}) = (\vertA_1, \vertA_2)$ and $\vertB_1, \vertB_2 \in c$,
		which implies that $\vertA_2 \in \plvI{c}{i}$ and by induction that $\plvI{c}{i}$ is an $\colB$-connected component.
		Because $\vertA_1 \in \plvI{c}{i}$,
		there is a vertex $\vertB_1 \in c$ such that $\plvI{\vertB_1}{i} = \vertA_1$.
		Assume w.l.o.g.~that $\vertA_1 \in \ccA^\Struct$.
		Then $\plvI{\vertB_1}{j} \in \ccB^\Struct$, where $\set{i,j} = [2]$.
		Because~$\vertA_1$ has an incident $\colB$-edge (namely $(\vertA_1, \vertA_2)$),
		$\plvI{\vertB_1}{j}$ has an incident $\colC$\nobreakdash-edge $(\plvI{\vertB_1}{j}, \vertC)$.
		Because~$\colB$ and~$\colC$ connect vertices of the same fiber,
		$\vertA_2 \in \ccA^\Struct$ and $\vertC \in \ccB^\Struct$.
		Then by Property~\ref{prop:an-direct-product-vertex}
		there is a vertex~$\vertB_2$ such that
		$\plvI{\vertB_2}{i} = \vertA_2$ and $\plvI{\vertB_2}{j} = \vertC$.
		Now, we have that $(\vertB_1,\vertB_2) \in \col^\StructB$
		by Property~\ref{prop:an-direct-product-edge},
		that $(\plvI{\vertB_1}{i},\plvI{\vertB_2}{i}) =(\vertA_1,\vertA_2) \in \colB^\Struct$,
		and that $(\plvI{\vertB_1}{j},\plvI{\vertB_2}{j}) = (\plvI{\vertB_1}{j},\vertC) \in \colC^\Struct$.
		That is,~$\vertB_1$ and~$\vertB_2$ are in the same $\col$-connected component.
		This, by Lemma~\ref{lem:coconf-scc-cc},
		implies that~$\vertB_1$ and~$\vertB_2$ are in the same $\col$-SCC
		and thus $\vertB_2 \in c$.
	\end{claimproof}

	\begin{claim}
		\label{clm:plain-vertices-sccs-alternating}
		Let~$\colB$ connect vertices of different fibers~$\ccA$ and~$\ccB$
		and let~$\colC$ connect vertices of the same different fibers
		but in the other direction,
		i.e., $\colB^{\Struct} \subseteq \ccA^\Struct \times \ccB^\Struct$
		and $\colC^\Struct \subseteq \ccB^\Struct \times \ccA^\Struct$ for $\ccA \neq \ccB$.
		Let $i \in [2]$ and~$c$ be an $\colA$-SCC.
		Then the set~$\plvI{c}{i}$
		is an $\set{\colB,\colC}$-SCC.
	\end{claim}
	\begin{claimproof}
		If $(\vertA_1, \dots, \vertA_k)$ is an $\col$-path,
		then $(\plvI{\vertA_1}{i}, \dots, \plvI{\vertA_k}{i})$ is an $\set{\colB,\colC}$-path.
		So~$\plvI{c}{i}$
		is contained in an $\set{\colB,\colC}$-SCC~$d$.
		For a sake of contradiction,
		assume that~$\plvI{c}{i}$ is strictly contained in~$d$.
		Then there is an $\set{\colB,\colC}$-path $(\vertB_1, \dots, \vertB_\ell)$ contained in~$d$
		with $\ell \geq 3$, $\vertB_1 \in \plvI{c}{i}$, $\vertB_\ell \in \plvI{c}{i}$,
		and $\vertB_k \notin \plvI{c}{i}$ for every $1 < k < \ell$.
		We observe the following:
		\begin{enumerate}[label=\alph*)]
			\item There is a (possibly empty) $\set{\colB,\colC}$-path
			$(\vertB_\ell, \dots, \vertB_{\ell'+1})$ contained in~$d$
			such that $\vertB_{\ell'+1} = \vertB_1$
			because~$d$ is an $\set{\colB,\colC}$-SCC.
			\item In every $\set{\colB,\colC}$-path
			the edge colors~$\colB$ and~$\colC$ alternate.
			Likewise, the fibers of the vertices alternate.
			Thus, every cycle consists of an even number of vertices.
		\end{enumerate}
		Let $(\vertA_1, \dots, \vertA_{m'}, \vertA_1)$ be a non-empty $\colA$-cycle (which possibly uses vertices multiple times)
		such that $\plvI{\vertA_1}{i} = \vertB_1$, $1 < m \leq m' $, and
		$\plvI{\vertA_m}{i} = \vertB_{\ell'}$.
		Such a cycle exists because~$c$ is an $\colA$-SCC and $\vertB_\ell \in \plvI{c}{i}$.	
		Consider the two following sequences $\tup{\beta} := (\vertB_1, \dots, \vertB_{\ell'})^{m'}$ and $\tup{\alpha} := (\plvI{\vertA_1}{j}, \dots, \plvI{\vertA_{m'}}{j})^{\ell'}$ of plain vertices of length $\ell' \cdot m'$,
		where~$j$ is chosen such that $\set{i, j} =[2]$,
		cf.\ Figure~\ref{fig:cycles-scc-in-coherent-conf}.
			\begin{figure}
			\centering
			\begin{tikzpicture}
				\tikzset{
					bicolor/.style n args={5}{
						dashed,dash pattern=on #3 off #3,#4,#1,
						postaction={draw,dashed,dash pattern=on #3 off #3,#5,#2,dash phase=#3}
					},
				}
			
				\tikzset{
					bicolor color fill/.code 2 args={
						\pgfdeclareverticalshading[%
						tikz@axis@top,tikz@axis@middle,tikz@axis@bottom%
						]{diagonalfill}{100bp}{%
							color(0bp)=(tikz@axis@bottom);
							color(50bp)=(tikz@axis@bottom);
							color(50bp)=(tikz@axis@middle);
							color(50bp)=(tikz@axis@top);
							color(100bp)=(tikz@axis@top)
						}
						\tikzset{shade, left color=#1, right color=#2, shading=diagonalfill}
					}
				}

				\tikzstyle{path} = [decorate,decoration={snake,amplitude=.4mm,segment length=2mm, post length = 1mm}, thick];
				
				\tikzstyle{sequence} = [black, draw,-Latex, double];
				
				\begin{scope}[scale=1.2]
					\node[vertex, label=173:{$\vertA_1$}] (u1) at (165:1) {};
					\node[vertex, label=00: {$\vertA_m$}] (um) at (0:1) {};
					\node[vertex, label=187:{$\vertA_{m'}$}] (ump) at (195:1){};
					
					\draw[->,draw, black, path]
					(159:1) arc(-21:-174:-1);
					\draw[->,draw, black, path]
					(-6:1) arc(174:21:-1) (ump);
					\path[->,draw, black, thick]
					(ump) edge node[right, font=\footnotesize] {$\col$}(u1);

					\node[align=center] (ai) at (0,-2) {\strut crossing vertices\\of $\applyplan{\buildplan}{\StructA}$\strut};
				\end{scope}

				\begin{scope}[shift={(5,0)},scale=1.2]
					
					\draw[gray, fill=gray!30!white]
					(135:0.8) arc (315:135:0.4)
					arc(135:40:1.6)
					arc(40:-140:0.4)
					arc(40:135:0.8);
					\node[gray!50!black] (dclab)  at (50:2.2) {$d\setminus \plvI{c}{i}$};
					
					\node[vertex, red, label=173:{$\vertB_1$}] (v1) at (165:1) {};
					\node[vertex, blue] (v2) at (135:1) {};
					\node[vertex, red] (v3) at (105:1) {};
					\node[vertex, blue] (v4) at (75:1) {};
					\node[vertex, blue, label=187:{$\vertB_{\ell'}$}] (vlp) at (195:1) {};
					\node[vertex,  bicolor color fill={red}{blue},shading angle=45, label=0  :{$\vertB_\ell$}]    (vl)  at (0:1){};
					
					\path[->, draw, magenta, thick]
					(v1) edge (v2)
					(v3) edge node[above, font=\footnotesize,] {$\colB$}  (v4);
					
					\path[->, draw, cyan, thick]
					(vlp) edge  (v1)
					(v2) edge node[above, font=\footnotesize, xshift=-1mm] {$\colC$}(v3);
					
					\path[->, draw, path, bicolor={cyan}{magenta}{13pt}{-}{->}]
					(69:1) arc(69:6:1) (vl);
					\path[->, draw, path, bicolor={cyan}{magenta}{13.6pt}{->}{-}]
					(-6:1) arc(-6: -159:1) (vlp);
					
					\draw[sequence]
					(170:0.45) arc (170:-170:0.45);
					\node (beta) at (0,0) {$\tup{\beta}$};
					\node[font=\scriptsize] (betaiter) at (0,-0.65) {\strut $m'$ times};
					
					\node[align=center] (ai) at (0,-2) {\strut plain vertices\\of $\StructA_i$\strut};

				\end{scope}
				
				\begin{scope}[shift={(-5,0)},scale=1.2]
					\node[vertex, blue, label=7:{$\plvI{\vertA_1}{j}$}] (ui1) at (15:1) {};
					\node[vertex, red] (ui2) at (45:1) {};
					\node[vertex, blue] (ui3) at (75:1) {};
					\node[vertex, red] (ui4) at (105:1) {};
					\node[vertex, red, label=-7:{$\plvI{\vertA_{m'}}{j}$}] (uimp) at (-15:1) {};
					
					\path[->, draw, cyan, thick]
					(ui1) edge (ui2)
					(ui3) edge  node[above, font=\footnotesize]  {$\colC$} (ui4);
					
					\path[->, draw, magenta, thick]
					(uimp) edge  (ui1)
					(ui2) edge node[above, font=\footnotesize, xshift=1mm] {$\colB$}(ui3);
					
					\draw[->, draw, path, bicolor={magenta}{cyan}{13pt}{-}{->}]
					(111:1) arc(111:339:1);
					
					\draw[sequence]
					(10:0.45) arc (10:350:0.45);
					\node (alpha) at (0,0) {$\tup{\alpha}$};
					\node[font=\scriptsize] (alphaiter) at (0,-0.65) {\strut $\ell'$ times};
					
					\node[align=center] (ai) at (0,-2) {\strut plain vertices\\of $\StructA_j$\strut};

				\end{scope}
				
				\path[->, draw, darkgreen, thick]
				(u1) to[bend right=5]  (ui1);
				\path[->, draw, darkgreen, thick]
				(u1) to[bend left=5] (v1);
				\path[->, draw, darkgreen, thick]
				(um) to[bend right=15] (vlp);
				\path[->, draw, darkgreen, thick]
				(ump) to[bend left=5] (uimp);
				
				\node[darkgreen] (prel1) at (2.5,0) {\strut $\pRel$};
				\node[darkgreen] (prel2) at (-2.5,0) {\strut $\pRel$};
				
			\end{tikzpicture}
			
			\caption{
				The situation in Claim~\ref{clm:plain-vertices-sccs-alternating} in the proof of Lemma~\ref{lem:contract-components-unique}:
				All vertices in~$\ccA$ are red, the ones in~$\ccB$ are blue.
				Whether the vertex~$\vertB_\ell$ is in~$\ccA$ or in~$\ccB$
				depends on whether~$\ell$ is even or odd.
				The sequence~$\tup{\alpha}$ iterates~$\ell'$~times the cycle $(\plvI{u_1}{j}, \dots ,\plvI{u_{m'}}{j},\plvI{u_1}{j})$
				and the sequence~$\tup{\beta}$ iterates~$m'$~times the cycle $(\vertB_1,\dots,\vertB_{\ell'},\vertB_1)$.
				This figure assumes that $\vertB_1 \in \ccA^\Struct$.
				In the case that $\vertB_1 \in \ccB^\Struct$,
				the colors~$\colB$ and~$\colC$ 
				and the fibers~$\ccA$ and~$\ccB$ need to be swapped.
			}
			\label{fig:cycles-scc-in-coherent-conf}
		\end{figure}

		For a vertex $\vertC \in \vertices{\StructA}$, denote by~$\ccA_\vertC$ the fiber containing~$\vertC$.
		For every pair ${(\vertC,\vertC') \in \vertices{\StructA}^2}$, denote by~$\col_{(\vertC,\vertC')}$ the color containing $(\vertC,\vertC')$.
		We show that
		\[\set{\ccA_{\beta_k},\ccA_{\alpha_k}} = \set{\ccA, \ccB}\]
		for every $k \in [\ell' \cdot m']$:
		First, consider the case $k= 1$.
		By construction,
		\[(\beta_1,\alpha_1) = (\vertB_1, \plvI{\vertA_1}{j}) = (\plvI{\vertA_1}{i}, \plvI{\vertA_1}{j}).\]
		Because~$\colA$ connects vertices in the same fiber,
		the corresponding fibers of~$\ccA_\col$ must be $\set{\ccA,\ccB}$
		and because $\vertA_1 \in \ccA_\col^\StructB$,
		it follows that $\set{\ccA_{\beta_1},\ccA_{\alpha_1}} = \set{\ccA,\ccB}$.
		Second, assume that $\set{\ccA_{\beta_k},\ccA_{\alpha_k}} = \set{\ccA, \ccB}$ for some $k  < \ell' \cdot m'$.
		We already have seen that the fibers~$\ccA$ and~$\ccB$
		alternate on $\set{\colB, \colC}$-paths,
		so in particular on the cycles $(\vertB_1, \dots, \vertB_{\ell'}, \vertB_1)$
		and $(\plvI{\vertA_1}{j}, \dots, \plvI{\vertA_{m'}}{j},\plvI{\vertA_1}{j})$.
		So if $\ccA_{\beta_k} = \ccA$, then $\ccA_{\beta_{k+1}} = \ccB$
		and vice versa
		and similar for~$\alpha_k$ and~$\alpha_{k+1}$.
		Hence, $\set{\ccA_{\beta_{k+1}},\ccA_{\alpha_{k+1}}} = \set{\ccA, \ccB}$.

		By Property~\ref{prop:an-direct-product-vertex},
		there is a sequence of crossing vertices
		$(\vertC_1, \dots, \vertC_{\ell'\cdot m'})$ such that
		$\plvI{\vertC_k}{i} = \beta_k$ and
		$\plvI{\vertC_k}{j} = \alpha_k$
		for all $k \in [\ell' \cdot m']$.
		We prove that $(\vertC_1, \dots, \vertC_{\ell'\cdot m'}, \vertC_1)$ is an $\colA$-cycle.
		We consider the sequences of colors
		\begin{align*}
		&\col_{(\plvI{\vertC_1}{i},\plvI{\vertC_2}{i})}, \dots ,
		\col_{(\plvI{\vertC_{\ell'\cdot m'-1}}{i},\plvI{\vertC_{\ell'\cdot m'}}{i})},
		\col_{(\plvI{\vertC_{\ell'\cdot m'}}{i},\plvI{\vertC_1}{i})}
		\text{ and} \\
		&\col_{(\plvI{\vertC_1}{j},\plvI{\vertC_2}{j})}, \dots ,
		\col_{(\plvI{\vertC_{\ell'\cdot m'-1}}{j},\plvI{\vertC_{\ell'\cdot m'}}{j})},
		\col_{(\plvI{\vertC_{\ell'\cdot m'}}{j},\plvI{\vertC_1}{j})}.
		\end{align*}
		Similar to the case of the fibers,
		the colors~$\colB$ and~$\colC$ alternate in both sequences
		and $\set{\col_{(\plvI{\vertC_k}{i},\plvI{\vertC_{k+1}}{i})},\col_{(\plvI{\vertC_k}{j},\plvI{\vertC_2}{{k+1}})}} =\set{\colB,\colC}$
		for every $k  < \ell'\cdot m'$ and
		\[\set[\big]{\col_{(\plvI{\vertC_{\ell'\cdot m'}}{i},\plvI{\vertC_1}{i})},\col_{(\plvI{\vertC_{\ell'\cdot m'}}{j},\plvI{\vertC_1}{j})}} =\set{\colB,\colC}.\]
		Because~$\colB$ and~$\colC$ are the corresponding colors of~$\colA$,
		$(\vertC_k,\vertC_{k+1}) \in \col^\StructB$ by Property~\ref{prop:an-direct-product-edge}
		for every $k <  \ell'\cdot m'$
		and $(\vertC_{\ell'\cdot m'},\vertC_1) \in \col^\StructB$.
		Thus,
		$(\vertC_1, \dots, \vertC_{\ell'\cdot m'},\vertC_1)$ is an $\colA$-cycle.
		
		Because $\beta_1 = \vertB_1 = \plvI{\vertA_1}{i}$
		and $\alpha_1 = \plvI{\vertA_{1}}{j}$,
		we have that $\vertC_1 = \vertA_1 \in c$.
		In particular, the cycle $(\vertC_1, \dots, \vertC_{\ell'\cdot m'},\vertC_1)$
		contains a vertex in~$c$.
		By construction,
		 $\plvI{\vertC_k}{i} = \vertB_{k}$ for every $1 < k \leq \ell'$.
		And by assumption on the path $(\vertB_1, \dots, \vertB_\ell)$,
		$\vertB_{k} \notin \plvI{c}{i}$ for every $1 < k < \ell$.
		Because $\ell \geq 3$,
		there is a $1 < k < \ell$ such that 
		$\vertB_{k} \notin \plvI{c}{i}$
		and thus $\vertC_k \notin c$.
		But this means that there is an $\col$-cycle containing a vertex
		in the $\col$-SCC~$c$ and a vertex not in~$c$,
		which is a contradiction.
	\end{claimproof}

	First consider the case that $\colB = \colC$.
	Then $K_1 = K_2 = \colA^\StructB$
	and the claim of the lemma follows immediately
	from Claim~\ref{clm:plain-vertices-sccs-one-side}.
	So consider the case that $\colB \neq \colC$.
	Let $\colB^\Struct \subseteq \ccA_\colB^\Struct \times \ccB_\colB^\Struct$, $\colC^\Struct \subseteq \ccA_\colC^\Struct \times \ccB_\colC^\Struct$,
	and recall that $\colA^\StructB \subseteq (\ccA_\colA^\StructB )^2$.
	Then, by Property~\ref{prop:an-direct-product-vertex},
	it follows that the corresponding fibers for~$\ccA_\colA$
	are $\set{\ccA_\colB, \ccA_\colC} = \set{\ccB_\colB, \ccB_\colC}$.
	We make the following case distinction:
	
	\begin{itemize}
		\item $\ccA_\colB \neq \ccB_\colB$ or $\ccA_\colC \neq \ccB_\colC$:
		Then 
		we have that $\ccA_\colB = \ccB_\colC \neq \ccB_\colB = \ccA_\colC$
		because $\set{\ccA_\colB, \ccA_\colC} = \set{\ccB_\colB, \ccB_\colC}$.
		The claim of the lemma follows immediately from Claim~\ref{clm:plain-vertices-sccs-alternating}.
		
		\item$\ccA_\colB= \ccB_\colB = \ccA_\colC = \ccB_\colC$:
		Let~$c$ be an $\colA$-SCC.
		By Lemma~\ref{lem:coconf-scc-cc},
		it suffices to show that~$\plvI{c}{i}$ is an $\set{\colB,\colC}$-connected component.
		So let $(\vertA, \vertB)$ be an $\set{\colB,\colC}$-edge
		and let $\vertA \in \plvI{c}{i}$.
		We show that then also $\vertB \in \plvI{c}{i}$
		which by induction shows that~$\plvI{c}{i}$ is an $\set{\colB,\colC}$-connected component.
		Assume w.l.o.g.~that $(\vertA, \vertB) \in \colB^\Struct$.
		Now for some $K_i$-SCC~$c'$ we have that
		$\vertA \in \plvI{c'}{i}$
		and by Claim~\ref{clm:plain-vertices-sccs-one-side}
		we also have that  $\vertB \in \plvI{c'}{i}$
		because by Lemma~\ref{lem:coconf-scc-cc} the vertices $\vertA$ and~$\vertB$ are in the same $\colB$-SCC.
		Because $c' \subseteq c$, we have that $\plvI{c'}{i} \subseteq \plvI{c}{i}$
		and thus that $\vertB \in \plvI{c}{i}$.\qedhere
\end{itemize}
\end{proof}

Now that we know that the SCCs of inter-crossing relations
correspond to SCCs of the corresponding plain relations,
we show that $\sccSym$-operations on inter-crossing relations
can be simulated using building plans.  

\begin{lemma}
	\label{lem:simulate-contract-buildplan}
	There is a normalized polynomial-time DeepWL-algorithm that
	for every  normalized HF-structure
	$\Struct = \Struct_1 \disunion \Struct_2$,
	every building plan~$\buildplan$ for~$\Struct$,
	and every inter-crossing color~$\col$ of $\applyplan{\buildplan}{\Struct}$
	halts with a normalized HF-structure~$\StructB$ in the cloud and
	writes a building plan~$\buildplan_\col$ for~$\StructB$ onto the tape
	which satisfies $\applyplancross{\buildplan_\col}{\StructB} \iso \applyplancross{\buildplan}{\Struct}_\col$,
	where $\applyplancross{\buildplan}{\Struct}_\col$
	denotes the structure obtained by executing $\scc{\col}$ on $\applyplancross{\buildplan}{\Struct}$.
	If~$\buildplan$ is efficient, then~$\buildplan_\col$ is efficient, too.
\end{lemma}
\begin{proof}
	Let $\Struct = \Struct_1 \disunion \Struct_2$ be a normalized HF-structure, $\buildplan = (\omega, \kappa)$ be a building plan for~$\Struct$,
	and~$\col$ be an inter-crossing color of $\applyplan{\buildplan}{\Struct}$.
	We distinguish the following two cases:
	Assume first that the color~$\col$ connects vertices in different fibers~$\ccA$ and~$\ccB$,
	that is, $\col^{\applyplan{\buildplan}{\Struct}} \subseteq \ccA^{\applyplan{\buildplan}{\Struct}} \times \ccB^{\applyplan{\buildplan}{\Struct}}$.
	Here clearly all $\col$-SCCs are trivial
	and the $\sccSym$-operation would create a new vertex
	for every $\ccA$- and every $\ccB$-vertex.
	To create these vertices with the building plan,
	let the corresponding fibers of~$\ccA$ and~$\ccB$
	be~$\ccA_1$,~$\ccA_2$,~$\ccB_1$, and~$\ccB_2$.
	We create copies of these fibers by executing $\addPairSym$ for the loops
	and obtain component relations $\rel_{\leftT}$ and $\rel_{\rightT}$,
	which coincide.
	The new vertices end up in new fibers~$\ccA_1'$,~$\ccA_2'$,~$\ccB_1'$, and~$\ccB_2'$.
	The algorithm updates
	\begin{align*}
		\omega &\leftarrow \omega \cup \set[\big]{\set{\ccA_1', \ccA_2'}, \set{\ccB_1', \ccB_2'}} \text{ and}\\
		\kappa &\leftarrow \kappa \cup \set[\big]{(\relB,\set{\rel_{\leftT}, \rel_{\rightT}})},
	\end{align*}
	where~$\relB$ serves as new membership relation
	(and just relates a $\ccA$-vertex to its copy).
	
	Otherwise,~$\col$ connects vertices in the same fiber.
	Let this fiber be~$\ccA$ and let~$\ccA_1$ and~$\ccA_2$ be the
	fibers corresponding to~$\ccA$
	and~$\col_1$ and~$\col_2$ be the colors corresponding to~$\col$.
	Note that in this case every $\set{\col_1, \col_2}$-SCCs is nontrivial:
	w.l.o.g.~$\col_1$ has to connect~$\ccA_1$ to~$\ccA_2$
	and~$\col_2$ has to connect~$\ccA_2$ to~$\ccA_1$.
	Because~$\col_1$ and~$\col_2$ are fibers,
	every $\ccA_1$-vertex and every $\ccA_2$-vertex has one outgoing and one incoming 
	$\set{\col_1,\col_2}$-edge,
	so there must be a cycle and in particular one nontrivial $\set{\col_1, \col_2}$-SCCs.
	But by the properties of  a coherent-configuration,
	every $\set{\col_1, \col_2}$-SCCs has the same size and is thus nontrivial.
	
	By Lemma~\ref{lem:contract-components-unique},
	we can construct vertices for SCCs of the corresponding plain vertices
	such that for every $\col$-SCC
	there is a corresponding pair of $\set{\col_1,\col_2}$-SCCs.
	So we start with executing $\scc{\set{\col_1,\col_2}}$ (formally we have to create a relation as the union of $\col_1$ and $\col_2$)
	and obtain a new plain vertex class~$\vcA$ containing the new SCC-vertices
	and a plain membership relation~$\rel$.
	We then create for each pair of $\vcA$-vertices in different components,
	so for such a pair of $\set{\col_1,\col_2}$-SCCs,
	new crossing vertices and the inter-crossing membership relation~$\relB$  by updating
	\begin{align*}
		\omega &\leftarrow \omega \cup \set[\big]{\set{\vcA, \vcA}} \text{ and}\\
		\kappa &\leftarrow \kappa \cup \set[\big]{(\relB,\set{\relA,\relA})}.
	\end{align*}
	Indeed, every newly created vertex by~$\omega$ corresponds to an $\col$-SCC:
	Because the $\col$-SCCs are nontrivial,
	every $\vcA$-vertex was obtained from an $\set{\ccA_1,\ccA_2}$-SCC
	containing at least one~$\ccA_1$-vertex and at least one~$\ccA_2$-vertex
	(if $\ccA_1 = \ccA_2$ this is trivial and if $\ccA_1 \neq \ccA_2$
	every nontrivial $\set{\col_1,\col_2}$-SCC has to contain one~$\ccA_1$- and one~$\ccA_2$-vertex because~$\col_1$ has to connect~$\ccA_1$ to~$\ccA_2$ and~$\col_2$ the other way around).
	So for every pair of $\vcA$-vertices in different components,
	we can find a $\ccA_1$\nobreakdash-vertex $\vertA \in \vertices{\StructA_1}$ in one component and a $\ccA_2$\nobreakdash-vertex $\vertB \in \vertices{\StructA_2}$ in the other component. 
	By Property~\ref{prop:an-direct-product-vertex},
	there is a $\ccA$-vertex~$\vertC$ such that $\plvA{\vertC} = \vertA$
	and $\plvA{\vertC} = \vertB$.
	But this means there is an $\col$-SCC, which corresponds to the two $\vcA$-vertices, namely the one containing~$\vertC$.
	In the same manner,~$\relB$ correctly defines the inter-crossing membership relation for the new $\col$-SCC vertices.
	
	Finally, to see that~$\StructB$ is normalized,
	note that we only executed a single $\addPairSym$-operation or a single $\sccSym$-operation for a plain relation.
	It is also clear that the property of being efficient is preserved
	because every newly created vertex is in a fiber~$\ccA_1'$,~$\ccA_2'$,~$\ccB_1'$, and~$\ccB_2'$ or in the vertex class~$\vcA$
	and all of them are used in the building plan.
	Obviously, the algorithm runs in polynomial time.
\end{proof}

\subsubsection{Simulation}
\label{sec:simulation}
Now, we finally want to simulate arbitrary DeepWL+WSC-algorithms with normalized ones.
Recall that we still have to define how normalized DeepWL+WSC-algorithms encode sets of witnessing automorphisms.
We do this now.

Let~$\Struct$ be a normalized HF-structure and~%
$\buildplan$ be a building plan for~$\Struct$.
A tuple of crossing relations $(\rel_{\text{aut}}, \rel_{\text{dom}}, \rel_{\text{img}})$
and a crossing vertex $\vertC_\auto$
\defining{encode the
partial map} $\auto \colon \StructV\to \StructV$ as follows (cf.\ Figure~\ref{fig:encoding-automorphism-normalized}):
we have $\auto(\vertA) = \vertB$ in case that
there exists exactly one crossing vertex~$\vertC$
such that $(\vertC_\auto, \vertC) \in \rel_{\text{aut}}^{\applyplan{\buildplan}{\Struct}}$,
$(\vertC,\vertA') \in \rel_{\text{dom}}^{\applyplan{\buildplan}{\Struct}}$, and $(\vertC,\vertB') \in \rel_{\text{img}}^{\applyplan{\buildplan}{\Struct}}$,
where~$\vertA'$ and~$\vertB'$ are crossing vertices
and~$\vertA$ and~$\vertB$ are the only atoms
for which $(\vertA',\vertA) \in \pRel^{\applyplan{\buildplan}{\Struct}}$
and $(\vertB', \vertB) \in \pRel^{\applyplan{\buildplan}{\Struct}}$.
Recall here that every crossing vertex has exactly two $\pRel$-neighbors.
So, the other $\pRel$-neighbor of~$\vertA'$ and~$\vertB'$ must not be an atom.
While introducing the vertices~$\vertA'$ and~$\vertB'$ into the definition seems odd at first, this will simplify technical aspects in the following.
A tuple of relations $(\rel_{\text{aut}}, \rel_{\text{dom}}, \rel_{\text{img}})$
\defining{encodes the set of partial maps}
$\setcond{\auto}{\vertC_\auto \text{ encodes } \auto \text{ for some }(\vertC_\auto, \vertC) \in \rel_{\text{aut}}^\Struct}$.
\begin{figure}
	\centering
	\begin{tikzpicture}
		\draw[draw=none, use as bounding box] (-6.8,0.2) rectangle (8.8,-4.2);
		
		\node[vertex,label=0:{$\vertC_\auto$}] (autoV) at (0,0) {};
		\node[vertex, label=0:{$\vertC$}, below=1 of autoV] (vertC) {};
		\node[vertex, right=1.5 of vertC] (vertC2) {};
		\node[right=0.5 of vertC2] (vertC3) {$\dots$};
		\node[vertex, label=180:{$\vertA'$}, below left=of vertC] (vertAP) {};
		\node[vertex, label=0:{$\vertB'$}, below right=of vertC] (vertBP) {};
		\node[vertex, label=180:{$\vertA$}, below left= 1 and 0.5 of vertAP] (vertA) {};
		\node[vertex, black!40!white, below right= 1 and 0.5 of vertAP] (vertA2) {};
		\node[vertex, label=0:{$\vertB$}, below right= 1 and 0.5 of vertBP] (vertB) {};
		\node[vertex,black!40!white, below left= 1 and 0.5 of vertBP] (vertB2) {};
		
		\node[below = 0.2 of vertA, font=\footnotesize](atomU){atom};
		\node[below = 0.2 of vertB, font=\footnotesize](atomV){atom};
		\node[font=\footnotesize,black!50!white] at ($0.5*(atomU) + 0.5*(atomV)$)  {non-atoms};
		
		\path[->]
		(autoV) edge node [left]{$\rel_{\text{aut}}$} (vertC)
		(autoV) edge node [right, xshift=0.1cm]{$\rel_{\text{aut}}$} (vertC2);
		\path[->,blue]
		(vertC) edge node [left, xshift=-0.1cm]{$\rel_{\text{dom}}$} (vertAP);
		\path[->,red]
		(vertC) edge node [right]{$\rel_{\text{img}}$}(vertBP);
		\path[->,darkgreen]
		(vertAP) edge node [left]{$\pRel$} (vertA)
		(vertBP) edge node [right]{$\pRel$} (vertB);
		\path[->,darkgreen!50!white]
		(vertAP) edge node [right]{$\pRel$} (vertA2)
		(vertBP) edge node [left]{$\pRel$} (vertB2);
		
		\draw [decorate,
		decoration = {brace}] let \p1 = (autoV),\p2 = (vertAP) in
		($(4, \y1+0.2cm)$) --  ($(4, \y2-0.2cm)$)
		node[pos=0.5,black, right=0.2cm,align=left]
		{crossing vertices of $\applyplan{\buildplan}{\StructA}$};
		
		\draw [decorate,
		decoration = {brace}] let \p1 = (vertA) in
		($(4, \y1+0.3cm)$) --  ($(4, \y1-0.3cm)$)
		node[pos=0.5,black, right=0.2cm]
		{plain vertices of $\StructA$};
	\end{tikzpicture}
	\caption{Encoding of an automorphism~$\autoA$
	of a normalized HF-structure~$\StructA$ for normalized DeepWL+WSC algorithms
	by a building plan~$\buildplan$, a tuple of crossing relations $(\rel_{\text{aut}}, \rel_{\text{dom}}, \rel_{\text{img}})$,
	and a vertex $\vertC_\auto$:
	The figure shows the encoding of $\auto(\vertA) = \vertB$.
	The vertices~$\vertC_\auto$,~$\vertC$,~$\vertA'$, and~$\vertB'$
	are crossing vertices of $\applyplan{\buildplan}{\StructA}$.
	The vertices~$\vertA$ and~$\vertB$ are atoms
	and the other $\pRel$-neighbor of~$\vertA'$ respectively~$\vertB'$
	is not an atom.}
	\label{fig:encoding-automorphism-normalized}
\end{figure}

The witnessing machine~$\dwlmwit$ of a normalized DeepWL+WSC-algorithm on input~$\StructA$
outputs a set of witnessing automorphisms
by writing a tuple $(\buildplan, \rel_{\text{aut}}, \rel_{\text{dom}}, \rel_{\text{img}})$
of a building plan~$\buildplan$ for the final content of the cloud of~$\dwlmwit$
and three relations for which~$\buildplan$ contains a relation plan
on the interaction-tape.
	
Note that, with this definition, a normalized DeepWL+WSC-algorithm is formally not a (non-normalized) DeepWL+WSC-algorithm anymore, because it encodes sets of  witnessing automorphisms differently.
Finally, we are ready to simulate an arbitrary DeepWL+WSC-algorithm
with a normalized one.

\begin{definition}[Simulating a Structure]
	A pair $(\Struct, \buildplan)$
	of a normalized HF-structure~$\Struct$ and
	a building plan $\buildplan=(\omega, \kappa)$ for~$\Struct$
	\defining{simulates} an HF-structure~$\hat{\Struct}$
	with the same atoms as~$\Struct$
	if
	\begin{enumerate}[label=(S\arabic*)]
		\item \label{prop:sim-iso} $\applyplancross{\buildplan}{\Struct} \iso \nonHF{\hat{\Struct}}$,
		\item \label{prop:sim-efficient} $\buildplan$ is efficient,
		\item \label{prop:sim-atoms} there is a relation plan $(\vcA, \set{\vcB_1, \vcB_2}) \in \kappa$ defining a crossing vertex class~$\vcA$,
		such that every isomorphism $\autoA \colon \nonHF{\hat{\Struct}} \to \applyplancross{\buildplan}{\Struct}$ satisfies $\autoA(\hat{\StructV}) = \vcA^{\applyplan{\buildplan}{\Struct}}$,
		i.e.,~the atoms of~$\hat{\Struct}$ are mapped precisely onto the $\vcA$-vertices,
		\item \label{prop:sim-bijection} the $\pRel$-relation is a perfect matching between the atoms~$\StructV$ of~$\Struct$ and $\vcA^{\applyplan{\buildplan}{\Struct}}$, and
		\item \label{prop:sim-auto} via this bijection between~$\StructV$ and  $\vcA^{\applyplan{\buildplan}{\Struct}}$ we have that
		$\autGroup{\nonHF{\Struct}} = \autGroup{\applyplancross{\buildplan}{\Struct}}$.
	\end{enumerate} 
\end{definition}
Note that the definition above only relates $\applyplancross{\buildplan}{\Struct}$ (and not $\applyplan{\buildplan}{\Struct}$) to~$\hat{\Struct}$.
This definition reduces the need for case distinctions in the simulation:
Crossing vertices of $\applyplan{\buildplan}{\Struct}$ are always used to simulate~$\hat{\Struct}$
and plain vertices are always used to create crossing vertices.
Now that we have a notion of simulating a structure, we can also simulate DeepWL+WSC-algorithms:

\begin{definition}[Simulating an Algorithm]
	Let $\dwla = (\dwlmout, \dwlmwit, \dwla_1, \dots, \dwla_\ell)$
	and $\hat{\dwla} = (\hatdwlmout, \hatdwlmwit, \hat{\dwla}_1, \dots, \hat{\dwla}_\ell)$
	be DeepWL+WSC-algorithms.
	The algorithm~$\dwla$ 
	\defining{simulates}~$\hat{\dwla}$ 
	if~$\dwla_i$ simulates~$\hat{\dwla}_i$ for all $i \in [\ell]$
	and for every structure~$\hat{\StructA}$ 
	and every pair $(\Struct, \buildplan)$ simulating $\hat{\StructA}$
	the algorithm~$\dwla$ on input $(\Struct, \buildplan)$
	accepts (or respectively rejects)
	whenever~$\hat{\dwla}$ on input~$\hat{\Struct}$ accepts (or respectively rejects).
\end{definition}
We do not care about~$\hat{\dwla}$ failing
because in the following we will always assume that this is not the case.

\begin{lemma}
	\label{lem:deepwl-wsc-simulate-normalized}
	For every polynomial-time DeepWL+WSC-algorithm~$\hat{\dwla}$,
	there is a normalized polynomial-time DeepWL+WSC-algorithm~$\dwla$ simulating~$\hat{\dwla}$.
\end{lemma}
\begin{proof}
	We will construct normalized DeepWL+WSC-algorithms,
	which uses the additional operations of Lemma~\ref{lem:deepwl-wsc-plus}.
	One easily sees that the reductions in the proof of Lemma~\ref{lem:deepwl-wsc-plus}
	preserve being normalized.
	 
	The proof is by induction on nesting DeepWL+WSC-algorithms.
	For this, suppose $\hat{\dwla} = (\hatdwlmout, \hatdwlmwit, \hat{\dwla}_1, \dots, \hat{\dwla}_\ell)$ and assume by induction hypothesis
	that there are normalized DeepWL+WSC-algorithms~$\dwla_i$
	simulating~$\hat{\dwla}_i$ for every $i \in [\ell]$.
	By Lemma~\ref{lem:pure}, we can assume that~$\hat{\dwla}$ is pure.
	
	Let $\hat{\dwlm} \in \set{\hatdwlmout, \hatdwlmwit}$.
	We will now construct a DeepWL+WSC-machine~$\dwlm$ simulating~$\hat{\dwlm}$.
	Let $\hat{\Struct}_0$ be the input HF-structure of~$\hat{\dwlm}$.
	Let $\hat{\Struct}_0, \dots, \hat{\Struct}_k$
	be the sequence of HF-structures in the cloud of~$\hat{\dwlm}$
	and let $(\Struct_0, \buildplan_0)$ simulate~$\hat{\Struct}_0$.
	We construct the machine~$\dwlm$ inductively (independently of the specific input~$\StructA_0$).
	We say that the machine~$\dwlm$ on input $(\Struct_0, \buildplan_0)$ \defining{simulates the $t$-th step}
	for $t \in [k]$
	if the content of the cloud of~$\dwlm$ is~$\Struct_t$
	and there is a building plan $\buildplan_t=(\omega_t, \kappa_t)$ for~$\Struct_t$ written onto the working tape
	such that $(\Struct_t, \buildplan_t)$ simulates~$\hat{\Struct}_t$.
	Assume that we constructed a machine~$\dwlm$ simulating the $k$-th
	step, then the content of the cloud does not change anymore
	and we can just track the run of the Turing machine of~$\hat{\dwlm}$
	until it halts.
	
	We construct by induction on~$t$ a machine~$\dwlm$ simulating the $t$-th step
	for every $t \leq k$.
	For $t = 0$, the claim holds by assumption
	that $(\Struct_0, \buildplan_0)$ simulates~$\hat{\Struct}_0$.
	Now assume that~$\dwlm$ simulates the $t$-th step and
	that $(\Struct_t, \buildplan_t)$ simulates~$\hat{\Struct}_t$.
	By Lemma~\ref{lem:compute-sketch-building-plan},
	the machine~$\dwlm$ computes $\sketch{\applyplan{\buildplan_t}{\Struct_t}}$
	in polynomial time.
	From this sketch, it computes
	$\sketch{\applyplancross{\buildplan_t}{\Struct_t}}$
	using Lemma~\ref{lem:restrict-sketch}.
	
	Because $(\Struct_t, \buildplan_t)$ simulates~$\hat{\Struct}_t$,
	by Property~\ref{prop:sim-iso}, it holds that
	$\applyplancross{\buildplan_t}{\Struct_t}
	\iso \nonHF{\hat{\Struct}_t}$
	and thus $\sketch{\applyplancross{\buildplan_t}{\Struct_t}}=\sketch{\hat{\Struct}_t}$.
	So~$\dwlm$ can track the run of~$\hat{\dwlm}$
	until~$\hat{\dwlm}$ executes an operation modifying the cloud.
 	We make a case distinction on this operation:

 	\paragraph{$\addPair{\col}$:}
 	The color~$\col$ is an inter-crossing color in $\applyplan{\buildplan_t}{\Struct_t}$
 	because $\applyplancross{\buildplan_t}{\Struct_t}  \iso \nonHF{\hat{\Struct}_t}$.
 	Let~$\col_1$ and~$\col_2$ be the two plain colors corresponding to~$\col$.
 	The machine~$\dwlm$ executes $\addPair{\col_1}$ and $\addPair{\col_2}$
 	(if $\col_1 = \col_2$, only one operation is executed)
 	and obtains vertex classes~$\vc_1$ and~$\vc_2$
 	and component relations~$\rel_{i,d}$ for every $i \in [2]$ and $d \in \set{\leftT, \rightT}$
 	(here if $\col_1 = \col_2$, we have $\vc_1 = \vc_2$ and $\rel_{1,d} = \rel_{2,d}$).
 	We set 
 	\begin{align*}
 		\omega_{t+1} &:= \omega_t \cup \set[\big]{\set{\vc_1, \vc_2}} \text{ and}\\
 		\kappa_{t+1} &:= \kappa_t \cup \setcond[\big]{(\col_{d}, \set{\set{\rel_{1,d}, \rel_{2, d}}})}{d \in \set{\leftT, \rightT}}.
 	\end{align*}
 	That is, precisely for every set $\set{\vertA, \vertB}$
 	of vertices $\vertA \in \vc_1^{\Struct_{t+1}}$
 	and $\vertB \in \vc_2^{\Struct_{t+1}}$,
 	which means by construction for every set $\set{e_1, e_2}$
 	of edges $e_1  \in  \col_1^{\Struct_t}$
 	and $e_2 \in  \col_2^{\Struct_t}$,
 	which again means by Property~\ref{prop:an-direct-product-edge}
 	for every edge $e \in \col^{\applyplan{\buildplan_t}{\Struct_t}}$,
 	a vertex is added in $\applyplan{\buildplan_{t+1}}{\Struct_{t+1}}$.
 	
 	We create the component relations from the~$\rel_{i,d}$.
 	That is, ${\applyplancross{\buildplan_{t+1}}{\Struct_{t+1}} \iso \nonHF{\hat{\Struct}_{t+1}}}$
 	and so we maintained~\ref{prop:sim-iso}.
 	Property~\ref{prop:sim-efficient} is satisfied
 	because all newly created plain vertices are either in~$\vcA_1$ or~$\vcA_2$.
 	Properties~\ref{prop:sim-atoms} and~\ref{prop:sim-bijection}
 	will always be maintained if they initially hold (and we do not remove entries from the building plan).
 	Finally, property~\ref{prop:sim-auto} is maintained
 	because we did not make any choices and thus the automorphism
 	of both structures stay exactly the same by Corollary~\ref{cor:operations-apart-choice-auto-invariant}
 	(note that by Lemma~\ref{lem:auts-flat-vs-hf}
 	both~$\Struct_{t+1}$ and~$\nonHF{\Struct_{t+1}}$ have the same automorphisms). 	
 	So we simulated the $(t+1)$-th step.

	\paragraph{$\scc{\col}$:}
	Again,~$\col$ is an inter-crossing color in $\applyplan{\buildplan_t}{\Struct_t}$.
	The machine~$\dwlm$ simulates the $\sccSym$-operation
	using Lemma~\ref{lem:simulate-contract-buildplan}:
	Properties~\ref{prop:sim-iso} and~\ref{prop:sim-efficient} are ensured by the lemma,
	all other properties hold by the same reasons as for the $\addPairSym$-operation,
	and so the $(t+1)$-th step is simulated.

 	\paragraph{$\create{\pi}$:}
 	Let $\pi = \set{\colA_1, \dots, \colA_k}$,
 	where all~$\col_i$ are inter-crossing colors in $\applyplan{\buildplan_t}{\Struct_t}$,
 	and let~$\rel$ be the relation to be created.
 	Let~$\colB_i$ and~$\colC_i$ be the corresponding colors for~$\colA_i$
 	for every $i \in [k]$
 	(Property~\ref{prop:an-direct-product-edge}).
 	We update
 	\begin{align*}
 		\omega_{t+1} &:= \omega_t \text{ and}\\
 		\kappa_{t+1} &:= \kappa_t \cup \set[\big]{(\rel, \setcond{\set{\colB_i, \colC_i}}{i \in [k]})}.
 	\end{align*}
 	Because the cloud of~$\dwlm$ is not modified,
 	$\Struct_{t+1} = \Struct_{t}$ is still normalized.
 	In particular, Property~\ref{prop:sim-efficient} is maintained because no new vertices are created.
 	By construction Property~\ref{prop:sim-iso} is satisfied.
 	Again by the same reasons as before, 
 	the remaining properties are satisfied and the $(t+1)$-th step is simulated.
 	
 	\paragraph{$\refine{\ccA}{k}$:}
 	Because~$\hat{\dwla}$ is pure, $\refineSym$ is
 	only executed for fibers.
 	Let~$\ccA_1$ and~$\ccA_2$ be the plain fibers corresponding to~$\ccA$
 	and let~$\rel_\ccA$ be the crossing color such that
 	${\rel_\ccA^\Struct = \rel_{\set{\ccA_1,\ccA_2}}^\Struct}$,
 	that is,~$\rel_\ccA$ precisely connects all the pairs of vertices corresponding to a $\ccA$-vertex.

 	We cannot simply execute $\refine{\rel_\ccA}{k, \buildplan_t}$
 	because this might create a relation
 	which connects the two components of~$\Struct_t$
 	and so the content of the cloud would not be normalized anymore.
 	(Recall the example in Section~\ref{sec:normalized-normalized}.)
 	We are going to refine~$\ccA_1$ and~$\ccA_2$,
 	then decompose~$\rel_\ccA$ into multiple colors,
 	such that either all edges of a color are accepted by~$\dwla_k$ or no edge of the color is.
 	This results in the structure~$\Struct_{t+1}$.
 	We then refine one (crossing) color after the other.
 	This way,
 	for every color, either no new relation or an empty relation is created
 	and the structure stays normalized.
 	
 	Let $\vertB,\vertB' \in \ccA_2^{\Struct_t}$ be in the same component of~$\StructA_t$,
 	say w.l.o.g.~$\vertB,\vertB' \in \ccA_2^{\Struct_t} \cap \verticesOf{2}{\Struct_t}$.
 	Then by Lemma~\ref{lem:distinguish-cross-relation-distinguish-vertices},
 	we can distinguish~$\vertB$ and~$\vertB'$
 	using the sets
 	\begin{align*}
 		&\setcond[\big]{\run{\dwla_k}{(\Struct, \set{\vertC,\vertB}}}{\vertC \in \ccA_1^{\Struct_t} \cap \verticesOf{1}{\StructA_t}} \text{ and }\\
 		&\setcond[\big]{\run{\dwla_k}{(\Struct, \set{\vertC,\vertB'}}}{\vertC \in \ccA_1^{\Struct_t} \cap \verticesOf{1}{\StructA_t}}
 	\end{align*}
 	if for some $\vertA\in \ccA_1^{\Struct_t}$
 	we have that $(\Struct_t, \set{\vertA,\vertB})$ is accepted by~$\dwla_k$
 	but $(\Struct_t, \set{\vertA,\vertB'})$ is not.
 	There is a normalized DeepWL+WSC-algorithm computing a function 
 	distinguishing the same (or more) vertices as these sets according to Lemma~\ref{lem:compute-distinguish-vertices-function}
 	(note that because we input all $(\vertA,\vertB)$ or  $\set{\vertA,\vertB}$
 	during the $\refineSym$-operation,~$\dwla_k$
 	indeed never fails because otherwise~$\hat{\dwla}_k$ would have failed).
 	We can refine~$\ccA_2$
 	according to these sets by Lemma~\ref{lem:deepwl-refine-with-function}.
 	The same procedure is performed for~$\ccA_1$.
 	After that, we obtain the structure~$\Struct_{t+1}'$ in the cloud.
 	
 	Before we can use~$\dwla_k$ to refine the edges,
 	we have to slightly modify it:
 	The input of~$\dwla_k$ is $\Struct_t' := ((\Struct_{t+1}', \set{\vertA, \vertB}), \buildplan_t')$
 	where some $\set{\vertA, \vertB} \in \rel_\ccA^{\Struct_t}$
 	is individualized, i.e.,~there is a new vertex class $\vc_{\set{\vertA,\vertB}}$ only containing~$\vertA$ and~$\vertB$.
 	(If~$\rel_\ccA$ is directed, the argument still applies because if~$\rel_\ccA$ is directed, then $\rel_\ccA^{\StructA_t}\subseteq \ccA_1^{\StructA_t} \times \ccA_2^{\StructA_t}$ and the ordered tuple $(\vertA,\vertB)$ can be recovered from the set $\set{\vertA,\vertB}$).
 	This corresponds to individualizing one $\ccA$-vertex
 	(recall that~$\pRel$ does \emph{not} distinguish the two components of~$\Struct_t'$).
	In order to represent this individualization of the $\ccA$-vertex,
	the algorithm~$\dwla_k$ modifies the building plan as follows:
 	$\kappa_t' := \kappa_t \cup \set{(\relB', \set{(\vc_{\set{\vertA,\vertB}},\vc_{\set{\vertA,\vertB}})})}$
	 and $\buildplan_t' := (\omega_t, \kappa_t')$.
 	This way, the vertex~$\vertC$ corresponding to $\set{\vertA, \vertB}$
 	in $\applyplan{\buildplan_t'}{\Struct_t'}$
 	gets individualized by the relation~$\relB'$.
 	Now $(\Struct_t', \buildplan_t')$ simulates the structure
	$(\hat{\Struct_t},\vertC)$, 
	where the $\ccA$-vertex~$\vertC$ is individualized.
 	All conditions for the induction hypothesis are satisfied
 	and we can run the algorithm~$\dwla_k$
 	by induction hypothesis.
 
 	So we decompose~$\rel$ into colors $\colA_1, \dots , \colA_m$
 	and  execute $\refine{\colA_i}{k, \buildplan_t}$
 	for every $i \in [m]$.
 	By Lemma~\ref{lem:distinguish-cross-relation-distinguish-vertices},
 	either all edges in a color are accepted or none of them.
 	So the machine can compute the set $I := \setcond{i \in [m]}{\dwla_k \text{ accepts the } \colA_i \text{-edges}}$
 	by distinguishing for every~$i$ whether no relation is created (so $i \in I$) or an empty relation is created (so $i \notin I$).
 	After that, the structure~$\Struct_{t+1}$ is in the cloud.
 	In particular,~$\Struct_{t+1}$ and all intermediate steps are normalized
 	because we only create empty relations.
 	Let $\set{\colB_i, \colC_i}$ be the corresponding colors for~$\colA_i$
 	for every $i \in [m]$.
 	We create the vertex class~$\vcB$ in $\applyplan{\buildplan_t}{\Struct_t}$,
 	which is the result of the $\refine{\ccA}{k}$-operation of~$\hat{\dwlm}$, as follows:
 	\begin{align*}
 		\omega_{t+1}&:= \omega_t, \\
 		\kappa_{t+1}&:= \kappa_t \cup \set[\big]{(\vcB, \setcond{\set{\colB_i,\colC_i}}{i \in I})}.
 	\end{align*}
 	This establishes Property~\ref{prop:sim-iso}.
 	The other properties still hold for the same reasons as before
 	and the $(t+1)$-th step is simulated.

 	\paragraph{$\choice{\ccA}$:}
 	Recall that in this case $\hat{\dwlm} = \hatdwlmout$
 	because~$\hatdwlmwit$ is choice-free.
 	Again because~$\hat{\dwla}$ is pure,~%
 	$\ccA$ is a crossing fiber in $\applyplan{\buildplan_t}{\Struct_t}$.
 	
 	The machine defines the undirected relation~$\rel_\ccA$
 	as in the $\refineSym$-case
 	and executes $\choice{\rel_\ccA}$.
 	By the semantics of the $\choiceSym$-operator,
 	this individualizes an undirected $\rel_\ccA$-edge
 	(recall that, to individualize this edge, 
 	we obtain a vertex class containing both endpoints of the edge
 	and the two components of~$\Struct_t$ are not connected).
 	Using this edge we can individualize the corresponding crossing vertex
 	in $\applyplan{\buildplan_t}{\Struct_t}$
 	by defining a singleton fiber~$\ccB$ similar to the $\refineSym$-case.
 	This shows Property~\ref{prop:sim-iso}.
 	Properties~\ref{prop:sim-efficient}, \ref{prop:sim-atoms}, and~\ref{prop:sim-bijection}
 	still hold as seen before.
 	To show Property~\ref{prop:sim-auto},
 	observe that
 	every automorphism of $\applyplancross{\buildplan_{t+1}}{\Struct_{t+1}}$ is an automorphism of 
 	$\applyplancross{\buildplan_t}{\Struct_t}$, which additionally fixes the singleton crossing vertex in~$\ccB$.
 	Thus, such an automorphism also has to fix the corresponding undirected $\rel_\ccA$-edge and so it is an automorphism of $\applyplan{\buildplan_{t+1}}{\Struct}$.
 	Vice versa, every automorphism fixing this edge also fixes the crossing vertex
 	and the $(t+1)$-th step is simulated.
 	
 	\bigskip
 	Finally, we have to alter~$\dwlmwit$: when~$\hatdwlmwit$
 	writes the tuple of automorphism-encoding relations $(\rel_{\text{aut}}, \rel_{\text{dom}}, \rel_{\text{img}})$
 	onto the interaction-tape,~$\dwlmwit$ only constructed relation plans for these relations. Now~$\dwlmwit$ writes $(\buildplan, \rel_{\text{aut}}, \rel_{\text{dom}}, \rel_{\text{img}})$ onto the interaction-tape,
 	where~$\buildplan$ is the building plan maintained by~$\dwlmwit$.
 	
 	To show that $(\dwlmout, \dwlmwit, \dwla_1, \dots, \dwla_\ell)$
 	indeed simulates $(\hatdwlmout, \hatdwlmwit, \hat{\dwla}_1, \dots, \hat{\dwla}_\ell)$,
 	it remains to show that all choices are witnessed
 	(if not all choices of~$\hat{\dwlm}$ are witnessed,~%
 	$\dwlm$ is allowed to do anything).
 	First, the machine~$\dwlmout$ only executes a $\choice{\rel_\ccA}$-operation
 	when~$\hatdwlmout$ executes the operation $\choice{\ccA}$.
 	Because~$\rel_\ccA$ is undirected
 	and due to the semantics of the choice operator for undirected edges,
 	we have to witness that ${\crossVertexRel{\ccA} := \crossVertexRel{\set{\ccA_1,\ccA_2}}}$ is an orbit,
 	where $\set{\ccA_1,\ccA_2}$ are the corresponding plain fibers of~$\ccA$
 	(cf.\ the definition of $\rel_\ccA = \rel_{\set{\ccA_1, \ccA_2}}$).
 	Assume that $\choice{\ccA}$ is executed with the HF-structure~$\hat{\StructA}_t$ in the cloud.
 	Let~$N$ be the set of automorphisms
 	encoded by $(\rel_{\text{aut}}, \rel_{\text{dom}}, \rel_{\text{img}})$,
 	which witnesses~$\ccA$ as orbit
 	fixing all previous steps $\hat{\Struct}_{j_1}, \dots, \hat{\Struct}_{j_\ell}$ (where $j_1 < \dots < j_\ell$),
 	in which $\choiceSym$-operations where executed by~$\hatdwlmout$.
 	Then the tuple $(\buildplan, \rel_{\text{aut}}, \rel_{\text{dom}}, \rel_{\text{img}})$
 	precisely encodes the same set of automorphisms~$N$ for~$\Struct_t$:
 	for every vertex $\hat{\vertC}_\auto$ encoding a partial map~$\auto$
 	via $(\rel_{\text{aut}}, \rel_{\text{dom}}, \rel_{\text{img}})$ in~$\hat{\Struct}_t$,
 	the corresponding vertex~$\vertC_\auto$ in $\applyplan{\buildplan}{\StructA_t}$
 	encodes the same partial map because of the following.
 	By the definition of encoding automorphisms by DeepWL+WSC-algorithms,
 	the vertex~$\hat{\vertC}_\auto$ encodes that $\auto(\hat{\vertA}) = \hat{\vertB}$
 	for atoms~$\hat{\vertA}$ and~$\hat{\vertB}$
 	if there is a unique vertex~$\hat{\vertC}$ such that 
 	$(\hat{\vertC}_\auto,\hat{\vertC})$ is in~$\rel_{\text{aut}}$,
 	$(\hat{\vertC},\hat{\vertA})$ is in~$\rel_{\text{dom}}$, and
 	$(\hat{\vertC},\hat{\vertB})$ is in~$\rel_{\text{img}}$.
 	By Properties~\ref{prop:sim-atoms} and~\ref{prop:sim-bijection},
 	the corresponding vertices~$\vertC$,~$\vertA$, and~$\vertB$
 	in $\applyplan{\buildplan}{\StructA}$ are crossing vertices
 	such that~$\vertA$ has~$\hat{\vertA}$ as $\pRel$-neighbor (and no other atom)
 	and likewise~$\vertB$ has~$\hat{\vertB}$ as $\pRel$-neighbor.
 	This exactly is the definition of encoding automorphisms for normalized DeepWL+WSC-algorithms.
 	By Property~\ref{prop:sim-auto} and Lemma~\ref{lem:auts-flat-vs-hf},
 	the automorphisms of~$\Struct_t$ and~$\hat{\Struct}_t$ are equal.
 	By Lemma~\ref{lem:buildplan-automorphisms},~%
 	$N$ witnesses $\crossVertexRel{\ccA}$ as orbit
 	fixing $\crossVertexRel{\ccA_1}^{\Struct_{j_1}}, \dots, \crossVertexRel{\ccA_\ell}^{\Struct_{j_\ell}}$ in the simulation.
 	
 	Finally, to see that the algorithm runs in polynomial time,
 	we note that~$\dwlm$ needs polynomially many steps to simulate a single step of~$\hat{\dwlm}$.
 	So it remains to argue that the structure in the cloud is of polynomial size. 
 	By Property~\ref{prop:sim-efficient},
 	the building plans are always efficient.
 	Then, by Lemma~\ref{lem:building-plan-efficient-size},
 	we have that $|\Struct_t| \leq 2|\hat{\Struct}_t|$.
 	Because additionally the number of relations in~$\Struct_t$ is polynomially bounded,
 	$|\sketch{\Struct_t}|$ is polynomially bounded, too.
\end{proof}

\begin{corollary}
	\label{cor:normalized-iso}
	Let~$\GraphClass$ be a class of binary structures.
	If there is a polynomial-time DeepWL+WSC-algorithm deciding isomorphism for~$\GraphClass$,
	then there is a normalized\linebreak polynomial-time DeepWL+WSC-algorithm deciding isomorphism for~$\GraphClass$.
\end{corollary}
\begin{proof}
	Let~$\hat{\dwla}$ be a polynomial-time DeepWL+WSC-algorithm
	deciding isomorphism for~$\GraphClass$
	and let~$\dwla$ be a normalized polynomial-time DeepWL+WSC-algorithm
	simulation~$\hat{\dwla}$ given by Lemma~\ref{lem:deepwl-wsc-simulate-normalized}.
	Let $\Struct_1, \Struct_2 \in \GraphClass$
	and $\Struct = \Struct_1 \disunion \Struct_2$.
	To execute~$\dwla$ to decide isomorphism,
	we first have to construct a building plan~$\buildplan$
	such that $(\Struct, \buildplan)$ simulates~$\Struct$.
	
	This is done as follows:
	First, we define the vertex class~$\vcB$ of all atoms.
	Second, we create a plain vertex~$\vertA_i$ for both components~$\Struct_i$,
	and a vertex class~$\vcB'$ only containing~$\vertA_1$ and~$\vertA_2$.
	To do so, we define a relation~$\relB$ connecting all atoms in the same component.
	We execute $\scc{\relB}$ and obtain a vertex class~$\vcB'$
	containing one vertex per component.
	Then we rename every relation $\rel \in \sigA$ to a fresh relation symbol~$\rel'$.
	We then set $\buildplan := (\omega, \kappa)$, where
	\begin{align*}
		\omega &:= \set[\big]{\set{\vcB,\vcB'}} \text{ and}\\
		\kappa &:= \set[\big]{(\vcA,\set{\vcB,\vcB'})} \cup \setcond[\big]{(\rel, \set{\set{\rel', \vcB'}})}{\rel \in \sigA}.
	\end{align*}
	We claim that $(\Struct, \buildplan)$ simulates $\Struct$.
	With~$\omega$, we create for every atom $\vertA\in \StructV_i$ in $\vcB$
	a single crossing vertex, because~$\vcB'$ contains exactly one vertex in $\Struct_{j}$ (where $\set{i,j} = [2]$).
	Two such crossing vertices~$\vertA$ and~$\vertB$ are in the relation~$\rel$
	if between the corresponding vertices there is a $\vcB'$-edge (the $\vcB'$-loop)
	and an $\rel'$-edge.
	That is, $\plvI{\vertA}{i} = \plvI{\vertB}{i}$ is a $\vcB'$-vertex,
	so~$\vertA$ and~$\vertB$ are copies of plain vertices of the same component $\Struct_j$,
	and $(\plvI{\vertA}{j}, \plvI{\vertB}{j}) \in \rel^\Struct$,
	where~$\rel$ was renamed to~$\rel'$.
	So we have established Property~\ref{prop:sim-iso}.
	Clearly, every vertex is used in the building plan because every vertex
	is either in~$\vcB$ or in~$\vcB'$, thus Property~\ref{prop:sim-efficient}
	is satisfied.
	Every crossing vertex is adjacent via~$\pRel$ to exactly one atom in~$\StructV$ because the~$\vcB'$-vertices are not atoms.
	The relation plan $(\vcA,\set{\vcB,\vcB'})$
	creates a vertex class $\vcA$ containing all crossing vertices.
	Because each crossing vertex is $\pRel$-connected to an atom and every atom to a $\vcA$-vertex,
	we established Properties~\ref{prop:sim-atoms} and~\ref{prop:sim-bijection}.
	Lastly, we also established Property~\ref{prop:sim-auto}
	because $\applyplan{\buildplan}{\Struct}$
	consists essentially of two copies of~$\Struct$ connected by a perfect matching.
	A similar construction (but without building plans) can be found in the proof of Lemma~11 in~\cite{GroheSchweitzerWiebking2021}.
\end{proof}

\subsubsection{Deciding Isomorphism and Internal Runs}
In this section we prove that for every class of structures~$\GraphClass$,
isomorphism is decidable by a polynomial-time DeepWL+WSC-algorithm
if and only if a complete invariant is computable by a polynomial-time  DeepWL+WSC-algorithm.
To do so, we consider the internal runs of normalized DeepWL+WSC-algorithms.
We say that the components of a normalized HF-structure $\StructA_1 \disunion \StructA_1$ are \defining{distinguished} if $\sketch{\StructA_1} \neq \sketch{\StructA_2}$.

\begin{lemma}
	\label{lem:same-run-normalized}
	For every  normalized polynomial-time DeepWL+WSC-algorithm~$\dwla$,
	there is a normalized polynomial-time DeepWL+WSC-algorithm~$\dwla'$ with the following properties:
	\begin{enumerate}[label=\alph*)]
		\item For every normalized HF-structure~$\StructA$, the algorithm~$\dwla'$ accepts (respectively rejects)~$\StructA$ whenever~$\dwla$ accepts (respectively rejects)~$\StructA$.
		\item For every normalized HF-structures~$\StructA = \StructA_1\disunion\StructA_2$ and~$\StructB = \StructB_1\disunion\StructB_2$,
		 if $\run{\dwla'}{\StructA} = \run{\dwla'}{\StructB}$,
		then $\run{\dwla'}{\StructA} \in \setcond{\run{\dwla'}{\StructA_1 \disunion \StructB_i}}{i \in [2]}$.
	\end{enumerate}
\end{lemma}
\begin{proof}
	A normalized DeepWL+WSC-algorithm is \defining{nice} on input~$\StructA$ if
	\begin{enumerate}
		\item whenever $\refine{\relColA}{i}$ is executed
		and the components of the current structure in the cloud~$\StructA'$ are not distinguished,~%
		then $\relColA^{\Struct'} = \set{\vertA_1,\vertA_2}$ 
		such that~$\vertA_1$ and~$\vertA_2$ are in different components of~$\StructA'$,
		\item for every other $\refine{\relColA}{i}$-execution,
		$\relColA$ is a plain color, and
		\item whenever $\choice{\col}$ is executed,
		then the two components of the current structure in the cloud are distinguished
		and~$\col$ is a plain color.
	\end{enumerate}
	Note that if $\dwla = (\dwlmout,\dwlmwit,\dwla_1, \dots, \dwla_\ell)$
	is nice on input~$\StructA$,
	then $\dwla_1, \dots, \dwla_\ell$ must be nice on all inputs
	during $\refineSym$-operations of~$\dwlmout$ or~$\dwlmwit$.
	
	\begin{claim}
		\label{clm:normalized-to-nice}
		For every normalized polynomial-time DeepWL+WSC-algorithm~$\dwla$,
		there is normalized polynomial-time DeepWL+WSC-algorithm~$\dwla'$
		such that for every normalized HF-structure~$\StructA$ on which~$\dwla$ does not fail,~%
		$\dwla'$ is nice on~$\StructA$
		and~$\dwla'$ accepts (respectively rejects)~$\StructA$
		if~$\dwla$ accepts (respectively rejects)~$\StructA$.
	\end{claim}
	\begin{claimproof}
		Let~$\dwla$ be a normalized polynomial-time DeepWL+WSC-algorithm.
		Let~$\dwla''$ be the normalized polynomial-time DeepWL+WSC-algorithm given for~$\dwla$
		by Lemma~\ref{lem:normalized-distinguish-components}.
		That is, for every normalized HF-structure $\StructA_1 \disunion \StructA_2$,~%
		$\dwla''$ accepts (respectively rejects) 
		$(\StructA_1 \disunion \StructA_2, \StructVA_i)$
		whenever~$\dwla$ accepts (respectively rejects) $\StructA_1 \disunion \StructA_2$ for every  $i \in [2]$
		(the lemma states it only for $i=1$, the case $i=2$  follows from exchanging the two components).
		Then define the DeepWL+WSC-algorithm $\dwla'= (\dwlmout, \dwlmwit, \dwla'')$
		as follows:
		First, the machine~$\dwlmout$ creates a vertex class~$\vc$ containing one vertex
		per component by executing $\scc{\rel}$ for the relation~$\rel$ of plain edges.
		Next,~$\dwlmout$ executes $\refine{\vc}{1}$,
		i.e.,~$\dwlmout$ refines~$\vc$ with~$\dwla''$
		(if~$\vc$ decomposes into two fibers~$\ccA$ and~$\ccB$,
		then~$\dwlmout$ executes $\refine{\ccA}{1}$ because we are only allowed to refine colors).
		Individualizing one $\vc$\nobreakdash-vertex corresponds to creating a relation containing the vertices of one component,
		so~$\dwla''$ will accept (respectively reject) the $\vc$-vertices
		whenever~$\dwla$ accepts (respectively rejects) the input to~$\dwla$.
		Then~$\dwlm$ accepts (respectively rejects) accordingly.
		The witnessing machine~$\dwlmwit$ is not used and immediately halts.
		
		It is clear that~$\dwla'$ accepts (respectively rejects) if~$\dwla$ does so.
		We show how we have to modify~$\dwla'$ such that~$\dwla'$ is nice on every input:
		\begin{enumerate}
			\item By construction,~$\dwla'$ executes a single $\refineSym$-operation when the two components are not distinguished.
			Indeed, the refined relation is a vertex class containing a single vertex of each component.
			After that, the two components are distinguished (and remain so)
			because on of the two vertices is individualized.
			\item A $\choice{\col}$-operation is only executed after the initial $\refineSym$-operation, so the components of the structure~$\StructA$ in the cloud are distinguished.
			Assume that~$\col$ is a crossing color,
			then $\col^\Struct \subseteq \ccA^\StructA \times \ccB^\StructA$ for two fibers $\ccA$ and $\ccB$ in different components
			because~$\StructA$ is normalized.
			We can equivalently execute
			$\choice{\ccA}$ and $\choice{\ccB}$.
			The automorphisms witnessing $\choice{\col}$
			will also witness the two other operations
			because $\col^\Struct$ is an orbit if and only if~%
			$\ccA^\StructA$ and~$\ccB^\StructA$ are orbits.
			\item  For all  $\refine{\relColA}{i}$-operations apart from the first one,
			the components of the structure~$\StructA$ in the cloud are distinguished.
			By decomposing~$\relColA$ into its colors,
			we can assume that~$\relColA$ is a color~$\col$.
			If~$\col$ is a crossing color,
			then again $\col^\Struct \subseteq \ccA^\StructA \times \ccB^\StructA$ for two different fibers~$\ccA$ and~$\ccB$ in different components because~$\StructA$ is normalized
			and its components are distinguished.
			Let~$\dwla_i$ be the DeepWL+WSC-algorithm used to refine~$\col$.
			We know that either every $(\vertA,\vertB) \in \col^\Struct$ is accepted by~$\dwla_i$ or 
			every $(\vertA,\vertB) \in \col^\Struct$ is rejected by~$\dwla_i$ because~$\dwla$ is normalized.
			These are the two only possibilities without creating a crossing relation, which would make the structure non-normalized.
			So we can equivalently execute two nested $\refineSym$-operations:
			we first refine~$\ccA$ with a new algorithm~$\dwla_i'$.
			The algorithm~$\dwla_i'$, which gets $(\Struct, \vertA)$ for some $\vertA \in \ccA^\Struct$ as input,
			immediately refines~$\ccB$
			with the algorithm~$\dwla_i$,
			which then gets $(\Struct,\vertA\vertB)$ for some $\vertB \in \ccB^\Struct$ as input.
			The algorithm~$\dwla_i'$ accepts $(\Struct, \vertA)$
			if~$\dwla_i$ accepts $(\Struct,\vertA\vertB)$ for every $\vertB \in \ccB^\Struct$.
			If all $\ccA$-vertices are accepted by~$\dwla_i'$,
			then $(\StructA, \vertA\vertB)$ is accepted by~$\dwla_i$
			for every 
			$(\vertA,\vertB) \in \col^\StructA$ 
			and thus no relation is created.
			Otherwise,~$\dwla_i$ rejects $(\StructA, \vertA\vertB)$ for every 
			$(\vertA,\vertB) \in \col^\StructA$.
			Thus,~$\dwla_i'$ rejects $(\StructA, \vertA)$ for every $\vertA \in \ccA^\Struct$ and an empty relation is created.
			\qedhere
		\end{enumerate} 
	\end{claimproof}

	\begin{claim}
		\label{clm:normalized-and-nice-combine-run}
		For every normalized DeepWL+WSC-algorithm~$\dwla$
		and every normalized HF-structures~$\StructA = \StructA_1 \disunion \StructA_2$ and~$\StructB = \StructB_1 \disunion \StructV_2$,
		if $\dwla$ is nice on~$\StructA$ and~$\StructB$ such that
		$\run{\dwla}{\StructA} = \run{\dwla}{\StructB} \neq \choiceError$,
		then $\run{\dwla}{\StructA} \in \setcond{\run{\dwla}{\StructA_1 \disunion \StructB_i}}{i \in [2]}$.
	\end{claim}
	\begin{claimproof}
	The proof is by induction on the nesting depth of DeepWL+WSC-algorithms.
	Let $\dwla = (\dwlmout, \dwlmwit, \dwla_1, \dots, \dwlm_\ell)$ be a normalized DeepWL+WSC-algorithm
	and let $\StructA = \StructA_1 \disunion \StructA_2$
	and $\StructB = \StructB_1 \disunion \StructB_2$ be two normalized HF-structures
	such that~$\dwla$ is nice on~$\StructA$ and~$\StructB$ and
	$\run{\dwla}{\StructA} = \run{\dwla}{\StructB}$.
	Let $\StructA^1, \dots, \StructA^m$ be the sequence of structures in the cloud of~$\dwlmout$ on input~$\StructA$,
	$\StructA^i = \StructA^i_1 \disunion \StructA^i_2$ for every $i \in [m]$,
	and similarly let $\StructB^1, \dots, \StructB^m$ be the same sequence on input~$\StructB$ and
	$\StructB^i = \StructB^i_1 \disunion \StructB^i_2$ for every $i \in [m]$.
	Because $\run{\dwla}{\StructA} = \run{\dwla}{\StructB}$ implies $\run{\dwlmout}{\StructA} = \run{\dwlmout}{\StructB}$,
	the two sequences have indeed the same length and satisfy
	$\sketch{\StructA^j} = \sketch{\StructB^j}$ for every $j \in [m]$.
	From Lemma~\ref{lem:normalized-direct-product}
	it follows that 
	$\setcond{\sketch{\StructA^j_i}}{i \in [2]} = 
	\setcond{\sketch{\StructB^j_i}}{i \in [2]}$
	for every $j \in [m]$.
	Moreover, if the components of~$\Struct^j$ are distinguished for some $j \in[m]$,
	then the components of~$\Struct^k$ are distinguished for every $k \geq j$
	(and likewise for the~$\StructB^j)$.
	So there is permutation ${\rho \colon [2] \to [2]}$ such that
	$\sketch{\StructA^j_i} = \sketch{\StructB^j_{\rho(i)}}$ for every $i \in[2]$
	and $j \in [m]$.
	Assume w.l.o.g.~that~$\rho$ is the identity map.
	
	Consider the HF-structure $\StructC = \StructA_1 \disunion \StructB_2$.
	We claim that $\run{\dwla}{\StructC} = \run{\dwla}{\StructA} = \run{\dwla}{\StructB}$.
	We first show that $\run{\dwlmout}{\StructC} = \run{\dwlmout}{\StructA} = \run{\dwlmout}{\StructB}$.
	Then in particular the sequence of structures $\StructC^1, \dots, \StructC^{m'}$ in the cloud of~$\dwlmout$ on input~$\StructC$
	satisfies $m' = m$ and $\StructC^j_1 = \StructA^j_1$ and $\StructC^j_2 = \StructB^j_2$ for every $j \in [m]$.
	This implies by Lemma~\ref{lem:normalized-direct-product}
	that $\sketch{\StructC^j} = \sketch{\StructA^j} = \sketch{\StructB^j}$ for every $j \in [m]$. 
	We show by induction on~$j$,
	that $\StructC^j_1 = \StructA^j_1$ and $\StructC^j_2 = \StructB^j_2$
	and that the  configuration of $\dwlmout$
	is equal when the $j$-th cloud modifying operation is performed
	on input~$\StructA$,~$\StructB$, and~$\StructC$
	for all $j \in [m]$
	(for~$\StructA$ and~$\StructB$ the claim follows from $\run{\dwlmout}{\StructA} = \run{\dwlmout}{\StructB}$).
	
	For $j = 1$,
	this is the case by construction: $\StructC^1 = \StructC =  \StructA_1 \disunion \StructB_2$ and~$\dwlmout$ is started in its initial state.
	
	So assume that the claim holds for $j \geq 1$
	and we show that it holds for $j+1$.
	Because the sketches of the structures in the cloud are equal
	and the Turing machine of~$\dwlmout$ is in the same configuration,
	the run of~$\dwlmout$ on~$\StructC_j$ is equal to the one on~$\StructA_j$ (or equally on~$\StructB_j$)
	until the next cloud-modifying operation is executed.
	Because the Turing machines are in the same configuration,
	the same operation is executed for~$\StructC^j$
	as for~$\StructA^j$ and~$\StructB^j$.
	\begin{enumerate}[label=\alph*)]
		\item If the operation modifying the cloud is an $\addPairSym$-, $\sccSym$-, or $\createSym$-operation,
		the effect of the operation is clearly given by the effect on each component because~$\dwlmout$ is normalized.
		So, since $\StructC^j_1 = \StructA^j_1$ and $\StructC^j_2 = \StructB^j_2$,
		we have that $\StructC^{j+1}_1 = \StructA^{j+1}_1$ and $\StructC^{j+1}_2 = \StructB^{j+1}_2$.
		
		\item  Assume that~$\dwlmout$ executes $\refine{\col}{k}$.
		Because $\run{\dwlmout}{\StructA} = \run{\dwlmout}{\StructB}$,
		we have
		\[\bigmsetcondition{\run{\dwla_k}{(\StructA^j, x)}}{x \in \tilde{\col}^{\StructA^j}} = \bigmsetcondition{\run{\dwla_k}{(\StructB^j, x)}}{x \in \tilde{\col}^{\StructB^j}}.\]
		Recall that $\tilde{\col}^{\Struct^j} = \relColA^{\Struct^j}$ if $\col$ is directed
		and $\tilde{\col}^{\Struct^j} = \setcond{\set{\vertA,\vertB}}{(\vertA,\vertB) \in \col^{\Struct^j}}$  if $\col$ is undirected.
		So there is a bijection $\zeta \colon \tilde{\col}^{\StructA^j} \to \tilde{\col}^{\StructB^j}$ such that 
		$\run{\dwla_k}{(\StructA^j, x)} = \run{\dwla_k}{(\StructB^j, \zeta(x))}$
		for every $x \in \tilde{\col}^{\StructA^j}$.
		
		Assume first that~$\col$ is a crossing color and undirected in~$\StructA^j$,~$\StructB^j$, and thus in~$\StructC^j$.
		So the components are not distinguished in all three structures
		(otherwise, all crossing colors are directed).
		Because~$\dwla$ is nice on~$\StructA$ and~$\StructB$,~%
		$\col$ is a fiber and contains one vertex per component.
		So assume $\col^{\StructA^j} =\set{\vertA_1, \vertA_2}$ and
		$\col^{\StructB^j} =\set{\vertB_1, \vertB_2}$,
		and thus $\col^{\StructC^j} =\set{\vertA_1, \vertB_2}$.
		
		We can assume that~$\zeta$ satisfies $\zeta(\vertA_i) = \vertB_i$ for every $i \in [2]$ because we assumed that $\rho$ is the identity map
		and so $\sketch{\StructA^j_i} = \sketch{\StructB^j_i}$ for every $i \in [2]$.
		That is, $\run{\dwla_k}{(\StructA^j, \vertA_i)} = \run{\dwla_k}{(\StructB^j, \vertB_i)}$ for every $i \in [2]$.
		
		Let $\vertC \in \col^{\StructC_j}$
		and w.l.o.g.~assume that $\vertC = \vertA_1$ (the case $\vertC= \vertB_2$ is symmetric).
		So $(\StructC^j_1,\vertC) = (\StructA^j_1, \vertA_1)$,
		$\StructC^j_2 = \StructB^j_2$,
		and thus $(\StructC^j,\vertC) = (\StructA^j_1, \vertA_1) \disunion \StructB^j_2$.
		Because $\run{\dwla_k}{(\StructA^j, \vertA_1)} = \run{\dwla_k}{(\StructB^j, \vertB_1)}$
		and~$\dwla_k$ is nice on~$\StructA^j$ and~$\StructB^j$,
		we can apply the outer induction hypothesis
		and conclude that
		\begin{align*}
			\run{\dwla_k}{(\StructC^j, \vertC)} &= \run{\dwla_k}{(\StructA^j, \vertA_1)} = \run{\dwla_k}{(\StructB^j, \vertB_1)}.
		\intertext{Note that we assumed that~$\rho$ is the identity map
		and hence}
			\text{if } \run{\dwla_k}{\StructA^j} &\in \setcond{\run{\dwla_k}{\StructA^j_1 \disunion \StructB^j_i}}{i \in [2]},\\
			\text{then } \run{\dwla_k}{\StructA^j} &= \run{\dwla_k}{\StructA^j_1 \disunion \StructB^j_2} = \run{\dwla_k}{\StructC^j}.
		\end{align*}
		In particular,~$\dwla_k$ accepts (respectively rejects) $(\StructC^j, \vertA_1)$
		if and only if~$\dwla_k$ accepts (respectively rejects) $(\StructA^j, \vertA_1)$.
		The same holds (by symmetry) for $(\StructC^j, \vertB_2)$ and $(\StructB^j, \vertA_2)$.
		Thus, the resulting vertex class~$\vcB$ of the 
		\refineSym-operation 
		satisfies $\vcB^{\StructC^{j+1}} = \vcB^{\StructA^{j+1}_1} \cup \vcB^{\StructB^{j+1}_2}$.
		That is, 
		$\StructC^{j+1}_1 = \StructA^{j+1}_1$,
		$\StructC^{j+1}_2 = \StructB^{j+1}_2$, and
		\begin{align*}
			\bigmsetcondition{\run{\dwla_k}{(\StructC^j,\vertC)}}{\vertC \in \set{\vertA_1,\vertB_2}}& = 
		\bigmsetcondition{\run{\dwla_k}{(\StructA^j,\vertC)}}{\vertC \in \set{\vertA_1,\vertA_2}}\\
		& = 
		\bigmsetcondition{\run{\dwla_k}{(\StructB^j,\vertC)}}{\vertC \in \set{\vertB_1,\vertB_2}}.
		\end{align*}
		If otherwise~$\col$ is not crossing or not undirected,
		then~$\col$ is a plain color in~$\StructA^j$
		(and thus in~$\StructB^j$ and~$\StructC^j$) because~$\dwla$ is nice.
		So every $x \in \tilde{\col}^{\StructC^j}$
		consists solely of vertices of either $\vertices{\StructC^j_1}$ or $\vertices{\StructC^j_2}$.
		Assume w.l.o.g.~that $x$ consists of vertices of $\vertices{\StructC^j_1}$.
		Then $(\StructC^j_1,x) = (\StructA^j_1,x)$,
		$\StructC^j_2 = \StructB^j_2$,
		and $(\StructC^j,x) = (\StructA^j_1,x) \disunion \StructB^j_2$.
		
		Because $\run{\dwla_k}{(\StructA^j, x)} = \run{\dwla_k}{(\StructB^j, \zeta(x))}$,
		we can apply the outer induction hypothesis to~$\dwla_k$ and
		obtain that
		\[\run{\dwla_k}{(\StructA^j, x)} = \run{\dwla_k}{(\StructB^j, \zeta(x))} = \run{\dwla_k}{(\StructC^j, x)}.\]
		In particular,~$\dwla_k$ accepts $(\StructC^j, x)$ if and only if it accepts $(\StructA^j, x)$.
		So let~$\relColB$ be the relation obtained from $\refine{\col}{k}$.
		Then $\relColB^{\StructC^{j+1}} = \relColB^{\StructC^{j+1}_1} \cup \relColB^{\StructC^{j+1}_2}$,
		$\relColB^{\StructC^{j+1}_1} = \relColB^{\StructA^{j+1}_1}$, and 
		$\relColB^{\StructC^{j+1}_2} = \relColB^{\StructB^{j+1}_2}$.
		That is, 
		$\StructC^{j+1}_1 = \StructA^{j+1}_1$ and $\StructC^{j+1}_2 = \StructB^{j+1}_2$.
		
		Because the components of~$\StructA^j$ and~$\StructB^j$ are distinguished,~%
		$\zeta$ maps an edge of the $i$-th component of~$\StructA^j$ to the $i$-th component of~$\StructB^j$ (otherwise the sketches differ immediately).
		That is,
		\begin{align*}
			\bigmsetcondition{\run{\dwla_k}{(\StructA^j, x)}}{x \in \tilde{\col}^{\StructA^j}} &= \bigmsetcondition{\run{\dwla_k}{(\StructA^j, x)}}{x \in \tilde{\col}^{\StructA^j} \cap \vertices{\StructA^j_1}} \cup {}\\
			&\hiddenEq \bigmsetcondition{\run{\dwla_k}{(\StructA^j, x)}}{x \in \tilde{\col}^{\StructA^j}\cap \vertices{\StructA^j_2}}
			\intertext{and likewise for $\StructB_j$.
				So we have}
			\bigmsetcondition{\run{\dwla_k}{(\StructC^j, x)}}{x \in \tilde{\col}^{\StructC^j}} &= \bigmsetcondition{\run{\dwla_k}{(\StructA^j, x)}}{x \in \tilde{\col}^{\StructA_j} \cap \vertices{\StructA^j_1}} \cup {}\\
			&\hiddenEq \bigmsetcondition{\run{\dwla_k}{(\StructB^j, x)}}{x \in \tilde{\col}^{\StructB^j}\cap \vertices{\StructB^j_2}}
			\intertext{and thus}
			\bigmsetcondition{\run{\dwla_k}{(\StructC^j, x)}}{x \in \tilde{\col}^{\StructC^j}} &= \bigmsetcondition{\run{\dwla_k}{(\StructA^j, x)}}{x \in \tilde{\col}^{\StructA^j}}\\
			& = \bigmsetcondition{\run{\dwla_k}{(\StructB^j, x)}}{x \in \tilde{\col}^{\StructB^j}}.
		\end{align*}
		\item
		Assume that~$\dwlmout$ executes $\choice{\col}$.
		Because~$\dwla$ is nice on~$\StructA$ and~$\StructB$,~%
		$\col$ is a plain color and
		the components of~$\StructA^j$ and~$\StructB^j$ are distinguished.
		That is,
		\[\sketch{\StructA_1} = \sketch{\StructB_1}=\sketch{\StructC_1} \neq \sketch{\StructA_2} = \sketch{\StructB_2} = \sketch{\StructC_2}.\]
		Hence, the components of~$\StructC$ are distinguished, too.
		Because~$\col$ is plain,~%
		$\col$ occurs solely in one component, say w.l.o.g.~the first.
		Let $x \in \tilde{\col}^{\StructC^j}$ be a chosen element
		and $\StructC^{j+1} =(\StructC^j, x)$.
		Then $(\StructC^j_1, x) = (\StructA^j_1, x)$,
		$\StructC^j_2 = \StructB^j_2$,
		and $(\Struct^j,x) = (\StructA^j_1, x) \disunion \StructB^j_2$
		given that we also chose~$x$ in the execution of~$\dwlmout$ on~$\StructA$
		(which for the sketch of course does not matter if all choices are witnessed).
	\end{enumerate}
	We have proven that~$\dwlmout$ satisfies that
	$\run{\dwlmout}{\StructC} = \run{\dwlmout}{\StructA} = \run{\dwlmout}{\StructB}$.
	To show that $\run{\dwla}{\StructC} = \run{\dwla}{\StructA} = \run{\dwla}{\StructB}$,
	it remains to show that all choices are witnessed and
	that whenever the $j$-th cloud-modifying operation was a $\choiceSym$-operation,
	then
	\[\run{\dwlmwit}{\StructC^m \labeledUnion \StructC^j} = \run{\dwlmwit}{\StructA^m \labeledUnion \StructA^j} =
	\run{\dwlmwit}{\StructB^m \labeledUnion \StructB^j}.\]
	So suppose that the $j$-th operation is a $\choice{\col}$-operation
	for an arbitrary $j \in [m]$.
	We have already seen that
	$\StructC^k_1 = \StructA^k_1$
	and 
	$\StructC^k_2 = \StructB^k_2$
	for every $k \in \set{j, m}$,
	so the same applies to $\StructC^m \labeledUnion \StructC^j$.
	By analogous reasoning as for~$\dwlmout$,
	the witnessing machine~$\dwlmwit$
	satisfies
	$\run{\dwlmwit}{\StructC^m \labeledUnion \StructC^j} = 
	\run{\dwlmwit}{\StructA^m \labeledUnion \StructA^j} =
	\run{\dwlmwit}{\StructB^m \labeledUnion \StructB^j}$.
	In particular,~$\dwlmwit$ writes the same tuple
	$(\buildplan, \rel_{\text{aut}}, \rel_{\text{dom}}, \rel_{\text{img}})$
	onto the interaction-tape in all cases.
	Assume that $(\buildplan, \rel_{\text{aut}}, \rel_{\text{dom}}, \rel_{\text{img}})$
	witnesses~$\col^{\StructA^j}$ and~$\col^{\StructB^j}$ as orbit
	(if that is not the case,~$\dwla$ fails on~$\StructA$ and~$\StructB$
	and there is nothing to show).
	Because~$\dwla$ is nice on~$\StructA$ and~$\StructB$,
	a $\choiceSym$-operation is
	executed only if both components are distinguished.
	So no automorphism maps one component to the other
	and every automorphism~$\autoA$ of~$\StructA^j$ is induced by two automorphisms~$\autoA_1$ and~$\autoA_2$,
	one for each component.
	We write $\autoA = (\autoA_1, \autoA_2)$.
	Similarly, every automorphism~$\autoB$ of~$\StructB^j$ decomposes into $\autoB = (\autoB_1, \autoB_2)$.
	Then the map $(\autoA_1, \autoB_2)$ defined by~$\autoA_1$ on the first and by~$\autoB_2$
	on the second component is an automorphism of~$\StructC^j$.
	Because the defined relations~$\rel_{\text{aut}}$,~$\rel_{\text{dom}}$, and~$\rel_{\text{img}}$ are isomorphism-invariant
	(the witnessing machine is choice-free),
	they encode a set of automorphisms~$N^{\StructA^j}$ in $\applyplan{\buildplan}{\StructA^m \labeledUnion \StructA^j}$
	which decomposes into sets of automorphisms~$N^{\StructA^j}_1 = \setcond{\autoA}{(\autoA,\autoB) \in N^{\StructA^j}_1}$ 
	and likewise in~$N^{\StructA^j}_2$
	such that $N^{\StructA^j} = N^{\StructA^j}_1 \times N^{\StructA^j}_2$.
	Similarly, $N^{\StructB^j} = N^{\StructB^j}_1 \times N^{\StructB^j}_2$ for $\applyplan{\buildplan}{\StructB^m \labeledUnion \StructB^j}$.
	So on $\StructC^m \labeledUnion \StructC^j$,
	$(\buildplan, \rel_{\text{aut}}, \rel_{\text{dom}}, \rel_{\text{img}})$
	encodes the set of automorphisms $N^{\StructC^j} = N^{\StructA^j}_1 \times N^{\StructB^j}_2$.
	Clearly,~$N^{\StructC^j}$ witnesses exactly the same plain relations (which are either contained in~$\StructA^j_1$
	or in~$\StructB^j_2$)
	as orbit,
	which~$N^\StructA_1$ and~$N^\StructB_2$ witness as orbits.
	\end{claimproof}
	Finally, 
	let~$\dwla$ be a normalized polynomial-time DeepWL+WSC-algorithm
	and let~$\dwla'$ be the nice and normalized polynomial-time DeepWL+WSC-algorithm
	given by Claim~\ref{clm:normalized-to-nice} for~$\dwla$.
	Let~$\StructA = \StructA_1 \disunion \StructA_2$ and~$\StructB = \StructB_1 \disunion \StructB_2 $ be two normalized HF-structures
	on which~$\dwla$ does not fail.
	Then~$\dwla'$ accepts (respectively rejects)~$\StructC$
	if~$\dwla$ accepts (respectively rejects)~$\StructC$
	for every $\StructC \in \set{\StructA,\StructB}$.
	Assume that $\run{\dwla'}{\Struct} = \run{\dwla'}{\StructB}$.
	Then $\run{\dwla'}{\Struct} \in \setcond{\run{\dwla'}{\StructA_1 \disunion \StructB_i}}{i \in [2]}$ by Claim~\ref{clm:normalized-and-nice-combine-run}
	because $\dwla'$ is nice.
\end{proof}

\begin{theorem}
	\label{thm:deepwl-wsc-iso-iff-canonical-inv}
	Let $\GraphClass$ be a class of binary $\sig$-structures.
	Then the following are equivalent:
	\begin{enumerate}
		\item \label{itm:iso-deepwl} There is a polynomial-time DeepWL+WSC-algorithm deciding isomorphism on~$\GraphClass$.
		\item \label{itm:compl-inv-deepwl} There is a polynomial-time DeepWL+WSC-algorithm computing some complete invariant for~$\GraphClass$.
	\end{enumerate}
	\end{theorem}
\begin{proof}
	To prove that Condition~\ref{itm:compl-inv-deepwl} implies
	let~$\dwla$ be a DeepWL+WSC-algorithm computing a complete invariant for~$\GraphClass$.
	We want, on input $\StructA\disunion\StructB$ for $\StructA , \StructB \in \GraphClass$,
	to run the algorithm~$\dwla$ on both structures in parallel
	and accept if the invariants are equal.
	Here we are faced with a similar issue as in Theorem~\ref{thm:cpt+wsc-all-equiv}:
	we cannot simulate the computation on one component in the disjoint union
	if the components are not distinguished because we then possibly cannot witness orbits if the components are isomorphic.
	In the DeepWL+WSC setting,
	a complete invariant is a function $f\colon \GraphClass \to \{0,1\}^*$.
	We first execute $\scc{\rel}$ for the relation of plain edges in $\StructA \disunion \StructB$.
	This way, we obtain two vertices $\vertA_\StructA$ and $\vertB_\StructB$
	related to all~$\StructA$-atoms respectively~$\StructB$-atoms
	in a vertex class~$\vc$ (similar to the proof of Lemma~\ref{lem:same-run-normalized}).
	Now, in $(\StructA \disunion \StructB, \vertA_\StructC)$ 
	for $\StructC \in \set{\StructA,\StructB}$,
	the components are distinguished and we can execute~$\dwla$
	to compute $f(\StructC)$ by ignoring the other component.
	Using Lemma~\ref{lem:deepwl-refine-with-function},
	we execute $\refine{\vcA}{f}$.
	If this results into two singleton vertex classes,
	then $f(\StructA) \neq f(\StructB)$ and thus $\StructA \not\iso \StructB$.
	Otherwise, $\StructA \iso \StructB$.
	
	To prove that Condition~\ref{itm:iso-deepwl} implies Condition~\ref{itm:compl-inv-deepwl}, let~$\dwla$ be a DeepWL+WSC-algorithm 
	deciding isomorphism.
	By Corollary~\ref{cor:normalized-iso}, we can assume that~$\dwla$ is normalized
	and, by Lemma~\ref{lem:same-run-normalized},
	we can assume that
	if $\run{\dwla}{\StructA} = \run{\dwla}{\StructB}$,
	then
	\[\run{\dwla}{\StructA} \in \setcond*{\run{\dwla}{\StructA[\verticesOf{1}{\StructA}] \disunion \StructB[\verticesOf{i}{\StructB}]}}{i \in [2]}\]
	for all normalized HF-structures~$\StructA$ and~$\StructB$.
	We show that $f(\StructA) := \run{\dwla}{\StructA \disunion \StructA}$ for every $\StructA \in \GraphClass$ is
	a complete invariant for~$\GraphClass$.
	Clearly, if $\StructA \iso \StructB$, then $f(\StructA) = f(\StructB)$
	(runs are isomorphism-invariant).
	For the other direction, assume that $f(\StructA) = \run{\dwla}{\StructA \disunion \StructA} = \run{\dwla}{\StructB \disunion \StructB} =  f(\StructB)$.
	Because~$\dwla$ decides isomorphism,~%
	$\dwla$ accepts $\StructA \disunion \StructA$ and $\StructB \disunion \StructB$.
	By Lemma~\ref{lem:same-run-normalized},
	\[\run{\dwla}{\StructA \disunion \StructA}  \in  \set{\run{\dwla}{\StructA \disunion \StructB}},\]
	i.e., $\run{\dwla}{\StructA \disunion \StructA}= \run{\dwla}{\StructA \disunion \StructB}$.
	So~$\dwla$ also accepts $\StructA \disunion \StructB$
	and we finally conclude $\StructA\iso \StructB$
	because~$\dwla$ decides isomorphism.
\end{proof}

\subsection{From DeepWL+WSC to CPT+WSC}
To finally prove Theorem~\ref{thm:cpt+wsc-iso-implies-inv},
it remains to show that CPT+WSC simulates polynomial-time DeepWL+WSC-algorithms.
\begin{lemma}
	\label{lem:deepwl-wsc-to-cpt-wsc}
	If a function~$f$ is computable (respectively a property~$P$ is decidable)
	by a polynomial-time DeepWL+WSC-algorithm,
	then~$f$ (respectively~$P$) is CPT+WSC-definable.
\end{lemma}
\begin{proof}
	We follow the same strategy as in Lemma~17 of~\cite{GroheSchweitzerWiebking2021}:
	Polynomial time DeepWL-algorithms can be simulated in CPT
	because CPT can execute the two-dimensional Weisfeiler-Leman algorithm
	to compute the needed coherent configurations and their algebraic sketches
	and because $\addPairSym$- and $\sccSym$-operations
	(the only ones modifying the vertex set of the cloud)
	can easily be simulated by set operations.
	So, we can maintain a set representing the vertex set of the structure in the cloud and further sets for the relations and colors in CPT.
	
	To extend this proof to DeepWL+WSC,
	we again proceed by induction on the nesting depth of the algorithm.
	We encode an HF-structure $\StructA=(\StructV, \StructHFV, \rel_1^\Struct, \dots ,\rel_k^\Struct)$ as an $\HF{\StructVA}$-set (using Kuratowski encoding for tuples).
	We then evaluate CPT+WSC-formulas or terms over the structure with atoms~$\StructV$ and no relations.
	We denote this structure with~$\StructA_0$.
	The CPT+WSC-formulas and terms will have a free variable~$x$,
	to which the $\StructA$-encoding $\HF{\StructVA}$-set is passed.
	
	We show that for every polynomial-time DeepWL+WSC-algorithm~$\dwla$,
	there is a CPT+WSC-term~$\termA(x)$
	such that for every
	HF-structure~$\StructA$,~%
	the term~$\termA(x)$ defines the final configuration of the output machine on input~$\StructA$ (using some appropriate encoding of $0/1$ strings in $\HF{\emptyset}$-sets) if~$\dwla$ does not fail on input~$\StructA$.
	To do so, let $\dwla = (\dwlmout, \dwlmwit, \dwla_1, \dots, \dwla_\ell)$ be a DeepWL+WSC-algorithm
	and let by induction hypothesis $\termA_1(x), \dots, \termA_\ell(x)$ be CPT+WSC-terms
	such that~$\termA_i(x)$ defines the final configuration of the output machine of~$\dwla_i$
	for every $i \in [\ell]$.
	Let, for every $i \in [\ell]$, 
	$\formA_i(x)$ be a CPT+WSC-formula which defines whether 
	the configuration defined by $\termA_i(x)$ is accepting,
	i.e., the head on the work-tape points on a~$1$.
	The Turing machine of~$\dwlmout$ gets simulated using a WSC-fixed-point operator, which uses a variable~$y$ to maintain the structure in the cloud and the configuration of~$\dwlmout$.
	As for DeepWL, $\addPairSym$- and $\sccSym$-operations are
	simulated using set constructions on the structure in the cloud.
	The algebraic sketch is defined
	using the two-dimensional Weisfeiler-Leman algorithm.
	A $\refine{\rel_j}{i}$-operation is simulated as follows:
	Let $\StructA=(\StructV, \StructHFV, \rel_1^\Struct, \dots ,\rel_k^\Struct)$ be the current HF-structure in the cloud as maintained by the WSC-fixed-point operator.
	First assume that~$\rel_j$ is directed.
	We then obtain with the term
	\[\setcond[\big]{z}{z \in \termB_{2+j}(y), \formA_i(\termC(y, \set{z_1}, \set{z_2}))}\]
	the output of the $\refineSym$-operation,
	where
	\begin{itemize}
		\item 	$\termB_{2+j}(y)$ defines the $(2+j)$-th entry in the tuple encoding~$\StructA$, i.e., the set~$\rel_j^\Struct$,
		\item $z_i$ defines the $i$-th entry of the pair~$z$, and
		\item $\termC(y, \set{z_1}, \set{z_2})$ extends the structure in~$y$ by the two new relations $\set{z_1}$ and $\set{z_2}$.
	\end{itemize}
	Here, we individualize~$z_1$ and~$z_2$ by putting them into new singleton relations.
	In the case that~$\rel$ is undirected, we proceed similarly but 
	only create one new relation $\set{z_1,z_2}$.
	So we can simulate the machine~$\dwlmout$ until it makes a $\choiceSym$- operation.
	We now use the variable
	\begin{itemize}
		\item $x$ for the input structure,
		\item $y$ for the pair of the current structure~$\StructA$ in the cloud and the current configuration~$c$ of~$\dwlmout$ (as before), and 
		\item $z$ for the chosen element for the last $\choiceSym$-operation (or~$\emptyset$ if the machine has to be started).
	\end{itemize}
	Let $\termA_{\text{step}}(x,y,z)$ be a CPT+WSC-term
	which simulates~$\dwlmout$ in configuration~$c$ and~$\StructA$ in the cloud
	until~$\dwlmout$ executes the next $\choiceSym$-operation or halts.
	The term $\termA_{\text{step}}(x,y,z)$
	outputs the pair of the obtained configuration and the obtained structure in the cloud.
	If~$\dwlmout$ has to be started, i.e.,~$y$ is assigned to~$\emptyset$,
	the machine is started in the initial configuration on the structure passed to~$x$.
	Second, let $\termA_{\text{choice}}(y)$ be a CPT+WSC-term
	which defines the choice set of the next $\choiceSym$-operation
	to be made
	(or the empty set if the machine halted or was started).
	Third, let $\termA_{\text{wit}}(y,y')$ be the CPT+WSC-term which simulates
	the witnessing machine~$\dwlmwit$
	on the labeled union of the structures passed to~$y$ and~$y'$
	and defines the set of automorphisms outputted by~$\dwlmwit$
	(or more precisely on the structures contained in the pairs passed to~$y$ and~$y'$).
	Last, let $\termA_{\text{out}}(y)$ be a CPT+WSC-term
	extracting the  configuration of the Turing machine passed to~$y$,
	which is encoded as $\HF{\emptyset}$-set.
	We use the WSC-fixed-point operator defining $\HF{\emptyset}$-sets
	from Lemma~\ref{lem:extended-iteration-term}
	to define the final configuration of $\dwlmout$
	when~$\dwla$ is executed on the structure passed to~$x$ via
	\[\formA(x) := \wscFormbig{y}{z}{\termA_{\text{step}}(x,y,z)}{\termA_{\text{choice}}(y)}{\termA_{\text{wit}}(y,z)}{\termA_{\text{out}}(y)},\]
	where we have to replace~$y'$ with~$z$ in $\termA_{\text{wit}}(y,y')$
	to satisfy the formal requirements of the WSC-fixed-point operator.

	Because the elements of the structure in the cloud are obtained by the same HF-sets as by the DeepWL+WSC-algorithm,
	all choices are witnessed successfully in the formula
	if they are witnessed in the algorithm.
	
	For the case of a DeepWL+WSC-computable function~$f$,
	we extract the content of the work-tape from the set defined by~$\formA$.
	For a DeepWL+WSC-computable property, we
	check whether the defined configuration is accepting.
\end{proof}

\begin{corollary}
	A function~$f$ is computable or a property~$P$ is decidable
	by a polynomial-time DeepWL+WSC-algorithm
	if and only if~$f$ or~$P$, respectively, are CPT+WSC-definable.
\end{corollary}
\begin{proof}
	The claim follows from Lemmas~\ref{lem:cpt-wsc-to-deepwl-wsc}
	and~\ref{lem:deepwl-wsc-to-cpt-wsc}.
\end{proof}

We finally prove Theorem~\ref{thm:cpt+wsc-iso-implies-inv}
and show that a CPT+WSC-definable isomorphism test implies
a CPT+WSC-definable complete invariant.
\begin{proof}[Proof of Theorem~\ref{thm:cpt+wsc-iso-implies-inv}]
	Let~$\GraphClass$ be a class of binary $\sig$-structures
	and let~$\formA$ be a CPT+WSC-formula defining isomorphism of~$\GraphClass$.
	Then there is a polynomial-time DeepWL+WSC-algorithm deciding isomorphism of~$\GraphClass$ by Lemma~\ref{lem:cpt-wsc-to-deepwl-wsc}.
	By Theorem~\ref{thm:deepwl-wsc-iso-iff-canonical-inv},
	there is a complete invariant computable by a polynomial-time DeepWL+WSC-algorithm. 
	Finally, from Lemma~\ref{lem:deepwl-wsc-to-cpt-wsc} it follows
	that this complete invariant is CPT+WSC-definable.
\end{proof}

We note that the translation from polynomial-time DeepWL+WSC-algorithms
into CPT-formulas cannot be effective
because the polynomial bounding the running time of a DeepWL+WSC-algorithm
is not given explicitly,
but CPT-formulas have to provide them explicitly.
However, this could be achieved by
equipping DeepWL+WSC-algorithms with explicit polynomial bounds.

\section{The CFI-Query}
\label{sec:cfi}

Fixed-point logic (IFP)  was shown not to capture \PTime{} by Cai, Fürer, and Immerman~\cite{CaiFI1992} using the so-called CFI graphs.
These graphs come with the problem to decide whether a given CFI graph is even or not. This is called the CFI-query.
CFI graphs and their generalizations turned out to be useful to separate various logics from \PTime{}, most recently rank logic~\cite{Lichter21}.
So it is of particular interest whether the CFI-query is CPT-definable. 

Already defining restricted versions of the  CFI-query is rather difficult in CPT:
the best currently known result is that the CFI-query for base graphs
of logarithmic color class size or base graphs with linear maximal degree is CPT-definable~\cite{PakusaSchalthoeferSelman2016}.
The technique for logarithmic color class size is based on constructing deeply-nested sets which are invariant under isomorphisms of the CFI graphs
and encode their parity.
For general base graphs, it is not clear how such sets can be constructed
while still satisfying polynomial bounds.
However, in the presence of symmetric choice, defining the CFI-query becomes easier.

For totally ordered base graphs,
IFP extended  with a suitable fixed-point operator with witnessed symmetric choice defines the CFI-query~\cite{GireHoang98}.
While the choice operator in~\cite{GireHoang98} has slightly different semantics
than the one in CPT+WSC, the same approach can be formulated in CPT+WSC,
which is not surprising since CPT captures IFP.

To construct a class of base graphs
for which the CFI-query it is not known to be definable in CPT,
we consider not necessarily totally ordered base graphs. 
Because every automorphism of the base graph translates to automorphisms of the CFI graph and because we need to choose from orbits,
we assume that the base graphs have CPT distinguishable orbits.
This turns out to be sufficient to define the CFI-query in CPT+WSC.

\paragraph{The CFI Construction.}
We first review the CFI construction and its crucial properties.
Afterwards, we review the approach to define the CFI-query on totally ordered base graphs
in IFP with witnessed symmetric choice from~\cite{GireHoang98}.

A \defining{base graph} is a simple, undirected, and connected graph.
For a base graph $G=(V,E)$
and a function
$g \colon E \to \FF_2$,
we define the \defining{CFI graph} $\CFI{G,g}$ as follows.
For every directed base edge $\bEdge=(\bVertA,\bVertB)$ of every $\set{\bVertA,\bVertB} \in E$,
there is a pair of \defining{edge vertices} $\vertA_{\bEdge,0}$ and $\vertA_{\bEdge,1}$
(note that, since~$G$ is undirected, we have these edge vertex pairs for every direction).
For every base vertex $\bVertA \in V$ of degree~$d$,
we add a degree-$d$ CFI gadget:
the vertex set of the gadget is $C_\bVertA := \setcond{\vertB_{\bVertA,\tup{a}}}{\tup{a}\in\FF_2^{\neighbors{G}{\bVertA}}, \sum \tup{a} = 0}$, where $\sum \tup{a} := \sum_{\bVertB \in \neighbors{G}{\bVertA}} a(\bVertB)$.
For every base vertex $\bVertA \in V$ and every $\bVertB \in \neighbors{G}{\bVertA}$, we add the edges $\set{\vertA_{(\bVertA,\bVertB),i}, \vertB_{\bVertA, \tup{a}}}$
whenever $\tup{a}(\bVertB) = i$ for every $i \in \FF_2$.
Finally, for every edge $\set{\bVertA,\bVertB} \in E$,
we add the edges $\setcond{{\set{\vertA_{(\bVertA,\bVertB),i},\vertA_{(\bVertB,\bVertA),j}}}}{i + j = g(\set{\bVertA,\bVertB})}$.
We say that a \defining{gadget vertex} $\vertB_{\bVertA, \tup{a}}$ has \defining{origin}~$\bVertA$
and write $\orig{\vertB_{\bVertA, \tup{a}}} = \bVertA$
and similar for edge vertices $\orig{\vertA_{(\bVertA, \bVertB)}} = (\bVertA,\bVertB)$.
We extend the notation to tuples
of vertices $\orig{\tup{\vertC}} = (\orig{\vertC_1}, \dots, \orig{\vertC_{|\vertC|}})$.

Its well-known~\cite{CaiFI1992} that up to isomorphism there are exactly two CFI graphs
for a given base graph $G=(E,V)$, namely the one with $\sum g := \sum_{e \in E}g(e) = 0$ called \emph{even}
and the one with $\sum g = 1$ called \emph{odd}.
The \defining{CFI-query} is to decide whether a given CFI graph is even.

We also allow colored base graphs:
If the base graph is colored (that is, a structure $G=(V,E,\spleq)$ where~$\spleq$ is a total preorder on~$V$),
then the vertices are colored according to their origin.
The color class size of~$G$ is the size of its largest $\spleq$-equivalence class.
For a class of base graphs~$\GraphClass$,
we denote by $\CFI{\GraphClass} := \setcond{\CFI{G,g}}{G=(V,E) \in \GraphClass, g \colon E \to \FF_2}$ the class
of CFI graphs over~$\GraphClass$.

We recall some well-known facts on automorphisms of CFI graphs.
Let ${G=(V,E, \leq)}$ be a totally ordered base graph, $g \colon E \to \FF_2$ be arbitrary, $\Struct = \CFI{G,g}$, and $\tup{\vertA} \in \StructV^*$.
All automorphisms of $(\Struct, \tup{\vertA})$ are composed out of cycles in the base graph as follows.
Let $(\bVertA_1, \dots , \bVertA_k)$ with $\bVertA_1 = \bVertA_k$ be a cycle in $G - \orig{\tup{\vertA}}$
(the graph obtained from~$G$ by deleting all the origin base vertex of all gadget vertices in~$\tup{\vertA}$ and all origin base edges of all edge vertices in~$\tup{\vertA}$).
Then there is an automorphism that swaps the edge vertex pairs for exactly the edges $\set{\bVertA_{i}, \bVertA_{i+1}}$ for all $i \in [k-1]$
and maps the gadget vertices $C_{\bVertA_i}$ accordingly.
It holds that the two edge vertices $\vertA_{\bEdge,0}$ and $\vertA_{\bEdge,1}$
are in the same $1$-orbit of $(\Struct, \tup{\vertA})$
if and only if there is a cycle in $G - \orig{\tup{\vertA}}$
containing~$\bEdge$.

The next important property is that if $\bVertA\in V$ is a base vertex of degree~$d$ and
for $d-1$ edge vertex pairs incident to~$C_\bVertA$ an edge vertex is contained in~$\tup{\vertA}$,
then we can define in IFP (and so in CPT) a total order on~$C_\bVertA$ and
the edge vertices of the single remaining incident edge.

Gire and Hoang~\cite{GireHoang98} showed that
IFP extended with a fixed-point operator with witnessed symmetric choice
defines the CFI-query for totally ordered base graphs.
Although the semantics of the choice operator is slightly different from the one presented in the article, the same arguments also work in CPT+WSC:
Maintain a tuple~$\tup{\vertA}$ of individualized edge vertices
and check whether there is still a base edge~$\bEdge$
contained in a cycle in $G - \orig{\tup{\vertA}}$.
If this is the case, use $\set{\vertA_{\bEdge,0},\vertA_{\bEdge,1}}$
as choice set and extend~$\tup{\vertA}$ by the chosen vertex.
Once no such edge exists, we can define a total order on the CFI graph
and e.g.~by the Immerman-Vardi Theorem decide whether it is even.
Finally, to define the necessary automorphisms for the edge vertices of a base edge~$\bEdge$, we pick a minimal (with respect to the total order in the base graph)
cycle in $G - \orig{\tup{\vertA}}$, swap the edge vertices and permute the gadget vertices accordingly.

\paragraph{Unordered Base Graphs.}
In the next step, we consider unordered base graphs.
Let $G = (V,E, \spleq)$ be a colored base graph and $g \colon E \to \FF_2$.
If~$G$ contains a vertex of degree at least~$3$,
every automorphism of $\CFI{G,g}$ is composed of an automorphism of the base graph~$G$ and a ``CFI-automorphism'' from a totally ordered version of~$G$~\cite{Pago2021CFI}.
In particular,
if an orbit of $\CFI{G,g}$ contains vertices with different origins,
then these origins form an orbit of~$G$.
Hence, defining orbits of $\CFI{G,g}$ is at least as hard as defining orbits of~$G$.

Thus, when we want to define the CFI-query for a class of base graphs~$\GraphClass$
using the choice operator,
it is reasonable to assume that~$\GraphClass$ has definable orbits.
We now show that assuming that~$\GraphClass$ has CPT-definable orbits suffices to define the CFI-query in CPT+WSC.

\begin{lemma}
	\label{lem:cfi-cpt-ready-for-individualization}
	Let~$\GraphClass$ be a class of (colored) base graphs with CPT-distinguishable $2$-orbits.
	Then $\CFI{\GraphClass}$ is ready for individualization in CPT.
\end{lemma}
\begin{proof}
	Let $\formA_{\text{orb}}$ be a CPT-formula
	distinguishing $2$-orbits for~$\GraphClass$.
	Let $G=(V,E,\spleq) \in \GraphClass$, $\Struct = \CFI{G,g}$ for some $g \colon E \to \FF_2$, and $\tup{\vertA} \in \StructV^*$.
	On input~$\Struct$ and~$\tup{\vertA}$,
	we first define the base graph $(G, \orig{\tup{\vertA}})$
	and define its $2$-orbits
	(note here that $\orig{\tup{\vertA}}$ is a tuple containing vertices or pairs of vertices of~$G$, so it can be encoded by a tuple of vertices of~$G$).
	This is done as follows:
	A base vertex~$\bVertA$ of the base graph $(G, \orig{\tup{\vertA}})$
	is represented by the set~$C_\bVertA$ of all vertices with origin~$\bVertA$.
	There is an edge between two base vertices
	represented by~$C_\bVertA$ and~$C_\bVertB$
	whenever there is an edge $\set{\vertA,\vertB}$
	for some $\vertA \in C_\bVertA$ and $\vertB \in C_\bVertB$.
	We then use the CPT-formula $\formA_{\text{orb}}$
	to define a total preorder on the base graph.
	Note that, because $\formA_{\text{orb}}$ is a CPT-formula,
	we do not have to consider choices here.
		
	We now check whether there is a $2$-orbit
	containing a directed edge $(\bVertA,\bVertB)$
	as part of a cycle of $G - \orig{\tup{\vertA}}$.
	If that is the case, we can canonically choose the minimal such $2$-orbit~$O$.
	Then the set of edge vertices $\setcond{\vertA_{\bEdge,i}}{\bEdge \in  O, i \in \ZZ_2}$ is a $1$-orbit of $(\Struct, \tup{\vertA})$.
	In particular, it is nontrivial and so disjoint with $\orig{\tup{\vertA}}$.
	
	In the other case that there is no such cycle,
	we can use the properties of CFI graphs
	to order the vertices of each gadget and each edge vertex pair.
	This does not necessarily define a total order on all vertices,
	because the base graph possibly has non-trivial orbits (but no cycles).
	So we check whether there is a nontrivial $2$-orbit in the base graph.
	If that is the case, we again choose the minimal such orbit~$O$.
	Now the set $\setcond{\vertA_{\bEdge,i}}{\bEdge \in O, \vertA_{\bEdge,i} \text { is minimal among } \vertA_{\bEdge,0} \text { and }\vertA_{\bEdge,1}}$
	is a nontrivial $1$-orbit of $(\Struct, \tup{\vertA})$
	and so again disjoint with $\orig{\tup{\vertA}}$.
	
	Finally, if there is neither a cycle in $G - \orig{\tup{\vertA}}$
	nor a nontrivial orbit in $(G, \orig{\tup{\vertA}})$,
	then the ordering of the vertices per gadget and edge vertex pairs
	extends to a total order of~$\Struct$
	and we output the minimal vertex not contained in~$\tup{\vertA}$. 
\end{proof}

\begin{corollary}
	\label{cor:cpt-to-cpt-wsc-distinguishable-orbits-cfi-query}
	For every class of possibly colored base graphs~$\GraphClass$ with CPT-distinguishable orbits, CPT+WSC defines the CFI-query.
\end{corollary}
\begin{proof}
	Because $\CFI{\GraphClass}$ is ready for individualization
	in CPT (and so in particular in CPT+WSC), CPT+WSC defines a canonization for $\CFI{\GraphClass}$ by Theorem~\ref{thm:cpt+wsc-all-equiv}.
	Then, by the Immerman-Vardi Theorem, IFP (and so CPT and so CPT+WSC)
	captures \PTime{} on $\CFI{\GraphClass}$
	and so in particular defines the CFI-query.
\end{proof}
Note here that we do not have to construct automorphisms of uncolored CFI graphs explicitly because we use Theorem~\ref{thm:cpt+wsc-all-equiv}.

We finally show that Corollary~\ref{cor:cpt-to-cpt-wsc-distinguishable-orbits-cfi-query}
covers graph classes
for which it is not known that the CFI-query is CPT-definable.
The best known results due to~\cite{PakusaSchalthoeferSelman2016}
are classes of base graph with logarithmic color class size or
with linear maximal degree.

\begin{corollary}
	There is a class of colored base graphs $\GraphClass = \setcond{G_n}{n \in \nat}$,
	such that CPT+WSC defines the CFI-query for $\GraphClass$
	and for every $n \in \nat$
	the graph~$G_n$ is $O(|G_n|^{0.5})$-regular
	and every color class of~$G_n$ has size $\Omega(|G_n|^{0.5})$.
\end{corollary}
\begin{proof}
	It suffices to construct a class $\GraphClass = \setcond{G_n}{n \in \nat}$ of colored base graphs~$\GraphClass$
	with the desired regularity and color-class size properties
	and which has CPT-definable $2$-orbits
	by Corollary~\ref{cor:cpt-to-cpt-wsc-distinguishable-orbits-cfi-query}.
	
	Let $n \in \nat$ be arbitrary but fixed.
	We define $G_n = (V,E, \spleq)$ as follows:
	start with~$n$ many
	disjoint cliques $K_1, \dots , K_n$ of size $n$.
	Then connect every~$K_i$ to~$K_{i+1}$ (and~$K_n$ to~$K_1$)
	with a complete bipartite graph.
	Finally, color the graph so that each~$K_i$ is a color class.
	We now consider the $2$-orbits of $(G_n, \tup{\vertA})$ for some arbitrary $\tup{\vertA} \in V^*$.
	For every vertex $\vertB \in V$, it is easy to see that
	the $1$-orbit $O(\vertB)$ containing~$\vertB$
	is $O(\vertB) = \set{\vertB}$ if $\vertB \in \tup{\vertA}$
	and otherwise $O(\vertB) = \setcond{\vertC}{\vertB,\vertC \in K_i, \vertC \notin \tup{\vertA}}$.
	Then the $2$-orbit of a tuple $(\vertA,\vertB)$
	is $O(\vertA,\vertB) = \setcond{(\vertC,\vertC')}{\vertC \in O(\vertA), \vertC' \in O(\vertB), \text{ and } \vertC \neq \vertC' \text{ if and only if }\vertA \neq \vertB}$.
	
	Clearly, the graph~$G_n$ has order~$n^2$, is $(3n-1)$-regular, and has color class size~$n$,
	satisfying the assertion of the lemma.
	It is also easy to see that CPT defines the partition of $2$-orbits
	shown above and can order it using the total order on the cliques.
\end{proof}

\section{Discussion}
\label{sec:conclusion}
We extended CPT with a witnessed symmetric choice operator
and obtained the logic CPT+WSC. We
proved that defining isomorphism in CPT+WSC is equivalent to defining canonization.
A crucial point was to extend the DeepWL computation model
to show that a CPT+WSC-definable isomorphism test yields a CPT+WSC-definable complete-invariant.

Thereby, CPT+WSC can be viewed as a simplification step in the quest for a logic capturing \PTime{} as now only isomorphism needs to be defined to be able to apply the Immerman-Vardi Theorem.

To turn a complete invariant into a canonization within CPT+WSC,
we used the canonization algorithm of Gurevich.
To implement it in CPT+WSC, we have to extend it to provide witnessing automorphisms.
For this to work, we needed to give to the witnessing terms
the defined fixed-points as input. This is different in other extensions of first order logic with symmetric choice~\cite{DawarRicherby03, GireHoang98}.
We actually have to require that choice sets are orbits respecting all previous intermediate steps in the fixed-point computation. It appears that this is only relevant if a formula actively forgets previous choices. But how could forgetting these be beneficial? In any case we are not sure whether the modification changes the expressiveness of the logic.

Another question is the relation of CPT+WSC to other logics.
Is CPT+WSC more expressive than CPT?
Do nested WSC-fixed-point operators increase the expressiveness of CPT+WSC?
In~\cite{DawarRicherby03}, it is proven that for  fixed-point logic extended with (unwitnessed) symmetric choice, the nesting fixed-point operators with choice increases expressiveness.

We should remark that any positive answer to our questions
separates CPT from \PTime{}
and hence all questions might be difficult to answer.

Finally, extending DeepWL with witnessed symmetric choice
turned out to be extremely tedious.
While proofs for DeepWL without choice are already complicated~\cite{GroheSchweitzerWiebking2021},
for our extensions the proofs got even more involved.
We would like to see more elegant techniques to prove Theorem~\ref{thm:cpt+wsc-iso-implies-inv}
for CPT+WSC (or even for CPT).

\bibliographystyle{plain}
\bibliography{cpt_symchoice}

\begin{thebibliography}{10}

\bibitem{bgs1999}
Andreas Blass, Yuri Gurevich, and Saharon Shelah.
\newblock {C}hoiceless {P}olynomial {T}ime.
\newblock {\em Ann. Pure Appl. Logic}, 100(1-3):141--187, 1999.

\bibitem{BlassGS02}
Andreas Blass, Yuri Gurevich, and Saharon Shelah.
\newblock On polynomial time computation over unordered structures.
\newblock {\em J. Symb. Log.}, 67(3):1093--1125, 2002.

\bibitem{CaiFI1992}
Jin{-}yi Cai, Martin F{\"{u}}rer, and Neil Immerman.
\newblock An optimal lower bound on the number of variables for graph
  identification.
\newblock {\em Combinatorica}, 12(4):389--410, 1992.

\bibitem{ChandraHarel82}
Ashok~K. Chandra and David Harel.
\newblock Structure and complexity of relational queries.
\newblock {\em J. Comput. Syst. Sci.}, 25(1):99--128, 1982.

\bibitem{ChenPonomarenko2019}
Gang Chen and Ilia Ponomarenko.
\newblock {\em Lectures on coherent configurations}.
\newblock Central China Normal University Press, Wuhan, 2019.
\newblock a draft is available at
  {\url{www.pdmi.ras.ru/\string~inp/ccNOTES.pdf}}.

\bibitem{DawarGraedelLichter22}
Anuj Dawar, Erich Grädel, and Moritz Lichter.
\newblock Limitations of the invertible-map equivalences.
\newblock {\em J. Log. Comput.}, 2022.

\bibitem{DawarRicherby03}
Anuj Dawar and David Richerby.
\newblock A fixed-point logic with symmetric choice.
\newblock In Matthias Baaz and Johann~A. Makowsky, editors, {\em Computer
  Science Logic, 17th International Workshop, {CSL} 2003, 12th Annual
  Conference of the EACSL, and 8th Kurt G{\"{o}}del Colloquium, {KGC} 2003,
  Vienna, Austria, August 25-30, 2003, Proceedings}, volume 2803 of {\em
  Lecture Notes in Computer Science}, pages 169--182. Springer, 2003.

\bibitem{DawarRicherby03Nondet}
Anuj Dawar and David Richerby.
\newblock Fixed-point logics with nondeterministic choice.
\newblock {\em J. Log. Comput.}, 13(4):503--530, 2003.

\bibitem{DawarRicherbyRossman2008}
Anuj Dawar, David Richerby, and Benjamin Rossman.
\newblock {C}hoiceless {p}olynomial {t}ime, counting and the
  {Cai-F{\"{u}}rer-Immerman} graphs.
\newblock {\em Ann. Pure Appl. Logic}, 152(1-3):31--50, 2008.

\bibitem{ebbinghaus1985}
Heinz{-}Dieter Ebbinghaus.
\newblock Extended logics: the general framework.
\newblock In Jon Barwise and Solomon Feferman, editors, {\em Model-Theoretic
  Logics}, Perspectives in MathematicalLogic, pages 25--76. Association for
  Symbolic Logic, 1985.

\bibitem{MR0371622}
Ronald Fagin.
\newblock Generalized first-order spectra and polynomial-time recognizable
  sets.
\newblock In {\em Complexity of Computation ({P}roc. {SIAM}-{AMS} {S}ympos.
  {A}ppl. {M}ath., {N}ew {Y}ork, 1973)}, pages 43--73. SIAM--AMS Proc., Vol.
  VII, 1974.

\bibitem{GireHoang98}
Fran{\c{c}}oise Gire and H.~Khanh Hoang.
\newblock An extension of fixpoint logic with a symmetry-based choice
  construct.
\newblock {\em Inf. Comput.}, 144(1):40--65, 1998.

\bibitem{GradelGrohe2015}
Erich Gr{\"{a}}del and Martin Grohe.
\newblock Is polynomial time choiceless?
\newblock In Lev~D. Beklemishev, Andreas Blass, Nachum Dershowitz, Bernd
  Finkbeiner, and Wolfram Schulte, editors, {\em Fields of Logic and
  Computation {II} - Essays Dedicated to Yuri Gurevich on the Occasion of His
  75th Birthday}, volume 9300 of {\em Lecture Notes in Computer Science}, pages
  193--209. Springer, 2015.

\bibitem{GradelPSK15}
Erich Gr{\"{a}}del, Wied Pakusa, Svenja Schalth{\"{o}}fer, and Lukasz Kaiser.
\newblock Characterising {C}hoiceless {P}olynomial {T}ime with first-order
  interpretations.
\newblock In {\em 30th Annual {ACM/IEEE} Symposium on Logic in Computer
  Science, {LICS} 2015, Kyoto, Japan, July 6-10, 2015}, pages 677--688. {IEEE}
  Computer Society, 2015.

\bibitem{Grohe2008}
Martin Grohe.
\newblock The quest for a logic capturing {PTIME}.
\newblock In {\em Proceedings of the Twenty-Third Annual {IEEE} Symposium on
  Logic in Computer Science, {LICS} 2008, 24-27 June 2008, Pittsburgh, PA,
  {USA}}, pages 267--271. {IEEE} Computer Society, 2008.

\bibitem{Grohe2017}
Martin Grohe.
\newblock {\em Descriptive Complexity, Canonization, and Definable Graph
  Structure Theory}.
\newblock Cambridge University Press, 2017.

\bibitem{GroheN19}
Martin Grohe and Daniel Neuen.
\newblock Canonisation and definability for graphs of bounded rank width.
\newblock In {\em 34th Annual {ACM/IEEE} Symposium on Logic in Computer
  Science, {LICS} 2019, Vancouver, BC, Canada, June 24-27, 2019}, pages 1--13.
  {IEEE} Computer Society, 2019.

\bibitem{GroheSchweitzerWiebking2021}
Martin Grohe, Pascal Schweitzer, and Daniel Wiebking.
\newblock Deep {W}eisfeiler {L}eman.
\newblock In D{\'{a}}niel Marx, editor, {\em Proceedings of the 2021 {ACM-SIAM}
  Symposium on Discrete Algorithms, {SODA} 2021, Virtual Conference, January 10
  - 13, 2021}, pages 2600--2614. {SIAM}, 2021.

\bibitem{Gurevich1988}
Yuri Gurevich.
\newblock Logic and the challenge of computer science.
\newblock In Egon Boerger, editor, {\em Current Trends in Theoretical Computer
  Science}, pages 1--57. Computer Science Press, 1988.

\bibitem{Gurevich97}
Yuri Gurevich.
\newblock From invariants to canonization.
\newblock {\em Bull. {EATCS}}, 63, 1997.

\bibitem{Immerman87}
Neil Immerman.
\newblock Languages that capture complexity classes.
\newblock {\em {SIAM} J. Comput.}, 16(4):760--778, 1987.

\bibitem{DBLP:conf/mfcs/KieferSS15}
Sandra Kiefer, Pascal Schweitzer, and Erkal Selman.
\newblock Graphs identified by logics with counting.
\newblock In Giuseppe~F. Italiano, Giovanni Pighizzini, and Donald Sannella,
  editors, {\em Mathematical Foundations of Computer Science 2015 - 40th
  International Symposium, {MFCS} 2015, Milan, Italy, August 24-28, 2015,
  Proceedings, Part {I}}, volume 9234 of {\em Lecture Notes in Computer
  Science}, pages 319--330. Springer, 2015.

\bibitem{Lichter21}
Moritz Lichter.
\newblock Separating rank logic from polynomial time.
\newblock In {\em 36th Annual {ACM/IEEE} Symposium on Logic in Computer
  Science, {LICS} 2021, Rome, Italy, June 29 - July 2, 2021}, pages 1--13.
  {IEEE} Computer Society, 2021.

\bibitem{LichterS21}
Moritz Lichter and Pascal Schweitzer.
\newblock Canonization for bounded and dihedral color classes in {C}hoiceless
  {P}olynomial {T}ime.
\newblock In Christel Baier and Jean Goubault{-}Larrecq, editors, {\em 29th
  {EACSL} Annual Conference on Computer Science Logic, {CSL} 2021, January
  25-28, 2021, Ljubljana, Slovenia (Virtual Conference)}, volume 183 of {\em
  LIPIcs}, pages 31:1--31:18. Schloss Dagstuhl - Leibniz-Zentrum f{\"{u}}r
  Informatik, 2021.

\bibitem{LichterSchweitzer22}
Moritz Lichter and Pascal Schweitzer.
\newblock {C}hoiceless {P}olynomial {T}ime with witnessed symmetric choice.
\newblock In Christel Baier and Dana Fisman, editors, {\em {LICS} '22: 37th
  Annual {ACM/IEEE} Symposium on Logic in Computer Science, Haifa, Israel,
  August 2 - 5, 2022}, pages 30:1--30:13. {ACM}, 2022.

\bibitem{DBLP:journals/ipl/Mathon79}
Rudolf Mathon.
\newblock A note on the graph isomorphism counting problem.
\newblock {\em Inf. Process. Lett.}, 8(3):131--132, 1979.

\bibitem{Pago21}
Benedikt Pago.
\newblock Choiceless computation and symmetry: Limitations of definability.
\newblock In Christel Baier and Jean Goubault{-}Larrecq, editors, {\em 29th
  {EACSL} Annual Conference on Computer Science Logic, {CSL} 2021, January
  25-28, 2021, Ljubljana, Slovenia (Virtual Conference)}, volume 183 of {\em
  LIPIcs}, pages 33:1--33:21. Schloss Dagstuhl - Leibniz-Zentrum f{\"{u}}r
  Informatik, 2021.

\bibitem{Pago2021CFI}
Benedikt Pago.
\newblock {C}hoiceless {P}olynomial {T}ime, symmetric circuits and
  {C}ai-{F}{\"{u}}rer-{I}mmerman graphs.
\newblock {\em CoRR}, abs/2107.03778, 2021.

\bibitem{pakusa2015}
Wied Pakusa.
\newblock {\em Linear Equation Systems and the Search for a Logical
  Characterisation of Polynomial Time}.
\newblock PhD thesis, RWTH Aachen University, 2015.

\bibitem{PakusaSchalthoeferSelman2016}
Wied Pakusa, Svenja Schalth{\"{o}}fer, and Erkal Selman.
\newblock Definability of {Cai-F{\"{u}}rer-Immerman} problems in {C}hoiceless
  {P}olynomial {T}ime.
\newblock In Jean{-}Marc Talbot and Laurent Regnier, editors, {\em 25th {EACSL}
  Annual Conference on Computer Science Logic, {CSL} 2016, August 29 -
  September 1, 2016, Marseille, France}, volume~62 of {\em LIPIcs}, pages
  19:1--19:17. Schloss Dagstuhl - Leibniz-Zentrum fuer Informatik, 2016.

\bibitem{rossman2010}
Benjamin Rossman.
\newblock Choiceless computation and symmetry.
\newblock In Andreas Blass, Nachum Dershowitz, and Wolfgang Reisig, editors,
  {\em Fields of Logic and Computation, Essays Dedicated to Yuri Gurevich on
  the Occasion of His 70th Birthday}, volume 6300 of {\em Lecture Notes in
  Computer Science}, pages 565--580. Springer, 2010.

\bibitem{DBLP:conf/stoc/SchweitzerW19}
Pascal Schweitzer and Daniel Wiebking.
\newblock A unifying method for the design of algorithms canonizing
  combinatorial objects.
\newblock In Moses Charikar and Edith Cohen, editors, {\em Proceedings of the
  51st Annual {ACM} {SIGACT} Symposium on Theory of Computing, {STOC} 2019,
  Phoenix, AZ, USA, June 23-26, 2019}, pages 1247--1258. {ACM}, 2019.

\bibitem{AbuZaidGraedelGrohePakusa2014}
Faried~Abu Zaid, Erich Gr{\"{a}}del, Martin Grohe, and Wied Pakusa.
\newblock {C}hoiceless {P}olynomial {T}ime on structures with small abelian
  colour classes.
\newblock In Erzs{\'{e}}bet Csuhaj{-}Varj{\'{u}}, Martin Dietzfelbinger, and
  Zolt{\'{a}}n {\'{E}}sik, editors, {\em Mathematical Foundations of Computer
  Science 2014 - 39th International Symposium, {MFCS} 2014, Budapest, Hungary,
  August 25-29, 2014. Proceedings, Part {I}}, volume 8634 of {\em Lecture Notes
  in Computer Science}, pages 50--62. Springer, 2014.

\end{thebibliography}

\end{document}